%% file: thesis.tex
\documentclass{eethesis}
\usepackage{graphicx}
\usepackage[stable,symbol]{footmisc}
\include{texcmds}

\renewcommand{\type}{Dissertation}

\renewcommand{\degree}{DOCTOR OF PHILOSOPHY}	

\renewcommand{\major}{Computer Science}	
\makeindex

\begin{document}
\pagenumbering{roman}
\include{title}

\include{approval}

\include{abstract}

\include{ded}
\include{ack}
\include{lists}
\pagenumbering{arabic}
\setlength{\headheight}{12pt}
\pagestyle{myheadings}
\include{chIntro}
\include{chBgnd}

\include{chStabq1}

\include{chStabq2}

\include{chSubsys1}

\include{chSubsys2}

\include{chSubsysEnc}

\include{chAqeccIsit}

\include{chBch}

\include{bib}		
\include{vita}

\end{document}

%% file: texcmds.tex
\usepackage{rotating}
\usepackage{algorithm}
\usepackage{algorithmic}
\input{Qcircuit}
\newcommand{\nix}[1]{}
\newcommand{\mc}[1]{{\mathcal{#1}}}

\usepackage{amsmath,amssymb,amsthm,comment}
\usepackage[alwaysadjust]{paralist} 
\usepackage{times}

\usepackage{epsfig}
\usepackage{cite}
\usepackage{lscape}

\newcommand{\ontop}[2]{\genfrac{}{}{0pt}{}{#1}{#2}}

\DeclareMathOperator{\tr}{tr}
\DeclareMathOperator{\Tr}{Tr}
\DeclareMathOperator{\Fix}{Fix}
\DeclareMathOperator{\Stab}{Stab}
\DeclareMathOperator{\swt}{swt}
\DeclareMathOperator{\wt}{wt}
\DeclareMathOperator{\w}{w}

\DeclareMathOperator{\puncture}{P}
\DeclareMathOperator{\pc}{P}
\DeclareMathOperator{\B}{BCH}

\DeclareMathOperator{\image}{image}

\newcommand{\C}{\mathbb{C}}
\newcommand{\F}{\mathbb{F}}
\newcommand{\Z}{\mathbb{Z}}

\newcommand{\RM}{{\mathcal{R}}}

\newcommand{\idtwo}{\mathbf{1}_2}
\newcommand{\onemat}{\mathbf{1}}

\newcommand{\mbf}{\mathbf}
\newcommand{\ds}{\displaystyle}

\newcommand{\sdual}{{\bot_s}}
\newcommand{\adual}{{\bot_a}}
\newcommand{\hdual}{{\bot_h}}

\newcommand{\stacked}[2]{\begin{tabular}{c} #1\\[-1ex] \scriptsize #2\end{tabular}}

\def\srow#1 #2{(#1_1,\dots,#1_n|#2_1,\dots,#2_n)}
\def\<#1|#2>s{\langle #1|#2 \rangle_s}
\def\(#1|#2)a{\langle#1|#2\rangle_a}

\newtheorem{theorem}{Theorem}[chapter]
\newtheorem{lemma}[theorem]{Lemma}
\newtheorem{proposition}[theorem]{Proposition}%
\newtheorem{corollary}[theorem]{Corollary}

\newtheorem{fact}[theorem]{Fact}
\newtheorem{conjecture}[theorem]{Conjecture}

\newtheorem{example}[theorem]{Example}
\newtheorem{remark}[theorem]{Remark}

\DeclareMathOperator{\Irr}{Irr}
\DeclareMathOperator{\supp}{supp}

\newcommand{\ol}{\overline}

\newcommand{\one}{\mathbb{I}}

\newcommand{\scal}[2]{\langle #1\mid #2\rangle_s}
\newcommand{\ip}[2]{\langle #1\mid #2\rangle}

\newcommand{\E}{\mathcal{E}}
\newcommand{\Hi}{\mathcal{H}}
\newcommand{\X}{{X}}
\newcommand{\Y}{{Y}}

\newcommand{\R}{\mathcal{R}}
\newcommand{\bch}{\mathcal{BCH}}
\newcommand{\eg}{{\textup{EG}}}

\newcommand{\Rm}{{ \textup{GRM}}}
\newcommand{\floor}[1]{\left\lfloor {#1}\right\rfloor}
\newcommand{\ceil}[1]{\left\lceil {#1}\right\rceil}

\usepackage{amsmath,paralist,amsthm,comment,amssymb}

\newtheorem{defn}[theorem]{Definition}

\DeclareMathOperator{\Bh}{\mathcal{B(\Hi)}}

\newcommand{\ord}{{ {ord}}}

%% file: title.tex
\maketitlepage
{Quantum Stabilizer Codes and Beyond}   
{Pradeep Kiran Sarvepalli} 
{Doctor of Philosophy}                
{August 2008}
{Computer Science}         

%% file: approval.tex
\approvalone
{Quantum Stabilizer Codes and Beyond}   
{Pradeep Kiran Sarvepalli} 
{Andreas Klappenecker}
{Donald K. Friesen}
{Jennifer L. Welch}
{Scott L. Miller}
{Valerie E. Taylor}
{August 2008}

%% file: abstract.tex
\absone
{Quantum Stabilizer Codes and Beyond}   
{August 2008}
{Pradeep Kiran Sarvepalli} 
{B.Tech., Indian Institute of Technology, Madras;\\
M.S., Texas A\&M University}  
{Dr. Andreas Klappenecker}
{
The importance of quantum error correction in paving the way to build a practical quantum computer is
no longer in doubt. Despite the large body of literature in quantum coding theory, many important
questions, especially those centering on the issue of ``good codes'' are unresolved. In this 
dissertation the dominant underlying theme is that of constructing good quantum codes. It 
approaches this problem from three rather different but not exclusive strategies. 
Broadly, its contribution to the theory of quantum error correction is threefold. 

Firstly, it extends the framework of an important class of quantum codes -- nonbinary stabilizer codes. It clarifies
the connections of stabilizer codes to classical codes over quadratic extension fields,
provides many new constructions of quantum codes, and  develops further the theory of optimal quantum codes and punctured quantum codes. In particular it provides many explicit constructions of stabilizer codes,
most notably it simplifies the criteria by which quantum BCH codes can be constructed from classical codes.

Secondly, it contributes to the theory of operator quantum error correcting codes also called as subsystem codes. These codes are expected to have efficient error recovery schemes than stabilizer codes. Prior to our work 
however, systematic methods to construct these codes were few and it was not clear how to fairly compare
them with other classes of quantum codes. 
This dissertation develops a framework for study and analysis of subsystem codes using character theoretic methods.
In particular, this work established a  close link between subsystem codes and classical codes and 
it became clear that the subsystem codes can be constructed from arbitrary classical codes. 

Thirdly, it seeks to exploit the knowledge of noise to design efficient quantum codes 
and considers more realistic channels than the commonly studied depolarizing channel. 
It gives systematic constructions of asymmetric quantum stabilizer
codes that exploit the asymmetry of errors in certain quantum channels. This approach is based
on a Calderbank-Shor-Steane construction that combines BCH and finite geometry LDPC codes. 

}

%% file: ded.tex
\dedicate{My Parents}

%% file: ack.tex
\acknow{

I owe a debt of gratitude to many who have directly or indirectly contributed
to bring this dissertation into shape. First and foremost, I would like to thank my
advisor, Andreas Klappenecker. 
He created and sustained a positive and encouraging atmosphere within which 
I could carry out my research. Even when I was a fledgling in 
the field, he treated me as an equal, while gently correcting me without
discouraging. I am greatly indebted to him for his guidance.

Many thanks to my co-authors, Avanti Ketkar, Santosh Kumar,
Salah Aly, Martin R\"{o}tteler, and Markus Grassl, for the many 
illuminating discussions and their role in enhancing my understanding. 
Martin, in particular,  was instrumental in my study of asymmetric quantum
codes when I was an intern at NEC Laboratories America, in the summer
of 2007. I would also like to thank Marcus Silva and Angad Kamat,  
with whom I had several interesting discussions. 
I would also like to take this opportunity to thank Profs. Jennifer Welch, 
Donald Friesen and  Scott Miller for serving on my Ph.D. committee. 

On a personal note, I thank 
Ammumma, Tataiah, Chelli, Aso Mama, Raja Mama  and Harsha Mama, for their 
encouragement and prayers. I was free to pursue my research partly because 
my brother Bobby stepped in to shoulder my responsibilities at home.
This work would not have been possible without the unwavering and undiminishing
support of my parents. I am deeply grateful for their faith and prayers.
As I look back, there is a sense in which this work seems to be
incomplete. Yet, to the extent that God has given me the grace to complete
it and for His many kindnesses, I am grateful to God. 
}

%% file: lists.tex
\pagestyle{headings}
\setlength{\headheight}{36pt}
\tableofcontents
\listoftables
\listoffigures
\clearpage

%% file: chIntro.tex
\chapter{Introduction}
\body
\section{Motivation}
In the 1980s and 1990s, it gradually became apparent that the  theory of
information founded by Claude Shannon was a purely classical theory in that it
did not take into account quantum mechanics. This realization crystallized the notion of 
quantum information as distinct from classical information. Despite the success of the
abstract formulation of information by Shannon, it is far 
more physical\footnote{R. Landauer.} than it appears. The representation of information 
{\em i.e.}, the mechanism/device
used to store does affect its behavior. Two level systems such as a switches or
more realistically transistors can be used to store and manipulate classical bits. 
One can also use  systems such as photons or electrons. In case of
photons for instance, information maybe 
stored on the polarization of the photon. The photon can be vertically  or 
horizontally polarized. Other quantum mechanical systems such as spin-$\frac{1}{2}$
systems {\em i.e.}, systems with two spin states can also be used for representing
information. These quantum mechanical representations give us something more than what we bargained for. 
Because they operate in a regime where the quantum mechanical effects can come into 
play\footnote{It might be argued that quantum mechanical effects are present even
when information is stored on a transistor (or any other device). That is true, 
however, when we speak of quantum mechanical effects we are not so much interested 
as to how they affect the functioning
of the device as much as how they affect the logical state of the device. In so far as the
logical state is considered, the transistor behaves classically.}, in addition to 
representing the usual logical states  they permit phenomena (such as linear combination
of the logical states), which have no classical analogues. These phenomena seem to 
confer additional power when it comes to information processing. A  far reaching
ramification due to differences between quantum and classical information is
that computers processing quantum information, if they were built, could provide
exponential speedups over computers that process classical information alone. 
For instance, Shor's algorithm for factoring integers provides an exponential 
speedup over the best known classical algorithms. A little less dramatically, 
Grover's search algorithm provides a quadratic speedup over its classical counterparts. 
Quantum computers therefore pose a challenge to one of the central tenets in theoretical
computer science -- the (modern) Church-Turing thesis which states:
\begin{quote}
Any reasonable model of computation can be simulated on a (probabilistic) Turing machine
with at most a polynomial overhead, (see \cite{vazirani98,bernstein93}).
\end{quote}
It must be emphasized that quantum computers cannot solve problems that are not solvable
on classical computers, for the simple reason that a quantum computer can be simulated
on a classical computer albeit with exponential slowdown. Quantum computers can potentially
change the landscape of tractable problems. But to realize their promise 
we have one important hurdle to cross -- which is the central theme of this 
dissertation -- that of protecting quantum information. 

\section{Quantum Error Correction}
 A quantum computer that can implement something nontrivial and
 useful as Shor's algorithm would require the  control and  manipulation of a large number
 of sensitive quantum mechanical systems. 
 Any practical quantum computer
 would require the ability to protect quantum information against not only
 noise but also the inevitable operational ({\em i.e.}, gate) errors that accompany its processing.
 It was initially supposed that it would be
 impossible to protect quantum information not only because of the scale of computation
 but because of reasons intrinsic to quantum information. Fortunately, such
 skepticism was laid to rest when Peter Shor \cite{shor95} and Andrew Steane \cite{steane96,steane96b} independently
proposed schemes to protect quantum information from noise and operational 
errors. Gottesman \cite{gottesman97} and independently Calderbank {\em et al.}, 
\cite{calderbank98} proposed methods to construct quantum codes from classical 
codes. Commonly referred to as ``stabilizer codes'', these codes are the most
studied class of quantum codes. 
Their work was followed with a substantial body of results related to 
quantum error correction. More importantly, it was shown that if 
the overall error rate was lower than a ``threshold'', it was possible to perform 
an arbitrarily long quantum computation with any desired accuracy with only 
a polylogarithmic overhead in time and space \cite{aharonov97}.

With these fundamental results in place, the focus of quantum coding theory 
shifted to  the design of good codes, systematic methods for construction,
efficient  decoding algorithms, passive error correction schemes, optimizing codes 
for realistic noise processes and the like. 
These questions are in some sense interrelated. This dissertation seeks to address these
questions\footnote{In this dissertation we do not focus so much on fault tolerance.} in varying degree as will be elaborated below. 
It explores various models and methods of quantum error correction.
Broadly, its contribution to the  theory of quantum error correction is threefold. Firstly, 
it extends the framework of nonbinary stabilizer codes. It clarifies
the connections of stabilizer codes to classical codes over quadratic extension fields,
provides many new constructions of quantum codes, and  develops further the theory of optimal quantum codes and punctured quantum codes. Secondly, it contributes to the theory of operator quantum error correcting codes (also called as subsystem codes).
These codes are expected to have efficient error recovery schemes compared to stabilizer codes.
This dissertation develops a framework for study and 
analysis of subsystem codes using character theoretic methods. The framework has made it possible
to study subsystem codes by translating them into classical codes. Thirdly, it seeks to 
exploit the knowledge of noise to design efficient quantum codes and considers more 
realistic channels than the commonly studied depolarizing channel. 
In addition to providing many explicit constructions for quantum codes, it seeks to integrate 
developments such as low density parity check (LDPC) codes into quantum coding theory.

\section{Outline and Contribution}
This dissertation is structured as follows. 
In Chapter~\ref{ch:stabq1}, we consider the theory of nonbinary stabilizer codes
initiated by Rains \cite{rains99} and Ashikhmin and Knill \cite{ashikhmin01}. 
This work was motivated in part by the comparatively little attention that 
codes over nonbinary alphabet had received. Currently it appears that binary quantum systems
are comparatively easier to control and implement than multi-level quantum systems.
However, the growing interest in nonbinary implementations suggests that 
nonbinary codes deserve a closer study,  especially as quantum technologies mature. Further, many of the
quantum mechanical systems naturally allow for a multi-dimensional representation
of quantum information. Instead of simply ignoring them as is often the case,
it might be to our benefit to exploit these additional degrees of freedom. It could
for instance lead to implementation of quantum processors with fewer systems. In fact, 
there are proposals to exploit these additional modes 
not only to implement nonbinary quantum systems \cite{bishop07} but also use them 
to simplify binary implementations \cite{ralph07,lanyon08}. 
It stands to reason that we need a systematic theory
to design good codes for nonbinary implementations. This chapter concerns itself
with generalizing many of the ideas of stabilizer codes to the nonbinary setting.
The nonbinary generalization turns out to be a nontrivial task and in fact there still
remain many open questions with respect to nonbinary quantum codes. We derive a number 
of important results with regard to structure and constructions of nonbinary stabilizer codes.

Armed with  the framework of nonbinary stabilizer codes developed in 
Chapter~\ref{ch:stabq1}, we then turn to a more constructive task
of designing good quantum codes in Chapter~\ref{ch:stabq2}. 
As in the classical case, quite often, imposing the constraint of linearity on the code
structure substantially simplifies our task. We have more control over the parameters
of the codes we design and more importantly, imposing the linearity constraint
simplifies the encoding and decoding complexity. 
Therefore, we focus on the construction of some linear quantum codes
bringing into bearing the machinery of the previous chapter. As in the case of classical
codes, optimal codes generate a lot of interest not only because of their optimality,
but because, not infrequently, they possess additional combinatorial structure that
leads to interesting mathematical problems. We also study the quantum MDS codes 
in this chapter, establishing some structural results related to them.

While error correcting codes address the problem of protecting quantum 
information, there are still certain hurdles to be crossed if we are to 
build a quantum computer. Unlike classical case where we can, with good
reason, assume that the encoding and decoding operations are noiseless
or at least that they are not as noisy as the channel, quantum information
processing does not allow us to do so. The process of encoding and decoding 
can be as noisy as the channel itself. Codes then have to designed to 
allow for fault tolerant computation not merely communication or storage. 
The theory of fault tolerant quantum computation was developed to address
this challenge. In keeping with this goal of fault tolerant quantum computation
some researchers have been investigating passive forms of quantum error 
correction, where information was encoded into subsystems that were immune to
noise. Kribs {\em et al.}, \cite{kribs05,kribs05b} proposed a generalized framework for understanding
both active and passive forms of quantum error correction. 
Such codes are called operator quantum error correcting codes or subsystem 
codes because in this model
information is protected by encoding into subsystems as against the subspaces.
Informally, this amounts to encoding each logical state into an equivalence
class rather a unique state in the codespace. The equivalence class is actually a subspace
and any state in the subspace is a representative of the logical state.
This is accomplished through the use of additional
qubits called gauge qubits. 
This method also generalizes the class of stabilizer codes studied in the 
Chapters~\ref{ch:stabq1},~\ref{ch:stabq2}. In view of its relevance to fault tolerant quantum
computing we devote Chapter~\ref{ch:subsys1} to the study of operator quantum error
correcting codes. Using character theoretic methods we establish a
connection with classical codes that enables us to construct these codes
systematically. In particular, we relax the constraint
of self-orthogonality on the classical codes used to construct stabilizer
codes.

In Chapter~\ref{ch:subsys2} we extend the theory of operator quantum
error correcting codes. The results are of interest in that they 
provide insight into the structure of subsystem codes. Additionally,
they enable us to compare the gains that subsystem codes provide over
stabilizer codes. An important question that had been raised when
the subsystem codes were first discovered was the possibility of improving upon
optimal stabilizer codes in the sense of requiring fewer syndrome measurements
than them. We demonstrate in this particular sense the subsystem
codes, at least the linear ones, cannot outperform the MDS
stabilizer codes. 

The presence of gauge qubits in subsystem codes not only simplifies 
error correction procedures, but it can potentially simplify the 
encoding process. Usually, the complexity of encoding is not
as large as the complexity of decoding and is often neglected. But
in the context of fault tolerant quantum computing, it is useful
to have simpler encoding schemes. Previous work on subsystem codes
contained claims that the encoding could also benefit due to subsystem
coding but the exact
circuits and the trade offs involved in achieving these gains were
either absent or not rigorously justified. In Chapter~\ref{ch:subsysEnc} 
we show how subsystem
codes can be encoded, and how to exploit the presence of the gauge 
qubits to simplify the encoding process. We contend these simplifications
in the encoding circuitry should also lead to additional benefits for 
fault tolerant quantum computation. 

Much of quantum coding theory followed the same path as the classical
coding theory did historically. That is it took on an algebraic 
outlook with great emphasis on the distance of the code. But modern
coding theory has gradually moved away from such a one dimensional
characterization of code performance. In the modern picture instead of
requiring that all errors up to a certain weight be correctable it
has shifted the focus to  achieving the capacity of the
channel while keeping the complexity of encoding and decoding low.
But these insights have not yet been fully absorbed by quantum 
coding theory. The reason is not that it has not been attempted.
Starting with the works of Postol \cite{postol01}, MacKay {\em et al}., \cite{mackay04},
Camara {\em et al}., \cite{camara05} 
and more recently Poulin and Chung \cite{poulin08}, there have been attempts to incorporate
these modern developments into quantum coding theory. The difficulty is addressing the conflicting 
requirements that are posed on the classical codes from which the quantum
codes are constructed. The additional constraints usually imply that 
these are bad codes classically and unlikely to lead to good quantum codes. 
In Chapter~\ref{ch:aqecc}, we contribute to the ongoing discussion on 
quantum LDPC codes by providing new constructions of algebraic quantum 
LDPC codes. 

In Chapter~\ref{ch:aqecc} we also study a problem that has generated a lot of
interest lately viz. the use of realistic noise models in quantum error correction.
Much of earlier work often assumed that the channels are depolarizing 
channels. The depolarizing channel while being particularly simple 
is not necessarily the most accurate noise
model which reflects many of the  current quantum technologies. In Chapter~\ref{ch:aqecc} 
we study the design of codes that are in some measure optimized to channels
that are asymmetric. For these channels we also address the problem
mentioned earlier, how to incorporate the modern developments such
as LDPC codes effectively. We study the theory 
of codes for asymmetric quantum channels and also provide systematic
constructions of classes of quantum codes for them. 
While it remains to be seen if these 
codes are suitable for quantum computation, they seem most suited for
quantum memories. 

In Chapter~\ref{ch:bch} we slightly change tracks to illustrate how 
the study of quantum codes can shed light on classical codes. 
In this chapter we show how studies in quantum codes led to us to
gain additional insight into the properties of BCH codes. Despite the
fact these codes have been known for more than forty years now, 
there remain open problems with regard to their properties. We make
some contribution to our understanding of these codes in the context
of quantum error correction. 
We characterize the dimension and duals of narrow-sense
BCH codes giving simple closed form expressions for their dimensions
and simple criteria to identify dual containing BCH codes. 

The material in Chapters~\ref{ch:stabq1} and \ref{ch:stabq2} is due to a joint 
work \cite{ketkar06} with  Andreas Klappenecker, Avanti Ketkar, and Santosh Kumar. 
Part of this material has appeared earlier in the theses of 
Avanti Ketkar and Santosh Kumar.
Chapters~\ref{ch:subsys1},~\ref{ch:subsys2} and \ref{ch:subsysEnc} 
are in collaboration with Andreas Klappenecker and are based on
\cite{pre06,pre07} and \cite{pre08b}.
The material in Chapter~\ref{ch:aqecc} is the outcome of 
a joint work \cite{pre08} with Martin R\"{o}tteler and
Andreas Klappenecker and was partly performed while at NEC Laboratories
America, Inc. The results in Chapter~\ref{ch:bch} are due to
a joint work with  Andreas Klappenecker and Salah Aly \cite{aly07}. 

To keep the dissertation of a manageable and readable size, I have not included
my investigations of algebraic geometric quantum codes (in collaboration with Andreas Klappenecker) \cite{klappenecker05p1,pre05},
quantum convolutional codes \cite{aly07b,aly07c} (together with 
Andreas Klappenecker, Salah Aly, Martin R\"{o}tteler and Markus Grassl), degenerate
quantum codes \cite{aly06b}, group algebra duadic codes \cite{aly07a},
some additional results on subsystem codes from \cite{aly06} which were due to
joint work with Andreas Klappenecker and Salah Aly.

%% file: chBgnd.tex
\chapter{Background}\label{ch:bg}
To make the dissertation self-contained and also to provide the context for the research
performed, this section provides a brief review of ideas relevant to quantum error correction.
Because the breadth of the contents precludes any possibility of covering it completely
in a short space, we recommend the lecture notes by Preskill \cite{preskill} and 
the textbook by Nielsen and Chuang \cite{nielsen00}
for an accessible introduction to quantum computation. Those familiar with 
quantum computing can skip this chapter and proceed directly to topics of interest. 
While there is a logical progression of ideas, effort has been made so that the
chapters can be read independently to some extent. 

\section{Quantum Computation}
\subsection{Qubits}\index{qubit}
Just as bits are abstractions of classical two level systems, qubits are an abstraction of two level 
quantum systems. We denote the basis states in the so-called Dirac notation where $\ket{0}$ (ket zero) 
and $\ket{1}$ (ket one), are simply column vectors
$\left[\begin{smallmatrix} 1\\0 \end{smallmatrix}\right]$ and
$\left[\begin{smallmatrix} 0\\1 \end{smallmatrix}\right]$ respectively. 
This notation also serves to distinguish them from the classical states. The first essential difference with respect to bits is that the qubits can be in superposition of the basis states
{\em i.e.}, they can be in any linear combination of the basis states subject to a normalization
constraint. For instance, consider a single qubit. This qubit can be in the state
$$a\ket{0} +b\ket{1},\quad\mbox{ where } a,b \in \C \mbox{ and } |a|^2+|b|^2=1.$$ 
So the state space of a qubit is $\C^2$. 

While the qubit can be put in any superposition of the basis states, the observed state
of the qubit is restricted to be either one of the states. We cannot observe the superposition itself. 
Any observation of the qubit ``collapses'' the state of the qubit to either $\ket{0}$ or
$\ket{1}$ with probability $|a|^2$ and $|b|^2$ respectively. This underscores the second 
difference between bits and qubits. Observation of qubits can change their state in
general.

If we have  $n$ qubits, then the state space is actually a tensor product of 
the individual state spaces. 
We refer to the state space of the system as the Hilbert space and denote it by $\Hi$. 
\index{Hilbert space}
We have $\Hi \cong \C^2\otimes \C^2\otimes \cdots\otimes \C^2$ with $\dim \Hi = 2^n$.
An orthonormal basis for $\Hi$ is given by $\ket{x_1}\otimes \ket{x_2}\otimes \cdots \otimes \ket{x_n}$. The basis states are also
sometimes denoted as $\ket{x_1x_2\dots x_n}$ or $\ket{x_1,x_2,\dots, x_n}$,
where the $x_i$ take the values zero or one. We can also label the basis elements
by $x\in \F_2^n$. Then a general state is given by 
\nix{
\begin{eqnarray}
\ket{\psi}  & = & \sum_{i=1}^n\sum_{x_i \in \F_2}\alpha_{x_1x_2\cdots x_n} \ket{x_1}\ket{x_2}\ldots \ket{x_n}; \quad \sum_{i=1}^{n} \sum_{x_k\in \F_2} |\alpha_{x_1x_2\cdots x_n} |^2=1.
\end{eqnarray}
}
\begin{eqnarray}
\ket{\psi}  & = & \sum_{x\in \F_2^n}\alpha_{x} \ket{x}; \quad \sum_{x\in \F_2^n} |\alpha_{x} |^2=1.
\end{eqnarray}
The state of the system is a unit vector of length one in $\Hi$. The probability of
observing the system in state $\ket{x}$ is given by $|\alpha_x|^2$.
The  normalization constraint is due to the fact on measurement some state
will be observed.
To describe a general state then, we require $2^n-1$ complex numbers.
This is in contrast to the classical case where the state space is only $n$ dimensional.
As an example, a two qubit system can be put in the state 
$$a_0\ket{00}+a_1\ket{01}+a_2\ket{10}+a_3\ket{11},$$
where $|a_0|^2+|a_1|^2+|a_2|^2+|a_3|^2=1$. The basis state $\ket{00}$ is actually 
$\ket{0}\otimes \ket{0}=\left[\begin{smallmatrix} 1\\0 \end{smallmatrix}\right]
\otimes \left[\begin{smallmatrix} 1\\0 \end{smallmatrix}\right]$. Other basis states are
given similarly. 

Often we will need to observe only a part of the system. This is a little more involved.
Assume that we have a system of $m+n$ qubits and we want to observe $m$ qubits.
An arbitrary state of the system is of the form
\begin{eqnarray}
\ket{\psi}  & = & \sum_{x\in \F_2^m, y\in \F_2^n}\alpha_{x,y} \ket{x}\ket{y}; \quad \sum_{x\in \F_2^m, y\in \F_2^n} |\alpha_{x,y} |^2=1.
\end{eqnarray}
Let us assume that we want to observe the qubits whose states 
correspond to $\ket{x}$. Then the 
probability of observing these qubits in state $\ket{x}$ is given by 
$$
p_x = \sum_{y\in\F_2^n}|\alpha_{x,y}|^2.
$$
Assuming that we observed $\ket{x}$, the state of the system after observation is given
by 
$$
\frac{1}{\sqrt{p_x}}\sum_{y\in \F_2^n}\alpha_{x,y}\ket{y}\ket{x}.
$$
Observing quantum systems can be described using the more powerful measurement
formalism, see for instance \cite{nielsen00}.

An important consequence of the fact that the qubits can be in superposition is a phenomenon known as entanglement. \index{entanglement}
Consider the following state. We ignore the normalization factors for
convenience.
$$
\ket{\psi} = \ket{01}+\ket{11}.
$$
We could also write this state as the product state {\em i.e.}, 
$$\ket{\psi} =\ket{0} \otimes (\ket{0}+\ket{1}).$$
When the states of the qubits can be written as product states then we can observe each of the
product states without disturbing the rest of the system. However there are states such as 
the following which cannot be written as the product of individual qubit states.
$$
\ket{\psi}= \ket{00}+\ket{11} \quad \ket{\varphi}= \ket{01}+\ket{10}.
$$
Such states are said to be entangled and this phenomenon is called entanglement. 
When qubits are entangled it is not possible to observe the state of one of the 
entangled qubits without disturbing the rest of the system. 
One could view the speedup provided by quantum computers as being due to entanglement.

We associate to every state $\ket{\psi}$ in $\Hi$ a row vector denoted as $\bra{\psi}$ 
which is simply the adjoint of the column vector corresponding to $\ket{\psi}$. Two vectors
$\ket{\psi}$ and $\ket{\varphi}$ are said to be orthogonal if their scalar product
denoted as $\ip{\varphi}{\psi}=0$. This is also called the inner product of two vectors.

\subsection{Quantum Gates}
Just as classical data is manipulated using gates, qubits are also manipulated using quantum 
gates. 
Since  the quantum states are unit vectors in $\C^{2^n}$, we could view the application of
gates on the qubits as matrices on $\C^{2^n}$. The postulates of quantum mechanics require the
matrices to be unitary, {\em i.e.}, they must satisfy $U^{-1}= U^{\dagger}$, where $U^{\dagger}$
is the adjoint of the matrix. 
We denote the action of a gate $U$ on a state $\ket{\psi}$ as $U\ket{\psi}$. We denote
the inner product of $U\ket{\psi}$ and $\ket{\varphi}$ as $\bra{\varphi}U\ket{\psi}$.
Some important operations on a single qubit are the following.
\begin{eqnarray}
X = \left[\begin{array}{cc}0&1\\1&0 \end{array}\right]; \quad 
Y = \left[\begin{array}{cc}0&-i\\i&0 \end{array}\right]; \quad 
Z = \left[\begin{array}{cc}1&0\\0&-1 \end{array}\right].
\end{eqnarray}
These operators are also called Pauli errors. 
We will denote the group generated by the Pauli errors by $\mc{P}$. \index{Pauli group}
Often $Y$ is redefined without
the $i$ for convenience in analysis. 
When we consider $n$ qubits we define the Pauli group
of matrices on them as 
\begin{eqnarray}
\mc{P}_n = \{i^{c} e_1\otimes e_2\otimes \cdots \otimes e_n \mid e_i \in \mc{P}, c\in \Z_4 \},
\mbox{ where }  \Z_4 = \{ 0,1 ,2 ,3\}.
\end{eqnarray}
In a subsequent chapter we will generalize the notion of Pauli group and use it
to define error operators and construct codes over prime power alphabet. 

Other important single qubit gates are the Hadamard gate, $H$, \index{Hadamard gate}
the phase gate $P$ and the $\pi/8$ gate (or $T$ gate) which are defined as \index{phase gate}
\begin{eqnarray}
H = \left[\begin{array}{cc}1&1\\1&-1\end{array}\right]; \quad 
P = \left[\begin{array}{cc}1&0\\0&i \end{array}\right];\quad
T = \left[\begin{array}{cc}1&0\\0&e^{i\pi/4} \end{array}\right].
\end{eqnarray}
Perhaps the most important two qubit gate is the CNOT (controlled-NOT) gate. \index{CNOT gate}
The action of the CNOT gate on the basis states is as follows. 
  \[
  \Qcircuit @C=1em @R=.7em {
\lstick{\ket{x}}& \ctrl{1} & \qw& \rstick{\ket{x}}\\
\lstick{\ket{y}}& \targ & \qw &\rstick{\ket{x\oplus y}}
}
\]
The top qubit is called the
control qubit and the bottom qubit is called the target qubit. 
A CNOT gate with control qubit $i$ and target qubit $j$ is denoted as {CNOT}$^{i,j}$ 
and acts as follows on these two qubits:
\begin{eqnarray}
\left[\begin{array}{cccc}
1&0&0&0\\
0&1&0&0\\
0&0&0&1\\
0&0&1&0
\end{array}\right].
\end{eqnarray}
The CNOT gate along with $H$, $P$ and $T$ gates forms a set of universal gates for quantum computation. 
Any arbitrary quantum gate can be realized efficiently using these set of gates to arbitrary accuracy
by the Solovay-Kitaev theorem. 
A graphic representation of the gates mentioned so far is given below:
\[
\Qcircuit @C=1em @R=.7em { 
& \gate{X}&\qw && \gate{Y}&\qw & &\gate{Z} &\qw& &\gate{H}&\qw & &\gate{P}&\qw &  &\gate{T}&\qw &  &\ctrl{1} &\qw \\
& & & &  & & & & & &&& && &  &&&&\targ &\qw \\
& &&&  & & & & & &&&&&& &\\
& \text{i)}&  & &\text{ii)} & & &\text{iii)} &&&\text{iv)}&&&\text{v)} & & &\text{vi)} & & & \text{vii)} & 
}
\]
An important point about the quantum gates is that they act linearly. Let us illustrate.
The $X$ gate acts as follows:
$$
X\ket{0}\mapsto \ket{1} \mbox{ and } X\ket{1}\mapsto \ket{0},
$$
so when it acts on an arbirary state such as $a\ket{0}+b\ket{1}$ we get 
$a\ket{1}+b\ket{0}$.
Later in Chapter~\ref{ch:subsysEnc} we will have occasion to give encoding and 
decoding circuits for subsystem codes. These ideas will be needed then. 

\subsection{Density Operators}\index{density operator}
The state of qubits can be viewed not only as a unit vector in the Hilbert space  but also as 
operators on $\Hi$. This approach makes it easy to analyze and study quantum channels.
Given two vectors $\ket{\psi}$ and $\ket{\phi}$ we can define what is known as the
outer product
of $\ket{\psi}$  and $\ket{\phi} $ as $\ket{\phi}\bra{\psi}$. For instance if $\ket{\psi}= \ket{0}$
and $\ket{\phi}=\ket{1}$. Then
$\ket{1}\bra{0}=\left[\begin{array}{cc}0&0\\1&0 \end{array}\right]$. 
We call the outer product obtained from $\ket{\psi}$ with itself {\em i.e.},
 $\rho=\ket{\psi}\bra{\psi}$ as the density matrix or the density operator. 
The density matrix is positive definite, {\em i.e.}, $\bra{\psi} \rho \ket{\psi} \geq 0$, 
and $\Tr(\rho)=1$ where $\Tr$ is the sum of the diagonal entries. Since the density operators are
matrices of size $2^n\times 2^n$, we can also view the states as being operators on
the system Hilbert space. A view which will be useful when defining quantum channels.
More generally if a system can be found in one of the
states $\ket{\psi_i}$ with probability $p_i$, the density operator associated 
to this system is given by 
$$
\rho = \sum_i p_i \ket{\psi_i}\bra{\psi_i}.
$$
A state is pure if $\Tr(\rho^2)=1$ and mixed otherwise.
The density operator approach will be helpful in understanding the 
motivation behind operator quantum 
error correction in Chapter~\ref{ch:subsys1} 
and also in Chapter~\ref{ch:aqecc}, where we design
codes optimized for a given channel. When a gate $U$ is applied to a state
with density matrix $\rho$, it transforms as $U\rho U^\dag$.

\subsection{Quantum Noise}
Noise on qubits is very different from the noise that we deal with
bits. The noise can be thought to be arising out of the
fact that the information bearing system cannot be completely isolated
from the environment and its interaction with the environment causes its state
to change. Sometimes this phenomenon is also called decoherence. \index{decoherence}

Since the state of a single qubit is given $a\ket{0}+b\ket{1}$, where $a,b$ are complex
numbers, one can expect that errors on quantum information form a continuum unlike the
classical bits where there exist only bit flip errors. 
In fact, we can view  noise on a qubit as a $2\times 2$ complex matrix and more generally,
noise on $n$ qubits is a $2^n\times 2^n $ complex matrix; for this reason we often refer to errors as error operators.  

While we have to protect quantum information from an  infinitude of
errors, in view of linearity of quantum mechanics, it suffices to correct
for only a basis of errors. The importance of the Pauli errors 
also stems from the fact that they form a basis for the error operators.
Of course, we cannot protect against all errors. 
We usually make the assumption that noise on each qubit is independent.
Under this assumption we can
decompose an error on the system into a tensor product of $n$ single qubit errors. 

Errors on the quantum states can also arise due to the finite precision with which
the quantum gates are implemented. Fortunately, the same mechanisms that are used to
correct decoherence can also be used to correct for these type of errors \cite{shor95,shor96}. 

\subsection{Quantum Channels}\index{quantum channel}
A quantum channel is a linear map on the density operators (on $\C^{2^{m}}$) to the set of
density operators (on $\C^{2^{n}}$); we usually assume that the input and output 
Hilbert spaces are same {\em i.e.}, $m=n$.
Sometimes quantum channels are also called ``superoperators'' \index{superoperator}
to indicate that they act on (density) operators. In this dissertation we
will confine ourselves to maps which are completely positive and trace preserving 
(CPTP) maps. A CPTP map $\mc{E}$ is usually given in terms of its
Kraus decomposition. \index{Kraus operators}
\begin{eqnarray}
\mc{E}(\rho) &=& \sum_{i}	E_i \rho E_i^\dagger \mbox{ where } \sum_i E_i^\dagger  E_i= I.
	\label{eq:kraus}
\end{eqnarray}

The quantum channel view is very convenient to understand errors. For instance if we
assume that the bit flip errors occur with a probability $p$ and the rest of the
time there are no errors. We can represent this as the following channel.
\begin{eqnarray}
\mc{E}(\rho) &= & (1-p)\rho +p X\rho X.
\end{eqnarray}
The Kraus operators are easily identified as $\sqrt{1-p} I $ and $\sqrt{p} X$.
The channel often studied in the context of quantum codes is the depolarizing channel
and it parallels the classical $4$-ary symmetric channel. This channel acts as 
\begin{eqnarray}
\mc{E}(\rho) &= & (1-3p)\rho + p X\rho X + p Y \rho Y + p Z\rho Z.
\end{eqnarray}
In this channel, each of the Pauli errors $X$, $Y$ or $Z$ act with a 
probability $p$ and with a probability of $1-3p$, the state is preserved.
The Kraus operators are simply given by $\sqrt{1-3p}I$,
$\sqrt{p}X$, $\sqrt{p}Y$ and $\sqrt{p}Z$.

\section{Quantum Error  Correction}
In this section we briefly review the elements of quantum error correction. 
The reader is also recommended to \cite{gottesman97,KnLa97,calderbank98} for more details.  Additionally, there are many expositions to the ideas of quantum error correction, see \cite{knill02,ashikhmin07,feng04,kempe06,martin04}. 
Here we summarize the main features. We will restrict our attention to additive 
quantum codes.

A binary quantum code \index{quantum code} is a linear subspace of the system Hilbert space i.e., $\C^{2^n}$.
The subspace structure arises due to the fact that we can have superpositions
of the encoded states. For instance, let us assume that the logical states are 
the following:
\begin{eqnarray}
\ket{\ol{0}} = \ket{000};\quad \ket{\ol{1}} = \ket{111}. \label{eq:3qubitEnc}
\end{eqnarray}
Since we are allowed to have linear combinations of states, this implies that 
$a\ket{000}+b\ket{111}$ is also a valid state and belongs to the code. The subspace
structure of the quantum code can be seen to emerge naturally.  
The typical questions that we have to address when dealing with error correcting
codes classical or otherwise are:
\begin{itemize}
\item Construction
\item Encoding
\item Error correction
\item Performance
\end{itemize}
In the case quantum codes, there is yet another component that plays a 
much more important role than in case of classical codes. The codes should
be suitable for fault tolerant computation {\em i.e.}, we should be
able to perform logical operations on the encoded data without having
to decode them. The encoded operations must also ensure that the errors must not
propagate catastrophically beyond the error correcting capability of
the code. In this thesis we will not get into the issues of fault tolerance. 
We shall address the problem of construction and performance
in more detail in the later chapters
of this dissertation. Let us look at the other two aspects. 

Since quantum codes are subspaces in $\C^{2^n}$, constructing quantum
codes can be viewed as packing of subspaces in $\Hi$. In fact, 
the original approaches to quantum error correction were along this
route. This geometric picture while intuitive is not very convenient; fortunately,
we can translate the problem of construction into one 
with a lot more algebraic flavor and more importantly, into a much more
familiar language involving construction of classical codes. 
We will have much more to say on this topic of how to link the classical
codes and the subspaces in $\Hi$ later 
in Chapters~\ref{ch:stabq1} and \ref{ch:subsys1}. 

Assume for now that we have some means to choose a subspace to be our quantum
code, $Q$. Then, from linear algebra we know that we can project onto a subspace by
 means of a projector. A projector $P$ satisfies $P^2=P$. \index{projector}
 A projector for $Q$ can be easily constructed by choosing an orthonormal 
basis of the subspace $Q$, say $\{\ket{\alpha_1},\ldots,\ket{\alpha_K}\}$, and forming 
the following matrix
$$
P = \sum_{i=1}^K \ket{\alpha_i }\bra{\alpha_i}.
$$
The dimension of the subspace is related to $P$ as $\dim Q= \Tr(P)$.
The subspace induces a decomposition of the Hilbert space into orthogonal subspaces. 
Encoding amounts to realizing $P$, though there are important subtleties 
to be addressed, (such as the nonunitariness of $P$). 
For instance, the encoding in equation~(\ref{eq:3qubitEnc}) can be easily 
accomplished using the following circuit.
\[ 
  \Qcircuit @C=1.3em @R=.6em @!R {
	\lstick{a\ket{0}+b\ket{1}}&\ctrl{1} &\ctrl{2}& \qw \\
	\lstick{\ket{0}}&\targ&\qw& \qw &\rstick{a\ket{000}+b\ket{111}}\\
	\lstick{\ket{0}}&\qw&\targ& \qw \gategroup{1}{1}{3}{4}{0.7em}{\}}\\
}
\]
We shall study encoding circuits in more detail in Chapter~\ref{ch:subsysEnc} 
when we discuss encoding of subsystem codes.

When it comes to quantum error correction, there are a few points worth highlighting. 
Error correction or error recovery 
implies that we correct the errors on the encoded information without finding
out what was the original information stored. 
By decoding we mean the process of extracting the 
information from the encoded qubits. It presumes that error correction 
has already been performed. Classically, we do not have to
make such fine distinction between error correction and decoding because 
once error correction is performed it is not difficult to obtain the information
that was encoded without affecting the encoded state.  In the quantum setting
decoding amounts to destroying the encoded state. In the context of fault tolerant
quantum computation we would not like to decode until the end of the computation 
as it would remove the protection afforded by the code. Unless explicitly mentioned
our focus will be on error recovery or correction. We will assume that the 
decoding of the encoded information is performed at the end of the computation. 
In this dissertation we will be concerned with error correction unless 
specified otherwise.

Let  us look at the error correction process in a little more detail. 
Assume that we use the encoding given in equation~(\ref{eq:3qubitEnc}).
Suppose that there is a bit flip error on the first qubit,  also called an $X$ error.
Then we have
$$
a\ket{000}+b\ket{111} \stackrel{\text{Bit flip}}{\mapsto} a\ket{100}+b\ket{011}.
$$
We cannot take a majority voting to figure out the error as in the classical case because
if we observed the state we would collapse the state to either $\ket{100}$ or $\ket{011}$.
Although we maybe able to find that there was an error on the first qubit,
we have also damaged the state. Thus error correction process is a little more complicated in
the quantum case. We must not perform a full measurement of the system. We solve this problem
by partial measurements and the use of additional qubits called ancilla. Let us illustrate this
for our running example. We can compute the parity of the first two qubits and the
second two qubits as follows. 
  \[ 
  \Qcircuit @C=1.3em @R=.6em {
&\mbox{Encoding}&& \mbox{Noise}& & &&& \mbox{Syndrome Measurement}& &  & &
\mbox{Correction}\\
\lstick{\ket{\psi}}&\ctrl{1} &\ctrl{2}&
\multigate{2}{\ \mathcal{N}\ } & \qw & \qw & \ctrl{3} & \qw & \qw & \qw &
\qw &  \qw &
\multigate{2}{\ \mathcal{R}\ } & \qw\\
\lstick{\ket{0}}&\targ&\qw&  \ghost{\ \mathcal{R}\ }& \qw & \qw & \qw & \ctrl{2} & \ctrl{3} & \qw &
\qw & \qw & \ghost{\ \mathcal{R}\ } &
\qw\\
\lstick{\ket{0}}&\qw&\targ&  \ghost{\ \mathcal{N}\ } & \qw & \qw & \qw & \qw & \qw & \ctrl{2} & \qw &
\qw &  \ghost{\ \mathcal{R}\ } \qw &
\qw\\
&&&& & \lstick{\ket{0}} & \targ \qw & \targ \qw &
\qw & \qw & \qw  & \meter &
\control \cw \cwx \gategroup{2}{1}{4}{3}{.8em}{--}\\
&&&& & \lstick{\ket{0}} & \qw & \qw & \targ \qw &
\targ \qw &  \qw & \meter &
\control \cw \cwx
\gategroup{2}{6}{6}{12}{.8em}{--} &
}
\] 
The state of the qubits changes as follows as we move across the circuit:
\begin{eqnarray*}
\begin{split}
& a\ket{0}+b\ket{1}\ket{00}\ket{00}\stackrel{\text{Encoder}}{\mapsto}
 (a\ket{000}+b\ket{111})\ket{00} \stackrel{\text{Noise}}{\mapsto}
 (a\ket{100}+b\ket{011})\ket{00}\\
&\stackrel{\text{CNOT}^{1,4}}{\mapsto} (a\ket{100}\ket{1}+b\ket{011}\ket{0})\ket{0}
\stackrel{\text{CNOT}^{2,4}}{\mapsto}(a\ket{100}\ket{1}+b\ket{011}\ket{1})\ket{0}\\
&=(a\ket{100}+b\ket{011})\ket{1}\ket{0}
\end{split}
\end{eqnarray*}
It will be seen that the first  ancilla qubit becomes entangled with the encoded state briefly
and then becomes unentangled. At this point we can make a measurement of the
ancilla without disturbing the rest of the encoded state.
The double lines indicate classical bits. We can then perform a 
correction operation based on the measurement of ancilla qubits. 
The value measured is usually called the syndrome.

The important thing to notice is that if we have an error then the codespace is taken to
an orthogonal subspace of $\C^{2^n}$, in the example considered it is the space spanned by
$\ket{100}$ and $\ket{011}$. 
On the other hand consider an error that flips all the qubits. This error takes $\ket{000}$ to 
$\ket{111}$ and vice versa. Its action on $Q$ is to  merely permute the basis vectors. Since
it takes valid codevectors to valid codevectors, it cannot be detected. 
Finally, let us consider an error which has no classical analogue. If we had a $Z$ error 
on the first two qubits, then it would take $\ket{000}$ to $\ket{000}$ and $\ket{111}$ to 
$\ket{111}$. So a nontrivial error can act trivially on the codespace. We consider
such errors to be harmless. 
This gives us a general principle for an error to be detectable. 
We shall make use of this lemma later, especially in Chapters~\ref{ch:stabq1},~\ref{ch:subsys1}. 

\begin{lemma}[\cite{KnLa97}]
Given a quantum code $Q$, with projector $P$, and $\ket{\alpha}$ and $\ket{\beta}$ two 
orthogonal vectors  in $Q$. An error $E$ is detectable if and only if 
$\bra{\alpha} E \ket{\beta} =\lambda_E \bra{\alpha} E \ket{\beta}$, where 
$\lambda_E$ depends only on $E$. Alternatively, 
an error is detectable if and only if $PEP = \lambda_E P$.
\end{lemma}
Given a set of errors $\{E_1, E_2,\ldots, E_l \}$ that are detectable by $Q$, their linear span
is also detectable by $Q$. 
The subspace $Q$  induces a decomposition of $\C^{2^n}$. Detectable errors take the
subspace one of the orthogonal subspaces, while undetectable errors take $Q$ to itself.

\section{Classical Coding Theory}
In this section we discuss some of the relevant aspects of classical codes
setting the stage for our work on quantum error correction. In view of vastness
of the subject, the reader is recommended standard textbooks in the field
such as \cite{macwilliams77,huffman03,lin04} for a comprehensive treatment of the
field. 

Let $\F_q$ denote a finite field with $q$ elements; we have $q=p^m$ for some prime $p$. 
If  $x=(x_1,\ldots,x_n) \in \F_q^n$, then we denote the Hamming weight of $x$
as \index{Hamming weight}
\begin{eqnarray}
\wt(x) = |\{ x_i\neq 0 \}|,
\end{eqnarray}
{\em i.e.}, it is the number of nonzero coordinates
of $x$. We say that a subset  $C \subseteq \F_q^n$ is an
\textsl{additive code} \index{additive code} if  for any
$x,y$ in $C$, $x+y$ is also in $C$. Additive codes 
play an important role in quantum error correction. If in addition to being 
additive, $C$ also satisfies $\alpha c \in C$ for any $\alpha \in \F_q$ and $c\in C$,
then $C$ is said to be an $\F_q$-linear code. \index{linear code}
Such codes  often have simpler encoding and 
decoding schemes while being tractable in terms of
construction and analysis. 
The \textsl{minimum distance} of a set $C\subset \F_q^n$ is defined as \index{distance!classical}
\begin{eqnarray}
\wt(C) = \min_{\stackrel{x,y\in C}{x\neq y}} \{\wt(x-y) \}.
\end{eqnarray}
The (minimum) distance of a code is indicative of the error correcting capabilities
of the code. If $C$ is an additive code, its distance is given by 
\begin{eqnarray}
\wt(C) = \min_{0\neq c\in C} \wt(c).
\end{eqnarray}

A classical $(n,K,d)_q$ code $C\subseteq \F_q^n$ is subset of $\F_q^n$ 
of size $|C|=K$ and distance $d=\wt(C)$. If  $|C|=q^k$, then we 
denote it by $[n,k,d]_q$. If  $C$ is also $\F_q$-linear code, \index{linear code}
then  $C$ is a 
$k$-dimensional subspace of $\F_q^n$. Linear codes are often described
by giving a basis of codewords in the form a matrix, often called as the generator matrix. For example,
consider the $[7,4,3]_2$ Hamming code with the generator matrix 
$$
G = \left[\begin{array}{cccc|ccc}
1&0&0&0&1&1&0\\ 
0&1&0&0&1&0&1\\
0&0&1&0&0&1&1\\
0&0&0&1&1&1&1
\end{array} \right].
$$
It consists of all the linear combinations of the rows of $G$. 
When the generator matrix is in the form $[I|P]$ we say that it is in the
standard form. We define the Euclidean inner product \index{inner product!Euclidean}
between two codewords
$x,y\in \F_q^n$ as 
\begin{eqnarray}
x\cdot y &= & x_1y_1+\cdots + x_ny_n = \sum_{i=1}^n x_iy_i. \label{eq:euclidIP}
\end{eqnarray}
The Euclidean inner product enables us to define a dual code. It is defined as 
\begin{eqnarray}
C^\perp = \{x\in \F_q^n \mid x\cdot c= 0 \mbox{ for all } c \in C \}. \label{eq:euclidD}
\end{eqnarray}
This is also called as the Euclidean dual of $C$. \index{dual code!Euclidean}
The dual code is itself
a linear code with  its own generator matrix $H$. A generator matrix
of  $C^\perp$ is also called a parity check matrix for $C$. For the
example just considered, a parity check matrix is given by 
$$
H = \left[\begin{array}{cccc|ccc}
1&1&0&1&1&0&0\\ 
1&0&1&1&0&1&0\\
0&1&1&1&0&0&1
\end{array} \right].
$$
When the generator matrix for $C$ is given in the standard form $[I_k|P]$, a parity 
check matrix is easily obtained as $[-P^t|I_{n-k}]$. 
One important relation between the generator matrix and the parity check matrix is
that $GH^t=0$. When a code $C\subseteq C^\perp$, we say that $C$ is a self-orthogonal
code. \index{self-orthognal code}
If $C=C^\perp$, then we say it is a self-dual code. 
In the context of quantum error correcting codes, dual codes 
and self-orthogonal codes play a much more 
significant role than in the classical case.
Additionally, we encounter far more general notions 
of inner products.

%% file: chStabq1.tex
\chapter{Theory of Nonbinary Stabilizer 
Codes\footnotemark } \label{ch:stabq1}
\footnotetext{\textcopyright 2006 IEEE.  Reprinted in part, with permission, 
from A. Ketkar, A. Klappenecker, S. Kumar and P. K. Sarvepalli,
``Nonbinary stabilizer codes over finite fields''. {\em IEEE Trans. Inform. Theory}, 
vol. 52, no. 11, pp.~4892--4914, 2006.}
As mentioned earlier, quantum codes were developed to make fault-tolerant
quantum computation possible. The most widely studied class of quantum error-correcting codes are
binary stabilizer codes,
see~\cite{ashikhmin00a,ashikhmin00b,beth98,calderbank96,calderbank97,
cleve97,cleve97b,cohen99,danielsen05,ekert96,freedman01,gottesman02,
gottesman97,gottesman05,gottesman96b,grassl99,grassl99b,grassl00,
grassl01,kim02,kim03,kitaev97,martin04,rains99c,shor95,steane96,steane96b,
steane99,steane99b,thangaraj01,xiaoyan04} and, in particular, the seminal
works~\cite{calderbank98,gottesman96}.  An appealing aspect of binary
stabilizer codes is that there exist links to classical coding theory
that facilitate the construction of good codes.  More recently, some
results were generalized to the case of nonbinary stabilizer
codes~\cite{aharonov97,arvind03,ashikhmin01,bierbrauer00,chau97,chau97b,
feng02,feng02b,gottesman99,grassl03,grassl04,kim04,li04,matsumoto00,
rains99,roetteler04,schlingemann00,schlingemann02}, but the theory is
not nearly as complete as in the binary case.

One would naturally ask why study nonbinary codes? There are at least 
three reasons for our interest in nonbinary codes. The first reason is 
the generalization is a nontrivial mathematical problem that is of 
interest in itself. Results which are considerably easy to prove in
the binary case turn out be much more formidable requiring the use
of elegant mathematical techniques to solve the problems. 
The second reason is a practical one and motivated by the behavior of 
classical codes. Many good classical codes like 
Reed-Solomon codes are nonbinary codes. Algebraic geometric codes that  
were the first shown to beat the Gilbert-Varshamov bound were once again
nonbinary codes. Even in the case LDPC codes it  has been shown that 
increasing the alphabet size improves the performance albeit at the expense
of complexity. As we shall see the close connections between the classical
and quantum codes tempt the conclusion that perhaps one would expect to find
good classes of quantum codes over a larger alphabet. Thirdly, quite often
many implementations naturally allow for a multilevel quantum system. These 
extra modes are usually ignored; but lately they have received interest, 
see \cite{bartlett02,gottesman01,bishop07,moreva06,ariano05} and references therein. 
Additionally, as shown in \cite{ralph07}, if properly
exploited, this can lead to efficient implementation of gates. All these 
reasons motivate our investigations of nonbinary quantum codes.

This chapter has two primary goals. On one hand we provide a review
of the theory of stabilizer codes and on the other we also extend and
generalize many of the results. This chapter is structured
as follows. 
We recall the basic principles of nonbinary stabilizer codes over
finite fields in Section~\ref{sec:stabCodes}. In Section~\ref{sec:galoisCnxn}, we introduce a Galois
theory for quantum error-correcting codes.  The original theory
developed by Evariste Galois relates field extensions to
groups. Oystein Ore  
derived a significantly more general theory for pairs of
lattices~\cite{ore44}. We use this framework and set up a Galois
correspondence between quantum error-correcting codes and groups.
This theory shows how some properties of general quantum codes, such
as bounds on the minimum distance, can be deduced from results about
stabilizer codes.

In Section~\ref{sect:additive}, we recall that stabilizer codes over a finite field
$\F_q$ correspond to additive codes over $\F_q$ that are
self-orthogonal with respect to a trace-symplectic
form~\cite{ashikhmin01}. We also establish the correspondence to
additive codes over $\F_{q^2}$ that are self-orthogonal with respect
to a trace-alternating form; remarkably, this basic construction had
been missing in the literature, in spite of the fact that it is a
generalization of the famous $\F_4$-codes~\cite{calderbank98}.

The MacWilliams relations for weight enumerators of stabilizer codes
are particularly easy to prove, as we show in Section~\ref{sec:wtEnums}. We then
derive upper and lower bounds on the minimum distance of
the best possible stabilizer codes in Section~\ref{sec:bounds}. 
Section~\ref{sec:const} details methods to construct new methods to construct
quantum codes from existing quantum codes. 
Unlike classical codes, puncturing quantum codes is a relatively complex task. 
So we include a generalization of the puncturing theory
introduced by Rains to additive codes that are not necessarily pure.
In a later chapter we show how to apply it. 

Apart from the basics of quantum computing, we
recommend \cite{calderbank98} and \cite{gottesman97} for background on
binary stabilizer codes, in addition to books on classical coding
theory, such as~\cite{huffman03,lin04,macwilliams77}. The general
theory of quantum codes is discussed in \cite{KnLa97}, and we assume
that the reader is familiar with the notion of a detectable error, as
introduced there.

\textit{Notations.}  We assume throughout this chapter that $\F_q$
denotes a finite field of characteristic~$p$; in particular, $q$
always denotes a power of a prime $p$.  The trace function from
$\F_{q^m}$ to $\F_q$ is defined as $\tr_{q^m/q}(x)=\sum_{k=0}^{m-1}
x^{q^k}$; we may omit the subscripts if $\F_q$ is the prime field. If
$G$ is a group, then we denote by $Z(G)$ the center of $G$. If
$S\subseteq G$, then we denote by $C_G(S)$ the centralizer of $S$ in
$G$.  We write $H\le G$ to express the fact that $H$ is a
subgroup of~$G$.  The trace~$\Tr(M)$ of a square matrix~$M$ is the sum of
the diagonal elements of~$M$.

\section{Stabilizer Codes}\label{sec:stabCodes}
Let $\C^q$ be a $q$-dimensional
complex vector space representing the states of a quantum mechanical
system. We denote by $\ket{x}$ the vectors of a distinguished
orthonormal basis of $\C^q$, where the labels $x$ range over the
elements of a finite field $\F_q$ with $q$ elements.  A quantum
error-correcting code $Q$ is a $K$-dimensional subspace of $\C^{q^n}=
\C^q\otimes \cdots \otimes \C^q$.

We need to select an appropriate error model so that we can measure
the performance of a code.  We simplify matters by choosing a basis
$\mathcal{E}_n$ of the vector space of complex $q^n\times q^n$
matrices to represent a discrete set of errors.  A stabilizer code is
defined as the joint eigenspace of a subset of~$\mathcal{E}_n$, so the
error operators play a crucial role.

\subsection{Error Bases}\label{ssec:errBases}
Let $a$ and $b$ be elements of the finite field $\F_q$.  We define the
unitary operators $X(a)$ and $Z(b)$ on~$\C^q$ by
$$ X(a)\ket{x}=\ket{x+a},\qquad
Z(b)\ket{x}=\omega^{\tr(bx)}\ket{x},$$ where $\tr$ denotes the trace
operation from the extension field $\F_q$ to the prime field $\F_p$,
and $\omega=\exp(2\pi i/p)$ is a primitive $p$th root of unity.

We form the set $\mathcal{E}=\{X(a)Z(b)\,|\, a,b\in \F_q\}$ of error
operators. The set $\mathcal{E}$ has some interesting properties,
namely (a) it contains the identity matrix, (b)~the product of two
matrices in $\mathcal{E}$ is a scalar multiple of another element in
$\mathcal{E}$, and (c)~the trace $\Tr(A^\dagger B)=0$ for distinct
elements $A,B$ of $\mathcal{E}$. A finite set of $q^2$ unitary matrices
that satisfy the properties (a), (b), and (c) is called a \textsl{nice
error basis}, see~\cite{knill96a}.

The set $\mathcal{E}$ of error operators forms a basis of the set of
complex $q\times q$ matrices due to property (c).  We include a
proof that $\mathcal{E}$ is a nice error basis, because parts of our
argument will be of independent interest in the subsequent sections.

\begin{lemma}
The set $\mathcal{E}=\{X(a)Z(b)\,|\, a,b\in \F_q\}$ is a nice error
basis on\/~$\C^q$.
\end{lemma}
\begin{proof}
The matrix $X(0)Z(0)$ is the identity matrix, so property (a) holds.
We also have $\omega^{\tr(ba)}X(a)Z(b)=Z(b)X(a)$, which implies that the
product of two error operators is given by
\begin{equation}\label{eq:multrule}
X(a)Z(b)\,X(a')Z(b')=\omega^{\tr(ba')}X(a+a')Z(b+b').
\end{equation}
This is a scalar multiple of an operator in $\mathcal{E}$,
hence property (b) holds.

Suppose that the error operators are of the form $A=X(a)Z(b)$ and
$B=X(a)Z(b')$ for some $a, b, b'\in \F_q$. Then
$$\Tr(A^\dagger B)= \Tr(Z(b'-b))=\sum_{x\in \F_q}
\omega^{\tr((b'-b)x)}.$$ The map $x\mapsto \omega^{\tr((b'-b)x)}$ is
an additive character of $\F_q$. The sum of all character values is 0
unless the character is trivial; thus, $\Tr(A^\dagger B)=0$ when
$b'\neq b$.

On the other hand, if $A=X(a)Z(b)$ and $B=X(a')Z(b')$ are two error
operators satisfying $a\neq a'$, then the diagonal
elements of the matrix $A^\dagger B=Z(-b)X(a'-a)Z(b')$ are 0, which
implies $\Tr(A^\dagger B)=0$.
Thus, whenever $A$ and $B$ are distinct element of $\mathcal{E}$, then
$\Tr(A^\dagger B)=0$, which proves~(c).
\end{proof}

\begin{example}
We give an explicit construction of a nice error basis with $q=4$
levels.  The finite field $\F_4$ consists of the elements $\F_4=\{0,
1, \alpha, \overline{\alpha}\}$. We denote the four standard basis
vectors of the complex vector space $\C^4$ by $\ket{0}, \ket{1},
\ket{\alpha},$ and $\ket{\overline{\alpha}}$.  Let $\idtwo$ denote the
$2\times 2$ identity matrix, $\sigma_x=\left(\begin{smallmatrix} 0&1\\
1&0
\end{smallmatrix}\right)$, and $\sigma_z=\left(\begin{smallmatrix}
1&\phantom{-}0\\ 0&-1\end{smallmatrix}\right)$. Then
\vspace*{-1ex}
$$
\begin{array}{c@{\,}c@{\,}lc@{\,}c@{\,}lc@{\,}c@{\,}lc@{\,}c@{\,}l}
X(0) &=& \idtwo\otimes \idtwo, &
X(1) &=& \idtwo\otimes \sigma_x, \\
X(\alpha)&=&\sigma_x\otimes \idtwo, &
X(\overline{\alpha})&=& \sigma_x\otimes \sigma_x, \\
Z(0) &=& \idtwo\otimes \idtwo, &
Z(1) &=& \sigma_z\otimes \idtwo, \\
Z(\alpha) &=& \sigma_z\otimes \sigma_z, &
Z(\overline{\alpha}) &=& \idtwo\otimes \sigma_z.
\end{array}
$$
We see that this nice error basis is obtained by tensoring the Pauli
basis, a nice error basis on $\C^2$. The next lemma shows that this is
a general design principle for nice error bases.
\end{example}

\begin{lemma}
If $\mathcal{E}_1$ and $\mathcal{E}_2$ are nice error bases, then
$$\mathcal{E}=\{ E_1\otimes E_2\,|\, E_1\in \mathcal{E}_1, E_2\in
\mathcal{E}_2\}$$ is a
nice error basis as well.
\end{lemma}
\noindent The proof of this observation follows directly from
the definitions.

Let $\mathbf{a}=(a_1,\dots, a_n)\in \F_q^n$. We write $ X(\mathbf{a})
= X(a_1)\otimes\, \cdots \,\otimes X(a_n)$ and $Z(\mathbf{a}) =
Z(a_1)\otimes\, \cdots \,\otimes Z(a_n)$ for the tensor products of
$n$ error operators.  Our aim is to provide an error model that
conveniently represents errors acting locally on one quantum
system. Using the new notations, we can easily formulate this model.

\begin{corollary}\label{th:nice}
The set
$\mathcal{E}_n= \{
X(\mathbf{a})Z(\mathbf{b})\,|\, \mathbf{a}, \mathbf{b}\in \F_q^n\}$ is
a nice error basis on the complex vector space~$\C^{q^n}$.
\end{corollary}

\textit{Remark.}  Several authors have used an error basis that is
equivalent to our definition of $\mathcal{E}_n$,
see~\cite{ashikhmin01,feng02b,kim04,matsumoto00}. We have defined the
operator $Z(b)$ in a slightly different way, so that the properties
relevant for the design of stabilizer codes become more
transparent. In particular, we can avoid an intermediate step that
requires tensoring $p\times p$--matrices, and that allows us to obtain
the trace-symplectic form directly, see Lemma~\ref{th:commute}.

\subsection{Stabilizer Codes}
Let $G_n$ denote the group generated by the matrices of the nice error
basis~$\mathcal{E}_n$.   It follows from
equation~(\ref{eq:multrule}) that
\begin{eqnarray} 
G_n = \{ \omega^{c}X(\mathbf{a})Z(\mathbf{b})\,|\, \mathbf{a, b}
\in \F_q^n, c\in \F_p\}.\label{eq:genPauli}
\end{eqnarray}
Note that $G_n$ is a finite group of order
$pq^{2n}$.  We call $G_n$ the \textsl{error group} associated with the nice
error basis $\mathcal{E}_n$.

A \textsl{stabilizer code} $Q$ \index{stabilizer code}
is a non-zero subspace of $\C^{q^n}$
that satisfies 
\begin{equation}\label{eq:stab}
Q = \bigcap_{E \in S} \{ v \in \C^{q^n} \mid Ev=v\}
\end{equation}
for some subgroup $S$ of $G_n$. In other words, $Q$ is the joint 
eigenvalue-$1$ eigenspace of a subgroup $S$ of the error group~$G_n$.
\smallskip

\textit{Remark.}  A crucial property of a stabilizer code is
that it contains \textsl{all} joint eigenvectors of~$S$ with
eigenvalue 1, as equation~(\ref{eq:stab}) indicates.  If the code is
smaller and does not exhaust all joint eigenvectors of $S$ with
eigenvalue 1, then it is not a stabilizer code for $S$.
\smallbreak

\subsection{Minimum Distance}
The error correction and detection capabilities of a quantum
error-correcting code $Q$ are the most crucial aspects of the code.
Recall that a quantum code~$Q$ is able to detect an error $E$ in the
unitary group $U(q^n)$ if and only if the condition $\langle c_1 | E|
c_2\rangle=\lambda_E\langle c_1 |c_2\rangle$ holds for all $c_1,
c_2\in Q$, see~\cite{KnLa97}.

It turns out that a stabilizer code~$Q$ with stabilizer $S$ can detect
all errors in $G_n$ that are scalar multiples of elements in $S$ or
that do not commute with some element of $S$, see
Lemma~\ref{th:detectable}. In particular, an error in $G_n$ that is
not detectable has to commute with all elements of the stabilizer.
Commuting elements in $G_n$ are characterized as follows:

\begin{lemma}\label{th:commute}
Two elements $E=\omega^cX(\mathbf{a})Z(\mathbf{b})$ and
$E'=\omega^{c'}X(\mathbf{a'})Z(\mathbf{b'})$ of the error group $G_n$
satisfy the relation
$$
EE' = \omega^{\tr(\mathbf{b\cdot a'-b'\cdot a})} E'E.
$$
In particular, the elements $E$ and $E'$ commute if and only if
the trace symplectic form
$\tr(\mathbf{b\cdot a'-b'\cdot a})$ vanishes.
\end{lemma}
\begin{proof}
It follows from equation (\ref{eq:multrule}) that
$EE'=\omega^{\tr(\mathbf{b\cdot a'})}X(\mathbf{a+a'})Z(\mathbf{b+b'})$
and $E'E=\omega^{\tr(\mathbf{b'\cdot a})}
X(\mathbf{a+a'})Z(\mathbf{b+b'})$. Therefore, multiplying $E'E$ by
the scalar $\omega^{\tr(\mathbf{b\cdot a'-b'\cdot a})}$ yields $EE'$,
as claimed.
\end{proof}

We define the \textsl{symplectic weight} $\swt$ \index{symplectic weight}
of a vector $(\mbf a|\mbf b)$
in $\F_q^{2n}$ as
$$\swt((\mbf a|\mbf b)) = | \{\, k\, |\, (a_k,b_k)\neq (0,0)\}|.$$ The
weight $\w(E)$ of an element $E=\omega^c X(\mbf{a})Z(\mbf{b})$ in the
error group~$G_n$ is defined to be the number of nonidentity tensor
components, $\w(E)=\swt((\mbf a|\mbf b))$. In particular, the weight
of a scalar multiple of the identity matrix is by definition zero.

A quantum code $Q$ has \textsl{minimum distance} $d$\/ \index{distance!quantum}
if and only if
it can detect all errors in $G_n$ of weight less than~$d$, but cannot
detect some error of weight~$d$.  We say that $Q$ is an $((n,K,d))_q$
code if and only if $Q$ is a $K$-dimensional subspace of $\C^{q^n}$
that has minimum distance~$d$.  An $((n,q^k,d))_q$ code is also called
an $[[n,k,d]]_q$ code. We remark that some authors are more
restrictive and use the bracket notation just for stabilizer codes.

We say that a quantum code $Q$ is \textsl{pure to} $t$\/ if and only
if its stabilizer group $S$ does not contain non-scalar matrices of
weight less than $t$.  A quantum code is called pure if and only if it
is pure to its minimum distance.  As in~\cite{calderbank98}, we
always assume that an $[[n,0,d]]_q$ code has to be pure.  
\medskip

\textit{Remarks.}  (a) If a quantum error-correcting code can detect a set
$\mathcal{D}$ of errors, then it can detect all errors in the linear
span of $\mathcal{D}$.  (b) A code of minimum distance~$d$ can correct all
errors of weight $t=\lfloor (d-1)/2\rfloor$ or less.

\section{Galois Connection}\label{sec:galoisCnxn}
We want to clarify the relation between stabilizer codes and more
general quantum codes before we proceed further.  Let us denote by
$\mathcal{Q}$ the set of all subspaces of $\C^{q^n}$. The set
$\mathcal{Q}$ is partially ordered by the inclusion relation.  Any two
elements of $\mathcal{Q}$ have a least upper bound and a greatest
lower bound with respect to the inclusion relation, namely
$$ \sup\{Q,Q'\} = Q+Q'\quad\mbox{and}\quad \inf\{Q,Q'\} = Q\cap Q'.$$
Therefore, $\mathcal{Q}$ is a complete (order) lattice. An element of this
lattice is a quantum error-correcting code or is equal to the vector
space $\{0\}$.

Let~$\mathcal{G}$ denote the lattice of subgroups of the error
group~$G_n$. We will introduce two order-reversing maps
between~$\mathcal{G}$ and $\mathcal{Q}$ that establish a Galois
connection. We will see that stabilizer codes are distinguished
elements of~$\mathcal{Q}$ that remain the same when mapped to the
lattice $\mathcal G$ and back.

Let us define a map $\Fix$ from the lattice $\mathcal{G}$ of subgroups
to the lattice $\mathcal{Q}$ of subspaces
that associates to a group $S$ its joint eigenspace with
eigenvalue~1,
\begin{equation}\label{eq:Fix}
\Fix(S) = \bigcap_{E\in S} \{ v\in \C^{q^n}\,|\, Ev=v\}.
\end{equation}
We define for the reverse direction a map $\Stab$ from the lattice
$\mathcal{Q}$ to the lattice $\mathcal{G}$ that associates to a
quantum code $Q$ its stabilizer group $\Stab(Q)$,
\begin{equation}\label{eq:Stab}
\Stab(Q) = \{ E\in G_n\,|\, Ev=v \text{ for all } v \in Q\}.
\end{equation}
We obtain four direct consequences of the definitions (\ref{eq:Fix})
and (\ref{eq:Stab}):
\begin{enumerate}
\item[\textbf{G1.}] If $Q_1\subseteq Q_2$ are subspaces of $\C^{q^n}$, then $\Stab(Q_2)\le \Stab(Q_1)$.
\item[\textbf{G2.}] If $S_1\le S_2$ are subgroups of $G_n$,
then $\Fix(S_2)\le \Fix(S_1)$.
\item[\textbf{G3.}] A subspace $Q$ of $\C^{q^n}$ satisfies
$Q\subseteq \Fix(\Stab(Q))$.
\item[\textbf{G4.}] A subgroup $S$ of $G_n$ satisfies
$S\le \Stab(\Fix(S))$.
\end{enumerate}
The first two properties establish that $\Fix$ and $\Stab$ are
order-reversing maps. The extension properties G3 and G4 establish
that $\Fix$ and $\Stab$ form a Galois connection,
see~\cite[page~56]{birkhoff61}. The general
theory of Galois connections establishes, among other results, that
$ \Fix(S)=\Fix(\Stab(\Fix(S)))$ and $
\Stab(Q)=\Stab(\Fix(\Stab(Q)))$
holds for all $S$ in $\mathcal{G}$ and all $Q$ in $\mathcal{Q}$.

A subspace~$Q$ of the vector space~$\C^{q^n}$ satisfying G3 with
equality is called a \textsl{closed subspace}, and a subgroup $S$ of
the error group~$G_n$ satisfying G4 with equality is called a
\textsl{closed subgroup}.  We record the main result of abstract
Galois theory in the following proposition.

\begin{proposition}
The closed subspaces of the vector space~$\C^{q^n}$ form a complete
sublattice $\mathcal{Q}_c$ of the lattice~$\mathcal{Q}$. The closed
subgroups of\/ $G_n$ form a complete sublattice~$\mathcal{G}_c$ of the
lattice\/ $\mathcal G$ that is dual isomorphic to the
lattice~$\mathcal{Q}_c$.
\end{proposition}
\begin{proof}
This result holds for any Galois connection, see
Theorem~10 in the book by Birkhoff~\cite[page 56]{birkhoff61}.
\end{proof}

We need to characterize the closed subspaces and subgroups
to make this proposition useful. We begin with the closed subspaces
because this is easier.

\begin{lemma}
A closed subspace is a stabilizer code or is 0-dimensional.
\end{lemma}
\begin{proof}
By definition, a closed subspace $Q$ satisfies
$$ Q = \Fix(\Stab(Q)) =
\bigcap_{E\in \Stab(Q)} \{ v\in \C^{q^n}\,|\, Ev=v\},$$
hence is a stabilizer code or $\{0\}$.
\end{proof}

\begin{lemma}\label{th:stab}
If $Q$ is a nonzero subspace of $\C^{q^n}$, then its stabilizer
$S=\Stab(Q)$ is an abelian group satisfying $S\cap Z(G_n)=\{1\}$.
\end{lemma}
\begin{proof}
Suppose that $E$ and $E'$ are non-commuting elements of
$S=\Stab(Q)$. By Lemma~\ref{th:commute}, we have $EE'=\omega^k E'E$
for some $\omega^k\neq 1$.  A nonzero vector $v$ in~$Q$ would have to
satisfy $ v=EE'v=\omega^kE'Ev=\omega^kv,$ contradiction. Therefore,
$S$ is an abelian group. The stabilizer cannot contain any element
$\omega^k\onemat$, unless $k=0$, which proves the second assertion.
\end{proof}

\begin{lemma}\label{th:projection}
Suppose that $S$ is the stabilizer of a vector space $Q$.
An orthogonal projector onto the joint eigenspace $\Fix(S)$ is given by
$$ P= \frac{1}{|S|} \sum_{E\in S} E. $$\index{projector!stabilizer code}
\end{lemma}
\begin{proof}
A vector $v$ in $\Fix(S)$ satisfies $Pv=v$, hence $\Fix(S)$ is
contained in the image of $P$. Conversely, note that $EP=P$ holds for
all $E$ in $S$, hence any vector in the image of $P$ is an eigenvector
with eigenvalue $1$ of all error operators $E$ in $S$. Therefore, $\Fix(S)=
\image P$. The operator $P$ is idempotent, because
$$ P^2 = \frac{1}{|S|} \sum_{E\in S} EP = \frac{1}{|S|} \sum_{E\in S} P = P$$
holds. The inverse $E^\dagger$ of $E$ is contained in the group $S$, hence $P^\dagger = P$. Therefore, $P$ is an orthogonal projector onto $\Fix(S)$.
\end{proof}

\textit{Remark.} If $S$ is a nonabelian subgroup of the group
$G_n$, then it necessarily contains the center $Z(G_n)$ of $G_n$; it
follows that $P$ is equal to the all-zero matrix. Note that the image
of $P$ has dimension $\Tr(P)=q^n/|S|$.

\begin{lemma}\label{th:closedsubgroup}
A subgroup $S$ of\/ $G_n$ is closed if and only if
$S$ is an abelian
subgroup that satisfies $S\cap Z(G_n)=\{1\}$ or if
$S$ is equal to $G_n$.
\end{lemma}
\begin{proof}
Suppose that $S$ is a closed subgroup of $G_n$. The vector space
$Q=\Fix(S)$ is, by definition, either a stabilizer code or a
0-dimensional vector space.  We have $\Stab(\{0\})=G_n$. Furthermore,
if $Q\neq \{0\}$, then $\Stab(Q)=S$ is an abelian group satisfying
$S\cap Z(G_n)=\{\onemat\}$, by Lemma~\ref{th:stab}.

Conversely, suppose that $S$ is an abelian subgroup of $G_n$ such that
$S$ trivially intersects the center $Z(G_n)$.  Let
$S^*=\Stab(\Fix(S)).$ We have
$\Fix(S^*)=\Fix(\Stab(\Fix(S)))=\Fix(S),$ because this holds for any
pair of maps that form a Galois connection.
It follows from Lemma~\ref{th:projection} that
$$ q^n/|S^*| = \Tr\left(\frac{1}{|S^*|}\sum_{E\in S^*} E\right) =
\Tr\left(\frac{1}{|S|}\sum_{E\in S} E\right) = q^n/|S|.$$ 
Since $S\le S^*$, this shows that $S=S^*=\Stab(\Fix(S))$; hence, $S$ is
a closed subgroup of $G_n$. We note that $\Fix(G_n)=\{0\}$, so that
$G_n=\Stab(\Fix(G_n))$ is closed.
\end{proof}

The stabilizer codes are easier to study than arbitrary quantum codes,
as we will see in the subsequent sections. If we know the error
correction capabilities of stabilizer codes, then we sometimes get a
lower bound on the minimum distance of an arbitrary code by the
following simple observation: 
\smallskip

\noindent\textbf{Fact.} An arbitrary quantum code $Q$
is contained in the larger stabilizer code given by $Q^*=\Fix(\Stab(Q))$.  If
an error $E$ can be detected by $Q^*$, then it can be detected by $Q$
as well. Therefore, if the stabilizer code $Q^*$ has minimum distance
$d$, then the quantum code $Q$ has at least minimum distance $d$.

\section{Additive Codes}\label{sect:additive}
The previous section explored the relation between stabilizer codes
and other quantum codes. We show next how stabilizer codes are related
to classical codes (namely, additive codes over $\F_q$ or
$\F_{q^2}$). The classical codes allow us to characterize the errors
in $G_n$ that are detectable by the stabilizer code.

In the binary case, the problem of finding stabilizer codes of length
$n$ had been translated into (a) finding binary classical codes of
length $2n$ that are self-orthogonal with respect to a symplectic
inner product or (b) finding classical codes of length $n$ over $\F_4$
that are self-orthogonal with respect to a trace-inner product,
see~\cite{calderbank98}. The approach (a) was generalized to prime
alphabets by Rains~\cite{rains99} and to prime-power alphabets by
Ashikhmin and Knill~\cite{ashikhmin01}. We simplify the arguments and
include a full proof of this connection.
There were many attempts to generalize
the approach (b) to nonbinary alphabets, but without complete success
(but see for instance \cite{rains99,matsumoto00,kim04} for notable
partial solutions). We fill this gap and introduce a natural
generalization of~(b). Furthermore, we discuss simpler constructions
for linear codes. Before exploring these connections to classical
codes, we first recall some facts about detectable errors.

If $S$ is a subgroup of $G_n$, then
$C_{G_n}(S)$ denotes centralizer of $S$ in $G_n$,
$$C_{G_n}(S)=\{ E\in G_n\,|\, EF=FE \text{ for all } F\in S\},$$ and
$SZ(G_n)$ denotes the group generated by $S$ and the center
$Z(G_n)$.  We first recall the following characterization of
detectable errors (see also~\cite{ashikhmin01}; the interested reader
can find a more general approach in~\cite{knill96b,klappenecker033}).

\begin{lemma}\label{th:detectable}
Suppose that $S \le G_n$ is the stabilizer group of a stabilizer
code~$Q$ of dimension $\dim Q>1$.  An error $E$ in $G_n$ is detectable by the
quantum code $Q$ if and only if either $E$ is an element of $SZ(G_n)$ or $E$
does not belong to the centralizer $C_{G_n}(S)$.
\end{lemma}
\begin{proof}
An element $E$ in $SZ(G_n)$ is a scalar multiple of a stabilizer;
thus, it acts by multiplication with a scalar $\lambda_E$ on $Q$.
It follows that $E$ is a detectable error.

Suppose now that $E$ is an error in $G_n$ that does not commute with some
element $F$ of the stabilizer~$S$; it follows that $EF= \lambda FE$ for some
complex number $\lambda \neq 1$, see Lemma~\ref{th:commute}.
All vectors $u$ and $v$ in $Q$ satisfy the condition
\begin{equation}\label{eq:noncommute}
\bra{u} E\ket{v} =\bra{u} EF\ket{v}=\lambda \bra{u} FE\ket{v}=
\lambda \bra{u}E\ket{v};
\end{equation}
hence, $\bra{u}E\ket{v}=0$. It follows that the error $E$ is detectable.

Finally, suppose that $E$ is an element of $C_{G_n}(S)\setminus
SZ(G_n)$.  Seeking a contradiction, we assume that $E$ is detectable;
this implies that there exists a complex scalar $\lambda_E$ such that
$Ev=\lambda_E v$ for all $v$ in $Q$.  The scalar $\lambda_E$ cannot be zero
because $E$ commutes with the elements of $S$, so
$EP=PEP=\lambda_EP$ and clearly $EP\neq 0$.
Let $S^*$ denote the abelian
group generated by $\lambda_E^{-1} E$ and by the elements of $S$.  The
joint eigenspace of $S^*$ with eigenvalue 1 has dimension $q^n/|S^*| <
\dim Q=q^n/|S|$. This implies that not all vectors in $Q$ remain invariant under
$\lambda_E^{-1} E$, in contradiction to the detectability of $E$.
\end{proof}

\begin{corollary}
If a stabilizer code $Q$ has minimum distance $d$ and is pure to~$t$,
then all errors $E\in G_n$ with $1\le \wt(E)<\min\{t,d\}$ satisfy
$\langle u|E|v\rangle=0$ for all $u$ and $v$ in $Q$.
\end{corollary}
\begin{proof}
By assumption, the weight of $E$ is less than the minimum distance, so
the error is detectable. However, $E$ is not an element of $Z(G_n)S$,
since the code is pure to $t>\wt(E)$. Therefore, $E$ does not belong
to $C_{G_n}(S)$, and the claim follows from equation~(\ref{eq:noncommute}).
\end{proof}

\subsection{Codes over $\F_q$}
Lemma~\ref{th:detectable} characterizes the error detection capabilities of a
stabilizer code with stabilizer group $S$ in terms of the groups
$SZ(G_n)$ and $C_{G_n}(S)$. The phase information of an element in
$G_n$ is not relevant for questions concerning the detectability,
since an element $E$ of $G_n$ is detectable if and only if $\omega E$
is detectable. Thus, if we associate with an element
$\omega^cX(\mathbf{a})Z(\mathbf{b})$ of $G_n$ an element
$(\mathbf{a}|\mathbf{b})$ of $\F_q^{2n}$, then the group $SZ(G_n)$ is
mapped to the additive code
$$C=\{ (\mathbf{a}|\mathbf{b})\,|\,
\omega^cX(\mathbf{a})Z(\mathbf{b})\in SZ(G_n)\}= SZ(G_n)/Z(G_n).$$ To
describe the image of the centralizer, we need the notion of a
trace-symplectic form of \index{inner product!trace-symplectic}
two vectors $(\mathbf{a}|\mathbf{b})$ and
$(\mathbf{a'}|\mathbf{b'})$ in $\F_{q}^{2n}$,
$$ \< (\mathbf{a}|\mathbf{b})\, |\, (\mathbf{a'}|\mathbf{b'}) >s =
\tr_{q/p}(\mathbf{b}\cdot \mathbf{a}' - \mathbf{b}'\cdot
\mathbf{a}).$$ The centralizer $C_{G_n}(S)$ contains all elements of
$G_n$ that commute with each element of $S$; thus, by
Lemma~\ref{th:commute}, $C_{G_n}(S)$ is mapped onto
the trace-symplectic dual code $C^\sdual$ of the code $C$,
$$ C^\sdual
=\{ (\mathbf{a}|\mathbf{b})\,|\, \omega^cX(\mathbf{a})Z(\mathbf{b})\in C_{G_n}(S)\}.
$$
The connection between these classical codes and the stabilizer code
is made precise in the next theorem. This theorem is essentially
contained in~\cite{ashikhmin01} and generalizes the well-known
connection to symplectic codes~\cite{calderbank98,gottesman96} of the
binary case.

\begin{theorem}\label{th:stabilizer}
An $((n,K,d))_q$ stabilizer code exists if and only if there exists an
additive code $C \le \F_q^{2n}$ of size $|C|=q^n/K$ such that $C\le
C^\sdual$ and $\swt(C^{\sdual} \setminus C)=d$ if $K>1$ (and
$\swt(C^\sdual)=d$ if $K=1$).
\end{theorem}
\begin{proof}
Suppose that an $((n,K,d))_q$ stabilizer code $Q$ exists. This implies
that there exists a closed subgroup $S$ of $G_n$ of order $|S|=q^n/K$
such that $Q=\Fix(S)$.  The group $S$ is abelian and satisfies $S\cap
Z(G_n)=1$, by Lemma~\ref{th:closedsubgroup}.  The quotient $C\cong
SZ(G_n)/Z(G_n)$ is an additive subgroup of $\F_q^{2n}$ such that
$|C|=|S|=q^n/K$. We have $C^\sdual =
C_{G_n}(S)/Z(G_n)$ by Lemma~\ref{th:commute}.  Since $S$ is an abelian
group, $SZ(G_n)\le C_{G_n}(S)$, hence $C\le
C^\sdual$.
Recall that the weight of an element $\omega^c X(\mbf a)Z(\mbf b)$ in
$G_n$ is equal to $\swt(\mbf a|\mbf b)$. If $K=1$, then
$Q$ is a pure quantum code, thus $\wt(C_{G_n}(S))=\swt(C^\sdual)=d$.
If $K>1$, then the elements of
$C_{G_n}(S)\setminus SZ(G_n)$ have at least weight $d$ by
Lemma~\ref{th:detectable}, so that $\swt(C^\sdual \setminus C)=d$.

Conversely, suppose that $C$ is an additive subcode of
$\F_q^{2n}$ such that $|C|=q^n/K$, $C\le C^\sdual$, and
$\swt(C^\sdual\setminus C)=d$ if $K>1$ (and $\swt(C^\sdual)=d$ if $K=1$).
Let
$$ N = \{ \omega^cX(\mbf a)Z(\mbf b)\,|\, c \in \F_p \text{ and }
(\mbf a|\mbf b)\in C\}.$$ Notice that $N$ is an abelian normal
subgroup of $G_n$, because it is the pre-image of $C=N/Z(G_n)$.
Choose a character $\chi$ of $N$ such that
$\chi(\omega^c\onemat)=\omega^c$. Then
$$ P_N = \frac{1}{|N|}\sum_{E\in N} \chi(E^{-1}) E$$ is an orthogonal
projector onto a vector space $Q$, because $P_N$ is an idempotent in
the group ring $\C[G_n]$, see~\cite[Theorem~1]{klappenecker033}. We have
$$\dim Q = \Tr P_N = |Z(G_n)| q^n/|N| = q^n/|C| = K.$$ Each coset of
$N$ modulo $Z(G_n)$ contains exactly one matrix $E$ such that $Ev=v$
for all $v$ in $Q$. Set $S=\{E\in N\,|\, Ev=v \text{ for all } v \in
Q\}$. Then $S$ is an abelian subgroup of $G_n$ of order
$|S|=|C|=q^n/K$. We have $Q=\Fix(S)$, because $Q$ is clearly a
subspace of $\Fix(S)$, but $\dim Q=q^n/|S|=K$. An element $\omega^c
X(\mbf a)Z(\mbf b)$ in $C_{G_n}(S)\setminus SZ(G_n)$ cannot have
weight less than $d$, because this would imply that $(\mbf a|\mbf
b)\in C^\sdual\setminus C$ has weight less than $d$, which is
impossible. By the same token, if $K=1$, then all nonidentity elements
of the centralizer $C_{G_n}(S)$ must have weight $d$ or higher.
Therefore, $Q$ is an $((n,K,d))_q$ stabilizer code.
\end{proof}

The results of this paragraph were established by
Ashikhmin and Knill~\cite{ashikhmin01}. It is instructive to compare
the two approaches, since their definition of the error basis
is different (but equivalent).

\subsection{Codes over $\F_{q^2}$}
A drawback of the codes in the previous paragraph is that the
symplectic weight is somewhat unusual. In the binary case,
reference~\cite{calderbank98} provided a remedy by relating binary
stabilizer codes to additive codes over $\F_4$, allowing the use of
the familiar Hamming weight. Somewhat surprisingly, the corresponding
concept was not completely generalized to~$\F_{q^2}$, although
\cite{matsumoto00,kim04} and~\cite{rains99} paved the way to our
approach.  
After an initial circulation of the results in this chapter, 
Gottesman drew our attention to another interesting approach that was
initiated by Barnum, see~\cite{barnum00,barnum02}, where a sufficient
condition for the existence of stabilizer codes is established using a
symplectic form.
\smallskip

Let $(\beta,\beta^q)$ denote a normal basis of $\F_{q^2}$ over $\F_q$.
We define a trace-alternating form of two vectors $v$ and $w$ in
$\F_{q^2}^n$ by
\begin{equation}\label{eq:alternating}\index{inner product!trace-alternating}
 \(v|w)a =
\tr_{q/p}\left(\frac{v\cdot w^q - v^q\cdot w }{\beta^{2q}-\beta^2}\right).
\end{equation}
We note that the argument of the trace is invariant under the
Galois automorphism $x\mapsto x^q$, so it is indeed an element
of~$\F_q$, which shows that (\ref{eq:alternating}) is well-defined.

The trace-alternating form is bi-additive, that is,
$\(u+v|w)a = \(u|w)a+\(v|w)a $ and $\(u|v+w)a = \(u|v)a+\(u|w)a$ holds
for all $u,v,w\in \F_{q^2}^n$. It is $\F_p$-linear, but not
$\F_q$-linear unless $q=p$ and it is alternating in the sense that
$\(u|u)a=0$ holds for all $u\in \F_{q^2}^n$. We write $u \adual w$ if
and only if $\(u|w)a=0$ holds.

At this point it might be helpful to see the form the trace-alternating
form takes in the binary case. A normal basis for $\F_4$ over $\F_2$
is given by $\{\omega, \omega^2 \}$. Since $\omega^2+\omega+1=0$, the 
trace-alternating form simplifies to 
\begin{equation}
 \(v|w)a = \tr_{2/2}\left(\frac{v\cdot w^2 + v^2\cdot w }{\omega^{4}+\omega^2}\right) = v\cdot w^q + v^q\cdot w,
\end{equation}
where we have used the facts that $\omega^3=1$ and  $x=-x$ over  $\F_4$.

We define a bijective map $\phi$ that takes an element $(\mbf a|\mbf b)$ of the
vector space $\F_q^{2n}$ to a vector in $\F_{q^2}$ by setting
$\phi((\mbf a|\mbf b)) = \beta \mbf a + \beta^q \mbf b.$ The map
$\phi$ is isometric in the sense that the symplectic weight of
$(\mbf a|\mbf b)$ is equal to the Hamming weight of $\phi((\mbf a|\mbf b))$.

\begin{lemma}\label{th:isometry}
Suppose that $c$ and $d$ are two vectors
of\/ $\F_q^{2n}$. Then
$$\< c\,|\,d>s=\(\phi(c)\,|\, \phi(d))a.$$ In particular, $c$ and $d$ are
orthogonal with respect to the trace-symplectic form if and
only if $\phi(c)$ and $\phi(d)$ are orthogonal with respect to the
trace-alternating form.
\end{lemma}
\begin{proof}
Let $c=(\mbf a|\mbf b)$ and $d=(\mbf a'|\mbf b')$. We calculate
\begin{equation*}
\begin{split}
\phi(c)\cdot \phi(d)^q =
\beta^{q+1} \,\mbf a\cdot \mbf a' +
\beta^{2}   \,\mbf a\cdot \mbf b' +
\beta^{2q}  \,\mbf b\cdot &\mbf a' +
\beta^{q+1} \,\mbf b\cdot \mbf b', \\[1ex]
\phi(c)^q\cdot \phi(d) =
\beta^{q+1} \, \mbf a\cdot \mbf a' +
\beta^{2q}  \, \mbf a\cdot \mbf b' +
\beta^2     \, \mbf b\cdot &\mbf a' +
\beta^{q+1} \, \mbf b\cdot \mbf b'.
\end{split}
\end{equation*}
Therefore, the trace-alternating form of $\phi(c)$ and $\phi(d)$ is given by
\begin{eqnarray*}
\(\phi(c)|\phi(d))a &=&
\tr_{q/p}\left(\frac{\phi(c)\cdot \phi(d)^q - \phi(c)^q\cdot \phi(d) }{\beta^{2q}-\beta^2}\right),\\
&=&\tr_{q/p}(\mbf b \cdot \mbf a' - \mbf a \cdot \mbf b' ),
\end{eqnarray*}
which is precisely the trace-symplectic form $\< c\,|\,d>s$.
\end{proof}

\begin{theorem}\label{th:alternating}
An\/ $((n,K,d))_q$ stabilizer code exists if and only if there exists
an additive subcode $D$ of\/ $\F_{q^2}^{n}$ of cardinality
$|D|=q^n/K$ such that $D\le D^\adual$ and
$\wt(D^{\adual} \setminus D)=d$ if $K>1$ (and $\wt(D^\adual)=d$ if $K=1$).
\end{theorem}
\begin{proof}
Theorem~\ref{th:stabilizer} shows that an $((n,K,d))_q$ stabilizer
code exists if and only if there exists a code $C\le \F_q^{2n}$ with
$|C|=q^n/K$, $C\le C^\sdual$, and $\swt(C^\sdual\setminus C)=d$ if
$K>1$ (and $\swt(C^\sdual)=d$ if $K=1$). We
obtain the statement of the theorem by applying the isometry $\phi$.
\end{proof}

We obtain the following convenient condition for the existence of a
stabilizer code as a direct consequence of the previous theorem.
\begin{corollary}\label{co:alternating}\index{stabilizer code!construction!additive}
If there exists a classical\/ $[n,k]_{q^2}$ additive code $D\le
\F_{q^2}$ such that $D\le D^\adual$ and $d^\adual=\wt(D^\adual)$,
then there exists an $[[n,n-2k,\geq d^\adual]]_{q}$ stabilizer 
code that is pure to~$d^\adual$.
\end{corollary}

\textit{Remark.} It is not necessary to use a normal basis in the
definition of the isometry $\phi$ and the trace-alternating
form. Alternatively, we could have used a polynomial basis $(1,\gamma)$
of $\F_q^2/\F_q$. In that case, one can define the isometry $\phi$ by
$\phi((\mathbf{a}|\mathbf{b}))=\mathbf{a}+\gamma \mathbf{b}$,
and a compatible trace-alternating form by
$$ \langle v\, | \, w\rangle_{a'} = \tr_{q/p}\left(\frac{
v \cdot w^q - v^q\cdot w
}{\gamma-\gamma^q}\right).$$
One can check that the statement of Lemma~\ref{th:isometry}
is satisfied for this choice as well.
Other variations on this theme are possible.

\subsection{Classical Codes}
Self-orthogonal codes with respect to the trace-alter\-nating form are
not often studied in classical coding theory; more common are codes
which are self-orthogonal with respect to a euclidean or hermitian
inner product.  We relate these concepts of orthogonality as follows.
Consider the hermitian inner product $\mathbf{x}^q\cdot \mathbf{y}$ of
two vectors $\mathbf{x}$ and $\mathbf{y}$ in $\F_{q^2}^n$; we write
$\mathbf{x} \,\hdual\, \mathbf{y}$ if and only if $\mathbf{x}^q \cdot
\mathbf{y}=0$ holds.
\begin{lemma}\label{th:hermitian} 
If two vectors $\mathbf{x}$ and $\mathbf{y}$ in $\F_{q^2}^n$ satisfy
$\mathbf{x} \, \hdual\, \mathbf{y}$, then they satisfy $\mathbf{x}
\,\adual\, \mathbf{y}$.  In particular, if\/ $D\le \F_{q^2}^n$, then
$D^\hdual \le D^\adual$.
\end{lemma}
\begin{proof}
It follows from $\mathbf{x}^q\cdot \mathbf{y}=0$ that
 $\mathbf{x}\cdot \mathbf{y}^q=0$ holds, whence
$$\(\mathbf{x}|\mathbf{y})a =
\tr_{q/p}\left(
\frac{\mathbf{x}\cdot \mathbf{y}^q - \mathbf{x}^q\cdot \mathbf{y}}{\beta^{2q}-\beta^2}\right)=0,$$
as claimed.
\end{proof}
Therefore, any self-orthogonal code with respect to the hermitian
inner product is self-orthogonal with respect to the trace-alternating
form. In general, the two dual spaces $D^\hdual$ and $D^\adual$ are not
the same. However, if $D$ happens to be $\F_{q^2}$-linear, then the
two dual spaces coincide.
\begin{lemma}\label{th:classical}
Suppose that $D\le \F_{q^2}^n$ is $\F_{q^2}$-linear, then $D^\hdual=D^\adual$.
\end{lemma}
\begin{proof}
Let $q=p^m$, $p$ prime. If $D$ is a $k$-dimensional subspace of
$\F_{q^2}^n$, then $D^\hdual$ is an $(n-k)$-dimensional subspace of
$\F_{q^2}^n$. We can also view~$D$ as a $2mk$-dimensional subspace of
$\F_p^{2mn}$, and $D^\adual$ as a $2m(n-k)$-dimensional subspace of
$\F_p^{2mn}$.  Since $D^\hdual \subseteq D^\adual$ and the
cardinalities of $D^\adual$ and $D^\hdual$ are the same, we can
conclude that $D^\adual =D^\hdual$.
\end{proof}

\begin{corollary}[Hermitian Construction]\label{co:classical}
\index{stabilizer code!construction!Hermitian}
If there exists an $\F_{q^2}$-linear $[n,k,d]_{q^2}$ code $B$ such
that $B^\hdual\le B$, then there exists an $[[n,2k-n,\geq d]]_{q}$
quantum code that is pure to~$d$.
\end{corollary}
\begin{proof}
The hermitian inner product is nondegenerate, so the hermitian dual of
the code $D:=B^\hdual$ is $B$.  The
$[n,n-k]_{q^2}$ code $D$ is $\F_{q^2}$-linear, so
$D^\hdual=D^\adual$ by Lemma~\ref{th:classical}, and the claim follows
from Corollary~\ref{co:alternating}.
\end{proof}

So it suffices to consider hermitian forms in the case of
$\F_{q^2}$-linear codes. We have to use the slightly more cumbersome
trace-alternating form in the case of additive codes that are not
linear over $\F_{q^2}$.

An elegant and surprisingly simple construction of quantum codes was
introduced in 1996 by Calderbank and Shor~\cite{calderbank96} and by
Steane~\cite{steane96}. The CSS code construction provides perhaps the
most direct link to classical coding theory.

\begin{lemma}[CSS Code Construction]\label{th:css}
\index{stabilizer code!construction!CSS}
Let $C_1$ and $C_2$ denote two classical linear codes with parameters
$[n,k_1,d_1]_q$ and $[n,k_2,d_2]_q$ such that $C_2^\perp\le C_1$.
Then there exists a $[[n,k_1+k_2-n,d]]_q$ stabilizer code with minimum
distance $d=\min\{ \wt(c) \mid c\in (C_1\setminus C_2^\perp)\cup
(C_2\setminus C_1^\perp)\}$ that is pure to $\min\{ d_1,d_2\}$.
\end{lemma}
\begin{proof}
Let $C=C_1^\perp\times C_2^\perp\le \F_q^{2n}$. If $(c_1\mid c_2)$ and
$(c_1'\mid c_2')$ are two elements of $C$, then we observe that
$$ \tr( c_2\cdot c_1' - c_2' \cdot c_1) = \tr(0-0)=0.$$ Therefore,
$C\le C^\sdual$. Furthermore, the trace-symplectic dual of $C$
contains $C_2\times C_1$, and a dimensionality argument shows that
$C^\sdual = C_2\times C_1$. Since the cartesian product
$C_1^\perp\times C_2^\perp$ has $q^{2n-(k_1+k_2)}$ elements, the
stabilizer code has dimension $q^{k_1+k_2-n}$ by
Theorem~\ref{th:stabilizer}. The claim about the minimum distance and purity of
the code is obvious from the construction.
\end{proof}

\begin{corollary}[Euclidean Construction]\label{th:css2}
\index{stabilizer code!construction!Euclidean}
If $C$ is a classical linear $[n,k,d]_q$ code containing its dual,
$C^\perp\le C$, then there exists an $[[n,2k-n,\geq d]]_q$ stabilizer code 
that is pure to~$d$.
\end{corollary}

\section{Weight Enumerators}\label{sec:wtEnums}
The Shor-Laflamme weight enumerators of an arbitrary $((n,K))_q$
quantum code~$Q$ with orthogonal projector $P$ are defined by the
polynomials
$$
\begin{array}{ll}
\ds \sum_{i=0}^n A_i^{\textsc{sl}} z^i, &\quad\text{with}\quad
\ds A_i^{\textsc{sl}}=\frac{1}{K^2}\sum_{\ontop{E\in G_n}{\wt(E)=i}} 
\Tr(E^\dagger P)\Tr(E P),
\end{array}
$$
and
$$
\begin{array}{ll}
\ds \sum_{i=0}^n B_i^{\textsc{sl}} z^i, &\quad\text{with}\quad
\ds B_i^{\textsc{sl}} = \frac{1}{K}\sum_{\ontop{E\in G_n}{\wt(E)=i}}
\Tr(E^\dagger PE P),\phantom{gggd}
\end{array}
$$ see~\cite{shor97} for the binary case. The definition given here 
differs from the original definition by Shor and Laflamme by a normalization
factor $p$, which is due to the sums running over the full error group
$G_n$.  
The theory of Shor-Laflamme weight enumerators~\cite{shor97} was
considerably extended by Rains in~\cite{rains98,rains99b,rains99d,rains00}.
In this section we give a simple proof for the relation between these weight
enumerators and the symplectic weight enumerators of the additive codes
associated with the stabilizer code.

The weights
$A_i^{\textsc{SL}}$ and $B _i^{\textsc{SL}}$ have a nice combinatorial
interpretation in the case of stabilizer codes.  Indeed, let $C\le
\F_q^{2n}$ denote the additive code associated with the stabilizer
code~$Q$.  Define the symplectic weights of $C$ and $C^\sdual$
respectively by
$ A_i=|\{ c\in C\,|\, \swt(c)=i\}|$ and
$ B_i=|\{ c\in C^\sdual\,|\, \swt(c)=i\}|.$
\nix{
$$ A_i=|\{ c\in C\,|\, \swt(c)=i\}| \quad\text{and}\quad
B_i=|\{ c\in C^\sdual\,|\, \swt(c)=i\}|.$$
}
The next lemma belongs to the folklore of stabilizer codes. 

\begin{lemma}
The Shor-Laflamme weights of an $((n,K))_q$ stabilizer code $Q$ are
multiples of the symplectic weights of the associated additive
codes~$C$ and~$C^\sdual$; more precisely,
$$ A_i^\textsc{sl}=pA_i\quad \text{and}\quad B_i^\textsc{sl}=pB_i
\quad\text{for}\quad  0\le i\le n,$$
where $p$ is the characteristic of the field\/ $\F_q$.
\end{lemma}
\begin{proof}
Recall that
$$ P = \frac{1}{|S|} \sum_{E\in S} S$$ for the stabilizer group $S$ of
$Q$. The trace $\Tr(EP)$ is nonzero if and only if $E^\dagger$ is an
element of $SZ(G_n)$. If $E^\dagger\in SZ(G_n)$, then $\Tr(E^\dagger
P)\Tr(EP)=(q^n/|S|)^2=K^2$. Therefore, $A_{i}^\textsc{sl}$
counts the elements in $SZ(G_n)$ of weight $i$, so
$A_i^\textsc{sl}=|Z(G_n)|\times |\{c\in C\,|\, \swt(c)=i\}|=p A_i.$

If $E$ commutes with all elements in $S$, then $\Tr(E^\dagger PE
P)=\Tr(P^2)=\Tr(P)=K$. If $E$ does not commute with some element of
$S$, then $E$ is detectable; more precisely, the proof of
Lemma~\ref{th:detectable} shows that $PEP=0 P$, hence $\Tr(E^\dagger
PEP)=0$.  Therefore, $B_i^\textsc{sl}$ counts the elements in
$C_{G_n}(S)$ of weight $i$, hence $B_i^\textsc{sl} = |Z(G_n)| \times
|\{c\in C^\sdual\,|\, \swt(c)=i\}| = p B_i.$
\end{proof}

Shor and Laflamme had been aware of the stabilizer case when they
introduced their weight enumerators, so the combinatorial
interpretation of the weights does not appear to be a coincidence.
Recall that the Shor-Laflamme enumerators of arbitrary quantum codes
are related by a MacWilliams identity, see~\cite{rains98,shor97}. For
stabilizer codes, we can directly relate the symplectic weight
enumerators of $C$ and $C^\sdual$,
$$ A(z) = \sum_{i=0}^n A_iz^i\quad\text{and}\quad B(z)=\sum_{i=0}^n
B_i z^i,$$ using a simple argument that is very much in the spirit of Jessie
MacWilliams' original proof for euclidean dual codes~\cite{macwilliams63}.

\begin{theorem}
Let $C$ be an additive subcode of $\F_q^{2n}$
with symplectic weight enumerator
$A(z)$. Then the symplectic weight enumerator of $C^\sdual$
is given by
$$ B(z) = \frac{(1+(q^2-1)z)^n}{|C|} A
\left(\frac{1-z}{1+(q^2-1)z)}\right). $$
\end{theorem}
\begin{proof}
Let $\chi$ be a nontrivial additive character of $\F_p$.  We define
for $b\in \F_q^{2n}$ a character $\chi_b$ of the additive group $C$ by
substituting the trace-symplectic form for the argument of
the character $\chi$, such that
$$ \chi_b(c) = \chi(\<c|b>s ). $$ The character $\chi_b$ is trivial if
and only if $b$ is an element of $C^{\sdual}$.
Therefore, we obtain from the orthogonality relations of characters that
$$ \sum_{c\in C} \chi_b(c) =
\left\{ \begin{array}{ll} |C| &
\text{ for } b\in C^{\sdual},\\
0 & \text{ otherwise.}
\end{array}\right.
$$
The following relation for polynomials is an immediate consequence
\begin{equation}\label{eq:mac}
\sum_{c\in C}\sum_{b\in \F_q^{2n}} \chi_b(c)z^{\swt(b)}  =
\sum_{b\in \F_q^{2n}} z^{\swt(b)} \sum_{c\in C} \chi_b(c) =
|C| B(z).
\end{equation}
The right hand side is a multiple of the weight enumerator of the code
$C^\sdual$. Let us have a closer look at the inner sum of the
left-hand side.  If we express the vector $c\in C$ in the form
$c=(c_1,\dots,c_n|d_1,\dots,d_n)$, and expand the character and its
trace-symplectic form, then we obtain
\begin{equation*}
\begin{split}
\ds\sum_{b\in \F_q^{2n}}  \chi_b(c)  z^{\swt(b)} &=
\ds\!\!
\sum_{(a_1,\dots,a_n|b_1,\dots,b_n)\in \F_q^{2n}}  \!\! \!\!\!\!
z^{\sum_{k=1}^n \swt(a_k|b_k)}
\chi\left(\sum_{k=1}^n \tr(d_ka_k-b_kc_k)\right)\\
&= \ds\sum_{(a_1,\dots,a_n|b_1,\dots,b_n)\in \F_q^{2n}}
\prod_{k=1}^n z^{\swt(a_k|b_k)} \chi\left( \tr(d_ka_k-b_kc_k)\right) \\
&=\; \ds\prod_{k=1}^n \sum_{(a_k|b_k)\in \F_q^2} z^{\swt(a_k|b_k)}
\chi\left( \tr(d_ka_k-b_kc_k)\right).
\end{split}
\end{equation*}
Recall that $\chi$ is a nontrivial character of $\F_p$,
hence the map $(a_k|b_k)\mapsto \chi(\tr(d_ka_k-b_kc_k))$ is a
nontrivial character of $\F_q^2$ for all $(c_k|d_k)\neq (0|0)$.
Therefore, we can simplify the inner sum to
\begin{equation*}
\begin{split}
\sum_{(a_k|b_k)\in \F_q^2}  z^{\swt(a_k|b_k)}
\chi\left( \tr(d_ka_k-b_kc_k)\right) 
 =\left\{ \begin{array}{l@{\,}l}
1+(q^2-1)z & \text{ if } (c_k|d_k)=(0,0),\\
1-z & \text{ if } (c_k|d_k)\neq (0,0).
\end{array}\right.
\end{split}
\end{equation*}
It follows that
$$
\sum_{b\in \F_q^{2n}} \chi_b(c)z^{\swt(b)}  =
(1-z)^{\swt(c)}(1+(q^2-1)z)^{n-\swt(c)}.$$
Substituting this expression into equation~(\ref{eq:mac}), we find that
$$
\begin{array}{lcl}
B(z) &=&\ds |C|^{-1}
\sum_{c\in C}\sum_{b\in \F_q^{2n}} \chi_b(c)z^{\swt(b)} \\
&=& \ds\frac{(1+(q^2-1)z)^n}{|C|}
\sum_{c\in C} \left(\frac{1-z}{1+(q^2-1)z}\right)^{\swt(c)}\\
&=& \ds\frac{(1+(q^2-1)z)^n}{|C|}\, A\!\left(\frac{1-z}{1+(q^2-1)z}\right),
\end{array}
$$
which proves the claim.
\end{proof}

The coefficient of $z^j$ in $(1+(q^2-1)z)^{n-x}(1-z)^x$ is given by
the Krawtchouk polynomial of degree $j$ in the variable $x$,
$$ K_j(x) =
\sum_{s=0}^j (-1)^s(q^2-1)^{j-s} {x \choose s}{n-x \choose j-s}.$$
\begin{corollary}\label{th:krawtchouk}
Keeping the notation of the previous theorem, we have
$$ B_j = \frac{1}{|C|}\sum_{x=0}^n K_j(x)A_x.$$
\end{corollary}
\begin{proof}
According to the previous theorem, we have
$$ \begin{array}{lcl}
B(z) &=& \displaystyle \frac{(1+(q^2-1)z)^n}{|C|} A
\left(\frac{1-z}{1+(q^2-1)z)}\right)\\
&=& \displaystyle\frac{1}{|C|} \sum_{x=0}^n A_x (1-z)^x(1+(q^2-1)z)^{n-x}.
   \end{array}
$$
We obtain the result by comparing the coefficients of $z^j$ on both sides.
\end{proof}

The weight enumerators turn out to be very useful in establishing the
bounds on quantum codes, as we will see in the next section.

\section{Bounds}\label{sec:bounds}
We need some bounds on the achievable minimum distance of a quantum
stabilizer code.  The main results in this section are the
generalization of the linear programming bounds \cite{calderbank98},
alternative proofs for the nonbinary quantum Singleton bound using a
generalization of the methods given in \cite{ashikhmin99}, a proof of
the validity of the quantum Hamming bound for single error-correcting
(degenerate) quantum codes (which generalizes an earlier result by
Gottesman \cite[Chapter~7]{gottesman97}), a simpler nonconstructive
proof for lower bounds on quantum codes, and an existence proof of a
class of optimal quantum codes.
\subsection{Upper Bounds}
We shall derive a series of upper bounds for nonbinary stabilizer 
codes. The first theorem yields a bound that is well-suited for computer
search.
\begin{theorem}
If an $((n,K,d))_q$ stabilizer code with $K>1$ exists, then there
exists a solution to the optimization problem: minimize
$\sum_{j=1}^{d-1} A_j$ subject to the constraints
\begin{enumerate}
\item $A_0=1$ and $A_j\ge 0$ for all $1\le j\le n$;
\item  $\ds\sum_{j=0}^n A_j = q^n/K$;
\item $B_j = \ds\frac{K}{q^n} \sum_{r=0}^n K_j(r)A_r$ holds for all $j$ in the range $0\le j\le n$;
\item $A_j=B_j$ for all $j$ in $0\le j<d$ and $A_j\le B_j$ for all $d\le j\le n$;
\item $(p-1)$ divides $A_j$ for all $j$ in the range $1\le j\le n$.
\end{enumerate}
\end{theorem}
\begin{proof}
If an $((n,K,d))_q$ stabilizer code exists, then the symplectic weight
distribution of the associated additive code~$C$ satisfies
conditions~1) and~2).  For each nonzero codeword $c$ in $C$, $\alpha
c$ is again in $C$ for all $\alpha$ in $\F_p^*$, so 5)
holds. Corollary~\ref{th:krawtchouk} shows that 3) holds. Since the
quantum code has minimum distance $d$, it follows that 4) holds.  
\end{proof}

\begin{remark}
If we are interested in bounds for $\F_{q^2}$ \-linear codes, then we
can replace condition 5) in the previous theorem by $q^2-1$ divides
$A_j$ for $1\leq j\leq n$.  This will even help in characteristic 2.
\end{remark}

The next bound is more convenient when one wants to find bounds by
hand. In particular, any function $f$ satisfying the constraints of
the next theorem will yield a useful bound on the dimension of a
stabilizer code. This approach was introduced by Delsarte for
classical codes~\cite{delsarte72}. Binary versions of
Theorem~\ref{th:lp2} and Corollary~\ref{th:singleton} were proved by
Ashikhmin and Litsyn~\cite{ashikhmin99}, see also~\cite{ashikhmin00b}.

\begin{theorem}\label{th:lp2}
Let $Q$ be an $((n,K,d))_q$ stabilizer code of dimension $K>1$.
Suppose that $S$ is a nonempty subset of $\{0,\dots,d-1\}$ and
$N=\{0,\dots,n\}$.  Let
$$ f(x)= \sum_{i=0}^n f_i K_i(x)$$
be a polynomial satisfying the conditions
\begin{enumerate}
\item[i)] $f_x> 0$ for all $x$ in $S$, and $f_x\ge 0$ otherwise;
\item[ii)] $f(x)\le 0$ for all $x$ in $N\setminus S$.
\end{enumerate}
Then
$$ K \le \frac{1}{q^n}\max_{x\in S} \frac{f(x)}{f_x}.$$
\end{theorem}
\begin{proof}
Suppose that $C\le \F_q^{2n}$ is the additive code associated with the
stabilizer code $Q$. If we apply Corollary~\ref{th:krawtchouk} to the
trace-symplectic dual code $C^\sdual$ of the code $C$, then we obtain
$$ A_i = \frac{1}{|C^{\sdual}|}\sum_{x=0}^n K_i(x)B_x.$$
Using this relation, we find that
$$ \begin{array}{lcl}
\ds|C^\sdual| \sum_{i\in S} f_i A_i &\le&
\ds |C^\sdual| \sum_{i=0}^{n} f_i A_i \\
&=& \ds |C^\sdual| \sum_{i=0}^{n} f_i
\left(\frac{1}{|C^{\sdual}|}\sum_{x=0}^n K_i(x)B_x\right)\\
&=& \ds \sum_{x=0}^n B_x \sum_{i=0}^n f_iK_i(x).
   \end{array}
$$
By assumption, $f(x)=\sum_{i=0}^n f_iK_i(x)$; thus,
we can simplify the latter inequality and obtain
$$
\ds|C^\sdual| \sum_{i\in S} f_i A_i \le
\ds \sum_{x=0}^n B_x f(x) \le
\sum_{x\in S} B_x f(x) = \sum_{x\in S} A_x f(x),$$
where the last equality follows from the fact that the stabilizer code
has minimum distance $d$, meaning that $A_x=B_x$ holds for all $x$ in the
range $0\le x<d$. We can conclude that
$$ |C^\sdual| \le \
\left(\ds\sum_{x\in S} A_x f(x)\right)\bigg/
\left(\ds\sum_{x\in S} f_x A_x\right)
\le \max_{x\in S} \frac{f(x)}{f_x},
$$
which proves the theorem, since $|C^\sdual|=q^n K$.
\end{proof}

The previous theorem implies the quantum Singleton bound. 
In general, linear programming yields better
bounds, but for short lengths one can
actually find codes meeting the quantum Singleton bound.
\begin{corollary}[Quantum Singleton Bound]\label{th:singleton}
An  $((n,K,d))_q$ stabilizer code with $K>1$ satisfies
$$ K\le q^{n-2d+2}.$$
\end{corollary}
The binary version of the quantum Singleton bound was first proved by
Knill and Laflamme in~\cite{KnLa97}, see
also~\cite{ashikhmin99,ashikhmin00b}, and later generalized by Rains
using weight enumerators in~\cite{rains99}.

A more interesting application of Theorem~\ref{th:lp2} is to derive 
the quantum Hamming bound.
The quantum Hamming bound states that any pure $((n,K,d))_q$
stabilizer code satisfies
\begin{eqnarray}
\sum_{i=0}^{\lfloor (d-1)/2\rfloor} \binom{n}{i}(q^2-1)^i \le
q^n/K,
\end{eqnarray}
 see~\cite{gottesman96,feng04}. Several researchers have tried
to find impure stabilizer codes that beat the quantum Hamming bound.
However, Gottesman has shown that impure single and double
error-correcting binary quantum codes cannot beat the quantum Hamming
bound~\cite{gottesman97}.  In the same vein, Theorem~\ref{th:lp2}
allows us to derive the Hamming bound for arbitrary stabilizer codes,
at least when the minimum distance is small. We illustrate the method
for single error-correcting codes, and note that the same approach
works for double error-correcting codes as well.

\begin{corollary}[Quantum Hamming Bound]\label{th:hamming}
An $((n,K,3))_q$ stabilizer code with $K>1$ satisfies
$$ K\le q^{n}\big/(n(q^2-1)+1).$$
\end{corollary}
\begin{proof}
Recall that the intersection number $p_{ij}^k$ of the Hamming
association scheme $H(n,q^2)$ is the integer
$ p_{ij}^k = | \{ z\in \F_{q^2}^n \,|\, d(x,z)=i, d(y,z)=j\}|,$
where $x$ and $y$ are two vectors in $\F_q^n$ of Hamming
distance~$d(x,y)=k$. The intersection numbers are related to
Krawtchouk polynomials by the expression
$$ p_{ij}^k = q^{-2n} \sum_{u=0}^n K_i^n(u)K_j^n(u)K_u^n(k),$$
see~\cite{barg00}.

After this preparation, we can proceed to derive the Hamming bound as
a consequence of Theorem~\ref{th:lp2}. Let
\begin{eqnarray*}
f(x) &=& \sum_{j,k=0}^1 \sum_{i=0}^n K_j^n(i)K_k^n(i)K_i^n(x),\\
& =&q^{2n}(p_{00}^x+ p_{10}^x+ p_{01}^x + p_{11}^x).
\end{eqnarray*}
The triangle
inequality implies that $p_{ij}^k=0$ if one of the three arguments
exceeds the sum of the other two; hence, $f(x)=0$ for $x>2$.
The coefficients of the Krawtchouk expansion
$f(x)=\sum_{i=0}^n f_i K_i(x)$ obviously satisfy $f_i=(K_0(i)+K_1(i))^2\ge 0$.
A straightforward
calculation gives
$$
\begin{array}{l@{\;}l}
f(0)= q^{2n}(n(q^2-1)+1), & f_0 = (n(q^2-1)+1)^2,\\
f(1)= q^{2n+2},           & f_1 = ((n-1)(q^2-1))^2,\\
f(2)= 2q^{2n},            & f_2 = ((n-2)(q^2-1)-1)^2.
\end{array}
$$
It follows that
$$ \max\{ f(0)/f_0, f(1)/f_1, f(2)/f_2\} \le q^{2n}/(n(q^2-1)+1)$$
holds for all $n\ge 5$. Using Theorem~\ref{th:lp2}, we obtain the
claim for all $n\ge 5$. For the lengths $n<5$, we obtain the claim
from the quantum Singleton bound.
\end{proof}

One real disadvantage of Theorem~\ref{th:lp2} is that the number of
terms increase with the minimum distance and this can lead to
cumbersome calculations. However, one can derive more consequences
from Theorem~\ref{th:lp2}; see, for
instance,~\cite{ashikhmin99,ashikhmin00b,levenshtein95,mceliece77}.

\subsection{Lower Bounds}
Feng and Ma have recently shown a quantum version of the classical
lower bounds by Gilbert and Varshamov \cite{feng04}. We conclude this
section by giving a simple proof for a weaker version of this result
based on a counting argument. It must be remembered that these lower
bounds are nonconstructive.

Our first lemma generalizes an idea used by Gottesman in his proof of
the binary case.
\begin{lemma}\label{th:gilbert}
An $((n,K,\ge d))_q$ stabilizer code with $K>1$ exists provided that
\begin{equation}\label{eq:gilbert}
(q^nK-q^n/K) \sum_{j=1}^{d-1} \binom{n}{j}(q^{2}-1)^j
<(q^{2n}-1)(p-1)
\end{equation}
holds.
\end{lemma}
\begin{proof}
Let $L$ denote the multiset
$$L=\{ C^\sdual\setminus C\,|\, C\le C^\sdual\le \F_q^{2n} \text{ with
} |C|=q^n/K\}.$$ The elements of this multiset correspond to
stabilizer codes of dimension $K$.  Note that $L$ is nonempty, since
there exists a code $C$ of size $q^n/K$ that is generated by elements
of the form $(a|0)$; the form of the generators ensures that $C\le
C^\sdual$.

All nonzero vectors in $\F_q^{2n}$ appear in the same number of sets
in $L$. Indeed, the symplectic group $\textup{Sp}(2n,\F_q)$ acts
transitively on the set $\F_{q}^{2n}\setminus \{ 0\}$,
see~\cite[Proposition~3.2]{grove01}, which means that for any nonzero
vectors $u$ and $v$ in $\F_q^{2n}$ there exists $\tau\in
\textup{Sp}(2n,\F_q)$ such that $v=\tau u$. Therefore, $u$ is
contained in $C^\sdual\setminus C$ if and only if $v$ is contained in
the element $(\tau C)^\sdual\setminus \tau C$ of $L$.

The transitivity argument shows that any nonzero vector in
$\F_{q}^{2n}$ occurs in $|L|(q^nK-q^n/K)/(q^{2n}-1)$ elements of $L$.
Furthermore, a nonzero vector and its $\F_p^\times$-multiples are contained
in the exact same sets of $L$.
Thus, if we delete all sets from $L$
that contain a nonzero vector with symplectic weight less than~$d$,
then we remove at most
$$
\frac{\sum_{j=1}^{d-1} \binom{n}{j}(q^{2}-1)^j}{p-1}
|L|\frac{(q^nK-q^n/K)}{q^{2n}-1}$$ sets from $L$. By assumption, this
number is less than $|L|$; hence, there exists an $((n,K,\ge d))_q$
stabilizer code.
\end{proof}

The Gilbert-Varshamov bound shows the existence of surprisingly good codes, even for
smaller lengths, when the characteristic of the field is
not too small. If $n\equiv k\bmod 2$, then we can significantly
strengthen the bound. 

\begin{lemma}\label{th:lingilbert}
If $k\ge 1$, $n\equiv k\bmod 2$ and
\begin{equation}\label{eq:lingilbert}
(q^{n+k}-q^{n-k}) \sum_{j=1}^{d-1} \binom{n}{j}(q^{2}-1)^{j-1}
<(q^{2n}-1)
\end{equation}
holds, then there exists an $\F_{q^2}$-linear $[[n,k,d]]_{q}$ stabilizer code.
\end{lemma}
\begin{proof}
The proof is almost the same as in the previous lemma, except that we
list only codes $C$ such that $\phi(C)$ is linear, meaning that $\phi(C)$ is a
vector space over $\F_{q^2}$. We repeat the previous argument with the
multiset
$$ L = \left\{ C^\sdual\setminus C\, \Bigg|\,\begin{array}{l} C\le C^\sdual \le F_q^{2n},
|C|=q^{n-k},\\ \phi(C) \text{ is $\F_{q^2}$-linear } 
\end{array} \right\}. $$
It is easy to see that $L$ is not empty.  
Note that
each set $\phi(C^\sdual) \setminus \phi(C)$ in $L$ contains now all
$\F_{q^2}^\times$-multiples of a nonzero vector, not just the
$\F_p^\times$-multiples, which proves the statement.
\end{proof}

Feng and Ma show that one can extend the previous
result to even prove the existence of pure stabilizer codes, but much
more delicate counting arguments are needed in that case,
see~\cite{feng04}. We are not aware of short proofs for this stronger
result.

The previous lemma allows us to show the existence of good quantum codes, especially
for larger alphabets. We illustrate this fact by proving the existence
of MDS stabilizer codes, see Section~\ref{sec:MDS} for more details on
such codes. 

\begin{corollary}
If $2\le d\le \lceil n/2\rceil $ and $q^2-1\ge \binom{n}{d}$, then
there exists a linear $[[n,n-2d+2,d]]_q$ stabilizer code.
\end{corollary}
\begin{proof}
The assumption $d\le \lceil n/2\rceil $ implies that $\binom{n}{1}\le
\binom{n}{2}\le \cdots \le \binom{n}{d}$, so the maximum value of
these binomial coefficients is at most $q^2-1$. Let $k=n-2d+2$. It
follows from the assumption that $k\ge 1$ and $n\equiv k\bmod
2$. It remains to show that (\ref{eq:lingilbert}) holds.
For the choice $k=n-2d+2$,
the left hand side of (\ref{eq:lingilbert}) equals
\begin{equation*}
\begin{split}
\ds(q^{2n-2d+2}-&q^{2d-2})\sum_{j=1}^{d-1} \binom{n}{j} (q^2-1)^{j-1}\\
&\le \ds (q^{2n-2d+2}-q^{2d-2})\sum_{j=1}^{d-1} (q^2-1)^{j}\\
&=\ds(q^{2n-2d+2}-q^{2d-2}) \frac{(q^2-1)^d-(q^2-1)}{q^2-2}.
\end{split}
\end{equation*}
We claim that the latter term is less than $q^{2n}-1$.
To prove this, it suffices to show that
\begin{equation}\label{eq:ineq}
q^{2n-2d+2}\frac{(q^2-1)^d-(q^2-1)}{q^2-2}\le q^{2n}
\end{equation}
holds. The latter inequality is equivalent to $ (q^2-1)^d \le
q^{2d}-2q^{2d-2}+q^2-1$, and it is not hard to see that this
inequality holds.  Indeed, note that
$$ q^{2d}=((q^2-1)+1)^d = (q^2-1)^d+ \sum_{j=0}^{d-1} \binom{d}{j}(q^2-1)^j.$$\
Recall that $\binom{d}{j}=\binom{d-1}{j-1}+\binom{d-1}{j}$; hence,
\begin{equation*} 
\begin{split}
q^{2d}-2q^{2d-2}- (q^2&-1)^d \\&= \sum_{j=0}^{d-1}\big(
\binom{d}{j}-2\binom{d-1}{j} \big)(q^2-1)^j,\\
&= \sum_{j=0}^{d-1}\big(
\underbrace{\binom{d-1}{j-1}-\binom{d-1}{j}}_{\alpha(j):=}\big)(q^2-1)^j.
\end{split}
\end{equation*}
We have $\alpha(j)=-\alpha(d-j)$ for $0\le j\le d-1$, and
$\alpha(j)\ge 0$ for $j\ge d/2$. This shows that all negative terms
get canceled by larger positive terms and we can conclude that
$q^{2d}-2q^{2d-2}-(q^2-1)^d\ge 0$ for $d\ge 2$; this implies
inequality~(\ref{eq:ineq}) and consequently shows
that~(\ref{eq:lingilbert}) holds.
\end{proof}

\begin{example} 
Recall that there does not exist a $[[7,1,4]]_2$ code,
see~\cite{calderbank98}. In contrast, the existence of a $[[7,1,4]]_q$
code for all prime powers $q\ge 7$ is guaranteed by the preceding
corollary. It also shows that there exist $[[6,2,3]]_q$ for all prime
powers $q\ge 5$ and $[[7,3,3]]_q$ for all prime powers $q\ge 7$, which
slightly generalizes~\cite{feng02}.
\end{example}

\section{Code Constructions}\label{sec:const}
Constructing good quantum codes is a difficult task. We need a quantum
code for each parameter $n$ and $k$ in our tables. In this section we collect
some simple facts about the construction of
codes. Lemmas~\ref{th:lengthening}--\ref{th:smallerdim}, (see also Table~\ref{table:cc}), show how to
lengthen, shorten or reduce the dimension of the stabilizer code. These generalize and
extend the constructions for binary quantum codes \cite[Theorem~6]{calderbank98}.

{
\begin{table}[htb]
\caption{The existence of a pure $[[n,k,d]]_q$ stabilizer code implies the existence of codes with other parameters.}\label{table:cc}
\begin{center}
\begin{tabular}{c||c|c|c}
n/k & $k-1$ & $k$ & $k+1$ \\
\hline\hline
$n-1$& \stacked{$\ge d-1$ pure}{Lemma~\ref{th:smallerdim}}   &  \stacked{$\ge d-1$ pure}{Lemma~\ref{th:smallerdim}}         & \stacked{$d-1$ pure}{Lemma~\ref{th:shorterlength}}     \\
\hline
$n$  & \stacked{$\ge d$ pure}{Lemma~\ref{th:smallerdim}} & \fbox{$d$ pure}  & \stacked{$d-1$ impure}{Lemma~\ref{th:lengthening}}     \\
\hline
$n+1$& \stacked{$\ge d$ impure}{Lemma~\ref{th:lengthening}} & \stacked{$d$ impure}{Lemma~\ref{th:lengthening}}& 
\end{tabular}
\end{center}
\end{table}
}

\begin{lemma}\label{th:lengthening}
If an $[[n,k,d]]_q$ stabilizer code exists for $k>0$, then there
exists an impure $[[n+1,k,d]]_q$ stabilizer code.
\end{lemma}
\begin{proof}
If an $[[n,k,d]]_q$ stabilizer code exists, then there exists an
additive subcode $C\le \F_q^{2n}$ such that $|C|=q^{n-k}$, $C\le
C^\sdual$, and $\swt(C^\sdual\setminus C)=d$. Define the additive code
$$ C' = \{ (a\alpha|b0)\,|\, \alpha\in \F_q, (a|b)\in C \}.$$
We have $|C'|=q^{n-k+1}$. The definition ensures that $C'$ is
self-orthogonal with respect to the trace-symplectic inner
product. Indeed, two arbitrary elements $(a\alpha|b0)$ and
$(a'\alpha'|b'0)$ of $C'$ satisfy the orthogonality condition
$$ \< (a\alpha|b0) | (a'\alpha'|b'0)>s = \< (a|b) | (a'|b')>s +
\tr(\alpha\cdot 0 - \alpha'\cdot 0) = 0.$$ A vector in the
trace-symplectic dual of $C'$ has to be of the form $(a\alpha|b0)$ with
$(a|b)\in C^\sdual$ and $\alpha\in \F_q$.
Furthermore,
$$ \swt(C'^\sdual \setminus C') = \min\{ \swt(a\alpha|b0)\,|\,
\alpha\in \F_q, a,b\in C^\sdual\setminus C\},$$ which coincides with
$\swt(C^\sdual\setminus C).$ Therefore, an $[[n+1,k,d]]_q$ stabilizer
code exists by Theorem~\ref{th:stabilizer}.  If $d>1$, then the code
is impure, because $C'^\sdual$ contains the vector
$(\mbf{0}\alpha|\mbf{0}0)$ of symplectic weight 1.
\end{proof}

\begin{lemma}\label{th:shorterlength}
If a pure $[[n,k,d]]_q$ stabilizer code exists with $n\ge 2$ and $d\ge
2$, then there exists a pure $[[n-1,k+1,d-1]]_q$ stabilizer code.
\end{lemma}
\begin{proof}
If a pure $[[n,k,d]]_q$ stabilizer code exists, then there exists an
additive code $D\le \F_{q^2}^n$ that is self-orthogonal with respect
to the trace-alternating form, so that $|D|=q^{n-k}$ and
$\wt(D^\adual)=d$.  Let $D_0^\adual$ denote the code obtained by
puncturing the first coordinate of $D^\adual$. Since the minimum
distance of $D^\adual$ is at least 2, we know that
$|D_0^\adual|=|D^\adual|=q^{n+k}$, and we note that the minimum
distance of $D_0^\adual$ is $d-1$. The dual of $D_0^\adual$ consists
of all vectors $u$ in $\F_{q^2}^{n-1}$ such that $0u$ is contained
in~$D$. Furthermore, if $u$ is an element of $D_0$, then $0u$ is
contained in $D$; hence, $D_0$ is a self-orthogonal additive code.  The
code $D_0$ is of size $q^{(n-1)-(k+1)}$, because
$$\dim D_0+ \dim D_0^\adual = \dim \F_{q^2}^{n-1}$$
when we view $D_0$ and its dual as $\F_p$--vector spaces.
It follows that there exists a pure $[[n-1,k+1,d-1]]_q$ stabilizer code.
\end{proof}

\begin{lemma}\label{th:smallerdim}
If a (pure) $[[n,k,d]]_q$ stabilizer code exists, with $k\ge 2$ ($k\ge
1$), then there exists 
an  $[[n, k-1,d^*]]_q$ stabilizer code (pure to $d$) such that $d^*\ge d$.
\end{lemma}
\begin{proof}
If an $[[n,k,d]]_q$ stabilizer code exists, then there exists an
additive code $D\le \F_{q^2}^n$ such that $D\le D^\adual$ with
$\wt(D^\adual\setminus D)=d$ and $|D|=q^{n-k}$. Choose an additive
code~$D_b$ of size $|D_b|=q^{n-k+1}$ such that $D\le D_b \le D_b^\adual \le D^\adual$.
Since $D\le D_b$, we have $D_b^\adual\le D^\adual$.  The set $\Sigma_b
= D_b^\adual \setminus D_b$ is a subset of $D^\adual \setminus D$,
hence the minimum weight $d^*$ of $\Sigma_b$ is at least $d$.
This proves the existence of an $[[n,k-1,d^*]]$ code.

If the code is pure, then $\wt(D^\adual)=d$; it follows from
$D_b^\adual \le D^\adual$ that $\wt(D_b^\adual)\ge d$, so the smaller
code is pure as well.
\end{proof}

\begin{corollary}
If a pure $[[n,k,d]]_q$ stabilizer code with $n\ge 2$ and
$d\ge 2$ exists, then there exists a pure $[[n-1,k,\ge d-1]]_q$ stabilizer
code.
\end{corollary}
\begin{proof}
Combine Lemmas~\ref{th:shorterlength} and \ref{th:smallerdim}.
\end{proof}

\begin{lemma}\label{th:directsum}
Suppose that an $((n,K,d))_q$ and an $((n',K',d'))_q$ stabilizer code
exist.  Then there exists an $((n+n', KK',\min(d,d'))_q$ stabilizer
code.
\end{lemma}
\begin{proof}
Suppose that $P$ and $P'$ are the orthogonal projectors onto the
stabilizer codes for the $((n,K,d))_q$ and $((n',K',d'))_q$ stabilizer
codes, respectively. Then $P\otimes P'$ is an orthogonal projector
onto a $KK'$-dimensional subspace $Q^*$ of $\C^d$, where $d=q^{n+n'}$.
Let $S$ and $S'$ respectively denote the stabilizer groups of the
images of $P$ and $P'$. Then $S^*=\{ E\otimes E'\,|\, E\in S, E'\in
S'\}$ is the stabilizer group of $Q^*$.

If an element $F\otimes F^*$ of $G_n\otimes G_{n'}=G_{n+n'}$ is not
detectable, then $F$ has to commute with all elements in $S$, and $F'$
has to commute with all elements in $S'$.  It is not possible that both
$F\in Z(G_n)S$ and $F'\in Z(G_{n'})S'$ hold, because this would imply that
$F\otimes F'$ is detectable. Therefore, either $F$ or $F'$ is not
detectable, which shows that the weight of $F\otimes F'$ is at least
$\min(d,d')$.
\end{proof}

\begin{lemma}
Let $Q_1$ and $Q_2$ be pure stabilizer codes that respectively have parameters
$[[n,k_1,d_1]]_q$ and $[[n,k_2,d_2]]$. If\/
$Q_2\subseteq Q_1$, then there exists a $[[2n,k_1+k_2,d]]_q$
pure stabilizer code with minimum distance $d\ge \min\{ 2d_2,d_1\}$.
\end{lemma}
\begin{proof}
The hypothesis implies that there exist
additive subcodes $D_1\le D_2$ of\/ $\F_{q^2}^n$ such that
 $D_m\le D_m^\adual$, $|D_m|=q^{n-k_m}$, and
$\wt(D_m^\adual)=d_m$ for $m=1,2$.  The additive code
$$ D = \{ (u,u+v)\;|\, u\in D_1, v\in D_2\}\le \F_{q^2}^{2n}$$ is of
size $|D|=q^{2n-(k_1+k_2)}$. The trace-alternating dual of the code
$D$ is $D^\adual = \{ (u'+v',v')\,|\, u'\in D_1^\adual, v'\in
D_2^\adual\}$.  Indeed, the vectors on the right hand side are
perpendicular to the vectors in $D$, because
$$ \( (u,u+v)\, |\, (u'+v',v') )a =
\( u| u'+v')a + \( u+v|v')a = 0
$$
holds for all $u\in D_1, v\in D_2$ and $u'\in D_1^\adual, v'\in
D_2^\adual$.  We observe that $D$ is self-orthogonal, $D\le D^\adual$.
The weight of a vector $(u'+v',v')\in D^\adual\setminus D$ is at least
$\min\{ 2d_2,d_1\}$; the claim follows.
\end{proof}

\begin{lemma}
\label{th:difference}
Let $q$ be a power of two.  If a pure $[[n,k_1,d_1]]_q$ stabilizer
code~$Q_1$ exists that has a pure subcode $Q_2\subseteq Q_1$ with
parameters $[[n,k_2,d_2]]_q$ such that $k_1>k_2$, then a
pure $[[2n,k_1-k_2,d]]_q$ stabilizer code exists such that $d \ge
\min{\{2d_1, d_2\}}$.
\end{lemma}
\begin{proof}
If an $[[n_m,k_m,d_m]]_q$ stabilizer code exists, then there exists an
additive code $D_m\le \F_{q^2}^n$ such that $D_m\le D_m^\adual$,
$\wt(D_m^\adual)=d$, and $|D_m|=q^{n-k_m}$ for $m=1,2$.
The inclusion $Q_2\subseteq Q_1$ implies that $D_1\le D_2$.
Let $D$ denote the additive code consisting of vectors of the form
$(u,u+v)$ such that $u \in D_2^\adual$ and $v \in D_1$.

We claim that $D^\adual$ consists of vectors of the form $(u',u'+v')$
such that $u' \in D_1^\adual$ and $v' \in D_2$.  Indeed, let
$v_1=(u,u+v)$ denote a vector in~$D$, and let $v_2 = (u',u'+v')$ be a
vector with $u' \in D_1^\adual$ and $v' \in D_2$.  We have
$$\(v_1|v_2)a= \(u|u')a + \(u|u')a + \(u|v')a + \(v|u')a + \(v|v')a.$$
The first two terms on the right hand side cancel because the
characteristic of the field is even; the next two terms vanish since
the vectors belong to dual spaces; the last term vanishes because $v$
and $v'$ are both contained in $D_2$, and $D_2$ is self-orthogonal.
Therefore, $v_1$ and $v_2$ are orthogonal. The set $\{ (u',u'+v')\,|\,
u' \in D_1^\adual, v' \in D_2\}\subseteq D^\adual$ has cardinality
$q^{2n+k_1-k_2}$, so it must be equal to $D^\adual$ by a dimension
argument.

The Hamming weight of a vector $(u',u'+v')$ in $D^\adual$
is at least $\min{\{2d_1, d_2\}}$, because
$u' \in D_1^\adual$ and $v'\in D_2\le D_2^\adual$.
\end{proof}

\begin{lemma} \label{th:codeexpansion} 
Let $q$ be a power of a prime.  If an $((n,K,d))_{q^m}$ stabilizer
code exists, then an $((nm,K,\ge d))_q$ stabilizer code
exists. Conversely, if an $((nm,K,d))_q$ stabilizer code exists,
then there exists an $((n,K,\ge \lfloor d/m\rfloor))_{q^m}$ stabilizer code.
\end{lemma}
This lemma is implicitly contained in the paper by Ashikhmin and
Knill~\cite{ashikhmin01}. 

\begin{proof}
Let $B=\{\beta_1,\dots,\beta_m\}$ denote a basis of
$\F_{q^m}/\F_q$. If $a$ is an element of $\F_{q^m}$, then we denote 
by $e_{\mbox{\tiny{B}}}(a)$ the coordinate vector in $\F_q^m$ given by 
$ e_{\mbox{\tiny{B}}}(a) = (a_1,\dots,a_m)$, where $a = \sum_{i=1}^m
a_i\beta_i.$ 

A nondegenerate symmetric form on the $\F_q$-vector space $\F_{q^m}$
is given by $\tr_{q^m/q}(xy)$. It follows that the Gram matrix
$M=(\tr_{q^m/q}(\beta_i\beta_j))_{1\le i,j\le m}$ is nonsingular.
We have $\tr_{q^m/q}(xy)= e_{\mbox{\tiny B}}(x)^tMe_{\mbox{\tiny B}}(y)$ for all $x, y$ in $\F_{q^m}$. 
We define an $\F_p$--vector space isomorphism $\varphi_{\mbox{\tiny B}}$ from
$\F_{q^m}^{2n}$ onto $\F_q^{2nm}$ by 
$$ \varphi_{\mbox{\tiny B}}((a|b)) = ((e_{\mbox{\tiny B}}(a_1),\dots,
e_{\mbox{\tiny B}}(a_n))|(Me_{\mbox{\tiny B}}(b_1),\dots,Me_{\mbox{\tiny B}}(b_n))).$$ It follows from the fact
that $\tr_{q^m/q}(\tr_{q/p}(x))=\tr_{q^m/p}(x)$ for all $x$ in
$\F_{q^m}$ and the definition of the isomorphism $\varphi_{\mbox{\tiny B}}$ that
$(a|b)\, \sdual\, (c|d)$ holds in $\F_{q^m}^{2n}$ if and only if
$\varphi_{\mbox{\tiny B}}((a|b)) \,\sdual\, \varphi_{\mbox{\tiny B}}((c|d))$ holds in
$\F_{q^{2nm}}$.

If an $((n,K,d))_{q^m}$ exists, then there exists an additive code
$C\le \F_{q^m}^{2n}$ of size $|C|=q^{nm}/K$ such that $C\le C^\sdual$,
$\swt(C^\sdual\setminus C)=d$ if $K>1$, and $\swt(C^\sdual)=d$ if
$K=1$. Therefore, the code $\varphi_{\mbox{\tiny B}}(C)$ over the alphabet $\F_q$ is
of size $q^{nm}/K$, satisfies $\varphi_{\mbox{\tiny B}}(C)\le \varphi_{\mbox{\tiny B}}(C)^\sdual\le
\F_q^{2nm}$, and $\swt(\varphi_{\mbox{\tiny B}}(C)^\sdual\setminus \varphi_{\mbox{\tiny B}}(C))=d$
if $K>1$ and $\swt(\varphi_{\mbox{\tiny B}}(C)^\sdual)=d$ if $K=1$. Thus, an
$((nm,K,d))_q$ stabilizer code exists. 

The existence of an $((nm,K,d))_q$ stabilizer code implies the
existence of an $((n,K))_{q^m}$ stabilizer code; the claim about the minimum
distance follows from the fact that $\varphi_{\mbox{\tiny B}}^{-1}$ maps each nonzero
block of $m$ symbols to a nonzero symbol in $\F_{q^m}$.
\end{proof}

We notice that there exists a basis $B$ such that $M$ is the identity
matrix if and only if either $q$ is even or both $q$ and $m$ are odd,
see~\cite{seroussi80}. In that
case, $\varphi_{\mbox{\tiny B}}$ simply expands each symbol into coordinates with
respect to~$B$.

\section{Puncturing Stabilizer Codes}
If we delete one coordinate in all codewords of a classical code, then
we obtain a shorter code that is called the punctured code. In
general, we cannot proceed in the same way with stabilizer codes,
since the resulting matrices might not commute if we delete one or
more tensor components. 

Rains~\cite{rains99} invented an interesting approach that solves the
puncturing problem for linear stabilizer codes and, even better, gives
a way to construct stabilizer codes from arbitrary linear codes. The
idea is to associate with a classical linear code a so-called puncture
code; if the puncture code contains a codeword of weight $r$, then a
self-orthogonal code of length $r$ exists and the minimum distance is
the same or higher than that of the initial classical code.  Further
convenient criteria for puncture codes are given in~\cite{grassl04}.

In this section, we generalize puncturing to arbitrary stabilizer
codes and review some known facts. Determining a puncture code is a
challenging task, and maynot always possible to find it in closed form.
In the next chapter we show how the results of this section can be
applied to puncture quantum BCH codes. 

It will be convenient to denote the the pointwise product of two
vectors $u$ and $v$ in $\F_q^n$ by $uv$, that is, $uv=(u_iv_i)_{i=1}^n$. 
Suppose that $C\le \F_q^{2n}$ is an arbitrary additive code. The
associated puncture code $\puncture_s(C) \subseteq \F_q^n$ is defined
as
\begin{eqnarray}
 \puncture_s(C)= \left\{ (b_k a_k' - b_k'a_k)_{k=1}^n\,
| \,(a|b), (a'|b')\in C\right\}^\perp.
\end{eqnarray}

\begin{theorem} \label{th:punc_symplectic}
Suppose that $C$ is an arbitrary additive subcode of $\F_q^{2n}$ of
size $|C|=q^n/K$ such that $\swt(C^\sdual\setminus C)=d$.  If the
puncture code $\puncture_s(C)$ contains a codeword of Hamming weight
$r$, then there exists an $((r,K^*,d^*))_q$ stabilizer code with
$K^*\ge K/q^{n-r}$ that has minimum distance $d^* \ge d$ when $K^*>1$. If $\swt(C^\sdual)=d$, then the resulting
punctured stabilizer code is pure to $d$.
\end{theorem}
\begin{proof}
Let $x$ be a codeword of weight $r$ in the $\puncture_s(C)$. Define an
additive code $C_x\le \F_q^{2n}$ by
$$ C_x = \{ (a|bx) \; |\; (a|b)\in C\}.$$ 
If $(a|bx)$ and $(a'|b'x)$ are arbitrary elements of $C_x$, then 
\begin{equation}\label{puncturedual}
\< (a|bx)\; |\;(a'|b'x)>s
= \ds \tr\left(\sum_{k=1}^n (b_ka'_k - b_k'a_k)x_k\right)=0
\end{equation}
by definition of $\puncture_s(C)$; thus, $C_x\le (C_x)^\sdual.$ 

Let $C_x^R= \{ (a_k|b_k)_{k\in S} | (a|b)\in C_x\}$ denote the
restriction of $C_x$ to the support $S$ of the vector $x$. Since
equation (\ref{puncturedual}) depends only on the nonzero coefficients
of the vector $x$, it follows that $C_x^R \le (C_x^R)^{\sdual}$ holds.

We note that $|C|\ge |C_x^R|$; hence, the dimension $K^*$ of the
punctured quantum code is bounded by
\[
 K^* \geq q^r/|C_x^R| \geq q^r/|C|=q^r/(q^n/K)=K/q^{n-r}.
\]

It remains to show that $\swt((C_x^R)^\sdual\setminus C_x^R)\ge d$.
Seeking a contradiction, we suppose that $u_x^R$ is a vector in
$(C_x^R)^\sdual\setminus C_x^R$ such that $\swt(u_x^R)<d$. Let
$u_x=(a|b)$ denote the vector in $(C_x)^\sdual$ that is zero outside
the support of $x$ and coincides with $u_x^R$ when restricted to the
support of $x$. It follows that $(ax|b)$ is contained in $C^\sdual$.
However $\swt(ax|b)<d$, so $(ax|b)$ must be
an element of $C$, since $\swt(C^\sdual\setminus
C)=d$.  This implies that $(ax|bx)$ is an element of $C_x\le
(C_x)^\sdual$. Arguing as before, it follows that $(ax^2|bx)$ is in
$C$ and $(ax^2|bx^2)$ is in $C_x$.  Repeating the process, we obtain
that $v_x=(ax^{q-1}|bx^{q-1})$ is in $C_x$, and we note that $x^{q-1}$
is the characteristic vector of the support of $x$.  Restricting $v_x$
in $C_x$ to the support of $x$ yields $u_x^R\in C_x^R$,
contradicting the assumption that $u_x^R\in (C_x^R)^\sdual\setminus
C_x^R$.

Finally, the last statement concerning the purity is easy to prove (a
direct generalization of the argument given in~\cite{grassl04} for
pure linear codes).
\end{proof}

If the code $C$ is a direct product, as in the case of CSS codes, then 
the expression for the puncture code simplifies somewhat. 
\begin{lemma}
If $C_1$ and $C_2$ are two additive subcodes of $\F_q^n$, then 
$$ \pc_s(C_1\times C_2) = \{ ab\mid a\in C_1, b\in
C_2\}^\perp\le \F_q^n.$$
\end{lemma}
\begin{proof}
Since $\langle ab\mid a\in C_1, b\in C_2\rangle = \langle
(ba'-b'a)\mid a,a'\in C_1, b,b'\in C_2\rangle$, the claim about the
orthogonal complements of these sets is obvious.
\end{proof}

Since many quantum codes are constructed from self-orthogonal codes $C
\le C^\perp$, we write 
\begin{equation}
\pc_e(C) =\pc_s(C\times C)= \{ ab \mid a,b\in C\}^\perp.
\end{equation}

\section{Conclusions}
In this chapter we have further developed the theory of nonbinary stabilizer codes.
After reviewing the basic theory of nonbinary
stabilizer codes over finite fields, we introduced Galois-theoretic
methods to clarify the relation between these and more general quantum
codes. We showed the most general class of codes over quadratic extension fields 
that can be used to construct quantum codes are those that are self-orthogonal 
with respect to the trace alternating product. 
 
We gave simpler proofs for the existence of nonbinary quantum codes.  
We also generalized the linear programming bounds for the nonbinary codes.
Following Gottesman's lead \cite{gottesman97}, we were able to show
that single and double error-correcting nonbinary stabilizer codes
cannot beat the quantum Hamming bound. We conjecture that no quantum
stabilizer code can exceed the quantum Hamming bound, but a
proof is still elusive.
We also gave methods to obtain new quantum codes from existing quantum codes.
In particular, we developed the theory of puncture codes.

There are open questions that the work in this chapter suggests. We
could for instance start with a different
choice of error basis~\cite{knill96a}, and one can develop a similar
theory for such stabilizer codes. For example, one choice leads to
self-orthogonal additive subcodes of $\Z_q^n\times \Z_q^n$ instead of
subcodes of $\F_q^n\times \F_q^n$.  It would be interesting to know
how the stabilizer codes with respect to different error bases
compare. To the best of our knowledge, such a comparison has not been made.

%% file: chStabq2.tex
\chapter{Classes of Stabilizer Codes\footnotemark}\label{ch:stabq2}
\footnotetext{\copyright  2006 IEEE.  Reprinted in part, with permission, 
from A. Ketkar, A. Klappenecker, S. Kumar and P. K. Sarvepalli,
``Nonbinary stabilizer codes over finite fields''. {\em IEEE Trans. Inform. Theory}, 
vol. 52, no. 11, pp.~4892--4914, 2006.}

In this chapter we shall take a constructive approach to our study of
stabilizer codes giving explicit constructions for many classes of 
codes. Much of the theory we developed in Chapter~\ref{ch:stabq1} will
be brought to bearing with additional simplifications for the classes
of linear codes.  In case of linear codes, our main methods of 
constructions will be the Hermitian construction and the CSS construction
(Lemmas \ref{co:classical}--\ref{th:css2}). Hence, we need to look for 
classical codes that are self-orthogonal with respect to the Hermitian or
the Euclidean product or families of nested codes like the BCH codes.
Additionally, we investigate the structural properties of nontrivial codes that 
meet the quantum Singleton bound and  establish bounds on the maximal 
length of such codes. We provide a concrete illustration of the theory
of puncture codes developed in the last chapter by puncturing the quantum
BCH codes.

\section{Quantum Cyclic Codes}
Cyclic codes are an interesting class of codes which have simple encoding
and efficient decoding algorithms. Consequently, quantum cyclic codes
have also generated interest. Before we construct quantum cyclic codes
we need the following results for identifying cyclic codes that
contain their duals. We have not been able to trace the references
that first proved these results, but we note that these conditions
have been established in various forms earlier, especially for codes
over $\F_2$ and $\F_4$; see
\cite[Chapter~4]{huffman03} for general results concerning classical
codes and \cite{calderbank98,grassl97} for results concerning binary
quantum codes. We provide a convenient and unified treatment while
giving the nonbinary equivalents.
 
Recall that a classical cyclic code with parameters $[n,k]_q$ is a
principal ideal in the ring $\F_q[x]/(x^n -1)$ and can be succinctly
described by its generator polynomial or its defining set.  The
polynomial $x^n-1$ of $\F_{q}[x]$ has simple roots if and only if $n$
and $q$ are coprime. If the latter condition is satisfied, then there
exists a positive integer $m$ such that the field $\F_{q^{m}}$
contains a primitive $n$th root of unity $\beta$.  In that case, one
can describe a cyclic code with generator polynomial $g(x)$ in terms
of its defining set $Z=\{ k\,|\, g(\beta^k)=0 \text{ for } 0\le k<
n\}$.  Further details on cyclic codes can be found in any standard
textbook on coding theory, see \cite{huffman03} or
\cite{macwilliams77}.
 
In the case of cyclic codes, identifying the self-orthogonal codes can be
translated into equivalent conditions on the generator polynomial 
of the code or its defining set. First we shall consider codes over $\F_{q^2}$.
Let $\sigma$ denote the automorphism of the field $\F_{q^2}$ given by
$\sigma(x)=x^q$. We can define an action of $\sigma$ on the polynomial
ring $\F_{q^2}[x]$ by
$$h(x)=\sum_{k=0}^n h_k x^k \longmapsto h^\sigma(x)
=\sum_{k=0}^n \sigma(h_k) x^k.$$

\begin{lemma}\label{th:cyclic}
Suppose that $B$ is a classical cyclic $[n,k,d]_{q^2}$ code with
generator polynomial $g(x)$ and check polynomial
$h(x)=(x^n-1)/g(x)$. If $g(x)$ divides $\sigma(h_0)^{-1}x^k
h^\sigma(1/x)$, then $B^\hdual\subseteq B$, and there exists an
$[[n,2k-n,\geq d]]_{q}$ stabilizer code that is pure to~$d$.
\end{lemma}
\begin{proof}
If $h(x)$ is the check polynomial of $B$, then $h^\sigma(x)$ is the
check polynomial of $\sigma(B)$. The generator polynomial of the dual
code $\sigma(B)^\perp=B^\hdual$ is given by $\sigma(h_0)^{-1}x^k
h^\sigma(1/x)$, the normalized reciprocal polynomial of $h^\sigma(x)$.
Therefore, the condition that the polynomial $g(x)$ divides
$\sigma(h_0)^{-1}x^k h^\sigma(1/x)$ is equivalent to the condition
$B^\hdual \subseteq B$. The stabilizer code follows from
Corollary~\ref{co:classical}.
\end{proof}

The following Lemma summarizes various equivalent conditions on
dual containing codes in terms of the generator polynomial $g(x)$ and the
defining set $Z$. 

\begin{lemma}\label{th:cyclic2}
Let $\gcd(n,q^2)=1$ and $C$ be a classical cyclic $[n,k,d]_{q^2}$ code
whose generator polynomial is $g(x)$ and defining set is $Z$. Suppose that any of
the following equivalent conditions are satisfied\\ (i) $x^n - 1
\equiv 0 \mod g(x)g^{*}(x)$ where $g^{*}(x)=x^{n-k}g^\sigma(1/x)$;\\
(ii) $Z\subseteq \{ -qz\,|\, z\in N\setminus Z\}$;\\ (iii) $Z\cap
Z^{-q}=\emptyset $, where $Z^{-q}=\{-qz\mid z\in Z\}$.\\ Then $C^\hdual
\subseteq C $ and there exists an $[[n,2k-n,\geq d]]_{q}$ stabilizer
code that is pure to~$d$.
\end{lemma}

\begin{proof}
Let $h(x)=(x^n-1)/g(x)$ be the check polynomial of $C$. 
Then $h^\sigma(x)=\sigma((x^n-1)/g(x))=(x^n-1)/g^\sigma(x)$. 
From Lemma~\ref{th:cyclic} we know that $C$ contains its Hermitian dual if 
g(x) divides $\sigma(h_0)^{-1}x^k h^\sigma(1/x)$ viz. 
$g(x) | \sigma(h_0)^{-1}(1-x^n)/(x^{n-k}g^\sigma(1/x))$, which 
implies $x^n-1 \equiv 0 \mod g(x) g^{*}(x)$ which proves (i).

The generator polynomial $g(x)$ of $C$ is given by
$g(x)=\prod_{z\in Z} (x-\beta^z)$, hence its check polynomial is
of the form
$$h(x)=(x^n-1)/g(x)=\prod_{z\in N\setminus Z} (x-\beta^z).$$
Applying the automorphism $\sigma$ yields
$h^\sigma(x)= \prod_{z\in N\setminus Z} (x-\beta^{qz}).$
Therefore, the generator polynomial of $C^\hdual$ is given by
$$ \begin{array}{lcl}
h^\sigma(0)^{-1} x^k h^\sigma(1/x)
&=& h^\sigma(0)^{-1} \prod_{z\in N\setminus Z} (1-\beta^{qz}x)\\
&=& \prod_{z \in N\setminus Z} (x -\beta^{-qz});
   \end{array}
$$
in the last equality, we have used the fact that
$h^\sigma(0)^{-1}=\prod_{z\in N\setminus Z} (-\beta^{-qz})$. By Lemma~\ref{th:cyclic}, $B^\hdual \subseteq B$ if
and only if the generator polynomial $g(x)$ divides $h^\sigma(0)^{-1}
x^k h^\sigma(1/x)$. The latter condition is equivalent to the fact
that $Z$ is a subset of $\{ -qz\mid z\in N\setminus Z\}$ and (ii) follows.
From (ii)  we know that $C^\hdual \subseteq C$ if and only if $Z\subseteq \{ -qz\mid z\in
N\setminus Z\}$. In other words $Z^{-q} \subseteq N\setminus Z$.
Hence $Z\cap Z^{-q}=\emptyset$. 
An  $[[n,2k-n,\geq d]]_q$ stabilizer code follows from Corollary~\ref{co:classical}.
\end{proof}

Cyclic codes that contain their Euclidean duals can also be nicely
characterized in terms of their generator polynomials and defining sets. 
The following Lemma is a very straight forward extension of the binary 
case and summarizes some of the known results in the nonbinary case
as well, but we include it because of its usefulness in
constructing cyclic quantum codes. 

\begin{lemma}\label{th:csscyclic}
Let $C$ be an $[n,k,d]_q$ cyclic code such that $\gcd(n,q)=1$. Let its defining set $Z$ and
generator polynomial  $g(x)$ be such that any of the following
equivalent conditions are satisfied\\
(i) $x^n-1 \equiv 0 \mod g(x)g^{\dagger}(x)$, where 
$g^{\dagger}(x)=x^{n-k}g(1/x)$; \\
(ii) $Z \subseteq \{-z\mid z\in N\setminus Z\}$; \\
 (iii) $Z\cap Z^{-1}=\emptyset$ where $Z^{-1}=\{ -z\mod n\mid z\in Z\}$.\\
Then $C^\perp \subseteq C$ and there exists an $[[n,2k-n,\geq d]]_q$
stabilizer code that is pure to~$d$.
\end{lemma}
\begin{proof}
The check polynomial of $C$ is given by $h(x)=(x^n-1)/g(x)$, from
which we obtain the (un-normalized) generator polynomial of $C^\perp$
as
$h^\dagger(x)=x^kh(x^{-1})=(1-x^n)/(x^{n-k}g(x^{-1}))=-(x^n-1)/g^\dagger(x)$.
If $C^\perp \subseteq C$, then $g(x)\mid h^\dagger(x)$; this means
that $g(x)$ divides $(x^n-1)/g^\dagger(x)$. In other words $ x^n-1
\equiv 0\mod g(x)g^\dagger(x)$.

The defining set of $C^\perp$ is given by $\{-z\mod n \mid z\in
N\setminus Z \}$, where $N=\{0,1,\ldots,n-1 \}$. Thus $C^\perp
\subseteq C$ implies $Z\subseteq \{ -z\mod n\mid N\setminus Z\}$.
Since this means that the inverses of elements in $Z$ are present in 
$N\setminus Z$, this condition can also be written as $Z\cap
Z^{-1}=\emptyset$. The
existence of  quantum code $[[n,2k-n,\geq d]]_q$ follows from Corollary
\ref{th:css2}.
\end{proof}

Although we have considered purely cyclic codes, a larger class of
cyclic quantum codes can be derived by considering constacyclic or
conjucyclic codes as in \cite{calderbank98}, \cite{xiaoyan04}.

\subsection{Cyclic Hamming Codes}
Binary quantum Hamming codes have been studied by various authors; see
for instance \cite{gottesman96,calderbank98,feng02b}.
We now derive stabilizer codes from nonbinary classical cyclic Hamming
codes. Let $m>1$ be an integer such that
$\gcd(q-1,m)=1$. A classical cyclic Hamming code $H_q(m)$ has
parameters $[n,n-m,3]_{q}$ with length $n=(q^{m}-1)/(q-1)$.
Let $\beta$ denote a primitive
$n$th root of unity in $\F_{q^{m}}$. The generator polynomial of
$H_q(m)$ is given by
\begin{eqnarray}
 g(x)=\prod_{i=0}^{m-1} \big(x-\beta^{q^{i}}\big), \label{eq:hamminggen}
\end{eqnarray}
an element of
$\F_{q}[x]$. Thus, the code $H_q(m)$ is defined by the cyclotomic coset
$C_1=\{ q^{i}\bmod n \,|\, i \in \Z\}$.

\begin{lemma}\label{th:hammingdual}
The Hamming code $H_{q^2}(m)$ contains its Hermitian dual, that is, 
$H_{q^2}(m)^\hdual\le H_{q^2}(m)$.
\end{lemma}
\begin{proof}
The statement $H_{q^2}(m)^\hdual\le H_{q^2}(m)$ is equivalent to the fact that the
cyclotomic coset $C_1$ satisfies $C_1\subseteq N_1=\{ -qz\bmod n\,|\,
z\in N\setminus C_1\}$, where $N=\{0,\dots,n-1\}$ and $n=(q^{2m}-1)/(q^2-1)$.  We note that $C_1$
can be expressed in the form
\begin{equation}\label{eq:cosets}
\begin{split}
C_1&=\left\{ (1-n)q^{2k}\bmod n \,\Big|\, k\in \Z\right\}\\
&=\left\{ -qzq^{2k} \bmod n\,\Big|\, k\in \Z\right\},
\end{split}
\end{equation}
where $z=q(q^{2m-2}-1)/(q^2-1)$.  Therefore, the condition
$C_1\subseteq N_1$ holds if and only if $C_z\subseteq N\setminus C_1$ holds, where $C_z=\{zq^{2j}\bmod n\,|\, j\in
\Z\}$.

Seeking a contradiction, we assume that the two cyclotomic cosets
$C_1$ and $C_z$ have an element in common, hence are the same.  This
means that there must exist a positive integer $k$ such that
$q^{2k}=q(q^{2m-2}-1)/(q^2-1)$. This implies that
$q^{2k-1}$ divides $q^{2m-2}-1$, which is absurd.
Thus, the sets $C_1$ and $C_z$ are disjoint, hence $C_z\subseteq
N\setminus C_1$, which proves the claim.
\end{proof}

\begin{theorem}
For each integer $m\ge 2$ such that $\gcd(m,q^2-1)=1$, there exists a
pure $[[n,n-2m,3]]_q$ stabilizer code of length
$n=(q^{2m}-1)/(q^2-1)$.
\end{theorem}
\begin{proof}
If $\gcd(m,q^2-1)=1$, then there exists a classical $[n,n-m,3]_{q^2}$
Hamming code $H_{q^2}(m)$. By Lemma~\ref{th:hammingdual}, we have
$H_{q^2}(m)^\hdual\le H_{q^2}(m)$, hence there exists a pure $[[n,n-2m,3]]_{q}$
stabilizer code by Corollary~\ref{co:classical}. The purity is due to the fact 
that the $H_{q^2}(m)^\hdual$ has minimum distance $q^{2m-2} \geq 3$ for $m\geq 2$ \cite[Theorem~1.8.3]{huffman03}.
\end{proof}

These quantum Hamming codes are optimal since they attain the quantum
Hamming bound, see Corollary~\ref{th:hamming}. A different approach
that allows construction of noncyclic perfect quantum codes can be
found in \cite{bierbrauer00}.  It is also possible to construct
quantum codes from Hamming codes that contain their Euclidean duals,
however these codes do not meet the quantum Hamming bound.

\begin{lemma}
If $\gcd(m,q-1)=1$ and $m\geq 2$, then there exists a pure
$[[n,n-2m,3]]_q$ quantum code, where $n=(q^m-1)/(q-1)$.
\end{lemma}
\begin{proof}
The generating polynomial of an $[n,n-m,3]_q$ Hamming code, with
n=$(q^m-1)/(q-1)$ is given by equation~(\ref{eq:hamminggen}) where
$\beta$ is an element of order $n$.  The code exists only if
gcd$(m,q-1)=1$.  By Lemma \ref{th:csscyclic} a cyclic code contains
its dual if $x^n-1 \equiv 0 \mod g(x)g^\dagger(x)$, where
$g^\dagger(x)=x^{n-k}g(x^{-1})$. If $g(x)$ is not self-reciprocal then
$g(x)g^\dagger(x)$ divides $x^n-1$ \cite{vatan99}. Since the
generating polynomial of the Hamming code is not self-reciprocal, the
code contains its Euclidean dual. By Lemma \ref{th:csscyclic} we can
construct a quantum code with the parameters $[[n,n-2m,3]]_q$.  Once
again the purity follows due to the fact the duals of Hamming codes
are simplex codes with weight $q^{m-1} \geq 3$
for $m\geq 2$ \cite[Theorem~1.8.3]{huffman03}.
\end{proof}

\subsection{Quantum Quadratic Residue Codes}

Another well known family of classical codes are the quadratic residue codes.
Rains constructed quadratic residue codes for prime alphabet in \cite{rains99}.
In this section we will construct two series of quantum codes based on the
classical quadratic residue codes over an arbitrary field using elementary methods. 

Let $\alpha$ denote a primitive $n$th root of unity from some
extension field of $\F_q$. We denote by $R=\{ r^2\bmod n\mid
r\in \Z \text{ such that } 1\le r\le (n-1)/2\}$ the set of
quadratic residues modulo $n$ and by $N=\{1,\ldots,n-1 \}\setminus R$ the set of quadratic non-residues modulo $n$.

Let $C_R$ and $C_N$ denote the cyclic
codes of length $n$ that are respectively generated by the polynomials
$q_R(x)$ and $q_N(x)$, where
$$ q_R(x) = \prod_{r\in R} (x-\alpha^r)\quad\mbox{and}\quad q_N(x)=
\prod_{r\in N} (x-\alpha^r).$$ Both codes have parameters
$[n,(n+1)/2,d]_q$ with $d^2\ge n$, see \cite[pp.~114-119]{betten98}
or \cite{huffman03}.
The codes with generator polynomials $(x-1)q_R(x)$ and $(x-1)q_N(x)$ are 
the even-like subcodes
of $C_R$ and $C_N$ respectively and have the parameters $[n,(n-1)/2,d']_q$ with 
$d'\geq d$. The relevance of these codes will become apparent in the following
theorems.

\begin{theorem}
Let $n$ be a prime of the form $n\equiv 3 \mod 4$, and let
$q$ be a power of a prime that is not divisible by~$n$.
If\/ $q$ is a quadratic residue modulo~$n$, then there exists
a pure $[[n,1,d]]_q$ stabilizer code with minimum distance $d$ satisfying
$d^2-d+1\ge n$.
\end{theorem}
\begin{proof}
\nix{
Let $\alpha$ denote a primitive $n$th root of unity from some
extension field of $\F_q$. Let $R=\{ r^2\bmod n\mid
r\in \Z \text{ such that } 1\le r\le (n-1)/2\}$ denote
the set of quadratic residues modulo
$n$. We define the quadratic residue code
$C_R$ as the cyclic code of length $n$ over $\F_q$
that is generated by the polynomial
$$ q(x) = \prod_{r\in R} (x-\alpha^r).$$ 
}
The code $C_R$ has parameters
$[n,(n+1)/2,d]_q$ and if $n\equiv 3 \bmod 4$, the dual code $C^\perp_R$ of $C_R$ is
given by the cyclic code generated by $(x-1)q_R(x)$, the even-like
subcode of $C_R$. The minimum distance $d$ is bounded by
$d^2-d+1\ge n$, see, for instance, \cite[pp.~114-119]{betten98}. Further 
$\wt(C_R\setminus C_R^\perp)=\wt(C_R)=d$ by \cite[Theorem~6.6.22]{huffman03}. 
We can deduce from Corollary~\ref{th:css2} that
there exists a pure $[[n,(n+1)-n,d]]_q$ stabilizer code.
\end{proof}

For example, the prime $p=3$ is a quadratic residue modulo $n=23$. The
previous proposition guarantees the existence of a $[[23,1,d]]_3$
stabilizer code with minimum distance $d\ge 6$.

If $n$ is an odd prime of the form $n\equiv 1 \bmod 4$, then we can
also construct quadratic residue codes, but now we need to
employ Lemma~\ref{th:css}, because $C_R$ does not contain
its dual.

\begin{theorem}
Let $n$ be a prime of the form $n\equiv 1 \mod 4$. Let $q$ be a power
of a prime that is not divisible by $n$. If\/ $q$ is a quadratic
residue modulo~$n$, then there exists a pure $[[n,1,d]]_q$ stabilizer code
with minimum distance $d$ bounded from below by $d\ge \sqrt{n}$.
\end{theorem}
\begin{proof}
\nix{
Let $\alpha$ denote a primitive $n$th root of unity from some
extension field of $\F_q$. We denote by $R$ denote the set of
quadratic residues modulo $n$ and by $N$ the set of quadratic
non-residues modulo $n$.

Let $C_R$ and $C_N$ denote the cyclic
codes of length $n$ that are respectively generated by the polynomials
$q_R(x)$ and $q_N(x)$, where
$$ q_R(x) = \prod_{r\in R} (x-\alpha^r)\quad\mbox{and}\quad q_N(x)=
\prod_{r\in N} (x-\alpha^r).$$ Both codes have parameters
$[n,(n+1)/2,d]_q$ with $d^2\ge n$, see \cite[pp.~114-119]{betten98}.
}
The dual code of $C_R$ is given by the even-like subcode of $C_N$; in
other words, $C_R^\perp$ is a cyclic code of length $n$ over $\F_q$
that is generated by the polynomial $(x-1)q_N(x)$; in particular,
$C_R^\perp \le C_N$. 
Moreover $\wt(C_R\setminus C_N^\perp)=\wt(C_N\setminus C_R^\perp)=\wt(C_R)=\wt(C_N)=d$ by \cite[Theorem~6.6.22]{huffman03}.
Therefore, we obtain a pure $[[n,(n+1)/2+(n+1)/2-n,d]]_q$ code by Lemma~\ref{th:css}.
\end{proof}

\section{Quantum BCH Codes}
In this section we consider a popular family of classical codes, the
BCH codes, and construct the associated nonbinary quantum stabilizer
codes.  Binary quantum BCH codes were studied in
\cite{calderbank98,cohen99,grassl99b, steane99}. The CSS construction
turns out to be especially useful, because BCH codes form a naturally
nested family of codes.  In case of primitive BCH codes over prime
fields, the distance of the dual is lower bounded by the generalized
Carlitz-Uchiyama bound, and this allows us to derive bounds on the
minimum distance of the resulting quantum codes.

\subsection{BCH Codes.}
Let $q$ be a power of a prime and $n$ a positive integer that is
coprime to~$q$. Recall that a BCH code $C$ of length $n$ and designed
distance $\delta$ over $\F_q$ is a cyclic code whose defining set $Z$
is given by a union of $\delta-1$ subsequent cyclotomic cosets,
$$Z=\bigcup_{x=b}^{b+\delta-2} C_x, \quad\text{where} \quad C_x =
\{xq^{r} \bmod n \mid r \in \Z, r\ge 0 \}.$$ 
The generator polynomial of the code is of the form 
$$ g(x) = \prod_{z \in Z} (x-\beta^{z}),$$ where $\beta$ is a primitive
$n$-th root of unity of some extension field of $\F_q$.  The
definition ensures that $g(x)$ generates a cyclic $[n,k,d]_{q}$ code
of dimension $k=n-|Z|$ and minimum distance~$d\ge \delta$.
If $b=1$, then the code $C$ is called a narrow-sense BCH code, and if
$n=q^m-1$ for some $m\ge 1$, then the code is called primitive. 
More precise statements can be made about the structure of primitive, 
narrow-sense codes than the other classes of BCH codes and we will restrict our
attention to these in this paper.
More details on BCH codes can be found in \cite{huffman03,macwilliams77}.

\subsection{Generalized Carlitz-Uchiyama Bound.} 
Our first construction derives stabilizer codes from BCH codes over
prime fields. We use the Knuth-Iverson bracket $[statement]$ in the
formulation of the Carlitz-Uchiyama bound that evaluates to 1 if
$statement$ is true and 0 otherwise.

\begin{lemma}[Generalized Carlitz-Uchiyama Bound]\label{th:carlitz}
Let $p$ be a prime. 
Let $C$ denote a narrow-sense BCH code of length $n=p^m-1$ over $\F_p$,
of designed distance $\delta=2t+1$. Then the minimum
distance $d^\perp$ of its Euclidean dual code $C^\perp$ is bounded by 
\begin{equation}\label{eq:carlitzdist} 
d^\perp \ge  \Big(1-\frac{1}{p}\Big)
\left( p^m-\frac{\delta-2-[\delta-1\equiv 0\bmod p]}{2}
\big\lfloor 2p^{m/2}\big\rfloor\right).
\end{equation}
\end{lemma}
\begin{proof}
See \cite[Theorem 7]{stichtenoth94}; for further background,
see~\cite[page 280]{macwilliams77}. 
\end{proof}

\begin{theorem} 
Let $p$ be a prime. Let $C$ be a $[p^m-1,k,\ge \delta]_p$ narrow-sense
BCH code of designed distance $\delta=2t+1$ and $C^*$ a
$[p^m-1,k^*,d^*]_p$ BCH code such that $C\subseteq C^*$. Then there
exists a $[[p^m-1,k^*-k,\ge \min\{d^*,d^\perp\}]]_p$ stabilizer code,
where $d^\perp$ is given by (\ref{eq:carlitzdist}).
\end{theorem}
\begin{proof}
The result follows from applying Lemma~\ref{th:carlitz} to $C$ and
Lemma~\ref{th:css} to the codes $C$ and $C^*$. 
\end{proof}
\begin{remark} \begin{inparaenum}[(i)]
\item The Carlitz-Uchiyama bound becomes trivial for larger design distances.
\item In \cite[Corollary 2]{moreno94} it was shown that for binary BCH
codes of design distance $d$, the lower bound in equation
(\ref{eq:carlitzdist}) is attained when $n=2^{2ab}-1$, where $a$ is
the smallest integer such that $d-2\mid 2^{a}+1$ and $b$ is odd.
\item For a further tightening of the Carlitz-Uchiyama bound see
\cite[Theorem 2]{moreno98}.  \end{inparaenum}
\end{remark}

\subsection{Primitive BCH Codes Containing Their Duals.}
We can extend the results of the previous section to BCH codes over finite
fields that are not necessarily prime. In fact, if we restrict
ourselves to smaller designed distances, then we can even achieve
significantly sharper results.\nocite{steane99} We will just review
the results and refer the reader to our companion
paper~\cite{aly06a} 
for the proofs. 
A generalization of the
following results is given in Chapter~\ref{ch:bch}, with a view to
demonstrate the fact that study of quantum codes can lead to interesting
insights into classical coding theory. 

In the BCH code construction, it is in general not obvious how large the
cyclotomic cosets will be. However, if the designed distance is small,
then one can show that the cyclotomic cosets all have maximal size. 

\begin{lemma}\label{th:bchdim} 
A narrow-sense, primitive BCH code with design distance 
$2\leq \delta \leq q^{\lceil m/2\rceil}+1$ has parameters
$[q^m-1,q^m-1-m\lceil (\delta-1)(1-1/q) \rceil,\geq \delta]_q$. 
\end{lemma}
\begin{proof}
See \cite[Theorem~7]{aly06a}; 
the binary case was already
established by Steane~\cite{steane99}.
\end{proof}

In the case of small designed distances, primitive, narrow-sense 
BCH codes contain their Euclidean duals. 
\begin{lemma}\label{th:bchEuclideandual}

A narrow-sense, primitive BCH code over $\F_q^n$ contains its
Euclidean dual if and only if its design
distance satisfies $2 \leq \delta \leq q^{\lceil m/2\rceil}-1-(q-2)[m \textrm{
odd}]$, where $n=q^m-1$ and $m\geq 2$.
\end{lemma}
\begin{proof}
See \cite[Theorem~2]{aly06a}. 
\end{proof}

A simple consequence is the following theorem: 

\begin{theorem}\label{co:bchEuclideandual}
If $C$ is a narrow-sense primitive BCH code over $\F_q$ with design
distance $2\leq \delta \leq q^{\lceil m/2 \rceil}-1-(q-2)[m \textrm{
odd}]$ and $m\geq 2$, then there exists an
$[[q^m-1,q^m-1-2m\lceil(\delta-1)(1-1/q)\rceil,\geq \delta]]_q$
stabilizer code that is pure to~$\delta$.
\end{theorem}
\begin{proof}
If we combine Lemmas~\ref{th:bchdim} and \ref{th:bchEuclideandual} and
apply the CSS construction, then we obtain the claim. See 
\cite{aly06a}
for details about purity.
\end{proof}

One can argue in a similar way for Hermitian duals of primitive,
narrow-sense BCH codes.
\begin{theorem}\label{th:bchHermitiandual} 
If $C$ is a narrow-sense primitive BCH code over $\F_{q^2}^n$ with
design distance $2\leq \delta \leq q^{m}-1$, then there exists an
$[[q^{2m}-1,q^{2m}-1-2m\lceil(\delta-1)(1-1/q^2)\rceil,\geq
\delta]]_q$ stabilizer code that is pure to $\delta$.
\end{theorem}
\begin{proof}
See~\cite{aly06a} for details. 
\end{proof}

When $m=1$, the BCH codes are the same as the Reed Solomon codes and
this case has been dealt with in \cite{grassl04}. An alternate perspective
using Reed-Muller codes is considered in \cite{klappenecker05p1}.

\subsection{Extending Quantum BCH Codes}
It is not always possible to extend a stabilizer code, because the
corresponding classical codes are required to be self-orthogonal.
We now show that it is possible to extend 
narrow-sense BCH codes of certain lengths.

\begin{lemma}\label{th:bch_extension}
Let $\F_{q^2}$ be a finite field of characteristic $p$.  If $C$ is a
narrow-sense $[n,k,\geq d]_{q^2}$ BCH code such that $C^\hdual
\subseteq C$ and $n\equiv -1 \mod p$, then there exists an
$[[n,2k-n,\geq d]]_q$ stabilizer code that is pure to~$d$ which can be
extended to an $[[n+1,2k-n-1,\geq d+1]]_q$ stabilizer code that is
pure to~$d+1$.
\end{lemma}

\begin{proof}
Since $C^\hdual \subseteq C$, Corollary~\ref{co:classical} implies the
existence of an $[[n,2k-n,\geq d]]_q$ quantum code that is pure to~$d$ and
being narrow-sense the parity check matrix of $C$ has the form
\begin{eqnarray*}
H&=&\left[ \begin{array}{ccccc}
1 &\alpha &\alpha^{2}&\cdots&\alpha^{(n-1)}\\
1 &\alpha^{2} &\alpha^{2(2)}&\cdots&\alpha^{2(n-1)}\\
\vdots &\ddots &\ddots &\ddots &\ddots\\
1&\alpha^{d-1}&\alpha^{2(d-1)}&\cdots&\alpha^{(n-1)(d-1)}
\end{array}\right],
\end{eqnarray*}
where $\alpha$ is a primitive $n^{th}$ root of unity. This can be
extended to give an $[n+1,k,d+1]$ code $C_e$, whose parity check
matrix is given as
\begin{eqnarray*}
H_{e} &=&\left[ \begin{array}{cccccc}
1 &1 &1&\cdots&1& 1\\
1 &\alpha &\alpha^{2}&\cdots&\alpha^{(n-1)}& 0\\
1 &\alpha^{2} &\alpha^{2(2)}&\cdots&\alpha^{2(n-1)}&0\\
\vdots &\ddots &\ddots &\ddots &\ddots&\vdots\\
1&\alpha^{d-1}&\alpha^{2(d-1)}&\cdots&\alpha^{(n-1)(d-1)}&0
\end{array}\right].
\end{eqnarray*}
We show that $C_e^\hdual$ is self-orthogonal. Let $R_i$ be the
$i^{th}$ row in $H_{e}$. For $2 \leq i\leq d$ the self-orthogonality
of $H$ implies that $\langle R_i|R_j\rangle_h=0$. We need to show
that $\langle R_i|\mbf{1}\rangle_h=0$, $1\leq i\leq d $. For $2\leq
i\leq d$ we have $\langle R_i|\mbf{1}\rangle_h
=\sum_{j=0}^{n-1}\alpha^{ij}=(\alpha^{in}-1)/(\alpha^{i}-1)=0$, as
$\alpha^n=1$ and $\alpha^i\neq 1$. For $i=1$ we have $\langle
\mbf{1}|\mbf{1}\rangle_h=n+1 \mod p$, which vanishes because of the
assumption $n\equiv -1 \mod p$.

Now we show that the rank of $H_{e}$ is  $d$, thus $C_e$ has a
minimum distance of at least $d+1$. Any $d$ columns of $H_{e}$
excluding the last column form a $d\times d$ vandermonde matrix
which is nonsingular, indicating that the $d$ columns are linearly
independent. If we consider any set of $d$ columns that includes the
last column, we can find the determinant of the corresponding matrix
by expanding by the last column. This gives us a $d-1 \times d-1$
vandermonde matrix with nonzero determinant. Thus any $d$ columns of
$H_e$ are independent and the minimum distance of $C_e$ is at least
$d+1$. Therefore $C_e$ is an $[n+1,k,\geq d+1]_{q^2}$ extended cyclic code such that
$C_e^\hdual \subseteq C_e$. By Corollary~\ref{co:classical} it defines
an $[[n+1,2k-n-1,\geq d+1]]_q$ quantum code pure to $d+1$.
\end{proof}

\begin{corollary}
For all prime powers $q$, integers $m\ge 1$ and all $\delta$ in the range 
$2\leq \delta \leq q^{m}-1$ there exists an 
$$[[q^{2m},q^{2m}-2-2m\lceil(\delta-1)(1-1/q^2)\rceil,\geq \delta+1]]_q$$ 
stabilizer code pure to $\delta+1$. 
\end{corollary}
\begin{proof}
The stabilizer codes from Theorem~\ref{th:bchHermitiandual} are
derived from primitive, narrow-sense BCH codes. If $p$ denotes the
characteristic of $\F_{q^2}$, then $q^{2m}-1\equiv -1\bmod p$, so the
stabilizer codes given in Theorem~\ref{th:bchHermitiandual} can be
extended by Lemma~\ref{th:bch_extension}.  
\end{proof}

A result similar to Lemma~\ref{th:bch_extension} can be developed for
BCH codes that contain their Euclidean duals.

\subsection{Puncturing BCH Codes.} 
In this section, let $\B_q^m(\delta)$ denote a primitive,
narrow-sense $q$-ary BCH code of length $n=q^m-1$ and designed
distance $\delta$. We illustrate the theory of puncture codes
developed in Chapter~\ref{ch:stabq1} by puncturing such BCH codes. 
Some knowledge about the puncture code is
necessary for this task, and we show in Theorem~\ref{th:bchP(C)} that
a cyclic generalized Reed-Muller code is contained in the puncture
code.

First, let us recall some basic facts about cyclic generalized Reed-Muller
codes, see~\cite{assmus92,assmus98,kasami68,pellikaan04} for
details. Let $L_m(\nu)$ denote the subspace of $\F_q[x_1,\dots,x_m]$ consisting
of polynomials of degree~$\le \nu$, and let $(P_0,\dots, P_{n-1})$ be
an enumeration of the points in $\F_q^m$ where $P_0=\mbf{0}$.  The $q$-ary cyclic
generalized Reed-Muller code $\RM^*_q(\nu,m)$ of order $\nu$ and
length $n=q^m-1$ is defined as
$$ \RM^*_q(\nu,m) = \{ ev\, f\,|\, f\in
L_m(\nu)\},$$ where the codewords are evaluations of the polynomials in
all but $P_0$ defined by $ev\, f=(f(P_1),\dots, f(P_{n-1}))$. 
The dimension $k^*(\nu)$ of the code 
$\RM^*_q(\nu,m)$ is given by the formula
$
k^*(\nu) = \sum_{j=0}^m (-1)^j \binom{m}{j}\binom{m+\nu-jq}{\nu-jq} 
$
and its minimum distance $d^*(\nu)=(R+1)q^Q-1,$
where $m(q-1)-\nu = (q-1)Q+ R$ with $0\le R<q-1$.
The dual code of $\RM_q^*(\nu,m)$ can be characterized by  
\begin{eqnarray}
\RM_q^*(\nu,m)^\perp
&=&\{ ev \mbox{ }f \,|\, f\in L_m^*(\nu^\perp)
\},\label{eq:RMdualdefn}
\end{eqnarray}
where $\nu^\perp= m(q-1)-\nu-1$ and $L_m^*(\nu)$ is the subspace of
all nonconstant polynomials in $L_m(\nu)$; 

It is well-known that a primitive, narrow-sense BCH code contains a
cyclic generalized Reed-Muller code, see~\cite[Theorem 5]{kasami68},
and we determine the largest such subcode in our next lemma.

\begin{lemma}\label{th:largestRMinbch}
Let $\nu = (m-Q)(q-1)-R$, with $Q=\lfloor\log_q(\delta+1) \rfloor$ and $R=\lceil
(\delta+1)/q^Q\rceil-1$, then $\RM_q^*(\nu,m)\subseteq \B_q^m(\delta)$. Also for all orders $\nu'>\nu$, we have 
$\RM_q^*(\nu',m)\not\subseteq \B_q^m(\delta)$.
\end{lemma}

\begin{proof}
First, we show that $\RM_q^*(\nu,m) \subseteq \B_q^m(\delta)$.  Recall
that the minimum distance $d^*(\nu)=(R+1)q^Q-1$, where
$m(q-1)-\nu=(q-1)Q+R$ with $0\le R<q-1$.  By \cite[Theorem
5]{kasami68}, we have $\RM_q^*(\nu,m) \subseteq \B_q^m((R+1)q^Q-1)$.
Notice that $(R+1)q^Q-1 =\lceil (\delta+1)/ q^Q\rceil q^Q -1 \geq
\delta$, so $\B_q^m((R+1)q^Q-1) \subseteq \B_q^m(\delta)$. Therefore,
$\RM_q^*(\nu,m) \subseteq \B_q^m(\delta)$, as claimed.

For the second claim, it suffices to show that $\RM_q^*(\nu+1,m)$ is
not a subcode of $\B_q^m(\delta)$. We prove this by showing that
the minimum distance $d^*(\nu+1)< \delta$.  Notice that 
$$
m(q-1)-(\nu+1) 
=\left\{\begin{array}{ll} (q-1)Q+R-1,
&
 R\geq 1,\\
(q-1)(Q-1)+q-2,& 
R=0 \end{array} \right.
$$ 
with $R$ and $Q$ as given in the hypothesis. 
Therefore, the distance $d^*(\nu+1)$ of 
$\RM_q^*(\nu+1,m)$ is given by 
\begin{eqnarray*}
d^*(\nu+1) &=& 
\left\{\begin{array}{ll} (\lceil (\delta+1)/q^Q\rceil-1)q^Q-1 
&\text{for } R\geq 1,\\
(q-1)q^{Q-1}-1&\text{for } R=0. \end{array} \right.
\end{eqnarray*}
In both cases, it is straightforward to verify that $d^*(\nu+1)<\delta$. 
\end{proof}

Explicitly determining the puncture code is a challenging task. For
the duals of BCH codes, we are able to determine large subcodes of the
puncture code.

\begin{theorem}\label{th:bchP(C)}
If $ \delta < q^{\lfloor m/2 \rfloor}-1$, then $\RM_q^*(\mu,m)
\subseteq \pc_e(\B_q^m(\delta)^\perp)$ for all orders $\mu$ in the range
$0\leq \mu\leq m(q-1)-2(R+(q-1)Q)+1$ with $Q=\lfloor\log_q
(\delta+1)\rfloor$ and $R=\lceil(\delta+1)/q^Q\rceil -1$.
\end{theorem}
\begin{proof}
By Lemma~\ref{th:largestRMinbch}, we have $\RM_q^*(\nu,m)
\subseteq \B_q^m(\delta)$ for $\nu= (m-Q)(q-1)-R$; hence,
$\B_q^m(\delta)^\perp \subseteq \RM_q^*(\nu,m)^\perp$.
It follows from the definition of the puncture code that 
$\pc_e(\B_q^m(\delta)^\perp) \supseteq \pc_e(\RM_q^*(\nu,m)^\perp)$. However, 
\begin{eqnarray*}
\pc_e(\RM_q^*(\nu,m)^\perp) &=&\{ev f \cdot ev \mbox{ } g \mid f,g \in
L_m^*(\nu^\perp)\}^\perp,\\ &\supseteq & \{ ev f \mid f \in
L_m^*(2\nu^\perp)\}^\perp,\\
&=&\RM_q^*((2\nu^\perp)^\perp,m),
\end{eqnarray*}
where the last equality follows from equation~(\ref{eq:RMdualdefn}).  This is
meaningful only if $(2\nu^\perp)^\perp \geq 0$ or, equivalently, if 
$\nu\geq(m(q-1)-1)/2$. Since $\delta < q^{\lfloor m/2\rfloor}-1$, it follows that $Q\leq \lfloor m/2\rfloor -1$, and the order $\nu$ satisfies
$$\begin{array}{lcl}
\nu&=&(m-Q)(q-1)-R \geq \lceil m/2+1\rceil(q-1)-R \\
&\geq& \lceil m/2\rceil (q-1)+1 \geq (m(q-1)-1)/2,
\end{array}
$$
as required. Since $\RM_q^*(\mu,m) \subseteq \RM_q^*((2\nu^\perp)^\perp,m)$ for 
$0\leq \mu\leq (2\nu^\perp)^\perp$, we have $\RM_q^*(\mu,m) \subseteq \pc_e(\B_q^m(\delta)^\perp)$.
\end{proof}

Unfortunately, the weight distribution of generalized cyclic
Reed-Muller codes is not known, see~\cite{charpin98}.  However, we
know that the puncture code of $\B_q^m(\delta)^\perp$ contains the
codes $\RM_q^*(0,m)\subseteq \RM_q^*(1,m)\subseteq \cdots \subseteq
\RM_q^*(m(q-1)-2(R+(q-1)Q)+1,m)$, so it must contain codewords of the
respective minimum distances.

\begin{corollary} \label{th:bchpunc}
If $\delta$ and $\mu$ are integers in the range $2\leq \delta <
q^{\lfloor m/2\rfloor}-1$ and $0\leq \mu\leq m(q-1)-2(R+(q-1)Q)+1$,
where $Q=\lfloor\log_q(\delta+1)\rfloor$ and
$R=\lceil(\delta+1)/q^Q\rceil-1$, then there exists
a 
$$[[d^*(\mu), \geq d^*(\mu)-2m\lceil(\delta-1)(1-1/q)\rceil,\geq
\delta]]_q$$ stabilizer code of length $d^*(\mu)=(\rho+1)q^\sigma-1$, 
where $\sigma$ and $\rho$ satisfy 
the relations $m(q-1)-\mu=(q-1)\sigma+\rho$ and $0\leq \rho<q-1$.
\end{corollary}

\begin{proof}
If $2\leq \delta < q^{\lfloor m/2\rfloor }-1 $, then from Theorem~\ref{co:bchEuclideandual} 
we know that there exists an
$[[q^m-1,q^m-1-2m\lceil(\delta-1)(1-1/q)\rceil,\geq \delta]]_q$ quantum code.
From Lemma~\ref{th:bchP(C)} we know that
$\pc_e(\B_q^m(\delta)^\perp)\supseteq \RM_q^*(\mu,m)$, where $0\leq
\mu\leq m(q-1)-2(q-1)Q-2R+1$.  By Theorem~\ref{th:punc_symplectic}, if
there exists a vector of weight $r$ in $\pc_e(\B_q^m(\delta)^\perp)$,
the corresponding quantum code can be punctured to give $[[r,\geq
r- 2m\lceil(\delta-1)(1-1/q)\rceil) ,d\geq \delta]]_q$.  The minimum distance
of $\RM_q^*(\mu,m)$ is $d^*(\mu)=(\rho+1)q^\sigma-1$, where 
$0\leq \rho<q-1$ \cite[Theorem 5]{kasami68}. Hence, it is always possible to puncture
the quantum code to $[[d^*(\mu), \geq d^*(\mu)-2m\lceil(\delta-1)(1-1/q)\rceil,
\geq \delta ]]_q$.
\end{proof}

It is also possible to puncture quantum codes constructed via
classical codes self-orthogonal with respect to the Hermitian inner
product. Examples of such puncturing can be found in \cite{grassl04}
and \cite{klappenecker05p1}.

\section{MDS Codes}\label{sec:MDS} 
A quantum code that attains the quantum Singleton bound is called a
quantum Maximum Distance Separable code or quantum MDS code for short.
These codes have received much attention, but many aspects have not
yet been explored in the quantum case (but
see~\cite{grassl04,rains99}).  In this section we study the
maximal length of MDS stabilizer codes.

An interesting result concerning the purity of quantum MDS codes was
derived by Rains~\cite[Theorem 2]{rains99}: 

\begin{lemma}[Rains] \label{th:d_purity}
An $[[n,k,d]]_q$ quantum MDS code with $k\ge 1$ is pure up to $n-d+2$.
\end{lemma}

\begin{corollary} \label{th:mds_purity}
All quantum MDS codes are pure.
\end{corollary}
\begin{proof}
An $[[n,k,d]]_q$ quantum MDS code with $k=0$ is pure by definition; if
$k\ge 1$ then it is pure up to $n-d+2$.  By the quantum Singleton
bound $n-2d+2=k\geq 0$; thus, $n-d+2\ge d$, which means that the code
is pure.
\end{proof}

\begin{lemma} \label{th:mds_classical}
For any $[[n,n-2d+2,d]]_q$  quantum MDS stabilizer code with $n-2d+2>0$, the
corresponding classical codes $C\subseteq C^\adual$ are also MDS. 
\end{lemma}
\begin{proof}
If an $[[n,n-2d+2,d]]_q$ stabilizer code exists, then
Theorem~\ref{th:alternating} implies the existence of an additive
$[n,d-1]_{q^2}$ code $C$ such that $C\subseteq C^\adual$.
Corollary~\ref{th:mds_purity} shows that $C^\adual$ has minimum
distance $d$, so $C^\adual$ is an $[n,n-d+1,d]_{q^2}$ MDS code. By
Lemma~\ref{th:d_purity}, the minimum distance of $C$ is $\ge n-d+2$,
so $C$ is an $[n,d-1,n-d+2]_{q^2}$ MDS code.
\end{proof}

A classical $[n,k,d]_q$ MDS code is said to be trivial if $k\leq 1 $
or $k\geq n-1$. A trivial MDS code can have arbitrary length, but a
nontrivial one cannot. The next lemma is a straightforward
generalization from linear to additive MDS codes.

\begin{lemma} \label{th:mds_nontrivial}
Assume that there exists a classical additive $(n,q^k,d)_q$ MDS code $C$.
\begin{compactenum}[(i)]
\item If the code is trivial, then it can have arbitrary length.
\item If the code is nontrivial, then its code parameters must be in the range 
$2\leq k\leq \min\{ n-2,q-1\}$ and
$n\leq q+k-1\leq 2q-2$.
\end{compactenum}
\end{lemma}
\begin{proof} The first statement is obvious. For (ii), we note that 
the weight distribution of the code $C$ and its dual are related by the
MacWilliams relations. The proof given in
\cite[p.~320-321]{macwilliams77} for linear codes applies without
change, and one finds that the number of codewords of weight $n-k+2$
in $C$ is given by
$$ A_{n-k+2}=\binom{n}{k-2}(q-1)(q-n+k-1).$$
Since $A_{n-k+2}$ must be a nonnegative number, we obtain the claim. 
\end{proof}

We say that a quantum $[[n,k,d]]_q$ MDS code is trivial if and only if
its minimum distance $d\le 2$.  The length of trivial quantum MDS
codes is not bounded, but the length of nontrivial ones is, 
as the next lemma shows.

\begin{theorem}[Maximal Length of MDS Stabilizer Codes]\label{th:mds_length}
A nontrivial $[[n,k,d]]_q$  MDS stabilizer code satisfies 
the following constraints: 
\begin{compactenum}[i)]
\item its length $n$ is in the range $4\leq n \leq q^2+d-2\leq 2q^2-2$;
\item its minimum distance satisfies $\max\{3,n-q^2+2\} 
\leq d \leq \min\{ n-1,q^2\}$. 
\end{compactenum}
\end{theorem}
\begin{proof}
By definition, a quantum MDS code attains the Singleton bound, so
$n-2d+2=k\ge 0$; hence, $n\ge 2d-2$. Therefore, a nontrivial quantum
MDS code satisfies $n\ge 2d-2\ge 4$.

By Lemma \ref{th:mds_classical}, the existence of an
$[[n,n-2d+2,d]]_q$  stabilizer code implies the existence of
classical MDS codes $C$ and $C^\adual$ with parameters
$[n,d-1,n-d+2]_{q^2}$ and $[n,n-d+1,d]_{q^2}$, respectively. If the
quantum code is a nontrivial MDS code, then the associated classical
codes are nontrivial classical MDS codes. Indeed, for $n\ge 4$ the quantum
Singleton bound implies $d\le (n+2)/2\le (2n-2)/2=n-1$, so $C$ is a
nontrivial classical MDS code. 

By Lemma \ref{th:mds_nontrivial}, the dimension of $C$ satisfies the
constraints $2\leq d-1 \leq \min\{n-2,q^2-1\}$, or equivalently $3\leq d\leq
\min\{n-1,q^2 \}$. Similarly, the length $n$ of $C$ satisfies 
$n\leq q^2+(d-1)-1 \leq 2q^2-2$. If we combine these inequalities then
we get our claim.
\end{proof}

\begin{example} 
The length of a nontrivial binary MDS stabilizer code cannot exceed
$2q^2-2=6$. In \cite{calderbank98} the nontrivial MDS stabilizer codes
for $q=2$ were found to be $[[5,1,3]]_2$ and $[[6,0,4]]_2$, so there
cannot exist further nontrivial MDS stabilizer codes.
\end{example}

In~\cite{grassl04}, the question of the maximal length of MDS codes
was raised. All MDS stabilizer codes provided in that reference had a
length of $q^2$ or less; this prompted us to look at the following
famous conjecture for classical codes (cf.~\cite[Theorem~7.4.5]{huffman03} or
\cite[pages~327-328]{macwilliams77}).

\smallskip

\noindent\textbf{MDS Conjecture.} \textit{
If there is a nontrivial $[n,k]_q$ MDS code, then $n\leq q+1$ except
when $q$ is even and $k=3$ or $k=q-1$ in which case $n\leq
q+2$.}
\smallskip

If the MDS conjecture is true (and much supporting evidence is
known), then we can improve upon the result of
Theorem~\ref{th:mds_length}.

\begin{corollary}
If the classical MDS conjecture holds, then there are no nontrivial
MDS stabilizer codes of lengths exceeding $q^2+1$ except when
$q$ is even and $d=4$ or $d=q^2$ in which case $n\leq q^2+2$.
\end{corollary}

\section{Conclusions}
In this chapter we applied the theory developed in Chapter~\ref{ch:stabq1} 
to derive classes of quantum codes. This work has also led to 
the construction of many more families of codes. The interested
reader can find the details in \cite{aly07}.
Table~\ref{tb:families} gives an overview and
summarizes the main parameters of these families.
We also illustrated the theory of puncture codes by deriving new codes
from quantum BCH codes. 
One central theme in quantum error-correction is the construction of
codes that have a large minimum distance. We were able to show that
the length of an MDS stabilizer code over $\F_q$ cannot exceed
$q^2+1$, except in a few sporadic cases, assuming that the classical
MDS conjecture holds.  An open problem is whether the length $n$ of a
$q$-ary quantum MDS code is bounded by $q^2+1$ for all but finitely
many $n$. Another related problem is to construct analytically quantum
MDS codes between lengths $q$ and $q^2$. Currently, constructions are
known only for a few lengths in this range. 

{
\begin{sidewaystable}[h]
\small
\caption{A compilation of known families of quantum
codes}\label{tb:families}
\begin{center}
\begin{tabular}{@{}l@{}||c|@{}c@{}|@{}c@{}}
Family & $[[n,k,d]]_q$ & Purity & Parameter Ranges and References \\
\hline\hline 
Short MDS & $[[n,n-2d+2,d]]_q$ & pure &$2\leq d\leq \lceil n/2\rceil$, $q^2-1\ge \binom{n}{d}$\\
\hline
Hermitian Hamming & $[[n,n-2m,3]]_q$& pure & $m\ge 2$, $\gcd(m,q^2-1)=1$, $n=(q^{2m}-1)/(q^2-1)$\\
\hline
 Euclidean Hamming & $[[n,n-2m,3]]_q$& pure & $m\ge 2$, $\gcd(m,q-1)=1$, $n=(q^m-1)/(q-1)$\\
\hline
Quadratic Residue  I &$[[n,1,d]]_q$ & pure & $n$ prime, $n\equiv 3 \mod 4$,
$q\not\equiv 0 \mod n$\\
& & & $q$ is a quadratic residue modulo~$n$,
$d^2-d+1\ge n$\\
\hline
Quadratic Residue II &$[[n,1,d]]_q$& pure&
$n$ prime, $n\equiv 1 \mod 4$, $q\not\equiv 0\mod n$\\
& & & $q$ is a quadratic residue modulo~$n$,  $d\ge \sqrt{n}$\\
\hline
Melas  &$[[n,n-4m,\ge
3]]_q$ &pure & $q$ even, $n=q^{2m}-1$, Pure to 3\\
\hline
Euclidean BCH & $[[n,n-2m\lceil(\delta-1)(1-1/q)\rceil,\geq \delta]]_q$&
 pure& $2\leq \delta \leq q^{\lceil m/2 \rceil}-1-(q-2)[m \textrm{ odd}]$\\
& & to $\delta$ & $n=q^m-1$ and $m\geq2$\\
Punctured BCH & $[[d^*(\mu),\geq d^*(\mu)-2m\lceil(\delta-1)(1-1/q)\rfloor,\geq \delta]]_q$ & pure?& $\delta <q^{\lfloor m/2\rceil}-1$, See Corollary~\ref{th:bchpunc}\\
\hline
Hermitian BCH & $[[n,n-2m\lceil(\delta-1)(1-1/q^2)\rceil,\geq \delta]]_q$ & pure&
 $2\leq \delta \leq q^{m}-1$, $n=q^{2m}-1$, Pure to $\delta$ \\
Extended BCH & $[[n+1,n-2m\lceil(\delta-1)(1-1/q^2)\rceil-1,\geq \delta+1]]_q$ & pure&Pure to $\delta+1$\\
\hline Trivial MDS & $ [[n,n-2,2]]_q$ &pure & $n\equiv 0 \mod p$\\
& $[[n,n,1]]_q$& pure& $n\geq 1$\\
\hline Character & $[[n,k(r_2)-k(r_1),\min \{ 2^{m-r_2},2^{r_1+1} \}]]_q$& pure&$n=2^m$,
$q$ odd, $0\leq r_1<r_2\leq m$, $k(r)=\sum_{j=0}^r\binom{m}{j}$\\
\hline CSS GRM &
$[[q^{m},k(\nu_2)-k(\nu_1),\min\{d(\nu_2),d(\nu_1^\perp)\} ]]_q$ &
pure& $k(\nu)= \sum_{j=0}^m (-1)^j
\binom{m}{j}\binom{m+\nu-jq}{\nu-jq}$, $\nu^\perp=m(q-1)-\nu-1$  \\
&$0\leq \nu_1\leq \nu_2\leq m(q-1)-1$ & & $\nu^\perp+1=(q-1)Q+R$, $d(\nu)=(R+1)q^Q$\\
Punctured GRM & $[[d(\mu),\geq k(\nu_2)-k(\nu_1)-(n-d(\mu)),\geq d ]]_q$ &
pure?& $d \geq \min
\{d(\nu_2),d(\nu_1^\perp)\}$, $0\leq \mu\leq \nu_2-\nu_1$; \cite{klappenecker05p1}\\
\hline Hermitian GRM & $[[q^{2m},q^{2m}-2k(\nu),d(\nu^\perp) ]]_q$& pure &
$ k(\nu) =\sum_{j=0}^m (-1)^j \binom{m}{j}\binom{m+\nu-jq^2}{\nu-jq^2}$, $\nu^\perp=m(q^2-1)-\nu-1$\\ 
&  $0 \leq \nu \leq m(q-1)-1$ & & $\nu^\perp+1 =(q^2-1)Q+R$, $d(\nu)=(R+1)q^{2Q}$\\
Punctured GRM & $[[d(\mu^\perp),\geq d(\mu^\perp)-2k(\nu),\geq d(\nu^\perp) ]]_q$& pure?&
$(\nu+1)q\leq \mu\leq m(q^2-1)-1$; \cite{klappenecker05p1}\\
\hline Punctured MDS & $[[q^2-q\alpha, q^2-q\alpha-2\nu-2,\nu+2]]_q$ &pure & $0\leq \nu\leq q-2$, $0\leq \alpha \leq q-\nu-1$; \cite{klappenecker05p1} \\
\hline Euclidean MDS & $[[n,n-2d+2,d]]_q$& pure&$3\leq n\leq q, 1\leq d\leq
n/2+1$; \cite{grassl03}\\
\hline Hermitian MDS & $[[q^2-s,q^2-s-2d+2,d]]_q$& pure&$1\leq d\leq
q, s=0,1$; \cite{grassl03}\\
\hline  Twisted & $[[q^2+1,q^2-3,3]]_q$& pure?&  \cite{bierbrauer00}\\
\hline
Extended Twisted & $[[q^r,q^r-r-2,3]]_q$ &pure& $r\geq 2$; \cite{bierbrauer00}\\
&$[[n,n-r-2,3]]_q$& pure &
$n=(q^{r+2}-q^3)/(q^2-1)$, $r\geq 1$, $r$ odd; \cite{bierbrauer00}\\
\hline Perfect & $[[n,n-r-2,3]]_q$&pure&
$n=(q^{r+2}-1)/(q^2-1)$, $r\geq 2$, $r$ even; \cite{bierbrauer00}
\end{tabular}
\end{center}
\end{sidewaystable}
}


%% file: chSubsys1.tex
\chapter{Subsystem Codes -- Beyond Stabilizer Codes\footnotemark}\label{ch:subsys1}
\footnotetext{Part of the material in this chapter has been submitted to IEEE and currently under review. Copyright maybe transferred to IEEE.}
\noindent 
In this chapter we study a recent generalization of quantum codes
that unifies many apparently disparate notions of quantum error
correction. This generalization called operator quantum error
correction gathers within its framework both passive and active
error correction schemes, among them decoherence free subspaces (DFS), 
noiseless subsystems (NS), and standard quantum error-correcting codes
(including stabilizer codes which formed the main theme of
the last two chapters). 
Our main contribution in this chapter is to provide a
natural construction of such codes in terms of Clifford codes, an
elegant generalization of stabilizer codes due to
Knill. Character-theoretic methods are used to derive a simple method
to construct operator quantum error-correcting codes from any
classical additive code over a finite field, which obviates the need for
self-orthogonal codes. In view of its importance and also to better appreciate our
contribution we shall spend a little time reviewing operator quantum 
error correction. The following review 
summarizes the key points of \cite{kribs05,kribs05b} relevant
for our discussion. A quick word about the nomenclature. These codes
were originally studied in the context of operator algebras and hence,
were named operator quantum error correcting codes. We shall often
use the descriptive term subsystem codes in view of brevity. Both will
be used interchangeably. 

\smallskip

{\small \textit{Notation.} If $N$ is a group, then $Z(N)$ denotes the
center of $N$.  We denote by $\Irr(N)$ the set of irreducible
characters of $N$. If $\chi$ and $\psi$ are characters of $N$, then
$(\chi,\psi)_N=|N|^{-1}\sum_{n\in N} \chi(n)\psi(n^{-1})$ defines a
scalar product on the vector space of class functions on $N$, and
$\Irr(N)$ is an orthonormal basis of this space. We denote by $\supp(\chi)=\{ n\in N|\, \chi(n)\neq 0\}$.
If $\chi\in \Irr(N)$,
then $Z(\chi)=\{ n\in N\,|\, \chi(1)=|\chi(n)|\}$ denotes the
quasikernel of $\chi$. Suppose that $G$ is a group that contains $N$
as a subgroup.  If $\phi\in \Irr(G)$, then $\phi_N$ denotes the
restriction of this character to $N$. If $x,y\in N$, then
$[x,y]=x^{-1}y^{-1}xy$ is the commutator. If $A$ and $B$ are subgroups
of a group, then $[A,B]=\langle [a,b]\,|\, a,\in A \text{ and } b\in
B\rangle$ is the commutator subgroup of $A$ and $B$. In particular,
$N'=[N,N]$ denotes the derived subgroup of $N$.  The reader can find
background material on finite groups in \cite{robinson95} and on
character theory in \cite{isaacs94}.
As usual let $\Hi$ be the system Hilbert space under consideration. Let
$\Bh$ denote bounded linear operators on $\Hi$.  
}

\section{Review of  Operator Quantum Error Correction}

The class of codes which we considered in the last two chapters 
come within the framework of a model often called the standard
model. Mathematically, this model 
is defined as a triple $(\R, \E, \mc{C})$, where
$\E $ is the quantum channel, $\mc{C}$  a subspace  of $\Hi$
and $\R $  a recovery operation. 
Additionally, we define a projector $P$ onto the codespace $\mc{C}$, thus 
$\mc{C}=P \Hi $. For any density operator
$\rho$ supported by $\mc{C}$ {\em i.e.} $\rho $ in $\mathcal{B}(\mc{C})$ or equivalently $\rho = P\rho P$, the triple satisfies the following relation:
\begin{eqnarray}
(\mathcal{R}\circ \mathcal{E} )(\rho) &=& \rho \mbox{ for all } \rho = P \rho P.
\end{eqnarray}

As we can see the standard model assumes a recovery operation $\mathcal{R}$.
In general $\mathcal{R}$ is  nontrivial which in turn implies some form of active monitoring of
the encoded quantum information in order to detect and correct the errors
that occur. An alternative approach is to rely on
passive error correction mechanisms, exemplified by decoherence free
subspaces and noiseless subsystems.  

If we want to avoid performing active error correction, we are naturally
led to the idea that the encoded states should not be affected by
the channel. In other words, we must encode into $\Fix(\mc{E})$, the fixed points of
$\E$ where
$$
\Fix(\E)= \{ \rho \in \Bh \mid \E(\rho)=\rho\}. 
$$
These fixed points can be nicely characterized for a certain class of quantum channels. 
Given a quantum channel $\mathcal{E}$, we can write the channel in 
terms of its Kraus operators as follows
\begin{eqnarray}
\mathcal{E}(\rho) &=& \sum_{i}E_i \rho E_i^{\dagger}.	\label{eq:krausRep}
\end{eqnarray}
Because of this decomposition we often write the channel $\mathcal{E} =\{E_i,E_i^{\dagger} \}$.
When the quantum  channels  satisfy the condition 
\begin{eqnarray}
\sum_{i}E_i  E_i^{\dagger} = I, \label{eq:unital}
\end{eqnarray}
we have a convenient way to characterize the fixed points. Channels satisfying equation~(\ref{eq:unital})
are called unital channels. 
Let $\rho E_i =E_i \rho $ for any $E_i$. Then under the unital assumption all such $\rho$
are fixed points of $\mc{E}$ as 
\begin{eqnarray}
\mathcal{E}(\rho) &=& \sum_{i}E_i \rho E_i^{\dagger} = \rho 
\sum_{i}E_i  E_i^{\dagger} = \rho.\label{eq:AsubsetFix}
\end{eqnarray}
We denote by $\mathcal{A}$, the matrix polynomials generated by $\{ E_i, E_i^{\dagger} \}$
{\em i.e.}, the algebra generated by $\{E_i, E_i^{\dagger}\}$. This is  called the
\textsl{interaction algebra} \index{interaction algebra} in the literature. 
The \textsl{noise communtant} $\mathcal{A}'$ \index{noise commutant}
is defined as 
\begin{eqnarray}
\mc{A}'= \left\{\rho \in \Bh \mid \rho E=E\rho \mbox{ for any } E \in \{E_i,E_i^{\dagger} \}\right\}.
\end{eqnarray}
From equation~(\ref{eq:AsubsetFix}), it follows that $\mc{A}' \subseteq \Fix(\E)$. In fact, for unital channels it was
shown that $\Fix(\E) = \mc{A}'$. 
Using results on $\C^*$ algebras, 
Kribs {\em et al.}, showed that the 
interaction algebra
has a representation of the form 
\begin{eqnarray}
\mc{A} \cong \bigoplus_j I_{K_j} \otimes \mc{B}(\Hi_{j}^B) \cong \bigoplus_j I_{K_j} \otimes M_{R_j},
\label{eq:chStruct}
\end{eqnarray}
where $M_{R_j}$ is $R_j$-dimensional matrix algebra (over $\C$).
This representation induces the following structure on $\Hi$
\begin{eqnarray}
\Hi \cong \bigoplus_j \Hi_{j}^A \otimes \Hi_{j}^{B},
\end{eqnarray}
where  $\dim \Hi_{j}^A = K_j$ and $\dim \Hi_{j}^B = R_j$. 
Since $\E$ (and $\mc{A}$) act trivially on $\Hi_{j}^A$, the
subsystems $\Hi_j^{A}$ are called noiseless subsystems. \index{noiseless subsystems}
To simplify matters we usually encode into only
one subsystem, which gives us the following decomposition 
\begin{eqnarray}
\Hi &=& \mc{C}\oplus \mc{C}^\perp = (\Hi^A\otimes \Hi^B) \oplus \mc{C}^\perp,
\end{eqnarray}
where $\mc{C}^\perp$ is the complement of $\mc{C}$. 
Let  $\dim \Hi^A = K$ and $\dim \Hi^B=R$. Then $\dim \mc{C}= KR $ and $\dim \mc{C}^\perp=\dim \Hi -KR$.
Let us denote operators in $\mc{B}(\Hi^A)$ and  $\mc{B}(\Hi^B)$ as $\rho^A$ and 
$\rho^B$ respectively. 
The (standard) noiseless subsystem \index{noiseless subsystems} 
given by $\mc{C}$  consists of operators in 
$\mc{B}(\Hi^A\otimes \Hi^B)$ that are of the form $\mc{B}(\Hi^A) \otimes I_{R}$
in other words  $\rho^A\otimes I_R$.
In this case the co-subsystem $B$ is in the maximally mixed state.
The codespace $\mc{C}$ is an algebra of operators.
Decoherence free subspaces are noiseless subsystems with the dimension of the
co-subsystem equal to one. \index{decoherence free subspaces}
In this case the codespace $\mc{C}$ 
is a subspace of $ \Hi$.
 
One of the insights of \cite{kribs05} was that we can relax the constraint
that the co-subsystem $B$ should be in the maximally mixed state. This led to the idea of
generalized noiseless subsystems. \index{noiseless subsystems!generalized}
In this case the noiseless subsystem code is
given by the operators in $\mc{B}(\Hi)$ that are of the form $(\rho^A \otimes \rho^B)$.
Comparing with equation~(\ref{eq:chStruct}) we can see that in this case we 
are not always encoding into the fixed points of $\mc{E}$. The codespace
instead of being an algebra of operators is now a monoid\footnote{In \cite{kribs05},
they refer to $\mc{C}$ as a semigroup even though $\mc{C}$ is equipped with identity.} 
of operators of the form $\rho^A\otimes \rho^B$. 
Given a decomposition of $\Hi= \Hi^A\otimes \Hi^B \oplus \mc{C}^\perp$ and
orthonormal bases $\{\ket{\alpha_i}\}_{i=1}^n$, 
and $\{\ket{\beta_j}\}_{j=1}^m$ for $\Hi^A$ and $\Hi^B$ respectively, we define a 
projector onto $\mc{C}=\Hi^A\otimes \Hi^B  = P \Hi$ as 
\begin{eqnarray}
P &=&  \mbf{1}^{AB} = \mbf{1}^A \otimes \mbf{1}^B = (\sum_i\ket{\alpha_i}\bra{\alpha_i}) \otimes (\sum_j\ket{\beta_j}\bra{\beta_j}).
\end{eqnarray}
The action of $P$ on $\rho$ is defined as $P \rho P$.
Then a generalized noiseless subsystem is defined as follows, see \cite[Lemma~2]{kribs05}.
\begin{lemma}[Generalized noiseless subsystems \cite{kribs05}] \index{noiseless subsystems!generalized}
Given a fixed decomposition of $\Hi = \Hi^A\otimes \Hi^B \oplus \mc{C}^\perp$ and  a CPTP map $\mc{E}$,
define $\mc{C} = \{\rho \in \mc{B}(\Hi)\mid \rho = \rho^A\otimes \rho^B \}$.
Then the following conditions are equivalent and define a generalized noiseless subsystem 
$\Hi^A$.
\begin{compactenum}[i)]
\item $\mc{E}(\rho^A\otimes\rho^B) = \rho^A\otimes \sigma^B$, for all $\rho^A\otimes \rho^B \in \mc{C}$
and  some $\sigma^B$.
\item $\mc{E}(\rho^A\otimes I_B) = \rho^A\otimes \sigma^B$, for all $\rho^A\otimes I_B \in \mc{C}$ and some $\sigma^B$.
\item $(\Tr_A \circ P\circ \mc{E})(\rho) = \Tr_A(\rho)$, for all $\rho \in \mc{C}$.
\end{compactenum}
\end{lemma} 
Kribs {\em et al.}, \cite{kribs05,kribs05b} generalized these ideas further by incorporating
active error correction also on the subsystem $A$. As in the standard model we now
define a recovery operation $\mc{R}$, that restores the subsystem $B$ after the 
error. The definition is as follows. 
\begin{lemma}[Operator quantum error correcting codes \cite{kribs05}]\index{subsystem code}
\index{operator quantum error-correcting code}
Given a fixed decomposition of $\Hi = \Hi^A\otimes \Hi^B \oplus \mc{K}$ and  a CPTP map $\mc{E}$,
define $\mc{C} = \{\rho \in \mc{B}(\Hi)\mid \rho = \rho^A\otimes \rho^B \}$.
Then the following conditions are equivalent and define an operator quantum error correcting
code $\mc{C}$ with recovery operation $\mc{R}$.
\begin{compactenum}[i)]
\item $\mc{R} \circ \mc{E}(\rho^A\otimes\rho^B) = \rho^A\otimes \sigma^B$, for all $\rho^A\otimes \rho^B \in \mc{C}$
and  some $\sigma^B$
\item $\mc{R} \circ \mc{E}(\rho^A\otimes I_B) = \rho^A\otimes \sigma^B$, for all $\rho^A\otimes I_B \in \mc{C}$ and some $\sigma^B$
\item $(\Tr_A \circ P\circ \mc{R} \circ \mc{E})(\rho) = \Tr_A(\rho)$, for all $\rho \in \mc{C}$.
\end{compactenum}
\end{lemma}

We are often more interested in a simple condition that identifies 
correctable errors for a given channel $\mc{E} =\{ E_a,E_a^\dagger\}$ or equivalently,
the detectable errors for a given subspace in $\Hi$. 
Recall that if a code corrects the set of errors in $\Sigma = \{ E_a\}$, it detects
all the errors in the algebra $\Sigma_D= \{ E_a^\dagger E_b \mid E_a,E_b\in \Sigma\}$.
\begin{theorem}[\cite{kribs05,nielsen05}]
Let $\Hi= \Hi^A\otimes \Hi^B\oplus \mc{K}$ and $P=\mbf{1}^A\otimes \mbf{1}^B$ be a
projector onto $\mc{C} = \Hi^A\otimes \Hi^B =P \Hi$. 
Then an  error $E$  is detectable by the
operator quantum error correcting code $\mc{C}$
if and only if 
\begin{eqnarray}
P E  P = \mbf{1}^A\otimes \rho_E^B \mbox{ for some } 
\rho_{E}^B \in \mc{B}(\Hi^B).\label{eq:oqecDetect}
\end{eqnarray}
\end{theorem}

Now that we have reviewed the salient ideas of operator quantum
error correction, we will address a very important question -- how
do we systematically construct these codes? Two important contributions
in this direction were the introduction of a stabilizer formalism 
and the notion of a gauge group by Poulin \cite{poulin05},  and
construction of a class of subsystem codes 
capable of encoding one qubit by Bacon \cite{bacon06a}. 
However the bigger question of 
systematic construction of good subsystem codes still remained open. 
Our work addresses this problem in more detail. Subsequent to the publication
of this work, Bacon and Cassacino independently proposed a class of
subsystem codes \cite{bacon06b}; these codes can be viewed as a special case
of the codes constructed in this chapter. More details on these codes will
be given in Chapter~\ref{ch:subsys2}.

Our approach is based on an elegant formalism to construct quantum
error-correcting codes that has been introduced in 1996 by Knill
as a generalization of the stabilizer code concept.  At the heart of
this quantum code construction is a famous theorem by Clifford
concerning the restriction of irreducible representations of finite
groups to normal subgroups, so these  codes were termed as
``Clifford codes'' in \cite{klappenecker031,klappenecker033}, although
``Knill codes'' is perhaps a more appropriate name.  Unexpectedly, it
turned out that Clifford codes are in many cases stabilizer codes, so
this construction did not become as widely known.

In our approach, we construct a Clifford code $C$ and give conditions
that ensure that this code decomposes into a tensor product $C=
A\otimes B$. The Clifford codes allow us to control the dimensions of
$A$ and $B$, and we get a simple characterization of the detectable
errors of the operator quantum error-correcting code.  Since there may
exist many different ways to construct the same Clifford code $C$, we
should note that these constructions can lead to different tensor
product decompositions. In fact, even if one is just interested in the
tensor decomposition of a stabilizer code $C$, then the Clifford codes
can provide a natural way to induce an operator quantum
error-correcting code on $C$.

\section{A Detour Through Clifford Codes} 
As we have seen in the previous sections and in Chapter~\ref{ch:stabq1},
the study of quantum codes is related to the operators acting on the
system Hilbert space. To simplify matters we can restrict our
attention to a basis of these operators and the group generated by
that basis, called the error group. In the binary case we deal with
the familiar Pauli matrices and the group generated by them on 
$n$ qubits.   Knill generalized this concept 
by introducing the notion of nice error bases and abstract error
groups which generalize the Pauli error group. We have already seen one
application of this generalization in Chapter~\ref{ch:stabq1}, where
we dealt with the generalization of  the Pauli group to nonbinary 
alphabet. The benefit of the abstract approach is that it will 
free us from having to deal with cumbersome matrix operators but 
instead work with groups. The representations of the groups
(in $\Hi$) will bring us  back to the concrete world of operators. 
In this chapter, we shall pursue this
abstract approach permitting different error groups
other than the Pauli error group. We
say that a finite group $E$ is an abstract error group \index{abstract error group}
if it has a
faithful irreducible unitary representation $\rho$ of degree
$d = |E:Z(E)|^{1/2}$. The irreducibility of the representation ensures that
one can express any error acting on $\C^d$ as a linear
combination of the matrices $\rho(g)$, with $g\in E$. The fact that
the representation is faithful and has the largest possible degree
ensures that the set of matrices $\{ \rho(g)\,|\, g\in T\}$, where $T$
is a set of representatives of $E/Z(E)$, forms a \textit{basis}\/ of
the vector space of $d\times d$ matrices.

A Clifford code is constructed with the help of a normal subgroup $N$
of the error group $E$ and an irreducible character $\chi$ of $N$.
Let $\phi$ denote the irreducible character corresponding to the
representation $\rho$ of the group $E$, that is, $\phi(g)=\Tr \rho(g)$
for $g\in E$. Suppose that $N$ is a normal subgroup of $E$ and that
$\chi$ is an irreducible character of $N$ such that $(\chi,\phi_N)_N>
0$.  

\begin{defn}[Clifford codes]\index{Clifford code}
A Clifford code $C$ corresponding to $(E,\rho,N,\chi)$ is
defined as the image of the orthogonal projector
$$ P = \frac{\chi(1)}{|N|} \sum_{n\in N} \chi(n^{-1})\rho(n),$$
see~\cite[Theorem~1]{klappenecker033}. 
\end{defn}

We emphasize that if we refer
to a Clifford code with data $(E,\rho,N,\chi)$, then it is assumed that
$(\chi,\phi_N)>0$, as this condition ensures that $\dim C>0$.
Recall that an error $e$ in $E$ is detectable by the (Clifford) quantum code $C$
if and only if $P\rho(e)P=\lambda_e P$ holds for some $\lambda_e\in \C$.

The image of $P$ is the homogeneous component that consists of the
direct sum of all irreducible $\C N$-submodules with character $\chi$
that are contained in the restriction of $\rho$ to~$N$.  The elements
$e$ in $E$ that satisfy $\rho(e)C= C$ form a group known as
the inertia group $I_E(\chi)= \{ g\in E\,|\, \chi(gxg^{-1})=\chi(x)
\text{ for all } x\in N\}.$ We note that $C$ is an irreducible
$\C[I_E(\chi)]$-module. Let $\vartheta$ be the irreducible character
corresponding to this module.

\begin{fact}
\label{fact:clifford}
Let $C$ be a Clifford code with data $(E,\rho,N,\chi)$. Then the
dimension of the code is given by $\dim C = |Z(E)\cap N| |E:
Z(E)|^{1/2}\chi(1)^2/|N|$.  An error $e$ in $E$ can be detected by $C$
if and only if $e$ is in $E-(I_E(\chi)-Z(\vartheta))$.
\end{fact}

For a proof of this fact see \cite{klappenecker033} and for more
background on Clifford codes see~\cite{klappenecker031} and the
seminal papers~\cite{knill96a,knill96b}.

\section{Constructing Operator Quantum Error-Correcting Codes} 
We are now concerned with the construction of a  decomposition of the
Hilbert space $\Hi$ in the form
$$ \Hi = (A\otimes B) \oplus C^\perp.$$ Put differently, we seek a
decomposition of the Clifford code $C$ as a tensor product $A\otimes
B$.

The next theorem gives a construction of operator quantum
error-correcting codes when one can express the inertia group
$I_E(\chi)$ as a central product $I_E(\chi)=LN$, where $L$ is a
subgroup of $E$ such that $[L,N]=1$.

\begin{theorem}\label{th:first}
Suppose that $C$ is a Clifford code with data $(E,\rho, N,\chi)$.  If
the inertia group $I_E(\chi)$ is of the form $I_E(\chi)=LN$, where $L$
is a subgroup of $E$ such that $[L,N]=1$, then $C$ is an operator
quantum error-correcting code $C= A\otimes B$ such that
\begin{compactenum}[i)]
\item $\dim A = |Z(E)\cap N| |E: Z(E)|^{1/2}\chi(1)/|N|$,
\item $\dim B=\chi(1)$.
\end{compactenum}
The subsystem $A$ is an irreducible $\C L$-module 
with character $\chi_A\in \Irr(L)$. 
An error $e$ in $E$ is detectable by
subsystem $A$ if and only if $e$ is contained in the set
$E-(I_E(\chi)-Z(\chi_A)N)$.
\end{theorem}
\begin{proof}
Since the Clifford code $C$ is an irreducible $\C[I_E(\chi)]$-module
and $I_E(\chi)=LN$, with $[L,N]=1$, there exists an irreducible $\C
L$-module $A$ and an irreducible $\C N$-module $B$ such that $C\cong
A\otimes B$, see~\cite[Proposition 9.14]{GLS2}. If $\chi_A\in \Irr(L)$
is the character associated with the module $A$, $\chi_B\in \Irr(N)$
the character associated with $B$, and $\vartheta\in \Irr(I_E(\chi))$
the character associated with $C$, then $\vartheta$ is of the form
$\vartheta(\ell n)=\chi_A(\ell)\chi_B(n)$ with $\ell \in L$ and $n\in
N$.

As the restriction of $C$ to a $\C N$-module contains an irreducible
$\C N$-module $W$ with character $\chi$, we must have
$$ \begin{array}{lcl}
\displaystyle 
(\vartheta_N , \chi)_N = \frac{1}{|N|} \sum_{n\in N}
\vartheta(1,n^{-1}) \chi(n) &=& 
\displaystyle \frac{1}{|N|} \sum_{n\in N}
\chi_A(1)\chi_B(n^{-1}) \chi(n) \\ 
&=& \chi_A(1)  (\chi_B, \chi)_N > 0.
   \end{array}
$$ 
Since $\Irr(N)$ forms an orthonormal basis with
respect to $(\,\cdot\,,\,\cdot\,)_N$, we
can conclude that the irreducible character $\chi_B$ must be equal to
$\chi$. It follows that $C\cong A\otimes W$.

The dimension of $W\cong B$ is $\chi(1)$, and by Fact~\ref{fact:clifford} 
the dimension of $C$ is given by 
$$ \Tr P = |Z(E)\cap N| |E:Z(E)|^{1/2} \chi(1)^2/|N|.$$
The dimension of $B$ follows from the formula $\dim C=\dim A\dim B$. 

Note that the projector for $C$ acts as $\mbf{1}^{AB} =\mbf{1}^A\otimes \mbf{1}^B$
on $C$. By \cite[Theorem~1]{klappenecker033}, an error $e\in E-I_E(\chi)$ 
maps $C$ to an orthogonal complement, so
$eP$ and $P$  project onto orthogonal subspaces and 
we have $PeP =0$; by equation~(\ref{eq:oqecDetect}) the error
$e$ is detectable\footnote{ Alternatively by Fact~\ref{fact:clifford}, the error $e$
is detectable when we view $C$ as a Clifford code. When viewed as an operator
quantum error correcting code, we encode only into a subspace of $C$, therefore
$e$ still remains detectable.}
An error $e$ in $Z(\chi_A)N$ acts by
scalar multiplication on $A$ and arbitrarily on $B$, so 
$eP = \mbf{1}^A\otimes \rho^B$ for some $\rho^B \in \mc{B}(B)$.
Thus $PeP = \mbf{1}^A \otimes \mc{B}(B)$; again by equation~(\ref{eq:oqecDetect})
these errors are detectable (harmless would be a better
word). Therefore, all errors in $E-(I_E(\chi)-Z(\chi_A)N)$ are
detectable. Conversely, an error $e$ in $I_E(\chi)-Z(\chi_A)N$
cannot be detectable, since $e$ does not act by scalar multiplication
on $A$. We have $eP \neq \mbf{1}^A\otimes \rho^B $. Therefore 
$PeP\neq \mbf{1}^A\otimes \rho^B $ and thus $e$ is an undetectable
error. 
\end{proof}

The data given in the previous theorem can be easily computed,
especially with the help of a computer algebra system such as GAP or
MAGMA. 

We will now consider some important special cases.  Recall that most
abstract error groups that are used in the literature satisfy the
constraint $E'\subseteq Z(E)$ (put differently, the quotient group
$E/Z(E)$ is abelian). In that case, we are able to obtain a
characterization of the resulting operator quantum error-correcting
codes that does not depend on the choice of the character $\chi$.

\begin{theorem}\label{th:second} 
Suppose that $E$ is an abstract error group such that $E'\subseteq Z(E)$. 
Suppose that $C$ is a Clifford code with data $(E,\rho, N,\chi)$.  
In this case, the inertia group is given by $I_E(\chi)=C_E(Z(N))$. 
If $C_E(Z(N))=LN$ for some subgroup $L$ of $E$ such that $[L,N]=1$, 
then $C$ is an operator
quantum error-correcting code $C= A\otimes B$ such that
\begin{compactenum}[i)]
\item $\dim A = |Z(E)\cap N| |E: Z(E)|^{1/2}|N:Z(N)|^{1/2}/|N|$,
\item $\dim B=|N:Z(N)|^{1/2}$.
\end{compactenum}
An error $e$ in $E$ is detectable by
subsystem $A$ if and only if $e$ is contained in the set
$E-(C_E(Z(N))-Z(L)N)$.
\end{theorem}
\begin{proof}
Since the abstract error group $E$ satisfies the condition
$E'\subseteq Z(E)$, the inertia group of the character $\chi$ in $E$
can be fully determined; it is given by $T:=I_E(\chi)=C_E(Z(N))$, see
\cite[Lemma~5]{klappenecker033}. 

Suppose that 
$$ P_1 = \frac{\chi(1)}{|N|}\sum_{n\in N} \chi(n^{-1})\rho(n)$$ is the
orthogonal projector onto $C$. The assumption $E'\subseteq Z(E)$
implies that there exists a linear character $\varphi$ of $\Irr(Z(N))$
such that 
$$ P_2 = \frac{1}{|Z(N)|}\sum_{n\in Z(N)} \varphi(n^{-1})\rho(n)$$
satisfies $P_1=P_2$, see \cite[Theorem~6]{klappenecker033}.

Let $\phi$ be the character of the representation $\rho$, that is,
$\phi(g)=\Tr \rho(g)$ for $g\in E$. We have $ \Tr P_1 = \chi(1)^2
\phi(1)|N\cap Z(E)|/|N|$ and $\Tr P_2 = \phi(1)|N\cap
Z(E)|/|Z(N)|$. Since $P_1=P_2$ project onto the codespace $C$, and
$\dim C>0$, we have $\Tr P_1/\Tr P_2=1$, which implies
$\chi(1)^2=|N\colon Z(N)|$. Therefore, the claims i) and ii) follow
from Theorem~\ref{th:first}.

Let $\vartheta\in \Irr(T)$ be the character associated with the 
$\C[T]$-module $C$; put differently, $\vartheta$ is the unique
character in $\Irr(T)$ that satisfies $(\vartheta_N,\chi)_N>0$ and
$(\phi_T,\vartheta)_T>0$. Since $Z(E)\le T$ and
$(\phi_T,\vartheta)_T>0$, it follows from Lemma~\ref{l:support2} that
$\supp(\vartheta)=Z(T)$. 

Since the inertia group $T$ is a central product given by $T=LN$ with
$[L,N]=1$, there exist characters $\chi_A\in \Irr(L)$ and
$\chi_B=\chi\in \Irr(N)$ such that $\vartheta(\ell
n)=\chi_A(\ell)\chi(n)$ for $\ell\in L$ and $n\in N$.  By
Lemma~\ref{l:center}, we have $Z(T)=Z(L)Z(N)$; thus,
$\supp(\vartheta)=Z(L)Z(N)$. This implies that $\supp(\chi_A)=L\cap
Z(L)Z(N)=Z(L)$; hence $Z(\chi_A)=Z(L)$.  The characterization of the
detectable errors is obtained by substituting these facts in
Theorem~\ref{th:first}.
\end{proof}

In the previous theorem, we still need to check whether $C_E(Z(N))$
decomposes into a central product of $N$ and some group $L$. In the
case of extraspecial $p$-groups (which is arguably the most popular
choice of abstract error groups) the decomposition of the inertia
group into a central product is always guaranteed, as we will show
next.

Recall that a finite group $E$ whose order is a power of a prime $p$
is called extraspecial if its derived subgroup $E'$ and its center
$Z(E)$ coincide and have order $p$. An extraspecial $p$-group is an
abstract error group. The quotient group $\overline{E}=E/Z(E)$ is the
direct product of two isomorphic elementary abelian
$p$-groups. Therefore, one can regard $\overline{E}$ as a vector space
$\F_p^{2n}$ over the finite field $\F_p$.

Let $\zeta$ be a fixed generator of the cyclic group $Z(E)$.  As the 
commutator $[x,y]$ depends only on the cosets $\overline{x}=xZ(E)$ and
$\overline{y}=yZ(E)$, one can determine a well-defined function $s\colon
\overline{E}\times \overline{E}\rightarrow \F_p$ by
$[x,y]=\zeta^{s(\overline{x},\overline{y})}$. The function $s$ is a
nondegenerate symplectic form. We note that two elements $x$ and $y$
in $E$ commute if and only if $s(\overline{x},\overline{y})=0$.
We write $\overline{x}\, \sdual\, \overline{y}$ if and only if 
$s(\overline{x},\overline{y})=0$.

For a subgroup $G$ of $E$, we will use $\overline{G}$ to denote $G/Z(E)$. 
\begin{lemma}\label{l:inertia}
If $E$ is an extraspecial $p$-group and $N$ a normal subgroup of $E$,
then $C_E(Z(N))=NC_E(N)$.
\end{lemma}
\begin{proof}
Since $Z(E)\le NC_E(N)\le C_E(Z(N))$, it suffices to show that the
dimensions of the $\F_p$-linear vector spaces 
$$\overline{NC_E(N)} \quad \text{and}\quad \overline{C_E(Z(N))}$$ are
the same. Suppose that $z=\dim\overline{Z(N)}$ and
$k=\dim\overline{N}$. 
Then
$$
\begin{array}{lcl}
\dim \ol{NC_E(N)}&=& \dim (\ol{N}+\ol{N}^\sdual)
= \dim\ol{N}+\dim\ol{N}^\sdual-\dim(\ol{N}\cap \ol{N}^\sdual)\\
&=& \dim\ol{N}+\dim\ol{N}^\sdual-\dim(\ol{Z(N)})\\
&=& k+ (2n-k)-z=2n-z,
  \end{array}
$$ which coincides with $\dim
\ol{C_E(Z(N))}=\dim\ol{Z(N)}^\sdual=2n-z$, and this proves our claim.
\end{proof}

The next theorem shows that it suffices to choose a normal subgroup
$N$ of the extraspecial $p$-group $E$, and this choice determines the
parameters of an operator quantum error-correcting code provided by a
Clifford code $C$.

\begin{theorem}\label{th:third} 
Suppose that $E$ is an extraspecial $p$-group.  If $C$ is a Clifford 
code with data $(E,\rho, N,\chi)$, with $N\neq 1$, 
then $C$ is an operator quantum
error-correcting code $C= A\otimes B$ such that
\begin{compactenum}[i)]
\item $\dim A = |Z(E)\cap N| |E: Z(E)|^{1/2}|N:Z(N)|^{1/2}/|N|$,
\item $\dim B=|N:Z(N)|^{1/2}$.
\end{compactenum}
An error $e$ in $E$ is detectable by
subsystem $A$ if and only if $e$ is contained in the set
$E-(NC_E(N)-N)$.
\end{theorem}
\begin{proof}
The inertia group $I_\chi(E)=C_E(Z(N))$, since $E'\subseteq Z(E)$,
see~\cite[Lemma~5]{klappenecker033}. By Lemma~\ref{l:inertia}, we have
$I_E(\chi)=LN=NL$ with $L=C_E(N)$. Thus, $C$ is an operator quantum
error-correcting code and the statements i) and ii) follow from
Theorem~\ref{th:second}. Furthermore, Theorem~\ref{th:second} shows
that an error $e$ in $E$ is detectable if and only if $e\in
E-(NC_E(N)-Z(L)N)$. Since $E$ is a $p$-group and $N\neq 1$, we have
$N\cap Z(E)\neq 1$; hence $Z(E)\le N$. We note that
$\ol{Z(L)}\subseteq \ol{L}\cap \ol{L}^\sdual=
\ol{N}^\sdual\cap\ol{N}\subseteq \ol{N}$; therefore, $N\subseteq Z(L)N\subseteq Z(N)N=N$, forcing $Z(L)N=N$.
\end{proof}

The normal subgroup $N$ used in the construction of subsystem codes 
will henceforth be called as the \textsl{gauge group}. \index{gauge group}
This definition 
coincides with the definition of the gauge group in \cite{poulin05}.

\section{Subsystem Codes from Classical Codes} 
We conclude this chapter by showing how the previous results can be
related to classical coding theory.
Let $a$ and $b$ be elements of the finite field $\F_q$ of characteristic $p$.  
Recall that in 
Section~\ref{ssec:errBases} we defined 
the unitary operators $X(a)$ and $Z(b)$ on~$\C^q$ by
$$ X(a)\ket{x}=\ket{x+a},\qquad Z(b)\ket{x}=\omega^{\tr(bx)}\ket{x},$$
where $\tr$ denotes the trace operation from the extension field
$\F_q$ to the prime field $\F_p$, and $\omega=\exp(2\pi i/p)$ is a
primitive $p$th root of unity.  Let $\mathbf{a}=(a_1,\dots, a_n)\in
\F_q^n$. We write $ X(\mathbf{a}) = X(a_1)\otimes\, \cdots \,\otimes
X(a_n)$ and $Z(\mathbf{a}) = Z(a_1)\otimes\, \cdots \,\otimes Z(a_n)$
for the tensor products of $n$ error operators. One readily checks that 
the group
$$ E = \langle X(a), Z(b)\,|\, a, b\in \F_q^n\rangle$$ is an
extraspecial $p$-group of order $pq^{2n}$. As a representation $\rho$,
we can take the identity map on $E$. We have $E/Z(E)\cong \F_q^{2n}$. 

We need to introduce a notion of weights of errors.  Recall that an
error in $E$ can be expressed in the form $\alpha X(a)Z(b)$ for some
nonzero scalar $\alpha$. The weight of $\alpha X(a)Z(b)$ is defined as
$|\{ i\,|\, 1\le i\le n, a_i\neq 0 \text{ or } b_i\neq 0\}|$, that is,
as the number of quantum systems that are affected by the
error. Similarly, we can introduce a weight on vectors of $\F_q^{2n}$
by
$$ \swt(a|b)=\{ i\,|\, 1\le i\le n, a_i\neq 0 \text{ or } b_i\neq 0\}|$$
for $a,b\in \F_q^n$.

Theorem~\ref{th:third} suggests 
the following approach to construct 
operator quantum error-correcting codes. 
\begin{theorem}\label{th:four}
Let $X$ be a classical additive subcode of\/ $\F_q^{2n}$ such that
$X\neq \{0\}$ and let $Y$ denote its subcode $Y=X\cap X^\sdual$. 
If $x=|X|$ and $y=|Y|$, then
there exists an operator quantum error-correcting code $C=
A\otimes B$ such that
\begin{compactenum}[i)]
\item $\dim A = q^n/(xy)^{1/2}$, 
\item $\dim B = (x/y)^{1/2}$.
\end{compactenum}
The minimum distance of subsystem $A$ is given by
$d=\swt((X+X^\sdual)-X)=\swt(Y^\sdual-X)$. Thus, the subsystem $A$ can
detect all errors in $E$ of weight less than $d$, and can correct all
errors in $E$ of weight $\le \lfloor (d-1)/2\rfloor$.
\end{theorem}
\begin{proof}
Let $E$ be the extraspecial $p$-group of order $pq^{2n}$, and let $N$
be the full preimage of $\ol{N}=X$ in $E$ under the canonical quotient
map.  Therefore, we can apply Theorem~\ref{th:third}. The remainder of
the proof justifies how the parameters given in Theorem~\ref{th:third}
can be expressed in terms of the code sizes $x$ and $y$.

Then $\ol{Z(N)}=X\cap X^\sdual=Y$.  By definition, $N$
contains $Z(E)$; hence, $Z(E)\le Z(N)$. 
It follows that
$|N:Z(N)|=|\overline{N}:\ol{Z(N)}|=x/y$, so ii) follows from
Theorem~\ref{th:third}. For the claim i), we remark that
$x=|X|=|N|/p$, which implies that $\dim A =(p/|N|)
|E:Z(E)|^{1/2}|N:Z(N)|^{1/2}=q^n(x/y)^{1/2}/x$.

The minimum distance of subsystem $A$ is the weight of the smallest
nondetectable error, so it is the minimum weight of an error in the
set $NC_E(N)-N=C_E(Z(N))-N$. Since the quotient map $E\rightarrow
\ol{E}$ maps an error $e$ of weight $w$ onto a vector $\ol{e}$ such
that $w=\swt{\ol{e}}$, the claim about the minimum distance follows
from the observations that $\ol{NC_E(N)-N}=(X+X^\sdual)-X$ and 
$\ol{C_E(Z(N))-N}=Y^\sdual-X$. 
\end{proof}

\begin{remark}\label{rm:symFormChoice}
As in the case of stabilizer codes, the most general symplectic form we
can choose is $\langle u|v \rangle_s = \tr_{q/p}(a'\cdot b - a\cdot b')$,
where $u=(a|b)$ and $v=(a'|b')$ are in $\F_q^{2n}$. 
We define the trace symplectic dual
as  $C^\sdual = \{x\in \F_q^{2n} \mid \langle x|y\rangle_s=0, \mbox{ for all }y \in C\}$. 
In case of $\F_q$-linear codes, the trace symplectic form 
$\langle (a|b)|(a'|b') \rangle_s$ vanishes if and only if 
$a'\cdot b - a\cdot b'$
vanishes. The trace symplectic dual for an $\F_q$-linear code therefore
coincides with its symplectic dual. 
So when dealing with $\F_q$-linear codes we indulge in an abuse of 
notation and denote $a'\cdot b - a\cdot b'$ also by $\langle (a|b)|(a'|b') \rangle_s$
and the duals with respect to both forms as $C^\sdual$.
\end{remark}

In the above the theorem we had been able to define the distance in terms
of the classical codes. Having made choice of the error group we can also
go back and recast the distance in terms of the gauge group as 
as $\wt(C_E(Z(N))-N)$. In addition, we can also extend the notion
of purity to subsystem codes also in a straightforward manner.
\begin{defn}[Pure and impure subsystem codes]\label{df:impureCodes}\index{subsystem code!pure}
\index{subsystem code!impure}
Let $N$ be the gauge group of a subsystem code
$Q$ with distance $d= \wt(C_E(Z(N))-N)$. We say that $Q$ is
\textit{pure to $d'$} if there is no error of weight less than $d'$ in
$N$. The code is said to be \textit{exactly pure to $d'$} if $\wt(N)$
is $d'$ and it is said to pure if $d'\geq d$ .  The code is said to be
impure if it is exactly pure to $d'< d$.  
\end{defn}
This refinement to the notion of purity was made in recognition of certain 
subtleties that had to addressed when constructing other subsystem codes 
from existing subsystem codes, see \cite{aly06} for details.

An operator quantum error-correcting code with parameters
$((n,K,R,d))_q$ is a subspace $C=A\otimes B$ of a $q^n$-dimensional
Hilbert space $H$ such that $K=\dim A$, $R=\dim B$, and the subsystem
$A$ has minimum distance $d$. The above theorem constructs an
$((n,q^n/(xy)^{1/2},(x/y)^{1/2},d))_q$ operator quantum
error-correcting code given a classical $(n,x)_q$ code $X$ and its
$(n,y)_q$ subcode $Y=X\cap X^{\sdual}$.
We write $[[n,k,r,d]]_q$ for an
$((n,q^k,q^r,d))_q$ operator quantum error-correcting code.

A further simplification of the above construction is possible which takes
any pair of classical codes to give a subsystem code. 

\begin{corollary}[Euclidean Construction]\label{th:cssoqec}
\index{subsystem code!construction!Euclidean}
Let $X_i \subseteq \F_q^n$, be $[n,k_i]_q$ linear codes where $i\in \{1,2\}$. Then 
there exists an $[[n,k,r,d]]_q$   Clifford subsystem code with
\begin{compactitem}
\item $k=n-(k_1+k_2+k')/2$, 
\item $r=(k_1+k_2-k')/2$, and 
\item 
$d=\min \{ \wt((X_1^\perp\cap X_2)^\perp\setminus X_1), 
\wt((X_2^\perp\cap X_1)^\perp\setminus X_2) \}$,
\end{compactitem}
where $k'= \dim_{\F_q}(X_1\cap X_2^\perp)\times (X_1^\perp\cap X_2)$. 
\end{corollary} 

The result follows from Theorem~\ref{th:third} by defining
$C=X_1\times X_2$; it follows that $C^\sdual = X_2^\perp\times
X_1^\perp$ and $D=C\,\cap\,C^\sdual = (X_1\cap X_2^\perp)\times (X_2\cap
X_1^\perp)$, and the parameters are easily obtained from these definitions, 
see \cite{aly06} for a detailed proof.

The notions of purity can be defined in terms of classical codes as well.
Let $C$ be an additive subcode of $\F_{q}^{2n}$ and $D=C\cap
C^\sdual$.  By theorem~\ref{th:third}, we can obtain an
$((n,K,R,d))_q$ subsystem code $Q$ from $C$
that has minimum distance $d=\swt(D^\sdual - C)$.  
If $d'\leq \swt(C)$, then we say that
the associated operator quantum error correcting code is \textit{pure
to $d'$}. 

Extending the ideas of purity to subsystem codes is useful because it facilitates
the analysis of the parameters of the subsystem codes, as will become clear
when we derive bounds in the next chapter. 

As in the case of stabilizer codes we would like one would like 
to characterize the minimum distance in terms
of the familiar Hamming weight. For this purpose, we reformulate
the above result in terms of codes of length $n$ over $\F_{q^2}$.

Let $(\beta,\beta^q)$ be a fixed normal basis of $\F_{q^2}$ over
$\F_q$. We can define a bijection $\phi$ from $\F_q^{2n}$ onto
$\F_{q^2}^n$ by setting
$$\phi((a|b))=\beta a+\beta^q b\quad \text{for}\quad (a|b)\in
\F_q^{2n}.$$ The map is chosen such that a vector $(a|b)$ of
symplectic weight $x$ is mapped to a vector $\phi((a|b))$ of Hamming
weight $x$. Recall the trace-alternating form 
$\langle v|w\rangle_a$ for vectors $v$ and $w$ in $\F_{q^2}^n$ given in 
equation~(\ref{eq:alternating}) 
$$\langle v|w\rangle_a=\tr_{q/p}\left(\frac{v\cdot w^q-v^q\cdot
w}{\beta^{2q}-\beta^q}\right).$$ It is easy to show that $
\langle c|d\rangle_s = \langle \phi(c)|\phi(d)\rangle_a$ holds for all
$c,d\in \F_q^{2n}$, see~Lemma~\ref{th:isometry}.  Specifically, we have
$c\, \sdual\, d$ if and only if $\phi(c)\, \adual\, \phi(d)$.
Therefore, the previous theorem can be reformulated terms of codes of
length $n$ over $\F_{q^2}$ as follows:

\begin{theorem}\label{th:five}\index{subsystem code!construction!additive}
Let $X$ be a classical additive subcode of\/ $\F_{q^2}^{n}$ such that
$X\neq \{0\}$ and let $Y$ denote its subcode $Y=X\cap X^\adual$. 
If $x=|X|$ and $y=|Y|$, then
there exists an operator quantum error-correcting code $C=
A\otimes B$ such that
\begin{compactenum}[i)]
\item $\dim A = q^n/(xy)^{1/2}$, 
\item $\dim B = (x/y)^{1/2}$.
\end{compactenum}
The minimum distance of subsystem $A$ is given by
$$d=\wt((X+X^\adual)-X)=\wt(Y^\adual-X),$$
where $\wt$ denotes the Hamming weight.  Thus, the subsystem $A$ can
detect all errors in $E$ of Hamming weight less than $d$, and can correct all
errors in $E$ of Hamming weight $\lfloor (d-1)/2\rfloor$ or less. 
\end{theorem}
\begin{proof}
This follows from Theorem~\ref{th:four} and the definition of the
isometry~$\phi$.
\end{proof}

The above connections of Clifford operator quantum error-correcting
codes to classical codes allow one to explore a plethora of code
constructions. Henceforth codes constructed by using Theorems~\ref{th:four},\ref{th:five}
will be referred to as \textsl{Clifford subsystem codes} \index{subsystem code}
or just subsystem codes. 
We shall give an example to illustrate the idea.
For simplicity we shall consider binary codes derived from codes
over $\F_4$ whose elements are given by $\{0,1,\omega,\omega^2 \}$, 
where $\omega^2+\omega+1=0$. Further, choosing $\beta=\omega$, the
trace alternating product simplifies as 
$\langle v|w\rangle_a =v\cdot w^2+v^2\cdot w$. Note that if  
$w=(w_1,\ldots,w_n)$, then we denote $w^2 = (w_1^2,\ldots,w_n^2)$.

\begin{example}
Let $X$ be the additive code given by the following generator matrix.
\begin{eqnarray*}
G_X &=& \left[\begin{array}{cccc}1&1&0&0\\0&0&1&1\\\omega&0&\omega&0\\
0&\omega&0&\omega \end{array} \right]
\end{eqnarray*}
Then it can be verified that $X^\adual$ is generated by 
\begin{eqnarray*}
G_{X^\adual} &=& \left[\begin{array}{cccc}\omega&\omega&0&0\\
0&0&\omega&\omega\\1&0&1&0\\
0&1&0&1 \end{array} \right].
\end{eqnarray*}
Further, $Y=X\cap X^\adual$ is generated by 
\begin{eqnarray*}
G_Y &=& \left[\begin{array}{cccc}1&1&1&1\\\omega&\omega&\omega&\omega
\end{array} \right].
\end{eqnarray*}
We see that $|X|=2^4$, while $|Y|=2^2$. Thus by Theorem~\ref{th:five} 
we have a $((4,K,R,d))_2$ Clifford subsystem code where $K=2^4/\sqrt{2^4\cdot 2^2}=1$
and $R=\sqrt{2^4/2^2}=2$. The distance of the code is 
$2$ because the $Y^\adual\setminus X$ contains
$(0,1,0,1)$ among other weight two elements. Thus we obtain a $((4,2,2,2))_2$ i.e. a 
$[[4,1,1,2]]_2$  code. This code is not a Clifford code. The associated Clifford
code is a $[[4,2,2]]_2$ code. Incidentally, this
code is the smallest error detecting subsystem code with nontrivial
dimensions for the subsystems.
\end{example}

Often linear codes are of more interest than the additive codes. So we 
shall consider a linear operator quantum error-correcting code. In this case we can look at 
Hermitian duals instead of the trace-alternating duals. Let 
$x,y\in \F_4^n$. Then we define the 
Hermitian inner product $\langle x|y\rangle_h = \sum_{i}^n x_iy_i^2$. 
Let $C\subseteq \F_4^n$ be an $\F_4$-linear code. The Hermitian dual of $C$ is defined as
$C^\hdual= \{x\in \F_4^n \mid \langle x|c\rangle_h= 0 \mbox{ for all } c \in C\}$.
From Lemma~\ref{th:classical}, 
we know that $C^\adual=C^\hdual$. So we can use
Hermitian duals in Theorem~\ref{th:five}. 

\begin{example}
Let $X\subseteq \F_4^{15}$ be a narrowsense BCH code of design distance 6. 
This code is neither
self-orthogonal nor does it contain its (Hermitian) dual. 
The generator polynomial of $X$
is given by 
$$g(x)= x^7 + x^6 + \omega x^4 + x^2 + \omega^2 x + \omega^2.$$ 
Thus $X$ is an $[15,8,\geq 6]_4$ code. A generator matrix for this
code is obtained as 
\begin{eqnarray*}
G&= &\left[\begin{array}{ccc|ccc|ccc|ccc|ccc}
1&1&0&\omega&0&1&\omega^2&\omega^2&0&0&0&0&0&0&0 \\
0&1&1&0&\omega&0&1&\omega^2&\omega^2&0&0&0&0&0&0 \\
0&0&1&1&0&\omega&0&1&\omega^2&\omega^2&0&0&0&0&0 \\
0&0&0&1&1&0&\omega&0&1&\omega^2&\omega^2&0&0&0&0 \\
0&0&0&0&1&1&0&\omega&0&1&\omega^2&\omega^2&0&0&0 \\
0&0&0&0&0&1&1&0&\omega&0&1&\omega^2&\omega^2&0&0 \\
0&0&0&0&0&0&1&1&0&\omega&0&1&\omega^2&\omega^2&0 \\
0&0&0&0&0&0&0&1&1&0&\omega&0&1&\omega^2&\omega^2 
\end{array} \right].
\end{eqnarray*}
The gauge group is the (full) preimage of $G$ under the isometry $\phi$.
The generator polynomial of its Hermitian dual is given by 
$$
x^8 + x^7 + \omega x^6 + x^5 + \omega x^4 + \omega^2 x^3 + \omega x^2 + \omega x + \omega.
$$
The generator polynomial of $Y=C\cap C^\hdual$ is given by 
$$
h(x)=x^9 + \omega x^8 + x^7 + x^5 + \omega x^4 + \omega^2 x^2 + \omega^2 x + 1.
$$
We see that $Y^\hdual$ is a $[15,9]_4$ code.
Again using Theorem~\ref{th:five} we can compute the dimensions of the subsystems $A$ and $B$ as
$2^{15}/\sqrt{4^8\cdot 4^6}=2$ and $\sqrt{4^8/4^6}=4$ respectively. The code $Y^\hdual$
has minimum weight $5$ (computed using MAGMA). Since $\wt(X)\geq 6$, it follows that
$\wt(Y^\hdual\setminus X)=5$. Thus, $X$ defines a $((15,2,4,5))_2$ code. But note that the
associated Clifford code has the parameters $((15,8,5))_2$. 
\end{example}

Further simplifications of Theorem~\ref{th:five} for constructing operator
quantum error-correcting codes can be found in \cite{aly06}. The reader can 
also find examples of Clifford subsytem codes derived from BCH codes, Reed-Solomon 
codes therein. Interested readers can also refer to \cite{bacon06b}
for a novel method to construct subsystem codes from a pair of classical codes.

\section{Conclusions} 
We have introduced a method for constructing operator quantum
error-correcting codes. We have seen that a Clifford codes $C$ offers
naturally a tensor-product decomposition $C= A\otimes B$, where
the dimensions of the subsystems are controlled by the choice of the
normal subgroup $N$ and its character $\chi$. 

Our construction in terms of classical codes is fairly simple: Any
classical (additive) code over a finite field can be used to construct
an operator quantum error-correcting code. In particular, we do not
require any self-orthogonality conditions as in the case of stabilizer
code constructions. 

The most prominent open problem concerning operator quantum
error-correcting codes is whether one can achieve better error
correction that by means of a quantum error-correcting code. The
construction given in Theorem~\ref{th:four} allows one to compare the
parameters of Clifford codes with the parameters of stabilizer codes.
One should note that a fair comparison should be made between
$[[n-r,k,d]]$ stabilizer codes and $[[n,k,r,d]]$ Clifford subsystem codes. 
In subsequent chapters we shall establish bounds on the parameters
of  subsystem codes and make a fair comparison of the subsystem codes
and stabilizer codes. Additionally, we shall also look into other aspects
which we have not considered here such as encoding subsystem codes,
the gains in encoding and decoding.

\section{Appendix}  
In this appendix, we prove some simple technical results on groups and
characters. 

\begin{lemma}\label{l:support}
Let $E$ be a finite group such that $E'\subseteq Z(E)$, and let $H$ be
a subgroup of $E$. If $\chi\in \Irr(H)$ satisfies $Z(E)\cap \ker \chi
= \{1\}$, then $\supp \chi = Z(H)$.
\end{lemma}
\begin{proof}
Let $h\in \supp(\chi)$. Seeking a contradiction, we assume that $h\in
H-Z(H)$. Since $E'\subseteq Z(E)$, there exists an element $g\in H$
such that $ghg^{-1}=zh$ with $z\in Z(E)$ such that $z\neq 1$.  Since 
$zh\in H$ and $h\in H$, we have $z\in H\cap Z(E)$. As $\chi$ is
irreducible, the element $z\in H\cap Z(E)$ is represented by $\omega
I$ for some $\omega\in \C$ by Schur's lemma; furthermore, $\omega \neq
1$, since $Z(E)\cap \ker \chi=\{1\}$.
We note that 
$ \chi(h)=\chi(ghg^{-1}) = \chi(zh)= \omega\chi(h)$, with
$\omega\neq 1$, forcing $\chi(h)=0$, contradiction. 

The elements of $Z(H)$ belong to the support of $\chi$, since they are
represented by scalar invertible matrices.
\end{proof}

\begin{lemma}\label{l:support2}
Let $E$ be a finite group such that $E'\subseteq Z(E)$, and let
$\phi\in \Irr(E)$ be a faithful character of degree
$\phi(1)=|E:Z(E)|^{1/2}$. Let $T$ be a subgroup of $E$ such that $Z(E)\le T$. 
If $\vartheta\in \Irr(T)$ and $(\phi_T,\vartheta)_T>0$, then 
$\supp(\vartheta)=Z(T)$. 
\end{lemma}
\begin{proof}
Let $Z=Z(E)$. We have $\supp(\phi)=Z$ by \cite[Lemma 2.29]{isaacs94}. 
Since the support of $\phi$ equals $Z$, it follows from the definitions that 
$$ 0<(\phi_T,\vartheta)_T = \frac{1}{|T:Z|}(\phi_Z,\vartheta_Z)_Z.$$
Clearly, $\phi_Z=\phi(1)\varphi$ and $\vartheta_Z=\vartheta(1)\theta$ 
for some linear characters $\varphi$
and 
$\theta$ of $Z$. As
$(\phi_Z,\vartheta_Z)_Z=\phi(1)\vartheta(1)(\varphi,\theta)_Z>0$, we
must have $\theta=\varphi$. Since $\phi$ is faithful, it follows that
$\varphi=\theta$ is faithful; hence, $\ker \vartheta \cap Z(E)=\{1\}$.
Thus, $\supp\vartheta=Z(T)$ by Lemma~\ref{l:support}.
\end{proof}

\begin{lemma}\label{l:center} 
Suppose that $T$ is a group with subgroups $L$ and $N$ such that
$T=LN$ and $[L,N]=1$. Then $Z(T)=Z(L)Z(N)$. 
\end{lemma}
\begin{proof} 
Since $T=LN$, an arbitrary element $z$ of $Z(T)$ can be expressed in
the form $z=ln$ for some $l\in L$ and $n\in N$.  For $n'$ in $N$, we
have $lnn'=n'ln=ln'n$, where the latter equality follows from
$[L,N]=1$. Consequently, $nn'=n'n$ for all $n'$ in $N$, so $n$ is an
element of $Z(N)$. Similarly, $l$ must be an element of $Z(L)$.  It
follows that $Z(T)=Z(L)Z(N)$.
\end{proof}


%% file: chSubsys2.tex
\chapter{Subsystem Codes -- Bounds and Constructions\footnotemark}\label{ch:subsys2}
\footnotetext{\copyright 2007. Part of the material in this chapter is reprinted from 
A. Klappenecker and P. K. Sarvepalli, ``On subsystem codes beating the 
quantum Hamming or Singleton bound'', {\em Proc. Royal Society London A}, vol 463, pp.~2887-2905, 2007.}
In this chapter we extend the theory of subsystem codes. One of our
goals is to clarify the benefits that can be gained from the use of
subsystem codes with respect to stabilizer codes. In this context we
derive bounds on the parameters of subsystem codes. These bounds help
in comparing the performance of subsystem codes with respect to 
stabilizer codes. Of course subsystem codes subsume stabilizer and
in that sense every stabilizer code is a subsystem code. However,
we use the term subsystem code to mean a code with nontrivial 
dimension of the gauge subsystem. 
 We generalize the quantum Singleton bound to
$\F_q$-linear subsystem codes.  It follows that no subsystem code over
a prime field can beat the quantum Singleton bound.  On the other
hand, we show the remarkable fact that there exist impure subsystem
codes beating the quantum Hamming bound.  A number of open problems
concern the comparison in performance of stabilizer and subsystem
codes. One of the open problems suggested by Poulin's work asks
whether a subsystem code can use fewer syndrome measurements than an
optimal $\F_q$-linear MDS stabilizer code while encoding the same
number of qudits and having the same distance. We prove that linear
subsystem codes cannot offer such an improvement under complete
decoding.


One of the promises of subsystem codes is their potential for simplifying 
error recovery. Perhaps the benefits of subsystem codes are best understood by an example.
Consider the first quantum error correcting
code proposed by  \cite{shor95}, which encodes one qubit into nine
qubits. This code which is capable of correcting a single error
on any of the qubits requires the measurement of eight syndrome qubits.
The Bacon-Shor subsystem code \cite{bacon06a} on the other hand, also encodes 
one qubit into nine but it requires only four syndrome measurements,
giving a simpler error recovery scheme.

In this context it becomes crucial to identify when subsystem codes
provide gains over the stabilizer codes. It also becomes necessary to
compare the stabilizer codes and the subsystem codes fairly and with
meaningful criteria. For instance, once again consider the
$[[9,1,3]]_2$ Shor code requiring $n-k=9-1=8$ syndrome
measurements. The $[[9,1,4,3]]_2$ Bacon-Shor code on the other hand
requires $n-k-r=9-1-4=4$ syndrome measurements.  Clearly, this code is
better than the Shor's code.  But the optimal single error correcting
binary quantum code that encodes one qubit is the $[[5,1,3]]_2$ code,
which also requires only $5-1=4$ syndrome measurements.  So it is
apparent that while a given subsystem code can be superior to some
stabilizer codes, it is not at all obvious that it is better than the
best stabilizer code for the same function, {\em viz.}, encoding $k$
qubits with a distance $d$.

The first part of our chapter seeks to address this issue for
$\F_q$-linear Clifford subsystem codes which might perhaps be the most
useful class of subsystem codes.  In this chapter we generalize the
quantum Singleton bound to $\F_q$-linear Clifford subsystem codes. It
follows that no Clifford subsystem code over a prime field can beat
the quantum Singleton bound. We then show how the quantum Singleton
bound can be applied to make the comparison between stabilizer and
subsystem codes (focusing on stabilizer codes that are optimal in the
sense that they meet the quantum Singleton bound).  This bound makes
it possible to quantify the gains that subsystem codes can provide in
error recovery. In particular, our results show that these gains
involve a trade off between the distance of the subsystem code and the
number of information and the gauge qudits. We show that if there exists an $\F_q$-linear MDS stabilizer code, {\em
i.e.}, a code meeting the quantum Singleton bound, then no
$\F_q$-linear subsystem code can outperform it in the sense of
requiring fewer syndrome measurements for error correction.

Then we shift our attention to a class of subsystem
codes on lattices. Bacon and Casaccino \cite{bacon06b} obtain a subsystem code
from two classical codes. We show that
this method is a special case of the Euclidean construction for subsystem
codes proposed in \cite{aly06} and give a coding theoretic
analysis of these codes.

Since the early works on quantum error-correcting codes, it has been
suspected that impure codes should somehow perform better than the
pure codes. However, it was shown that the quantum Singleton bound
holds true for both pure and impure stabilizer codes. But it was not
so clear with respect to the quantum Hamming bound. In fact, it was
often conjectured that there might exist impure quantum
error-correcting codes beating the quantum Hamming bound, but a proof
remained elusive.  At least in the case of binary stabilizer codes
there exists some evidence that the conjecture might not be true, as
\cite{ashikhmin99} showed that asymptotically the quantum Hamming
bound was obeyed by impure codes as well, and \cite{gottesman97}
showed that no single error correcting binary stabilizer code can beat
the quantum Hamming bound.  In this context it is not surprising that
questions were raised \cite{bacon06a} if subsystem codes are any
different. In \cite{aly06} we proved the quantum Hamming bound for pure
subsystem codes.  We show here that impure subsystem codes can indeed
beat the quantum Hamming bound for pure subsystem codes. For example,
we demonstrate that the lattice subsystem codes can provide examples
of impure subsystem codes that beat the quantum Hamming bound.

The chapter is structured as follows. 
We assume that the reader is familiar with the notion of 
subsystem code introduced in the last chapter. 
We prove the quantum Singleton bound
for subsystem codes in Section~\ref{sec:boundsSubsys}. The lattice subsystem codes
are focus of attention in Section~\ref{sec:latticeCodes} and
Section~\ref{sec:packing}, wherein it is shown that there exist impure
subsystem codes that beat the quantum Hamming bound. We conclude with
a few open questions on subsystem codes.

\section{Quantum Singleton Bound for $\F_{q}$-linear Subsystem Codes}\label{sec:boundsSubsys}

Recall that the quantum Singleton bound states that an $[[n,k,d]]_q$
quantum code satisfies $2d\leq n-k+2$, \cite{KnLa97,rains99}.  In this context
it is natural to ask if subsystem codes also obey a similar
relation. The usefulness of such a bound is obvious. Apart from
establishing the bounds for optimal subsystem codes, they also make it
possible to compare stabilizer and subsystem codes, as we shall see
subsequently.  We prove that the $\F_q$-linear subsystem codes with
the parameters $[[n,k,r,d]]_q$ satisfy a quantum Singleton like bound
{\em viz.,} $k+r\leq n-2d+2$. It will be seen that this reduces to the
quantum Singleton bound if $r=0$.  More interestingly, this reveals
that there is a trade off in the size of subsystem $A$ and the gauge
subsystem. One pays a price for the gains in error recovery. The cost
is the reduction in the information to be stored.

Our proof for
this result is quite straightforward, though the intermediate details are a little
involved. First we show that a linear $[[n,k,r>0,d]]_q$ subsystem code that is 
exactly pure to 1 can be punctured to an $[[n-1,k,r-1,d]]_q$ code which retains the relationship
between $n,k,r,d$. If $d=2$ by repeated puncturing we either arrive at a pure code or a stabilizer code, both of which have upper bounds.
For $d>2$, two cases can arise, if the code is exactly pure to 1, we simply puncture it to get
a smaller code as in $d=2$ case. Otherwise, we puncture it
to get an $[[n-1,k,r+1,d-1]]_q$ code. By repeatedly shortening we either get a stabilizer code or a distance 2 code both of which have an upper bound. Keeping
track of the change in the parameters will give us an upper bound on the parameters of
the original code.

Let $w=(a_1,a_2,\ldots,a_n|b_1,b_2,\ldots,b_n)\in \F_q^{2n}$. We denote by 
$\rho(w) \in \F_q^{2n-2}$, the vector obtained by deleting
the first and the $n+1^{th}$ coordinates of $w$. Thus we have
$$\rho(w)=(a_2,\ldots,a_n|b_2,\ldots,b_n) \in \F_q^{2n-2}.$$
Similarly, given a classical code 
$C\subseteq \F_q^{2n}$ we denote  the puncturing of a codeword or code in the 
first and $n+1$ coordinates by $\rho(C)$.

In Theorem~\ref{th:four} subsystem codes are constructed using a trace symplectic
product. Following Remark~\ref{rm:symFormChoice} for $\F_q$-linear codes instead 
of considering the trace symplectic inner product
we can consider the relatively simpler symplectic product. Recall that the symplectic product 
of $u=(a|b)$ and $v=(a'|b')$ in $\F_q^{2n}$ is defined as
$\langle u|v\rangle_s = \langle(a|b)|(a'|b') \rangle_s =a'\cdot b-a\cdot b'$. 
The symplectic dual of a code $C\subseteq \F_q^{2n}$ is defined as 
$C^\sdual = \{x\in \F_q^{2n} \mid \langle x|y\rangle_s=0, \mbox{ for all }y \in C\}$. 
As we shall be concerned with $\F_q$-linear codes in this chapter, we will 
focus only on the symplectic inner product in the rest of the chapter.
\begin{lemma}\label{lm:xzBasis}
Let $C \subseteq \F_q^{2n}$ be an $\F_q$-linear code.
Then $C$ has an $\F_q$-linear basis of the form
$$B = \{ z_1,\ldots, z_k,z_{k+1},x_{k+1},z_{k+2},x_{k+2},\ldots, z_{k+r},x_{k+r} \}$$
where $\langle x_i |x_j \rangle_s  = 0 = \langle z_i |z_j \rangle_s$ and 
$\langle x_i| z_j\rangle_s = \delta_{i,j}$.
\end{lemma}
\begin{proof}
First we choose a basis $B=\{z_1,\ldots, z_k, z_{k+1},\ldots, z_{k+r} \}$ for a maximal isotropic subspace $C_0$ of $C$.
If $C_0\neq C$, then we can choose a codeword $x_{k+1}$ in $C$ that is orthogonal to all of the $z_i$ except one, say $z_{k+1}$ (renumbering if necessary).  We can scale $x_{k+1}$ by an element in $\F_q^\times$ so that 
$\langle z_{k+1} |x_{k+1}\rangle_s=1$.  If $\langle C_0,x_{k+1}\rangle\neq C$, then we repeat the process by choosing another codeword $x_{k+i}$  that is orthogonal to all the previously chosen $\{x_{k+1},\ldots, x_{k+i-1}\}$ and all $z_i$ except $z_{k+i}$, 
until we have a basis of the desired form.
\end{proof}

\medskip
For the remainder of the section, we fix the following notation. 
By Theorem~\ref{th:four}, we can associate with an 
$\F_q$-linear $[[n,k,r,d]]_q$ subsystem code two 
classical $\F_q$-linear codes
$C,D\subseteq \F_{q}^{2n}$ such that $D=C\cap C^\sdual$,
$|C|=q^{n-k+r}$, $|D|=q^{n-k-r}$ and $\swt(D^\sdual\setminus C)=d$.  
By lemma~\ref{lm:xzBasis}, we can also
assume that $C$ is generated by
$$C =\langle z_1,\ldots, z_s,z_{s+1},x_{s+1},\ldots, z_{s+r},x_{s+r}\rangle,$$ 
where $s=n-k-r$
and the vectors $x_i$, $z_i$ in $\F_q^{2n}$ satisfy the relations 
$\langle x_i|x_j \rangle_s =0 = \langle z_i|z_j\rangle_s$ and
$\langle x_i|z_j \rangle_s=\delta_{i,j}$. These relations on $x_i,z_i$  imply that 
\begin{eqnarray*}
C^\sdual &=&\langle z_1,\ldots, z_s,z_{s+r+1},x_{s+r+1},\ldots, z_{s+r+k},x_{s+r+k}\rangle,\\
D=C\cap C^\sdual &=&\langle z_1,\ldots, z_s \rangle,\\ 
D^\sdual &=& \langle z_1,\ldots, z_s,z_{s+1},x_{s+1},\ldots, z_{n},x_{n}\rangle.\end{eqnarray*}

\begin{lemma}\label{lm:punc1lin}
An $\F_{q}$-linear $[[n,k,r>0,d\geq 2]]_q$ Clifford subsystem code exactly pure to $1$ can be punctured to 
an $\F_{q}$-linear $[[n-1,k,r-1, \geq d]]_q$ code. 
\end{lemma}
\begin{proof}
As mentioned above, we can associate to the subsystem code two classical codes 
$C,D\subseteq \F_q^{2n}$. 
Two cases arise depending on $\swt(D)$. 
\begin{compactenum}
\item[a)] If $\swt(D)=1$, then without loss of generality we can assume that $\swt(z_1)=1$.
Further, $z_1$ can be taken to be of the form $(1,0,\ldots,0|a,0,\ldots,0)$. And 
for $i\neq 1$, because of $\F_q$-linearity of the codes we can 
pick 
all 
$x_i,z_i$ to be of the form 
$(0,a_2,\ldots,a_n|b_1,b_2,\ldots,b_n)$. 
Further, as $x_i,z_i$ must satisfy the orthogonality relations 
with $z_1$ {\em viz.,} $\langle z_1|z_i \rangle_s=0= \langle z_1|x_i \rangle_s$, for 
$i>1$ we can choose $x_i,z_i$ to be of the form $(0,a_2,\ldots,a_n|0,b_2,\ldots,b_n)$.
It follows that because of the form of $x_i$ and $z_i$ puncturing the first and $n+1^{th}$ coordinate
will not alter these orthogonality relations, in particular $\langle \rho(x_i)|\rho(z_i)\rangle_s\neq 0$
for ${s+1}\leq i\leq n$. 

Letting $\rho(x_i)=x_i'$, $\rho(z_i)=z_i'$ and observing that
$\rho(z_1)=(0,\ldots,0|0,\ldots,0)$, we see that the code 
$\rho(C) = \langle z_2',\ldots, z_s', z_{s+1}',x_{s+1}',\ldots,z_{s+r}',x_{s+r}'\rangle$.
Denoting by $D_p=\rho(C)\cap \rho(C)^\sdual$ it is immediate that $D_p$ is generated by
$\{ z_2',\ldots, z_s'\}$ while 
$D_p^\sdual= \langle z_2',\ldots, z_s', z_{s+1}',x_{s+1}',\ldots,z_{n}',x_{n}'\rangle$.
Hence $\rho(C)$ defines an $[[n-1,k,r,\swt(D_p^\sdual\setminus \rho(C))]]_q$ code. 

Next we show that $\swt(D_p^\sdual\setminus \rho(C))\geq d$. 
Let $u=(a_2,\ldots,a_n|b_2,\ldots, b_n)$ be in
$D_p^\sdual \setminus \rho(C)$, then we can easily verify that 
$(0,a_2,\ldots,a_n|0,b_2,\ldots, b_n)$ is orthogonal to all
$z_i$, $1\leq i\leq s$ and hence it is in $D^\sdual$. It cannot be in $C$ as that would imply
that $u$ is in $\rho(C)$. But $\swt(D^\sdual \setminus C)\geq d$. Therefore $\swt(u)\geq d$.
and $\rho(C)$ defines an $[[n-1,k,r,\geq d]]_q$ code.
By choosing $C'= \langle z_2',\ldots, z_s',z_{s+1}',z_{s+2}',x_{s+2}',\ldots, z_{s+r}',x_{s+r}'\rangle$
we can conclude that there exists an $[[n-,k,r-1,d]]_q$ code. Alternatively, apply 
Theorem~16 in 
\cite{aly06}.
\item[b)] If $\swt(D) >1$, then we can assume that $\swt(z_{s+1})=1$ and form the code 
$C'= \langle z_1,\ldots, z_s,z_{s+1},z_{s+2},x_{s+2},\ldots, z_{s+r},x_{s+r}\rangle$. It is clear that 
$C'$ defines an $[[n,k,r-1,d]]_q$ code that is pure to $1$ with $\swt(C'\cap C'^\sdual)=1$. But this is just
the previous case, from which we can conclude that there exists an $[[n-1,k,r-1,\geq d]]_q$ code.
\end{compactenum}
\end{proof}
Lemma~\ref{lm:punc1lin} allows us to establish a bound for distance 2 codes which can then be
used to prove the bound for arbitrary distances. 
\begin{lemma}\label{lm:boundlind2}
An impure $\F_{q}$-linear $[[n,k,r,d=2]]_q$ Clifford subsystem code satisfies
$$ 
k+r\leq n-2d+2.
$$
\end{lemma}
\begin{proof}
Suppose that there exists an $\F_{q}$-linear $[[n,k,r,d=2]]_q$ impure
subsystem code such that $k+r> n-2d+2$; in particular, this code must
be pure to~$1$. By Lemma~\ref{lm:punc1lin} it can be punctured to give
an $[[n-1,k,r-1,\geq d]]_2$ subsystem code. If this code is pure, then
$k+r-1 \leq n-1-2d+2$ holds, contradicting our assumption $k+r>
n-2d+2$; hence, the resulting code is once again impure and pure to 1.

Now we repeatedly apply Lemma~\ref{lm:punc1lin} to puncture the
shortened codes until we get an $[[n-r,k,0,\geq d]]_q$ subsystem
code. But this is a stabilizer code which must obey the Singleton bound 
$k\leq n-r-2d+2$, contradicting our initial assumption 
$k+r>n-2d+2$. Therefore, we can conclude that $k+r\leq n-2d+2$.
\end{proof}

If the codes are of distance greater than 2, then we puncture the code
until it either has distance 2 or it is a pure code. The following
result tells us how the parameters of the subsystem codes vary on
puncturing.
\begin{lemma}\label{lm:punc1linp2}
An impure $\F_{q}$-linear $[[n,k,r,d\geq 3]]_q$  Clifford subsystem code exactly
pure to $d'\geq 2$ implies the existence of an $\F_{q}$-linear
$[[n-1,k,r+1,\geq d-1]]_q$ subsystem code.
\end{lemma}
\begin{proof}
Recall that the existence of an $[[n,k,r,d\ge 3]]_q$ subsystem code
implies the existence of $\F_q$-linear codes $C$ and $D$ such that 
$$C=\langle z_1,\ldots, z_s,z_{s+1},x_{s+1},\ldots,
z_{s+r},x_{s+r}\rangle,$$ with $s=n-k-r$, and $D=C\cap C^\sdual$, see above. 

The stabilizer code defined by $D$ satisfies $k+r=n-s\leq n-2d'+2$, or
equivalently $s\geq 2d-2$; it follows that $s\ge 2$, since $d'\ge
2$. Without loss of generality, we can take $z_1$ to be of the form
$(1,a_2,\ldots,a_n|b_1,b_2\ldots,b_n)$ for if no such codeword exists
in $D$, then $(0,0,\ldots,0|1,0,\ldots,0)$ is contained in $D^\sdual$,
contradicting the fact that $\swt(D^\sdual)\geq 2$.  Consequently, we
can choose $z_2$ in $D$ to be of the form
$(0,c_2,\ldots,c_n|1,d_2,\ldots,d_n)$, and we may further assume that
$b_1=0$ in $z_1$. The form of $z_1$ and $z_2$ allows us to assume
that any remaining generator of $C$ is of the form 
$(0,u_2,\ldots,u_n|0,v_2,\ldots,v_n)$.  

Let $\rho$ be the map defined by puncturing the first and $(n+1)^{th}$
coordinate of a vector in $C$. Define for all $i$ the punctured
vectors $x_i'=\rho(x_i)$ and $z_i'=\rho(z_i)$. Then one easily checks
that $\scal{\rho(x_i)}{\rho(x_j)}=0=\scal{\rho(z_i)}{\rho(z_j)}$ for
all indices $i$ and $j$, and
$\scal{\rho(x_i)}{\rho(z_j)}=\delta_{i,j}$ if $i\ge s+1$ or $j\ge 3$,
and that $\scal{\rho(z_1)}{\rho(z_2)}=-1$.

Let us look at the punctured code $\rho(C)$,
$$\rho(C)=\langle z_3',\ldots,z_s',z_{s+1}',x_{s+1}',\ldots,
z_{s+r}',x_{s+r}',z_{1}',z_{2}'\rangle.$$ Since
$\scal{\rho(z_1)}{\rho(z_2)}=-1$ we have $D_p=\rho(C)\cap
\rho(C)^\sdual =\langle z_3',\ldots,z_s'\rangle $, whence
$|D_p|=|D|/q^2$.  As $\swt(C)\geq 2$, it follows that $|\rho(C)|=|C|$.
Thus $\rho(C)$ defines an $[[n-1,k,r+1,\swt(D_p^\sdual\setminus
\rho(C))]]_q$ subsystem code.  

Recall that the code $D$ is generated by $s\ge 2$ vectors; we will
show next that our assumptions actually force $s\ge 3$.  Indeed, if
$s=2$, then $|D|=q^2$ and $|D^\sdual|=q^{2n-2}$. Under the assumption
$\swt(D^\sdual)\geq 2$, it follows that $|\rho(D^\sdual)|
=|D^\sdual|=q^{2n-2}$.  But as $\rho(D^\sdual) \subseteq \F_q^{2n-2}$
this implies that $\rho(D^\sdual)=\F_q^{2n-2}$.  Since $\F_{q}^{2n-2}$
has $2n-2$ independent codewords of symplectic weight one, $D^\sdual$
must have $2n-2$ independent codewords of symplectic weight
two. However, this contradicts our assumptions on the minimum distance
of the subsystem code: 
\begin{compactenum}[(a)]
\item If $C$ is a proper subspace of $D^\sdual$, then the minimum
distance $d$ is given by $d=\swt(D^\sdual\setminus C)\ge 3$; thus, the
weight 2 vectors must all be contained in $C$, which shows that
$|C|=q^{2n-2}=|D|$, contradicting $|C|<|D^\sdual|$.
\item If $C=D^\sdual$, then the minimum distance is given by
$d=\swt(D^\sdual)=2$, contradicting our assumption that $d\ge 3$. 
\end{compactenum}
Thus, from now on, we can assume that $s\geq 3$.

Before bounding the minimum distance of the punctured subsystem code,
we are going to show that $D_p^\sdual=\rho(D^\sdual)$.  Let
$w=(u_1,u_2,\ldots,u_n|v_1,v_2,\ldots,v_n)$ be a vector in
$D^\sdual$. For $3\leq i\leq s$, the vectors $z_i$ are of the
form $(0,a_2,\ldots,a_n|0,b_2,\ldots,b_n)$;
thus, it follows from $\langle w|z_i\rangle_s=0$ that $\langle
\rho(w)|z_i'\rangle_s=0$.  Hence $\rho(w)$ is in $D_p^\sdual$, which
implies $\rho(D^\sdual) \subseteq D_p^\sdual$. We have $|D_p^\sdual|
=q^{2n-2}/|D_p|=q^{2n}/|D|=|D^\sdual|$, and we note that
$|D^\sdual|=|\rho(D^\sdual)|$, because $\swt(D^\sdual)\geq 2$; hence,
$D_p^\sdual=\rho(D^\sdual)$.

Let $w'=(u_2,\ldots,u_n|v_2,\ldots,v_n)$ be an arbitrary vector in
$\rho(D^\sdual) \setminus \rho(C)$.  It follows that there exist some
$\alpha, \beta$ in $\F_q$ such that
$w=(\alpha,u_2,\ldots,u_n|\beta,v_2,\ldots,v_n)$ is in $D^\sdual;$ it is
clear that $w$ cannot be in $C$, since then $\rho(w)=w'$ would be in
$\rho(C)$; hence, $\swt(w)\geq d$. It immediately follows that 
$\swt(D_p^\sdual\setminus \rho(C))\geq d-1$.  Hence $\rho(C)$ defines
an $[[n-1,k,r+1,\geq d-1]]_q$ subsystem code.
\end{proof}

Now we are ready the prove the upper bound for an arbitrary subsystem code.
Essentially we reduce it to a pure code or distance two code by repeated puncturing
and bound the parameters by carefully tracing the changes.
\begin{theorem}\label{th:boundlin}
An $\F_{q}$-linear $[[n,k,r,d\geq 2]]_q$ Clifford subsystem code satisfies 
\begin{eqnarray}k+r\leq n-2d+2.\label{eq:singBound}\end{eqnarray}
\end{theorem}
\begin{proof}
The bound holds for all pure codes, see  
\cite{aly06}. So
assume that the code is impure.  If $d=2$, then the relation holds by
Lemma~\ref{lm:boundlind2}; so let $d\geq 3$. If the code is exactly pure to $1$,
then it can be punctured using Lemma~\ref{lm:punc1lin} to give an
$[[n-1,k,r-1,d'=d]]_q$ code, otherwise it can be punctured using
Lemma~\ref{lm:punc1linp2} to obtain an $[[n-1,k,r+1,d'\geq d-1]]_q$
code. If the punctured code is pure, then it follows that either 
$k+r-1\leq n-1-2d+2$ or $k+r+1\leq n-1-2d'+2\leq n-1-2(d-1)+2$ holds; 
in both cases, these inequalities imply that $k+r\leq n-2d+2$.

If the resulting code is impure, then if it is exactly pure 
to $1$ we puncture the code again using Lemma~\ref{lm:punc1lin}, 
if not we puncture using Lemma~\ref{lm:punc1linp2}, until we get a
pure code or a code with distance two. Assume that we punctured $i$ times
using Lemma~\ref{lm:punc1lin} and $j$ times using
Lemma~\ref{lm:punc1linp2}, then the resulting code is an
$[[n-i-j,k,r+j-i,d'\geq d-j]]_q$ subsystem code. Since pure subsystem
codes and distance 2 subsystem codes satisfy $$k+r+j-i\leq
n-i-j-2d'+2\leq n-i-j-2(d-j)+2,$$ it follows that $k+r\leq n-2d+2$ holds. 
\end{proof}

\medskip
When the subsystem codes are over a prime alphabet, this bound holds
for all codes over that alphabet.  In the more general case where the
code is not linear, numerical evidence indicates that it is unlikely
that the additive subsystem codes have a different bound. We have
shown that a large class of impure codes already satisfy this bound.
This prompts the following conjecture.

\begin{conjecture}
Any $[[n,k,r,d]]_q$ Clifford subsystem code satisfies $k+r\leq n-2d+2$. 
\end{conjecture}

\section{Comparing Subsystem Codes with Stabilizer Codes}\label{sec:mds}

In this section, we compare stabilizer codes with subsystem codes.
We first need to establish the criteria for the comparison, since
subsystem codes cannot be universally better than stabilizer codes.
For example, it is known that a subsystem code can be
converted to a stabilizer code \cite{kribs05b,poulin05}. See also 
Lemma~10 in \cite{aly06} for a simple
proof to convert an $[[n,k,r,d]]_q$ code to an $[[n,k,d]]_q$ code. This implies that no
$[[n,k,r,d]]_q$ subsystem code can beat an optimal $[[n,k,d']]_q$
stabilizer code in terms of minimum distance, as $d'\geq d$.  One of
the attractive features of subsystem codes is a potential reduction of
the number of syndrome measurements, and we use this criterion as the
basis for our comparison.

First, we must highlight a subtle point on the required number of
syndrome bits for an $\F_q$-linear $[n,k,d]_q$ code. A complete decoder,
will require $n-k$ syndrome bits. Complete decoders are also optimal
decoders. A bounded
distance decoder on the other hand can potentially decode with 
fewer syndrome bits. Bounded distance decoders typically 
decode  up to $\floor{(d-1)/2}$.
However, to the best of our knowledge, except for
the lookup table decoding method, all bounded distance decoders
also require $n-k$ syndrome bits.
As the complexity of decoding using a  lookup table increases exponentially 
in $n-k$ it is  highly impractical for long lengths. We therefore assume
that for practical purposes, that we need  $n-k$ syndrome bits.

Similarly, for an $\F_q$-linear $[[n,k,r,d]]_q$ subsystem code, a complete 
decoder will require $n-k-r$ syndrome measurements, as is shown 
in~\ref{sec:appendix}. We are not aware of any quantum code,
stabilizer or subsystem, for which there exists a bounded distance
decoder that uses less than $n-k-r$ syndrome measurements to
perform bounded distance decoding.
The work by Poulin \cite{poulin05} prompts the following question: Given an
optimal $[[k+2d-2,k,d]]_q$ MDS stabilizer code, is it possible to find
an $[[n,k,r,d]]_q$ subsystem code that uses fewer syndrome
measurements?

There exist numerous known examples of subsystem codes that improve
upon nonoptimal stabilizer codes. The fact that the stabilizer code is
assumed to be optimal makes this question interesting.  The Singleton
bound $k+r\le n-2d+2$ of an $\F_q$-linear $[[n,k,r,d]]_q$ subsystem
code implies that the number $n-k-r$ of syndrome measurements is
bounded by $n-k-r\geq 2d-2$; thus, for fixed minimum distance $d$,
there exists a trade off between the dimension $k$ and the difference
$n-r$ between length and number of gauge qudits.

\begin{corollary}\label{th:betterQMDS} Under complete decoding
an $\F_q$-linear $[[n,k,r,d\geq 2]]_q$ Clifford subsystem code cannot use
fewer syndrome measurements than an $\F_q$-linear $[[k+2d-2,k,d]]_q$
stabilizer code.
\end{corollary}
\begin{proof}
Seeking a contradiction, we assume that there exists an
$[[n,k,r,d]]_q$ subsystem code that requires fewer syndrome
measurements that the optimal $[[k+2d-2,k,d]]_q$ MDS stabilizer code.
In other words, the number of syndrome measurement yield the
inequality $k+2d-2-k> n-k-r$, which is equivalent to $k+r > n-2d+2$,
but this contradicts the Singleton bound.
\end{proof}

Poulin \cite{poulin05} showed by exhaustive computer search that there does
not exist an $[[5,1,r>0,3]]_2$ subsystem code. The above result
confirms his computer search and shows further that not even allowing
longer lengths and more gauge qudits can help in reducing the number
of syndrome measurements. In fact,
we conjecture that corollary~\ref{th:betterQMDS} holds for bounded distance
decoders also.

We wish to caution the reader that gains in error recovery cannot be
quantified purely by the number of syndrome measurements.  In practice,
more complex measures such as the simplicity of the decoding algorithm
or the resulting threshold in fault-tolerant quantum computing are
more relevant. The drawback is that the comparison of large classes of
codes becomes unwieldy when such complex criteria are used. 

\section{Subsystem Codes on a Lattice}\label{sec:latticeCodes}
Bacon gave the first family of subsystem codes generalizing the ideas
of Shor's $[[9,1,3]]_2$ code \cite{bacon06a}.  Recently, he and
Casaccino gave another construction which generalizes this further by
considering a pair of classical codes \cite{bacon06b}.  We show that
this method is a special case of Theorem~\ref{th:cssoqec}. Since this
construction is not limited to binary codes and our proofs remain
essentially the same, we will immediately discuss a generalization to
nonbinary alphabets.

\begin{theorem}\label{th:latticeCodes}
For $i\in \{1,2\}$, let $C_i \subseteq \F_q^{n_i}$ be $\F_q$-linear
codes with the parameters $[n_i,k_i,d_i]_q$. Then there exists a
Clifford subsystem code with the parameters
$$[[n_1n_2, k_1k_2, (n_1-k_1)(n_2-k_2), \min\{d_1,d_2 \}]]_q$$
that is pure to $d_p=\min\{d_1^\perp,d_2^\perp\}$, where $d_i^\perp$ denotes the minimum distance of $C_i^\perp$. 
\end{theorem}

\begin{proof}
Let $C$ be the classical linear code given by $C=(\F_q^{n_1}\otimes
C_2^\perp) \times (C_1^\perp \otimes \F_q^{n_2})$.  Then $\dim C =
n_1(n_2-k_2)+n_2(n_1-k_1)$ and $\swt(C\setminus \{0\})\ge
\min\{d_1^\perp,d_2^\perp\}$.  The symplectic dual of $C$ is given by
\begin{eqnarray*}
C^\sdual &=&  (C_1^\perp\otimes \F_q^{n_2})^\perp \times (\F_q^{n_1}\otimes C_2^\perp)^\perp\\
&=& (C_1\otimes \F_q^{n_2}) \times (\F_q^{n_1}\otimes C_2).
\end{eqnarray*}
We have $\dim C^\sdual= k_1n_2+n_1k_2$. 
The code $D=C\cap C^\sdual$ is given by  
\begin{eqnarray*}
\begin{split}
D &=\left((\F_q^{n_1}\otimes C_2^\perp) \times (C_1^\perp \otimes \F_q^{n_2}) \right)\cap
\left((C_1\otimes \F_q^{n_2}) \times (\F_q^{n_1}\otimes C_2) \right)\\
&=\left((\F_q^{n_1}\otimes C_2^\perp) \cap (C_1\otimes \F_q^{n_2}) \right)\times
\left(  (C_1^\perp \otimes \F_q^{n_2})\cap (\F_q^{n_1}\otimes C_2) \right)\\
&= (C_1\otimes C_2^\perp) \times (C_1^\perp \otimes C_2), 
\end{split}
\end{eqnarray*}
and $\dim D = k_1(n_2-k_2)+k_2(n_1-k_1)$.  It follows that $\dim C -
\dim D = 2(n_1-k_1)(n_2-k_2)$ and $\dim C^\sdual -\dim D =
2k_1k_2$. Using corollary~\ref{th:cssoqec}, we can get a subsystem
code with the parameters $$[[n_1n_2, k_1k_2, (n_1-k_1)(n_2-k_2),
d=\swt(D^\sdual\setminus C)]]_q $$ that is pure to $d_p=\min\{d_1^\perp,d_2^\perp\}$. It remains to show
that $d=\min\{d_1,d_2\}$. 

Since $D= (C_1\otimes
C_2^\perp) \times (C_1^\perp \otimes C_2)$, we have
\begin{eqnarray*}
D^\sdual &=&  (C_1^\perp \otimes C_2)^\perp \times (C_1\otimes C_2^\perp)^\perp\\
&=&  
\left((C_1 \otimes \F_q^{n_2})  + (\F_q^{n_1}\otimes C_2^\perp)\right)
\times 
\left((\F_q^{n_1}\otimes C_2) + 
(C_1^\perp\otimes \F_q^{n_2})\right).
\end{eqnarray*}
In the last equality, we used the fact that vectors $u_1\otimes
u_2$ and $v_1\otimes v_2$ are orthogonal if and only if $u_1\perp v_1$
or $u_2\perp v_2$.

For $i\in \{1,2\}$, let $G_i$ and $H_i$ respectively denote the
generator and parity check matrix of the code $C_i$. Without loss of
generality, we may assume that these matrices are in standard form
$$
H_i=\left[\begin{array}{cc}I_{n_i-k_i} & P_i\end{array}\right] \mbox{ and }  
G_i=\left[\begin{array}{cc}-P_i^t &I_{k_i} \end{array} \right],
$$
where $P_i^t$ is the transpose of $P_i$. Let $H_i^c=\left[\begin{array}{cc} 0& I_{k_i} \end{array}\right]$. Using these notations, the generator matrices of $C$ and $D^\sdual$ can be written as 
$$
G_C=\left[\begin{array}{cc}I_{n_1}\otimes H_2 &0\\ 0&H_1\otimes I_{n_2} \end{array}\right]\quad \text{and} \quad 
G_{D^\sdual} =  
\left[ \begin{array}{cc}
G_1 \otimes H_{2}^c &0 \\I_{n_1}\otimes H_2&0\\
0&H_{1}^c\otimes G_2\\0&H_1\otimes I_{n_2}
\end{array}\right].
$$ 
It follows that the minimum distance $d$ is given by 
\begin{eqnarray*}
\begin{aligned}
\swt(D^\sdual\setminus C)=\min&\left\{\wt\left(\left\langle  \begin{array}{c}
G_1 \otimes H_{2}^c \\I_{n_1}\otimes H_2
\end{array}\right\rangle \setminus \left\langle \begin{array}{c}
I_{n_1}\otimes H_2
\end{array} \right\rangle \right)\right.,\\
&\;\;\left. \wt\left(\left\langle  \begin{array}{c}
H_{1}^c\otimes G_2\\H_1\otimes I_{n_2}
\end{array} \right\rangle\setminus  \left\langle\begin{array}{c}
H_1\otimes I_{n_2}
\end{array} \right\rangle\right) \right\}.
\end{aligned}
\end{eqnarray*}
Let us compute 
$$\wt\left(\left\langle  \begin{array}{c}
H_{1}^c\otimes G_2\\H_1\otimes I_{n_2}
\end{array} \right\rangle\setminus  \left\langle\begin{array}{c}
H_1\otimes I_{n_2}
\end{array} \right\rangle\right).$$
If minimum weight codeword is present in $D^\sdual\setminus C$, it must be expressed as linear 
combination of at least one row from $\left[H_1^c\otimes G_2 \right]$ otherwise the 
codeword is entirely in $C$. Recall that $ H_1=[\begin{array}{cc}I_{n_1-k_1} & P_1\end{array}]$ and 
$H_1^c=[ \begin{array}{cc}0 &I_{k_1}\end{array}]$. Letting $P_1=(p_{ij})$,
we can write  
\begin{eqnarray*}
\left[  \begin{array}{c}
H_{1}^c\otimes G_2\\H_1\otimes I_{n_2}
\end{array} \right] = \left[\begin{array}{cccccccc}
0&0& \dots&0&G_2&0\\
0&0& \dots&0&0&G_2&0\\
\dots & \dots&\dots&\dots &\dots&\dots&\dots&\dots\\
0&0&\dots& 0&0&\dots&\dots&G_2\\\hline
I_{n_2}&0&\dots&0&p_{11}I_{n_2}&\dots&\dots&p_{1k_1}I_{n_2}\\
0&I_{n_2}&\dots&\dots&p_{21}I_{n_2}&\dots&\dots&p_{2k_1}I_{n_2}\\
\dots & \dots&\dots&\dots &\dots&\dots&\dots&\dots\\
0&0&\dots&I_{n_2}&p_{(n_1-k_1)1}I_{n_2}&\dots&\dots&p_{(n_1-k_1)k_1}I_{n_2}
\end{array} \right].
\end{eqnarray*}
Now observe that any row below the line in the above matrix can has a weight of only one in each of the last
$k_1$ blocks of size $n_2$. And any linear combination of them involving less than $d_2$ 
and at least one generator from the rows above must have a weight $\geq d_2$. If on the other hand there are more than
$d_2$ rows involved, then the first $n_2(n_1-k_1)$ columns will have a weight $\geq d_2$. Thus in either
case the weight of an element that involves a generator from $\left[ H_1^c\otimes G_2\right]$ must have a weight $ \geq d_2$. On the other hand, the minimum weight of the span of $\left[H_1^c\otimes G_2 \right]$ is 
$\wt(C_2)=d_2$, from which we can conclude that
$$\wt\left(\left\langle  \begin{array}{c}
H_{1}^c\otimes G_2\\H_1\otimes I_{n_2}
\end{array} \right\rangle\setminus  \left\langle\begin{array}{c}
H_1\otimes I_{n_2}
\end{array} \right\rangle\right)=d_2.$$
Because of the symmetry in the code we can argue that 
$$\wt\left(\left\langle  \begin{array}{c}
G_1 \otimes H_{2}^c \\I_{n_1}\otimes H_2
\end{array}\right\rangle \setminus \left\langle \begin{array}{c}
I_{n_1}\otimes H_2
\end{array} \right\rangle \right)=d_1$$
and consequently $d=\min\{ d_1,d_2\}$, which proves the theorem.
\end{proof}

\subsection{Bacon-Shor Codes}
Bacon \cite{bacon06a} proposed one of the first families of subsystem codes based on square
lattices.  A trivial modification using rectangular lattices instead of square
ones gives the following codes, see also \cite{bacon06b}. The relevance of these codes will be
seen later in Section~\ref{sec:packing}. Using the same notation as in 
Theorem~\ref{th:latticeCodes}, let $G_i=[1,\ldots,1]_{1\times i}$
and $H_i$ be the matrix defined as
$$
H_i=\left[ \begin{array}{ccccccc}
1&1&&&&&\\&1&1&&&&\\
&& &\ddots &&&\\
&&&&1&1&\\&&&&&1&1
 \end{array}\right]_{i-1\times i}
$$
and $C$, the additive code generated by the following matrix.
$$
G= \left[\begin{array}{cc}
I_{n_1}\otimes H_{n_2}&0\\
0&H_{n_1}\otimes I_{n_2}
 \end{array}\right].
$$
Observe that $G_i$ generates an $[i,1,i]_q$ code with distance $i$. 
By Theorem~\ref{th:latticeCodes}, $G_{n_1}$ and $G_{n_2}$ will give us the 
following family of codes 
\begin{corollary}\label{co:rectLattice}
There exist $[[n_1n_2,1,(n_1-1)(n_2-1),\min\{n_1,n_2\}]]_q$ Clifford subsystem codes.
\end{corollary}

\section{Subsystem Codes and Packing}\label{sec:packing}
We investigate whether subsystem codes lead to better codes because of the 
decomposition of the code space. Since the early days of quantum codes,
it has recognized that the degeneracy of quantum codes could lead to 
a more efficient quantum code and allow for a much more compact 
packing of the subspaces in the Hilbert space. But so far it has not
been shown for stabilizer codes. We can derive  similar bound for 
subsystem codes. 
\cite{aly06} 
showed the following theorem for pure subsystem codes.
\begin{theorem}
A pure $((n,K,R,d))_q$ Clifford subsystem code satisfies 
\begin{eqnarray} 
\sum_{j=0}^{\lfloor (d-1)/2\rfloor}\binom{n}{j} (q^2-1)^j \leq q^n/KR.
\end{eqnarray}
\end{theorem}
It is natural to ask if impure subsystem codes also satisfy this bound. We show 
that they do not by giving an explicit counterexample. This counter example comes 
from the codes proposed by 
\cite{bacon06a}. 
Recall the Bacon-Shor codes are $[[n^2,1,(n-1)^2,n]]_2$ subsystem codes. The $[[9,1,4,3]]_2$ 
is an interesting code. 
We can check that it satisfies the Singleton bound for subsystem codes as
$$ k+r=1+4 = n-2d+2=9-6+2.$$ So it is an optimal code. 
More interestingly, substituting the 
parameters of the $[[9,1,4,3]]_2$ Bacon-Shor code in the above inequality we get
$$
\sum_{j=0}^{1}\binom{9}{j}3^j = 28 >  2^{9-5}=16.
$$
Therefore the $[[9,1,4,3]]_2$ Bacon-Shor code beats the quantum Hamming bound
for the pure subsystem codes proving the following result.

\begin{theorem}
There exist impure $((n,K,R,d))_q$ Clifford subsystem codes that do not satisfy 
$$ 
\sum_{j=0}^{\lfloor (d-1)/2\rfloor}\binom{n}{j} (q^2-1)^j \leq q^n/KR.
$$
\end{theorem}

An obvious question is why impure codes can potentially pack more efficiently
than the pure codes. Let us understand this by looking at the $[[9,1,4,3]]_2$
code a little more closely. This code encodes information into a subspace, $Q$
where  $\dim Q= 2^{k+r}=2^5$. As it is a subsystem code $Q$ can be decomposed as $Q=A\otimes B$,
with $\dim A =2^k=2$ and $\dim B= 2^r=2^4$. In a pure single error 
correcting code all single errors must take the code space into orthogonal 
subspaces. In an impure code this is not required two or more distinct errors
can take the code space to the same orthogonal space. In the Bacon-Shor code
a phase flip error on any of the first three qubits will take the 
code space to same orthogonal subspace and because of this we cannot distinguish
between these errors. However, it is not a problem because 
we can restore the code space with respect to $A$
even though we cannot restore $B$. Thus instead of requiring $9$ orthogonal subspaces
as in a pure code, we only require 3 orthogonal subspaces to correct for 
any single phase flip error. Considering the bit flip errors and the combinations we need only
$9$ orthogonal subspaces. Thus with the original code space this means we need to pack
ten $2^5$-dimensional subspaces in the $2^n=2^9$ dimensional ambient space, which is achievable
as $10\cdot 2^5< 2^9$. 

More generally, in a sense degeneracy allows distinct errors to share the same 
orthogonal subspace and thus pack more efficiently. 
It must be pointed out though that this better packing is attained at the cost of $r$
gauge qudits compared to a stabilizer code.

In fact there exists another code among the Bacon-Shor codes which also 
beats the Hamming bound for the subsystem codes. This is the $[[25,1,16,5]]_2$ code.
The family of codes given in corollary~\ref{co:rectLattice} provides us 
with $[[12,1,6,3]]_2$, yet another example of a code
that beats the quantum Hamming bound like the $[[9,1,4,3]]_2$ code. 
 We can check that 
 $$ 
 \sum_{j=0}^1 \binom{12}{j}3^j = 37 > 2^{12-1-6}=2^5=32.
 $$
But note that unlike $[[9,1,4,3]]_2$ this code does not meet the Singleton bound for 
pure subsystem codes  as $6+1< 12-6+2$.
Naturally we can ask if there is a systematic method to construct codes that
beat the quantum Hamming bound. 
Ashikhmin and Litsyn showed that all binary stabilizer codes -- pure
or impure -- of sufficiently large length obey the quantum Hamming
bound, ruling out the possibility that impure codes of large length
can outperform pure codes with respect to sphere packing. In contrast
we show that impure subsystem codes do not obey the quantum Hamming
bound for pure subsystem codes, not even asymptotically.  We show that
there exist arbitrarily long Bacon-Shor codes that violate the quantum
Hamming bound.


Degenerate quantum error-correcting codes pose many interesting
questions in the theory of quantum error-correction. The early
discovery of the phenomenon of degeneracy raised the question whether
degenerate quantum codes can perform better than nondegenerate quantum
codes. One of the unresolved questions to this day in the theory of
stabilizer codes is whether the bounds that hold for nondegenerate
codes also hold for degenerate codes. Some bounds like the quantum
Singleton bound do. But for others, like quantum Hamming bound, an
answer remains elusive. Partial answers were provided by Gottesman
\cite{gottesman97} for single error-correcting and double
error-correcting codes. Ashikhmin and Litsyn \cite{ashikhmin99} showed
that asymptotically degenerate codes cannot beat the quantum Hamming
bound. \textit{This leaves only a small range of degenerate binary stabilizer
codes of moderate length that can potentially beat the quantum Hamming
bound, but we conjecture that no such examples can be found. }

We show that the situation is markedly different in the case of
subsystem codes (also known as operator quantum error-correcting codes
\cite{kribs05,knill06,kribs05b}). The quantum Hamming for pure
subsystem codes was derived in \cite{aly06}. 
We have already shown that there exist impure subsystem codes that beat the quantum
Hamming bound for pure subsystem codes. Now we address the question 
whether impure subsystem codes asymptotically obey the quantum Hamming
bound, as in the case of binary stabilizer codes. 
We show that there exist impure subsystem codes of arbitrarily
large length that beat the quantum Hamming (or sphere-packing) bound.

For the binary cases the quantum Hamming bound for subsystem codes states that a
pure $[[n,k,r,d]]$ subsystem code satisfies
\begin{eqnarray}\label{hamming} 
2^{n-k-r}	\geq \sum_{j=0}^{\floor{(d-1)/2}}\binom{n}{j}3^j.\label{eq:qhb}
\end{eqnarray}
We claim that all the Bacon-Shor codes \cite{bacon06a,bacon06b} of odd lengths {\em i.e.}, 
$[[(2t+1)^2,1,4t^2,2t+1]]$ violate the quantum Hamming bound, namely that
\begin{eqnarray*}
2^{(2t+1)^2-1-4t^2}=	2^{4t}&\not\geq &\sum_{j=0}^{t}\binom{(2t+1)^2}{j}3^j
\end{eqnarray*}
holds for all positive integers $t$. It suffices to show that 
\begin{eqnarray}\label{ineq}
2^{4t}&< &\binom{(2t+1)^2}{t}3^t
\end{eqnarray}
holds for all positive integers $t$. Since $0<4(t-1/6)^2+8/9=4t^2-4t/3+1$, we have 
$$ \frac{16t}{3} <{4t^2+1+4t}$$ for all $t>0$. Multiplying both sides
by $3/t$ and raising to the $t^{th}$ power yields
$$2^{4t}< \frac{3^t(2t+1)^{2t}}{t^t}, $$ which proves the inequality
(\ref{ineq}), as $\binom{n}{k} \geq n^tk^{-t}$.  Thus, we can conclude
that the Bacon-Shor codes of odd length do not obey the quantum
Hamming bound.

\begin{theorem}
Asymptotically, the quantum Hamming bound~(\ref{hamming}) does not
hold for impure subsystem codes.
\end{theorem}

It is remarkable that there exist such families of subsystem codes
that can pack more densely than any pure subsystem code. Further
examples of such densely packing subsystem codes can be found among
the family with parameters 
$[[n_1n_2,1,(n_1-1)(n_2-1),\min\{n_1,n_2\}]]$, which contains for
instance a $[[12,1,6,3]]$ subsystem code.

\section{Conclusions}\label{sec:conc}
We have proved that any $\F_q$-linear $[[n,k,r,d]]_q$ Clifford
subsystem code obeys the Singleton bound $k+r\le n-2d+2$.
Furthermore, we have shown earlier that pure Clifford subsystem codes
satisfy this bound as well. Our results provide much evidence for the
conjecture that the Singleton bound holds for arbitrary subsystem
codes. Proving this for all additive subsystem codes will be an
interesting problem. 

Pure Clifford subsystem codes obey the Hamming (or sphere packing)
bound. In this chapter, we have shown the amazing fact that 
there exist impure Clifford subsystem codes beating
the Hamming bound. This is the first
illustration of a case when impure codes pack more efficiently than
their pure counterparts. One example of a code beating the Hamming
bound is provided by the $[[9,1,4,3]]_2$ Bacon-Shor code; this
remarkable example also illustrates the following noteworthy facts:
\begin{compactenum}[a)]
\item The $[[9,1,4,3]]_2$ code requires $9-1-4=4$ syndrome measurements just like
the perfect $[[5,1,3]]_2$ code. 
\item Since $k+r\leq n-2d+2$ for all prime alphabet codes, $[[9,1,4,3]]_2$ code is also 
an optimal  subsystem code. 
This is interesting because the underlying classical
codes are not MDS. In MDS stabilizer codes, the underlying classical codes are required
to be MDS codes. 
\item The Bacon-Shor code is also impure. So unlike MDS stabilizer codes which must be pure,
MDS subsystem codes can be impure.
\item The maximal length of a $q$-ary stabilizer MDS code is $2q^2-2$, see Theorem~\ref{th:mds_length}
whereas for subsystem codes it is larger as the $[[9,1,4,3]]_2$ code
indicates.
\end{compactenum}
The implication of b)--d) is that optimal subsystem codes can be
derived from suboptimal classical codes, unlike stabilizer
codes. It would be an interesting problem to determine what are the conditions 
under which a non-MDS classical code will lead to an MDS subsystem code.


%% file: chSubsysEnc.tex
\chapter{Encoding and Decoding of Subsystem Codes}\label{ch:subsysEnc}

\section{Introduction}
In this chapter we investigate encoding and to some extent decoding of subsystem codes. 
Our main result is that encoding of a subsystem code can 
be reduced to the encoding of a related stabilizer code, thereby making use of the 
previous theory on encoding stabilizer codes \cite{cleve97,gottesman97, grassl03}. 
We shall prove this in two steps. First,
we shall show that Clifford codes can be encoded using the same methods used for stabilizer codes.
Secondly, we shall show how these methods can be adapted to encode Clifford subsystem codes.
Since subsystem codes subsume  stabilizer codes, noiseless subsystems and decoherence free
subspaces, these results imply that we can essentially use the
same methods to encode all these codes. 
In fact, while the exact details were not provided it was suggested in \cite{poulin06} that encoding
of subsystem codes can be achieved by Clifford unitaries. Our treatment is comprehensive and gives 
 proofs for  all the claims. 

Subsystem codes can potentially lead to simpler  error recovery schemes. 
In a similar vein, they can also simplify 
the encoding process, though perhaps not as dramatically\footnote{In general, decoding  
is usually of greater complexity than encoding and for this reason it is often neglected
in comparison. This parallels the classical case where also the decoding is studied much
more extensively than encoding.}. 
These simplifications have not been investigated thoroughly, neither have 
the gains in encoding been fully characterized. Essentially, these gains are in 
two forms. In the encoded state there need not exist a one to one correspondence between
the gauge qubits and the physical qubits. However, prior to encoding such a
correspondence exists. 
We can exploit this identification between the virtual qubits and the physical
qubits before encoding to tolerate errors on the gauge qubits, a fact which
was recognized in \cite{poulin06}. Alternatively, 
we  can optimize the encoding circuits by eliminating certain encoding operations. 
The encoding operations that are saved correspond to the encoded operators on the gauge
qubits. This is a slightly subtle point and will be elaborated at length subsequently.  
We argue that optimizing the encoding circuit for the latter is much more beneficial than
simply allowing for random initialization of gauge qubits.

{\em Notation.} 
The inner product of two characters of 
a group $N$, say  $\chi$ and $\theta$, is defined as $(\chi,\theta)_N=1/{|N|} \sum_{ n\in N}\chi(n)\theta(n^{-1})$. We shall denote the center of a group $N$ by $Z(N)$. 
Given a subgroup  $N\le E$, we shall denote the centralizer of $N$ in $E$ by $C_E(N)$. 
Given a matrix $A$, we consider another matrix $B$ obtained from $A$ 
by column permutation $\pi$ as being equivalent and denote this by 
$B=_{\pi} A $. Often we shall represent the basis of a group by the
rows of a matrix. In this case we will regard another basis obtained
by any row operations or permutations as being equivalent and by a slight
abuse of notation continue to denote $B=_{\pi}A$. 
The commutator of two operators  $ A$, $B$ is defined as $[A,B] = AB-BA$.
This can potentially conflict with our definition of commutator in 
Chapter~\ref{ch:subsys1} as $[x,y]=xyx^{-1}y^{-1}$. However, in this chapter we will not have occasion 
to use this definition. 

\section{Encoding Stabilizer Codes -- A Review} 
Recall the Pauli matrix operators\footnote{We consider the real version 
of the Pauli group in this chapter.}, 
\begin{eqnarray}
X = \left[\begin{array}{cc}0 &1\\1&0\end{array} \right], \quad
Z = \left[\begin{array}{cc}1 &0\\0&-1\end{array} \right], \quad 
Y = \left[\begin{array}{cc}0 &-1\\1&0\end{array} \right] = XZ.
\end{eqnarray}
Let $\mc{P}_n$ be the Pauli group on $n$ qubits. An element 
element $e=  (-1)^cX^{a_1}Z^{b_1}\otimes \cdots \otimes X^{a_n}Z^{b_n}$
in $\mc{P}_n$, can be mapped to $\F_2^{2n}$ by 
 $\tau : \mc{P}_n \rightarrow \F_2^{2n}$  as
\begin{eqnarray}
\tau(e) =  (a_1,\ldots, a_n|b_1,\ldots, b_n).\label{eq:iso}
\end{eqnarray}

Given an $[[n,k,d]]_2$ code with stabilizer $S$, we can associate to $S$ (and therefore
the code), a matrix in $\F_2^{(n-k)\times 2n}$ obtained by taking the image of
any set of its generators under the mapping $\tau$. We shall refer to this 
matrix as the  \textsl{stabilizer matrix}. \index{stabilizer!matrix}
We shall refer to the stabilizer as well as any
set of generators as the stabilizer. Additionally, because of the mapping $\tau$,
we shall refer to the stabilizer matrix or any matrix obtained 
from it by row reduction or column permutations also as the stabilizer.
The stabilizer matrix can be put in the so-called ``standard form'', see 
\cite{cleve97, gottesman97}. This form also allows us to compute the encoded operators 
for the stabilizer code. 
Recall that the encoded operators allow us to 
perform computations on the encoded data without having to decode the data and then
compute.

\begin{defn}[Encoded operators]\index{encoded operators}\index{logical operators}
Given a $[[n,k,d]]_2$ stabilizer code with stabilizer $S$,
let $\ol{X}_i$, $\ol{Z}_i$ for $1\leq i\leq k$ be a set of $2k$ linearly 
independent operators in $C_{\mc{P}_n}(S)\setminus S Z(\mc{P}_n)$. The operators
$\ol{X}_i$, $\ol{Z}_i$ are said to be encoded operators for the code 
if they satisfy the following requirements. 
\begin{compactenum}[i)]
\item $[\ol{X}_i,\ol{X}_j] =0$
\item $[\ol{Z}_i,\ol{Z}_j]=0$
\item $[\ol{X}_i,\ol{Z}_j] = 2\delta_{ij}\ol{X}_i\ol{Z}_i$
\end{compactenum}
\end{defn}
The operators $\ol{X}_i$ and $\ol{Z}_j$ are referred to as encoded or logical $X$ and $Z$ operators
on the $i$th and $j$th logical qubits, respectively.
The choice of which of the $2k$ linearly independent elements of $C_{\mc{P}_n}(S)\setminus SZ(\mc{P}_n)$
we choose to call encoded $X$ operators and $Z$ operators is arbitrary;
as long as  the generators satisfy the conditions above, any choice is valid. 
Different choices lead to different sets of encoded logical states; alternatively,
a different orthonormal basis for the codespace.

\begin{lemma}[Standard form of stabilizer matrix \cite{cleve97,gottesman97}]\label{lm:stabStdForm}
\index{encoding!standard form}
Up to a permutation $\pi$, the 
stabilizer matrix of an $[[n,k,d]]_2$ code can be put in the 
following form, 
\begin{eqnarray}
S=_{\pi}\left[\begin{array}{ccc|ccc}
I_{s'}& A_1 & A_2 & B & 0 & C\\
0 & 0 & 0& D&I_{n-k-s'} & E
\end{array} \right],\label{eq:stdForm}
\end{eqnarray}
while the associated encoded operators can be derived as 
\begin{eqnarray}
\left[\begin{array}{c}
\ol{Z}\\
\ol{X} 
\end{array}\right]
=_{\pi}\left[\begin{array}{ccc|ccc}
0 & 0 &0 &A_2^t&0 &I_{k}\\
\hline
0 & E^t &I_{k}&C^t&0 &0 
\end{array}\right].\label{eq:encOps}
\end{eqnarray}
\end{lemma}
\begin{remark}Encoding 
using essentially same ideas is possible even if the  
identity matrices $I_{s'}$ in the stabilizer matrix or $I_{k}$ in the encoded operators 
are replaced by upper triangular matrices.
\end{remark}
The standard form of the stabilizer matrix prompts us to distinguish between
two types of the generators for the stabilizer as they affect the encoding
in different ways (although it can be shown that they are of equivalent complexity). 

\begin{defn}[Primary generators]
A generator $G_i=(a_1,\ldots,a_n|b_1,\ldots,b_n)$ with at least one nonzero 
$a_i$ is  called a primary generator.
\end{defn}
In other words, primary generators contain at least one $X$ or $Y$
operator on some qubit. 
The primary generators determine to a large extent the complexity of the 
encoding circuit along with
the encoded $X$ operators. 
The operators $\ol{X}$ are also called
seed generators and they also figure in the encoding circuit.
The encoded $Z$
operators do not.

\begin{defn}[Secondary generators]
A generator of the form $(0,\ldots,0|b_1,\ldots,b_n)$ is called
secondary generator.
\end{defn}
In the standard form encoding, the complexity of the encoded $X$ operators is determined 
by the secondary generators. 
Therefore they indirectly contribute\footnote{Indirect because the submatrix 
$E$,  figures in both 
the secondary generators, see equation~(\ref{eq:stdForm}), and also the encoded 
$X$ operators, see equation~(\ref{eq:encOps}).} to  the complexity of encoding. 
We mentioned earlier that different choices of the encoded operators amounts to 
choosing different orthonormal basis for the codespace. 
However, the choice in Lemma~\ref{lm:stabStdForm} is particularly 
suitable for encoding. We can represent our input in the form 
$\ket{0}^{\otimes^{n-k}}\ket{\alpha_1\ldots\alpha_k}$ which allows us to make
the identification that $\ket{0}^{\otimes^n}$ is mapped to $\ket{\ol{0}}$, 
the logical all zero code word. 
This state is precisely the state stabilized by the stabilizer generators and 
logical $Z$ operators, (which in Lemma~\ref{lm:stabStdForm} can be seen to be consisting
of only $Z$ operators).
Given the stabilizer matrix in the standard form and the encoded operators as in 
Lemma~\ref{lm:stabStdForm}, the encoding circuit is given as follows. 
\begin{lemma}[Standard form encoding stabilizer codes \cite{cleve97,gottesman97}]\label{lm:implement}
\index{encoding!standard form}
Let $S$ be the  stabilizer matrix of an $[[n,k,d]]$ stabilizer code 
in the standard form {\em i.e.}, as in equation~(\ref{eq:stdForm}).
Let $G_i$ denote the $i$th primary generator of $S$ and $\ol{X_j}$ denote the $jth$ encoded
$X$ operator as in equation~(\ref{eq:encOps}). Then these operators are in the form\footnote{We allow some
freedom in the primary generators, in that instead of $I_{s'}$ in equation~(\ref{eq:stdForm}), we
allow it be an upper triangular matrix also.}
\begin{eqnarray*}
G_i&=&(0,0,\ldots,1,a_{i+1},\ldots, a_n|b_1,\ldots,b_{s'},0,\ldots,0,b_{n-k+1},\ldots,b_n),\\
\ol{X}_j&=&(0,\ldots,0,c_{s'+1},\ldots,c_{n-k}0,\ldots,0,1=c_{n-k+j},0,\ldots,0|d_1,\ldots,d_{s'},0,\ldots,0).
\end{eqnarray*}
To encode the stabilizer code we implement the following circuits
corresponding to each of the primary generators
and the encoded operators.
The generator $G_i$ is implemented after $G_{i+1}$. 
The encoded operators precede the primary generators
in their implementation but we can implement $\ol{X}_j$ before or after $\ol{X}_{j+1}$.
\[
\Qcircuit @C=.7em @R=.3em @!R {
\lstick{\ket{0}_1} & \qw & \qw & &\dots & & \qw& \qw & \qw &\\
\lstick{\vdots} & \qw & \qw & & \dots& & \qw&\qw & \qw& \\
\lstick{\ket{0}_i} & \qw & \qw &  &\dots& &\gate{H}& \ctrl{1} &\qw& \\
\lstick{\ket{0}_{i+1}} & \qw & \qw & &\dots & & \qw&\gate{X^{a_{i+1}}Z^{b_{i+1}}} \qwx[1]& \qw & \\
\lstick{\vdots} & \qw & \qw &  &\dots&& \qw&\gate{ \vdots }\qwx[1] & \qw \\
\lstick{\ket{0}_{s'}} & \qw & \qw & &\dots && \qw& \gate{X^{a_{s'}}Z^{b_{s'}}} \qwx[1]& \qw \\
\lstick{\ket{0}_{s'+1}} &  \gate{X^{c_{s'+1}}}\qwx[1]& \qw & &\dots& &\qw&  \gate{X^{a_{s'+1}}Z^{b_{s'+1}}} \qwx[1]& \qw \\
\lstick{\vdots} & \gate{ \vdots }  \qwx[1] & \qw &  &\dots&& \qw&\gate{ \vdots }\qwx[1] & \qw \\
\lstick{\ket{0}_{n-k}} &  \gate{X^{c_{n-k}}} & \qw & & \dots & &\qw&\gate{X^{a_{n-k}}Z^{b_{n-k}}} \qwx[1]& \qw & \\
\lstick{\ket{\psi_1}}  & \qw & \qw &  &\dots & & \qw&\gate{X^{a_{n-k+1}}Z^{b_{n-k+1}}} \qwx[1]& \qw & \\
\lstick{\vdots} & \qw & \qw & & \dots&& \qw&\gate{ \vdots } \qwx[1]& \qw \\
\lstick{\ket{\psi_j}}  & \ctrl{-3}& \qw &  &\dots && \qw&\gate{X^{a_{n-k+j}}Z^{b_{n-k+j}}}\qwx[1] & \qw \\
\lstick{\vdots} & \qw & \qw & &\dots && \qw&\gate{ \vdots } \qwx[1]& \qw \gategroup{7}{2}{12}{2}{.7em}{--}\\
\lstick{\ket{\psi_k}}  & \qw & \qw&  &\dots&& \qw&\gate{X^{a_{n}}Z^{b_{n}}}& \qw \gategroup{3}{7}{15}{9}{.7em}{_\}}& \\
 & \ol{X}_j  &  & & && &G_i&&&&\\
}
\]
\end{lemma}

To encode a stabilizer code, we first put the stabilizer matrix in the standard form,
then implement the seed generators i.e., the encoded $X$ operators, followed by the 
primary generators $i=s'$ to $i=1$ as per Lemma~\ref{lm:implement}. The complexity of
encoding the $ith$ primary generator is at most $n-i$ two qubit gates and one $H$ gate.
The complexity of encoding an encoded operator is at most $n-k-s'$ CNOT gates. This means
the complexity of standard form encoding  is upper bounded by $(2n-1-k-s')s'/2$ two 
qubit gates and $s'$ Hadamard gates; $O(n(n-k))$ gates. A minor modification 
(\cite{grassl01}) must
be incorporated when $Y$ is defined as 
$\Y= \left[\begin{smallmatrix} 0&-i \\i&0\end{smallmatrix}\right]$ as the
following example illustrates. See \cite{grasslweb} for more examples.
\begin{example}
Consider the  $[[5,1,3]]$ code with following stabilizer, 
with $Y=\left[\begin{smallmatrix} 0&-i \\i&0\end{smallmatrix}\right]$.
\begin{eqnarray*}
S & =  &\left[ \begin{array}{ccccc}
	X&I&X&X&X\\
	I&X&Z&X&Y\\
	Z&I&Z&Z&Z\\
	I&Z&Y&Z&X
\end{array}\right]
\end{eqnarray*}
The associated stabilizer matrix is given by 
\begin{eqnarray*}
S & =  &\left[ \begin{array}{ccccc|ccccc}
	1&0&1&1&1&0&0&0&0&0\\
	0&1&0&1&1&0&0&1&0&1\\
	0&0&0&0&0&1&0&1&1&1\\
	0&0&1&0&1&0&1&1&1&0 
\end{array}\right]
\end{eqnarray*}
Writing $S$ in standard form we get 
\nix{
\begin{eqnarray*}
S&=&\left[ \begin{array}{ccccc|ccccc}
	1&0&0&1&0&0&1&1&1&0\\
	0&1&0&1&1&0&0&1&0&1\\
	0&0&1&0&1&0&1&1&1&0\\ 
	0&0&0&0&0&1&0&1&1&1
\end{array}\right].
\end{eqnarray*}
}
\begin{eqnarray*}
S&=&\left[ \begin{array}{ccccc|ccccc}
	1&0&0&1&0&1&1&0&0&1\\
	0&1&0&1&1&0&0&1&0&1\\
	0&0&1&0&1&1&1&0&0&1\\ 
	0&0&0&0&0&1&0&1&1&1
\end{array}\right] =
\left[ \begin{array}{ccccc}
	Y&Z&I&X&Z\\
	I&X&Z&X&Y\\
	Z&Z&X&I&Y\\
	Z&I&Z&Z&Z
\end{array}\right] =
\left[\begin{array}{c} G_1\\G_2\\G_3\\G_4\end{array} \right].
\end{eqnarray*}
The encoded operators for this code are 
\begin{eqnarray*}
\left[ \begin{array}{c}\ol{Z}\\\ol{X}
\end{array}\right]&=&\left[ \begin{array}{ccccc|ccccc}
	0&0&0&0&0&0&1&1&0&1\\
	0&0&0&1&1&1&1&1&0&0
\end{array}\right].
\end{eqnarray*}
In addition to following the procedure described in Lemma~\ref{lm:implement}, one must
throw in a $P$ gate, for every  $Y$ on the diagonal of 
the stabilizer (in standard form). The encoding circuit is given by 
\[
\Qcircuit @C=.7em @R=.4em @! {
\lstick{\ket{0}} &\qw&\qw& \qw &\qw &\qw&\gate{H}& \ctrl{1}  &\gate{P}& \qw\\
\lstick{\ket{0}} &\qw&\qw&\qw& \gate{H} & \ctrl{1}&\qw & \gate{Z}\qwx[2] & \qw & \qw\\
\lstick{\ket{0}} &\qw&\gate{H}& \ctrl{2} & \qw&\gate{Z}\qwx[1]&\qw & \qw& \qw & \qw\\
\lstick{\ket{0}}&\gate{X}& \qw& \qw &\qw& \gate{X}\qwx[1] &\qw& \gate{X}\qwx[1] &\qw &\qw\\
\lstick{\ket{\psi}}&\ctrl{-1}&\qw& \gate{Y}&\qw & \gate{Y} & \qw&\gate{Z} &\push{\rule{0em}{1.3em}} \qw &\qw\\
&\ol{X}  \gategroup{4}{2}{5}{2}{.7em}{--}&& G_3  \gategroup{3}{3}{5}{4}{.7em}{_\}}&& G_2  \gategroup{2}{5}{5}{6}{.7em}{_\}}& &G_1 \gategroup{1}{7}{5}{9}{.7em}{_\}}\\
}
\]
\end{example}

\section{Encoding Clifford Codes}
In this section, we show that a Clifford code can be encoded using its stabilizer
and therefore the methods used for encoding stabilizer codes are applicable. 
So that this chapter can be read independently of Chapter~\ref{ch:subsys1}, we
briefly recapitulate some facts about Clifford subsystem codes.
\index{Clifford code}
Let $E$ be an abstract error group {\em i.e.}, it is a finite group with a faithful 
irreducible unitary representation $\rho$ of degree $|E:Z(E)|^{1/2}$. Denote by $\phi$, 
the irreducible character afforded by $\rho$. Let $N$ be
a normal subgroup of $E$. Further, let $\chi$ be an irreducible character $\chi$ 
of $N$ such that $(\phi_N,\chi)_N >0$.
Then the Clifford code defined by $(E,\rho,N, \chi)$ is the image of the orthogonal 
projector \index{projector!Clifford code}
\begin{eqnarray}
P=\frac{\chi(1)}{|N|}\sum_{n\in N}\chi(n^{-1})\rho(n).\label{eq:cliffProj}
\end{eqnarray}

Under certain conditions we can construct a subsystem code from the Clifford code,
in particular when $E$ is the extraspecial $p$-group, the Clifford code $C$ has a
tensor product decomposition\footnote{Strictly speaking the equality should be
replaced by an isomorphism.} as $C=A\otimes B$, where $B$ is an irreducible 
$\C N$-module, $A$ is an irreducible $\C L$-module and $L=C_E(N)$.
In this case we can encode information only into the subsystem $A$, while the
co-subsystem $B$ provides additional protection. When encoded this
way we say $C$ is a Clifford subsystem code. The normal subgroup $N$ 
consists of all errors in $E$ that act trivially on $A$. It is also 
called the gauge group of the subsystem code. Our main goal will be to show how to encode into 
the subsystem $A$. Therefore,  our interest will center on the projectors
for the Clifford code and the subsystem code and not so much on the 
parameters of the codes themselves. 

An alternate projector for a Clifford code with data $(E,\rho,N,\chi)$ can be defined 
in terms of $Z(N)$, the center of $N$. The proof of this can be found in 
\cite[Theorem~6]{klappenecker033}. This projector is given as \index{projector!Clifford code}
\begin{eqnarray}
P'=\frac{1}{|Z(N)|}\sum_{n\in Z(N)}\varphi(n^{-1})\rho(n),\label{eq:cliffZProj}
\end{eqnarray}
where $\varphi$ is an irreducible character of $Z(N)$, that satisfies 
$(\chi \downarrow Z(N))(x)=\chi(1)\varphi(x)$.  
In this case $Q$ can be thought of as a stabilizer
code in the sense of \cite{calderbank98} i.e. 
\begin{eqnarray}
\rho(m) \ket{\psi} = \varphi(m)\ket{\psi} \mbox{ for any $m$ in } Z(N).
\end{eqnarray}
In addition to the assumption that the error group is an
extraspecial $p$-group we also assume that $Z(E) \le N$.  The inclusion of the center of 
$E$ does not change the code but helps in analysis. Thus we have the following lemma.

\begin{lemma}\label{lm:phivalues}
Let $(E,\rho,N,\chi)$ be the data of a Clifford code and $\varphi$ an irreducible character of $Z(N)$,
the center of $N$, satisfying $(\chi\downarrow Z(N))(x)=\chi(1)\varphi(x)$. If $E$ is an extraspecial $p$-group, then for all $n$ in $Z(N)$,  $\varphi(n) \in\{ \zeta^k \mid \zeta=e^{j2\pi k/p}, 0\leq k<p\}$. Further, 
if $Z(E)\le N$, then for any  $n\in Z(N)$, we have $\varphi(n^{-1}) \rho(n) \in \rho(Z(N))$. 
\end{lemma}
\begin{proof}
First we note that the irreducibilty of $\rho$ implies that for any $z$ in $Z(E)$ we have 
$\rho(z)=  \omega I$ for some $\omega \in \C$ by Schur's lemma. The 
assumption that $E$ is an extraspecial $p$-group
forces $\omega \in \{ \zeta^k \mid 0\leq k<p \}$ where $\zeta=e^{j2\pi/p}$. 
This is because $|Z(E)|=p$ for extraspecial $p$-groups.
Secondly, we observe that $\varphi$ is an irreducible additive character of $Z(N)$ 
(an abelian subgroup of an extraspecial $p$-group) which implies that we must have $\varphi(n) = \zeta^{l}$ for some $0\leq l<p$, \cite{lidl97}.
Together these observations imply that we can assume 
$\varphi(n^{-1}) I= \zeta^{l} I =\rho(z)$  for some $0\leq l\le p$ and
$z\in Z(E)$. Since $Z(E) \le N$, it follows that $Z(E)\le Z(N)$ and $\varphi(n^{-1})\rho(n)$ 
is in $\rho(Z(N))$.
\end{proof}

Our goal is to use the stabilizer of $Q$ for encoding and as a first step
we will show that it can be computed from $Z(N)$.
The usefulness of such a projector is
that it obviates the need to know the character $\varphi$. Let $S \le \rho(E)$ be the
stabilizer of $Q$. Then we claim that $S$ is given as 
$$
S= \{\varphi(n^{-1})\rho(n) \mid n \in Z(N) \}.
$$
We claim that $S$ can be used for encoding the 
associated Clifford code. Then we will show how the encoding circuit of the
Clifford code is to be modified so that
we can encode the subsystem code derived from the Clifford code. 
\begin{theorem}\label{th:projClifford}
Let $Q$ be a Clifford code with the data $(E,\rho,N,\chi)$ and
$\varphi$ a constituent of the restriction of $\chi$ to $Z=Z(N)$. 
Let $E$ be an extraspecial $p$-group and $Z(E)\le N$ and 
\begin{eqnarray}
S=\left\{ \varphi(n^{-1}) \rho(n) \mid n\in Z(N) \right\} \quad \mbox{and}\quad
P=\frac{1}{|S|}\sum_{s\in S} s .\label{eq:stabProj}
\end{eqnarray}
Then $S$ is the stabilizer of $Q$ and $\text{Im } P =Q$.
\end{theorem}
\begin{proof}
We will show this in a series of steps.
\begin{enumerate}[1)]
\item
First we will show that $S\leq \rho(Z)$. By Lemma~\ref{lm:phivalues} we know that
$\varphi(n^{-1})\rho(n)$ is in $\rho(Z)$, therefore $S\subseteq \rho(Z)$. For any two
elements $n_1,n_2\in Z$, we have $s_1=\varphi(n_1^{-1})\rho(n_1), s_2=\varphi(n_2^{-1})\rho(n_2) \in S$ and we can easily verify that  $s_1^{-1}s_2 = \varphi(n_1)\rho(n_1^{-1})$ $\varphi(n_2^{-1})\rho(n_2)=
\varphi(n_2^{-1}n_1)\rho(n_1^{-1}n_2) \in S$, as  $\rho(n_1^{-1}n_2)$ is in $\rho(Z)$. 
Hence $S \leq \rho(Z)$.
\item Now we show that $S$ fixes $Q$. Let $s\in S$ and $\ket{\psi}\in Q$. Then
$s=\varphi(n^{-1})\rho(n)$ for some $n\in Z$. The action of $s$ on $\ket{\psi}$
is given as $s\ket{\psi} =\varphi(n^{-1})\rho(n)\ket{\psi}= \varphi(n^{-1})\varphi(n)\ket{\psi}=\ket{\psi}$, in other words $S$ fixes $Q$. 
\item  Next, we show that $|S|=|Z|/|Z(E)|$.
If two elements $n_1$ and $n_2$ in $Z$ map to the same element in $S$, then 
$\varphi(n_1^{-1})\rho(n_1) = \varphi(n_2^{-1})\rho(n_2)$, that is 
$\rho(n_2) = \varphi(n_1^{-1}n_2)\rho(n_1)$. From Lemma~\ref{lm:phivalues} it follows that
$\rho(n_2)=\zeta^l \rho(n_1)$ for some $0\leq l<p$. 
Since $\rho(Z(E))=\{e^{j2\pi k/p } I\mid 0\leq k<p \}$, we must have $n_2=zn_1$ for some $z\in Z(E)$. 
Thus, $|S|=|Z|/|Z(E)|$.
\item
Let $T$ be a traversal of  $Z(E)$ in $Z$, then every element in $Z$
can be written as $zt$ for some $z\in Z(E)$ and $t\in T$. From step 3)
we can see that all elements in a coset of $Z(E)$ in $Z$ map to the same element
in $S$, therefore, 
$$S=\{\varphi(t^{-1})\rho(t) \mid t\in T \}.$$
Recall that a projector for $Q$ is given by  
\begin{eqnarray*}
P' &=& \frac{1}{|Z|}\sum_{n\in Z}\varphi(n^{-1})\rho(n) ,\\
&=& \frac{1}{|Z|}\sum_{t\in T}\sum_{z\in  Z(E)}\varphi((zt)^{-1})\rho(zt).
\end{eqnarray*}
But we know from step 3) that if $z\in Z(E)$, then $\varphi(n^{-1})\rho(n)  = \varphi((zn)^{-1})\rho(z n)$. 
So we can simplify $P'$ as 
\begin{eqnarray*}
P'&=&\frac{1}{|Z|}\sum_{t\in T}\sum_{z\in  Z(E)}\varphi(t^{-1})\rho(t),\\
&=&\frac{|Z(E)|}{|Z|}\sum_{t\in T}\varphi(t^{-1})\rho(t)\\
&=&\frac{1}{|S|}\sum_{s\in S} s = P.
\end{eqnarray*}
Thus the projector defined by $S$ is precisely the same as $P'$
and $P$ is also a projector for $Q$. 
%
\end{enumerate}
From step 3) it is clear that $S\cap Z(E) =\{ \mathbf{1}\}$ and by 
Lemma~\ref{th:closedsubgroup},
$S$ is a closed subgroup of $E$. By
Lemma~\ref{th:projection},
$\text{Im }P =Q$ is a stabilizer code. 
Hence $S$ is the stabilizer of $Q$.
\end{proof}

\begin{corollary}\label{co:projSubsys}
Let $Q$ be an $[[n,k,r,d]]$ Clifford subsystem code and $S$ its stabilizer. Let
\begin{eqnarray}
P = \frac{1}{|S|}\sum_{s\in S} s.\label{eq:subsysProj}
\end{eqnarray}
Then $P$ is a projector for the subsystem code {\i.e.} $Q = \text{Im } P$. 
\end{corollary}
\begin{proof}
By \cite[Theorem~4]{pre06}, we know that an $[[n,k,r,d]]$ Clifford subsystem code is derived from a
Clifford code with data $(E,\rho,N,\chi)$. This construction assumes that 
$E$ is an extraspecial $p$-group and $Z(E)\le N  \trianglelefteq E$. 
Since as subspaces the Clifford code and subsystem code are identical, 
by Theorem~\ref{th:projClifford} we conclude that
the projector defined from the stabilizer of the subspace is also a projector for the 
subsystem code. 
\end{proof}

Theorem~\ref{th:projClifford} shows that any Clifford code can be encoded using its
stabilizer. 
As to a subsystem code, while Corollary~\ref{co:projSubsys} shows that 
there exists a projector that can be defined from its stabilizer, it is not clear 
how to use it so that one respects the subsystem structure during encoding.
More precisely, how do we use the projector defined in Corollary~\ref{co:projSubsys} to encode
into the information carrying subsystem $A$ and not the gauge subsystem. This will
be the focus of the next section.

\section{Encoding Subsystem Codes}

For ease of  presentation and clarity henceforth we will focus on binary codes, 
though the results can be  extended to nonbinary alphabet using  methods 
similar to stabilizer codes, see \cite{grassl03}. 
Theorem~\ref{th:projClifford} shows that in order to encode Clifford codes we can use 
a projector derived from the underlying stabilizer to project onto the codespace. 
But in case of Clifford subsystem codes 
we know that $Q=A\otimes B$ and the information is to be actually encoded in $A$.
Hence, it is not sufficient to merely project onto $Q$, we must also show that 
we encode into $A$ when we encode using the projector defined in 
Corollary~\ref{co:projSubsys}.

Let us clarify what we mean by encoding the information in $A$
and not in $B$. Suppose that $P$ maps $\ket{0}$  to $\ket{\psi}_A\otimes \ket{0}_B$
and $\ket{1}$ to $\ket{\psi}_A\otimes \ket{1}_B$. Then the
information is actually encoded into $B$. Since the gauge group acts nontrivially on $B$,
this particular encoding does not protect information. 
Of course a subsystem code should not encode (only) into $B$, but we have to show that the 
projector defined by $P$ in equation~\ref{eq:subsysProj} does not do that.

We need the following result on the structure of the gauge group and the encoded operators
of a subsystem code. Poulin \cite{poulin05} proved a useful result on the structure of the
gauge group and the encoded operators of the subsystem code. 
But first a little notation.
A basis for $\mathcal{P}_n$ is $X_i,Z_i$, $1\leq i\leq n$, where $X_i$ and $Z_i$
are given as 
$$
X_i =\bigotimes_{j=1}^n X^{\delta_{ij}} \quad \mbox{ and } \quad Z_i =\bigotimes_{j=1}^n Z^{\delta_{ij}}.
$$
They satisfy the relations  $[X_i,X_j]=0=[Z_i,Z_j]$; $[X_i,Z_j]=2\delta_{ij}X_iZ_j$. However,
we can choose other generating sets $\{x_i,z_i\mid 1\leq i \leq n \}$ for $\mathcal{P}_n$
that satisfy similar commutation relations {\em i.e.}, 
$[x_i,x_j]=0=[z_i,z_j]$ and $[x_i,z_j]=2\delta_{ij}x_iz_j$. These operators
may act nontrivially on many qubits. 
Given an $[[n,k,r,d]]$ code we could view the state space of the physical 
$n$ qubits as that of $n$ virtual qubits on which these $x_i,z_i$ act as 
$X$ and $Z$ operators. In particular $k$ of these virtual qubits are the 
logical qubits and $r$ of them  gauge qubits. The usefulness of these
operators is that we can specify the structure of the stabilizer, the
gauge group and the encoded operators. The following lemma makes
this specification precise. 
\begin{lemma}\label{lm:struct}
Let $Q$ be an $[[n,k,r,d]]_2$ subsystem code with gauge group, $G$ and stabilizer $S$.
Denote the encoded operators by $\ol{X}_i,\ol{Z}_i$, $1\leq i\leq k$, where 
$[\ol{X}_i,\ol{X}_j]=0=[\ol{Z}_i,\ol{Z}_j]; [\ol{X}_i,\ol{Z}_j ]=2  \delta_{ij}\ol{X}_i \ol{Z}_j$. Then 
there exist operators $\{ x_i,z_i \in \mathcal{P}_n \mid 1\leq i\leq n\}$ such that 
\begin{compactenum}[i)]
\item $S= \langle z_1, z_2, \ldots, z_s \rangle$, 
\item $G = \langle S, z_{s+1},x_{s+1},\ldots, z_{s+r},x_{s+r}, Z({\mathcal{P}_n}) \rangle$, 
\item $C_{\mathcal{P}_n}(S) = \langle G, \ol{X}_1,\ol{Z}_1,\ldots,  \ldots, \ol{X}_k, \ol{Z}_k \rangle$,
\item $\ol{X}_i =x_{s+r+i}$ and $\ol{Z}_i=z_{s+r+i}$, $1\leq i\leq k$,
\end{compactenum}
where 
$[z_i,z_j]=[x_i,x_j]=0; [x_i,z_i]=2\delta_{ij}x_iz_i$. Further, $S$ defines an $[[n,k+r]]$ stabilizer
code encoding into the same space as the subsystem code and its encoded operators are given
by $\{x_{s+1},z_{s+1},\ldots, x_{s+r},z_{s+r},  \ol{X}_1,\ol{Z}_1, \ldots, \ol{X}_k,\ol{Z}_k\}$
\end{lemma}
\begin{proof}
See \cite{poulin05} for proof on the structure of the groups.
Let $Q=A\otimes B$, then $\dim A =2^k$ and $\dim B= 2^r$.
From Corollary~\ref{co:projSubsys} we know that the projector defined by $S$  also 
projects onto $Q$ (which is $2^{k+r}$-dimensional) and therefore it defines an 
$[[n,k+r]]$ stabilizer code. 
From the definition of the operators $x_i,z_i$  and  $\ol{X}_{i}, \ol{Z}_i$ and the fact that
$$C_{\mc{P}_n}(S) = \langle S,x_{s+1},z_{s+1},\ldots, x_{s+r},z_{s+r} \ol{X}_1,\ol{Z}_1,\ldots, \ol{X}_k,\ol{Z}_k, Z({\mathcal{P}_n}) \rangle $$ 
we see that 
$x_i,z_i$, for $ s+1\leq i\leq r$ act like encoded operators on the gauge qubits, 
while  $\ol{X}_{i}, \ol{Z}_i$ continue to be the encoded operators on the information qubits.
Together they exhaust the set of $2(k+r)$ encoded operators of the $[[n,k+r]]$ stabilizer code. 
\end{proof}

We observe that the logical operators of the subsystem code are also logical operators for the
underlying stabilizer code. so if the 
stabilizer code and the subsystem code have the same logical all zero state, then 
Lemma~\ref{lm:struct} suggests that in order to encode the subsystem code, we can treat it as
stabilizer code and use the same techniques to encode. If the logical all zero code word was the same for
both the codes, then  because they have the same logical operators we can encode any given 
input to the same logical state in both cases. Using linearity we could then encode any
arbitrary state. Encoding  the all zero state seems to be the key. Now, even in the case
of the stabilizer codes, there is no unique all zero logical state. There are many possible
choices. The reader can refer to the appendix for examples. Given the 
encoded operators it is easy to define the logical all zero state as the following definition
shows: 

\begin{defn}\label{def:logZero}
A logical all zero state of an $[[n,k,r,d]]$ subsystem code is any state that is
fixed by its stabilizer and $k$ logical $Z$ operators. 
\end{defn}
This definition is valid in case of stabilizer codes also. This definition might appear a little circular. After all, we seem to have assumed the definition of the logical $Z$ operators. Actually, 
this is a legitimate definition because, depending on the choice
of our logical operators, we can have many choices of the  logical all zero state. In case of the
subsystem codes, this definition implies that the logical all zero state is fixed by $n-r$ operators,
consequently it can be any state in that $2^r$-dimensional subspace. If we consider the $[[n,k+r]]$
stabilizer code that is associated to the subsystem code, then its logical zero is additionally
fixed by $r$ more operators. So any logical zero of the stabilizer code is also a logical all zero  state
of the subsystem code. It follows that if we know how to encode the stabilizer code's logical all zero,
we know how to encode the subsystem code. We are
interested in more than merely encoding the subsystem code of course. We also want  to leverage
the gauge qubits to simplify and/or make the encoding process more robust. 
Perhaps a few examples will clarify the ideas.

\subsection{Illustrative Examples}
Consider the following $[[4,1,1,2]]_2$ subsystem code, with the gauge group $G$,
stabilizer $S$ and encoded operators given by $L$.
\begin{eqnarray*}
S&=&\left[ \begin{array}{cccc}X&X&X&X\\Z&Z&Z&Z 
\end{array}\right] =\left[ \begin{array}{c}z_1\\z_2\end{array}\right],\\
G&=&\left[ \begin{array}{cccc} 
X&X&X&X\\Z&Z&Z&Z \\  \hline
I&X&I&X \\I&I&Z&Z
\end{array}\right] = \left[ \begin{array}{c}  z_1\\z_2\\\hline x_3\\z_3\end{array}\right].
\end{eqnarray*} 
The encoded operators of this code are given by 
\begin{eqnarray*}
L&=&\left[ \begin{array}{cccc} I&I&X&X \\I&Z&I&Z\end{array}\right]=\left[ \begin{array}{c}\ol{X}_1\\\ol{Z}_1\end{array}\right].
\end{eqnarray*} 
The associated $[[4,2]]$ stabilizer code has the following encoded operators.
\begin{eqnarray*}
T&=&\left[ \begin{array}{cccc}  I&X&I&X \\ I&I&X&X\\
 I&I&Z&Z\\ I&Z&I&Z 
\end{array}\right]=\left[ \begin{array}{c}x_3\\\ol{X}_1 \\
z_3 \\\ol{Z}_1 \end{array}\right].
\end{eqnarray*} 
It will be observed that the encoded $X$ operators of $[[4,2]]$ are in a
form convenient for encoding. 
We treat the $[[4,1,1,2]]$ code as $[[4,2]]$ code and encode it as in 
Figure~\ref{fig:stdFormEncEx1}. The gauge qubits are permitted to be 
in any state. 
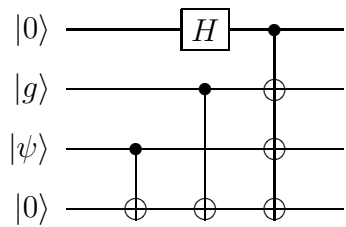
\begin{figure}[ht]
\[
\Qcircuit @C=.7em @R=.4em @! {
\lstick{\ket{0}} & \qw & \gate{H}& \ctrl{3} & \qw \\
\lstick{\ket{g}} & \qw & \ctrl{2} & \targ & \qw & \\
\lstick{\ket{\psi}} & \ctrl{1} & \qw & \targ & \qw & \\
\lstick{\ket{0}} & \targ & \targ & \targ & \qw & \\
}
\]
\caption{Encoding the $[[4,1,1,2]]$ code (Gauge qubits can be in any state)}\label{fig:stdFormEncEx1}
\end{figure}

Assuming $g=a\ket{0}+b\ket{1}$, the logical states up to a normalizing constant are
\begin{eqnarray*}
\ket{\ol{0}} &=&a(\ket{0000}+\ket{1111})+b(\ket{0101}+\ket{1010}),\\
\ket{\ol{1}} &=&a(\ket{0011}+\ket{1100})+b(\ket{0110}+\ket{1001}).
\end{eqnarray*}
It can be easily verified that $S$ stabilizes the above state and while 
the gauge group acts in a nontrivial fashion, the resulting states are still
orthogonal. 
In this example we have encoded as if we were encoding the $[[4,2]]$ code. Prior to 
encoding the gauge qubits can be identified with physical qubits. After the encoding
however such a correspondence between the physical qubits and gauge qubits does not 
necessarily exist in a nontrivial subsystem code. Since the encoded operators
of the subsystem code are also encoded operators for the stabilizer code, we are
guaranteed that the information is not encoded into the gauge subsystem.

As the state of gauge qubits is of no consequence, we can initialize them to any state.
Alternatively, if we initialized them to zero, we can simplify the circuit as shown in 
Figure~\ref{fig:stdFormEncEx1Opt}.
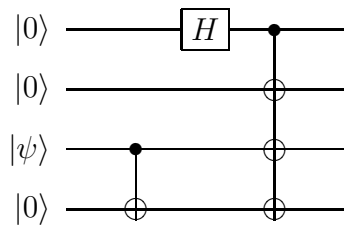
\begin{figure}[ht]
\[
\Qcircuit @C=.7em @R=.4em @! {
\lstick{\ket{0}} & \qw & \gate{H}& \ctrl{3} & \qw \\
\lstick{\ket{0}} & \qw & \qw & \targ & \qw & \\
\lstick{\ket{\psi}} & \ctrl{1} & \qw & \targ & \qw & \\
\lstick{\ket{0}} & \targ & \qw & \targ & \qw & \\
}
\]
\caption{Encoding the $[[4,1,1,2]]$ code (Gauge qubits initialized to zero)}\label{fig:stdFormEncEx1Opt}
\end{figure}

The encoded states in this case are (again, the normalization factors are ignored)
\begin{eqnarray*}
\ket{\ol{0}} &=&\ket{0000}+\ket{1111},\\
\ket{\ol{1}} &=&\ket{0011}+\ket{1100}.
\end{eqnarray*}
The benefit with respect to the previous version is that at the cost of initializing
the gauge qubits, we have been able to get rid of all the encoded operators associated
with them. This seems to be a better option than randomly initializing the gauge qubits.
Because it is certainly easier to prepare them in a known state like $\ket{0}$, rather
than implement a series of controlled gates depending on the encoded operators associated
with those qubits. 

At this point we might ask if it is possible to get both the benefits of random initialization
of the gauge qubits as well as avoid implementing the encoded operators associated with them.
To answer this question let us look a little more closely at the previous two encoding circuits for the
subsystem codes. We can see from them that it will not work in general. Let us see why. If 
we initialize the gauge qubit to $\ket{1}$ instead of $\ket{0}$ in the encoding given in 
Figure~\ref{fig:stdFormEncEx1Opt}, then  the  encoded state is 
\begin{eqnarray*}
\ket{\ol{0}} &=&\ket{0100}+\ket{1011},\\
\ket{\ol{1}} &=&\ket{0111}+\ket{1000}.
\end{eqnarray*}
Both these states are not stabilized by $S$, indicating that these states are not in the
code space. 

In general, an encoding circuit where it is simultaneously possible 
initialize the gauge qubits to random states and also avoid the encoded operators is likely
to be having more complex primary generators. For instance, let us consider the 
following $[[4,1,1,2]]$ subsystem code:
\begin{eqnarray*}
S&=&\left[ \begin{array}{cccc}X&Z&Z&X\\Z&X&X&Z 
\end{array}\right] =\left[ \begin{array}{c}z_1\\z_2\end{array}\right],\\
G&=&\left[ \begin{array}{cccc} 
X&Z&Z&X\\Z&X&X&Z \\  \hline
Z&I&X&I \\I&Z&Z&I
\end{array}\right] = \left[ \begin{array}{c}  z_1\\z_2\\\hline x_3\\z_3\end{array}\right].
\end{eqnarray*} 
The encoded operators of this code are given by 
\begin{eqnarray*}
L&=&\left[ \begin{array}{cccc} I&Z&I&X \\Z&I&I&Z\end{array}\right]=\left[ \begin{array}{c}\ol{X}_1\\\ol{Z}_1\end{array}\right].
\end{eqnarray*} 
The associated $[[4,2]]$ stabilizer code has the following encoded operators.
\begin{eqnarray*}
T&=&\left[ \begin{array}{cccc}  Z&I&X&I \\
I&Z&I&X \\
I&Z&Z&I\\
Z&I&I&Z
\end{array}\right]=\left[ \begin{array}{c}x_3\\\ol{X}_1 \\
z_3 \\\ol{Z}_1 \end{array}\right].
\end{eqnarray*} 
The encoding circuit for this code is given in Figure~\ref{fig:enc4112rand}.
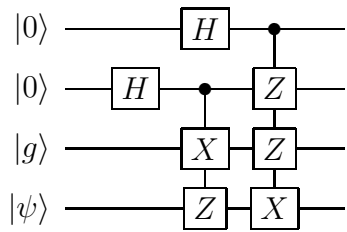
\begin{figure}[htb]
\[
\Qcircuit @C=.7em @R=.4em @! {
\lstick{\ket{0}} & \qw & \gate{H}& \ctrl{1} & \qw \\
\lstick{\ket{0}} & \gate{H} & \ctrl{1} & \gate{Z}\qwx[1] & \qw & \\
\lstick{\ket{g}} & \qw & \gate{X}\qwx[1] & \gate{Z} \qwx[1]& \qw & \\
\lstick{\ket{\psi}} & \qw & \gate{Z} & \gate{X} &\qw \\
}
\]
\caption{Encoding $[[4,1,1,2]]$ code (Encoded operators for the gauge qubits are trivial and 
gauge qubits can be initialized to random states)}\label{fig:enc4112rand}
\end{figure}


In this particular case, the gauge qubits (as well as the information qubits)
do not require any additional encoding circuitry. In this case we can 
initialize the gauge qubits to any state we want.
But, the reader would have observed we did not altogether end up with 
a simpler circuit. The primary generators are two as against one and the
complexity of the encoded operators has been shifted to them.
So even though we were able to  get rid of the encoded operator on the gauge qubit 
and also get the benefit of initializing it to a random state, this is still more 
complex compared to either of encoders in Figures~\ref{fig:stdFormEncEx1} and \ref{fig:stdFormEncEx1Opt}.
Our contention is that it is better to initialize
the gague qubits to zero state and not implement the encoded operators associated
to them. 

\subsection{Encoding Subsystem Codes by Standard Form Method}
The previous two examples might lead us to conclude that we can take the stabilizer
of the given subsystem code and form the encoded operators by reducing the stablizer
to its standard form and encode as if it were a stabilizer code. However, there are certain
subtle points to be kept in mind.
When we form the encoded operators we get $k+r$
encoded operators; we cannot from the stabilizer alone conclude which are the encoded
operators on the information qubits and which on the gauge qubits. Put differently,
these operators belong to the space $C_{\mc{P}_n}(S)\setminus S = G C_{\mc{P}_n}(G) \setminus S Z(\mc{P}_n)$. It is not guaranteed that they are entirely in $C_{\mc{P}_n}(G)$ {\em i.e.}, we cannot say if they act as encoded operators on the logical qubits. This implies that in general all these operators act nontrivially
on both $A$ and $B$. Consequently, we must be careful in choosing the encoded operators and
the gauge group must be taken into account. We give two slightly different methods for 
encoding subsystem codes. The difference between the two methods is subtle. Both methods
require the gauge qubits to be initialized to zero. In the second method (see Algorithm~\ref{alg:subsysEncOpt}) however,
we can avoid the encoded operators associated to them. Under certain circumstances, we can also 
permit initialization to random states.

\begin{algorithm}[h!]
\caption{{\ensuremath{\mbox{\scshape Encoding subsystem codes -- Standard form method 1}}}
}\label{alg:subsysEnc}
\begin{algorithmic}[1]
\REQUIRE Gauge group, $G=\langle S, x_{s+1},z_{s+1},\ldots, x_{s+r},z_{s+r}, \pm I \rangle$
 and stabilizer, $S =\langle z_1,\ldots, z_{n-k-r}\rangle$ of the $[[n,k,r,d]]$ subsystem code.

\ENSURE $[x_i,x_j] =[z_i,z_j] =0$; $[x_i,z_j ]=2x_iz_i \delta_{ij}$
\medskip
\STATE Form $S_A= \langle S, z_{s+1}, \ldots, z_{s+r}\rangle $, where $s=n-k-r$
\STATE Compute the standard form of $S_A$ as per Lemma~\ref{lm:stabStdForm}
$$S_A =_{\pi} \left[\begin{array}{ccc|ccc}
I_{s'}& A_1 & A_2 & B & 0 & C\\
0 & 0 & 0& D&I_{s+r-s'} & E
\end{array} \right] $$

\STATE Compute the encoded operators $\ol{X}_1,\ldots, \ol{X}_k$ as
$$
\left[\begin{array}{c}
\ol{Z}\\
\ol{X} 
\end{array}\right]
=_{\pi}\left[\begin{array}{ccc|ccc}
0 & 0 &0 &A_2^t&0 &I_{k}\\
\hline
0 & E^t &I_{k}&C^t&0 &0 
\end{array}\right]
$$
\STATE Encode using the primary generators of $S_A$ and $\ol{X}_i$ as encoded operators, 
see Lemma~\ref{lm:implement}; all the other $(n-k)$ qubits are initialized to $\ket{0}$.
\end{algorithmic}
\end{algorithm}

\textbf{Correctness of Algorithm~\ref{alg:subsysEnc}.}
Since stabilizer $S_A \ge S$, the space stabilized by $S_A$ is  a subspace of the $A\otimes B$,
the subspace stabilized by $S$. As $|S_A|/|S|=2^r$, the dimension of the subspace stabilized by $S_A$ is
$2^{k+r}/2^r=2^k$. 
Additionally, the generators $z_{s+1},\ldots, z_{s+r}$ act trivially on 
$A$. The encoded operators as computed in the algorithm act nontrivially on $A$ and give 
$2^k$ orthogonal states; thus we are assured that the information is encoded into $A$. 

Let us encode the $[[9,1,4,3]]$ Bacon-Shor code using the method just proposed. The stabilizer
and the gauge group are given by  
\begin{eqnarray*}
S &= &\left[ \begin{array}{ccc|ccc|ccc}
X&X&X&I&I&I&X&X&X\\
I&I&I&X&X&X&X&X&X\\
Z&I&Z&Z&I&Z&Z&I&Z\\
I&Z&Z&I&Z&Z&I&Z&Z
\end{array}\right],
\end{eqnarray*}
\begin{eqnarray*}
G &= &\left[ \begin{array}{ccc|ccc|ccc}
X&X&X&I&I&I&X&X&X\\
I&I&I&X&X&X&X&X&X\\
Z&I&Z&Z&I&Z&Z&I&Z\\
I&Z&Z&I&Z&Z&I&Z&Z\\ \hline 
I&X&I&I&X&I&I&I&I\\
I&I&X&I&I&X&I&I&I\\
I&I&I&I&I&X&I&I&X\\
X&X&X&X&X&X&I&I&I\\\hline 
Z&I&Z&I&I&I&I&I&I\\
I&I&I&Z&I&Z&I&I&I\\
I&Z&Z&I&I&I&I&I&I\\
I&I&I&I&Z&Z&I&I&I
 \end{array}\right]
 =\left[ \begin{array}{c}
S\\\hline G_x\\ \hline G_z\end{array}\right].
\end{eqnarray*}
Let us form $S_A$ by augmenting  $S$ with $G_z$. Then
\begin{eqnarray*}
S_A&=&\left[ \begin{array}{ccc|ccc|ccc}
X&X&X&I&I&I&X&X&X\\
I&I&I&X&X&X&X&X&X\\
Z&I&Z&Z&I&Z&Z&I&Z\\
I&Z&Z&I&Z&Z&I&Z&Z\\ \hline 
Z&I&Z&I&I&I&I&I&I\\
I&I&I&Z&I&Z&I&I&I\\
I&Z&Z&I&I&I&I&I&I\\
I&I&I&I&Z&Z&I&I&I
\end{array}\right].
\end{eqnarray*}
The encoded $X$ and $Z$ operators are $X_7X_8X_9$ and $Z_1Z_4Z_7$, respectively. 
After putting $S_A$ in the standard form, and encoder for this code is given 
in Figure~\ref{fig:shorCodeEnc}.
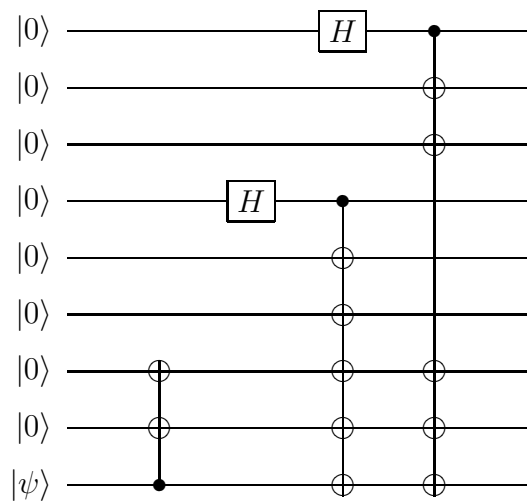
\begin{figure}[htb]
\[
\Qcircuit @C=1.4em @R=0.3em @!{
\lstick{\ket{0}}&\qw&\qw&\gate{H} &\ctrl{8}&\qw\\
\lstick{\ket{0}}&\qw&\qw&\qw &\targ&\qw\\
\lstick{\ket{0}}&\qw&\qw&\qw &\targ&\qw\\
\lstick{\ket{0}}&\qw&\gate{H} &\ctrl{5}&\qw&\qw\\
\lstick{\ket{0}}&\qw&\qw &\targ&\qw&\qw\\
\lstick{\ket{0}}&\qw&\qw &\targ&\qw&\qw\\
\lstick{\ket{0}}&\targ&\qw &\targ&\targ&\qw\\
\lstick{\ket{0}}&\targ&\qw&\targ&\targ&\qw\\
\lstick{\ket{\psi}}&\ctrl{-2}&\qw&\targ&\targ&\qw\\
}
\]
\caption{Encoder for the $[[9,1,4,3]]$ code. This is also an encoder for the
$[[9,1,3]]$ code.}\label{fig:shorCodeEnc}
\end{figure}

If on the other hand we had formed $S_A$ by adding $G_x$ instead, then 
$S_A$ would have been 
\begin{eqnarray*}
S_A=\left[ \begin{array}{ccc|ccc|ccc}
X&I&I&I&I&I&X&I&I\\
I&X&I&I&X&I&I&X&I\\
I&I&X&I&I&X&I&I&X\\
I&I&I&X&I&I&X&I&I\\
I&I&I&I&X&I&I&X&I\\
I&I&I&I&I&X&I&I&X\\
Z&I&Z&Z&I&Z&Z&I&Z\\
I&Z&Z&I&Z&Z&I&Z&Z\\ 
\end{array}\right].
\end{eqnarray*}
The encoded operators remain the same. In this case the encoding circuit is given 
in Figure~\ref{fig:shorCodeEncFewer}.
\begin{figure}[htb]
\[
\Qcircuit @C=1.4em @R=0.3em @!{
\lstick{\ket{0}}&\qw&\qw&\qw&\qw&\qw &\qw&\gate{H}&\ctrl{6}&\qw\\
\lstick{\ket{0}}&\qw&\qw&\qw&\qw&\qw &\gate{H}&\ctrl{6}&\qw&\qw\\
\lstick{\ket{0}}&\qw&\qw&\qw&\qw&\gate{H}&\ctrl{6}&\qw&\qw&\qw\\
\lstick{\ket{0}}&\qw&\qw&\qw&\gate{H} &\ctrl{3}&\qw&\qw&\qw&\qw\\
\lstick{\ket{0}}&\qw&\qw &\gate{H}&\ctrl{3}&\qw&\qw&\qw&\qw&\qw\\
\lstick{\ket{0}}&\qw&\gate{H} &\ctrl{3}&\qw&\qw&\qw&\qw&\qw&\qw\\
\lstick{\ket{0}}&\targ&\qw &\qw&\qw&\targ&\qw&\qw&\targ&\qw\\
\lstick{\ket{0}}&\targ&\qw&\qw&\targ&\qw&\qw&\targ&\qw&\qw\\
\lstick{\ket{\psi}}&\ctrl{-2}&\qw&\targ&\qw&\qw&\targ&\qw&\qw&\qw\\
}
\]
\caption{Encoder for the $[[9,1,4,3]]$ code with fewer CNOT gates.}\label{fig:shorCodeEncFewer}
\end{figure}
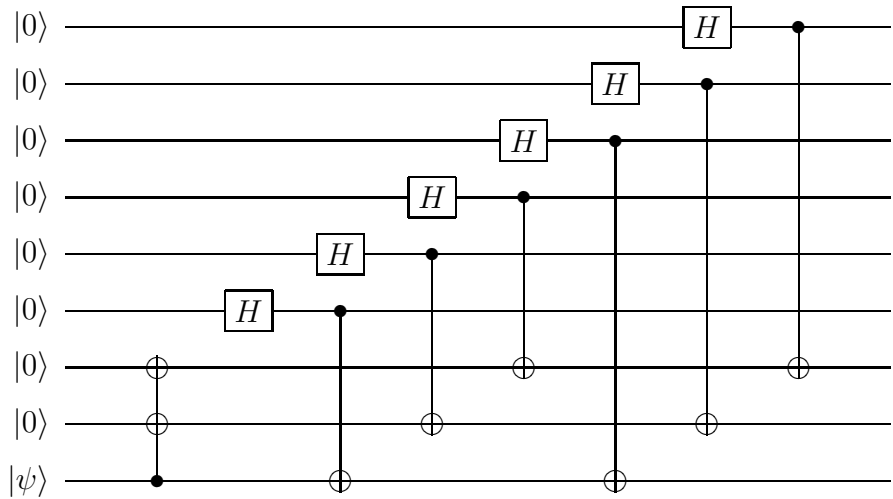
This circuit has fewer CNOT gates, though the number of single qubit gates has
increased. Since we expect the implementation of the CNOT gate to be
more complex than the $H$ gate, this might be a better choice. In any case,
this demonstrates that by exploiting the gauge qubits one can find ways
to reduce the complexity of encoding circuit.

The gauge qubits provide a great degree of freedom in encoding.  We consider the following
variant on standard form encoding, where we try to minimize the 
the number of primary generators. This is not guaranteed to reduce
the overall complexity, since that is determined by both the primary generators
and the encoded operators. Fewer primary generators might usually imply encoded operators 
with larger complexity. In fact we have already seen, that in the case of 
$[[9,1,4,3]]_2$ code that a larger number of primary generators does not 
necessarily imply higher complexity.
However, it has the potential for lower complexity.

\begin{algorithm}[h!]
\caption{{\ensuremath{\mbox{\scshape Encoding subsystem codes -- Standard form method 2}}}}\label{alg:subsysEncOpt}
\begin{algorithmic}[1]
\REQUIRE Gauge group, $G=\langle S, x_{s+1},z_{s+1},\ldots, x_{s+r},z_{s+r}, \pm I \rangle$
 and stabilizer, $S =\langle z_1,\ldots, z_{n-k-r}\rangle$ of the $[[n,k,r,d]]$ subsystem code.

\ENSURE $[x_i,x_j] =[z_i,z_j] =0$; $[x_i,z_j ]=2x_iz_i \delta_{ij}$
\medskip
\STATE Compute the standard form of $S$ as per Lemma~\ref{lm:stabStdForm}
$$S =_{\pi_1} \left[\begin{array}{ccc|ccc}
I_{s'}& A_1 & A_2 & B & 0 & C\\
0 & 0 & 0& D&I_{s-s'} & E
\end{array} \right] $$

\STATE Form $S_A= \langle S, z_{s+1}, \ldots, z_{s+r}\rangle $, where $s=n-k-r$
\STATE Compute the standard form of $S_A$ as per Lemma~\ref{lm:stabStdForm}
$$S_A =_{\pi_2} \left[\begin{array}{ccc|ccc}
I_{l}& F_1 & F_2 & G_1 & 0 &G_2 \\
0 & 0 & 0& D'&I_{s+r-l} & H
\end{array} \right] $$

\STATE Compute the encoded operators $\ol{X}_1,\ldots, \ol{X}_k$ as
$$
\left[\begin{array}{c}
\ol{Z}\\
\ol{X} 
\end{array}\right]
=_{\pi_2}\left[\begin{array}{ccc|ccc}
0 & 0 &0 &F_2^t&0 &I_{k}\\
\hline
0 & H^t &I_{k}&G_2^t&0 &0 
\end{array}\right]
$$

\STATE Encode using the primary generators of $S$ and $\ol{X}_i$ as encoded operators,
accounting for $\pi_1$ and $\pi_2$, see Lemma~\ref{lm:implement}; all the other $(n-k)$ qubits are initialized to $\ket{0}$.
\end{algorithmic}
\end{algorithm}

The main difference in the second method comes in lines 1 and 5. We encode using the 
primary generators of the stabilizer of the subsystem code instead of the augmented
stabilizer. The encoded operators however remain the same as before. 

\textbf{Correctness of Algorithm~\ref{alg:subsysEncOpt}.}
The correctness of this method lies in the observation we made earlier 
(see discussion following Definition~\ref{def:logZero}), that any logical all zero
state of the stabilizer code is also a logical all zero of the subsystem code and the fact that
both share the encoded operators on the encoded qubits.

The encoded operators are given modulo the elements of
the gauge group as in Algorithm~\ref{alg:subsysEnc}, which implies that the their action might be nontrivial on the gauge qubits. The benefit of the second method is when $S$ and $S_A$ have different number of 
primary generators. 
The following aspects of both the methods are worth highlighting. 
\begin{compactenum}[1)]
\item The gauge qubits must be initialized to $\ket{0}$ in both methods. 
\item In Algorithm~\ref{alg:subsysEnc},  the number of primary generators of $S$ 
and $S_A$ can be different leading to a potential increase in complexity compared to 
encoding with $S$.
\item In both methods, the encoded operators as computed are modulo $S_A$. Consequently, 
the  encoded operators might act nontrivially on the gauge qubits. 
\end{compactenum}

\subsection{Encoding Subsystem Codes by Conjugation Method} 
The other benefit of subsystem codes is
the random initialization of the gauge qubits.  We now give circuits where we can encode 
the subsystem codes to realize this benefit. But instead of using the standard form method
we will use the conjugation method proposed by Grassl {\em et al.}, \cite{grassl03}
for stabilizer codes. After
briefly reviewing this method we shall show how it can be modified for encoding subsystem
codes. 

The conjugation encoding method can be understood as follows.
It is based on the idea that the Clifford group acts transitively on the Pauli error group. 
It is possible to transform the stabilizer matrix of any $[[n,k,d]]$ stabilizer code 
into the matrix $(0 0 | I_{n-k} 0 )$. For a code with this stabilizer matrix
the encoding is trivial. We simply map $\ket{\psi} $ to $\ket{0}^{\otimes^{n-k}} \ket{\psi}$.
The associated encoded $\ol{X}$ and $\ol{Z}$ operators are given by $(0I_k| 0 0 ) $ and 
$(0 0 |0 I_k)$ respectively. Here we give a sketch of the method for the binary case,
the reader can refer to \cite{grassl03} for details. 
Assume that the stabilizer matrix is given by $S$. 
Then we shall transform it into $(00|I_{n-k}0)$ using the following 
sequence of operations. 
\begin{eqnarray}
(X|Z)  \mapsto (I_{n-k}0|0) \mapsto (00|I_{n-k}0).	
	\label{eq:conjSteps}
\end{eqnarray}
This can be accomplished through the action of $H=\left[\begin{smallmatrix}1&1\\1&-1\end{smallmatrix}\right]$, 
$P=\left[\begin{smallmatrix}1&0\\0&i\end{smallmatrix}\right]$  
and CNOT gates on the Pauli group under conjugation.
The action of $H$ on the $i$th qubit of
$(a_1,\ldots, a_n|b_1,\ldots,b_n)$
transforms it as
\begin{eqnarray}
(a_1,\ldots, a_n|b_1,\ldots,b_n)\stackrel{H_i}{\mapsto}
(a_1,\ldots,\mathbf{b_i},\ldots, a_n|b_1,\ldots,\mathbf{ a_i}, \ldots, b_n).
\end{eqnarray}
These modified entries have been highlighted for convenience.
The phase gate $P$ on the  $i$th qubit transforms $(a_1,\ldots, a_n|b_1,\ldots,b_n)$
as
\begin{eqnarray}
(a_1,\ldots, a_n|b_1,\ldots,b_n)\stackrel{P_i}{\mapsto}
(a_1,\ldots,\mathbf{a_i},\ldots, a_n|b_1,\ldots,\mathbf{ a_i+b_i}, \ldots, b_n).
\end{eqnarray}
We denote the CNOT gate  with the control on the
$i$th qubit and the target on the $j$th qubit by $\text{CNOT}^{i,j}$. The action of the 
$\textup{CNOT}^{i,j}$ gate on $(a_1,\ldots, a_n|b_1,\ldots,b_n)$ is to transform it to 
\begin{eqnarray}
(a_1,\ldots,a_{j-1},\mathbf{ a_j+a_i} ,a_{j+1}\ldots,a_n|b_1,\ldots,b_{i-1},\mathbf{b_{i}+b_{j}},b_{i+1},\ldots, a_n).
\end{eqnarray}
Note that the $j$th entry is changed in the $X$ part while the $i$th entry is changed in the
$Z$ part. For example, consider 
\begin{eqnarray*}
(1,0,0,1,0|0,1,1,0,0) \stackrel{\text{CNOT}^{1,4}}{\mapsto} (1,0,0,\mbf{0},0|0,1,1,0,0),\\
(1,0,0,1,0|0,1,1,1,0) \stackrel{\text{CNOT}^{1,4}}{\mapsto} 
(1,0,0,\mbf{0},0|\mbf{1},1,1,1,0).
\end{eqnarray*}
\nix{\begin{eqnarray*}
(1,0,0,\mathbf{1},0|0,1,1,0,0) \stackrel{\text{CNOT}^{1,4}}{\mapsto} (1,0,0,\mathbf{0},0|0,1,1,0,0),\\
(1,0,0,\mathbf{1},0|\mathbf{0},1,1,1,0) \stackrel{\text{CNOT}^{1,4}}{\mapsto} (1,0,0,\mathbf{0},0|\mathbf{1},1,1,1,0).
\end{eqnarray*}
}
Based on the action of these three gates we have the following lemmas to transform
error operators.
\begin{lemma}\label{lm:qubitConj}
Assume that we have a error operator of the form $(a_1,\ldots, a_n|b_1,\ldots, b_n)$. Then
we apply the following gates on the $i$th qubit to transform the stabilizer, transforming
$(a_i,b_i)$ to $(\alpha,\beta)$ as per the following table. 
\medskip
\begin{center}
\begin{tabular}{c|c|c}
	$(a_i,b_i)$ & Gate& $(\alpha,\beta)$ \\ \hline
	(0,0) & $I$ & (0,0)\\
	(0,1) & $H$ & (1,0)\\
	(1,0) & $I$ & (1,0)\\
	(1,1) & $P$ & (1,0)
\end{tabular}
\end{center}
Let $\bar{x}$ denote $1+x$, then the transformation to $(a_1,\ldots, a_n|0,\ldots,0)$ is achieved by 
$$
\bigotimes_{i=1}^n H^{\bar{a}_ib_i} P^{a_ib_i}.
$$
\end{lemma}
For example, consider the following generator $(1,0,0,1,0|0,1,1,1,0)$. This can be transformed to
$(1,1,1,1,0|0,0,0,0,0)$ by the application of $I\otimes H \otimes H \otimes P \otimes I $. 

\begin{lemma}\label{lm:cnotConj}
Let $e$ be an error operator of the form $(a_1,\ldots,a_i=1,\ldots, a_n|0,\ldots, 0)$. 
Then $e$ can be transformed to $(0,\ldots,0, a_i=1,0,\ldots, 0 |0,\ldots, 0)$ by
$$\prod_{j=1, i\neq j}^n \left[\textup{CNOT}^{i,j}\right]^{a_j}.$$
\end{lemma}
As an example $(1,1,1,1,0|0,0,0,0,0)$ can be transformed to 
$(0,1,0,0,0|0,0,0,0,0)$ by 
$$
\textup{CNOT}^{2,1} \cdot  \textup{CNOT}^{2,3} \cdot \textup{CNOT}^{2,4}.
$$

The first step involves making the $Z$ portion of the stabilizer matrix all zeros. This
is achieved by single qubit operations consisting of $H$ and $P$
performed on each row one by one.

Note that we must also modify the other rows of the stabilizer matrix according to the action of the
gates applied. 

Once we have a row of stabilizer matrix in the form $(a|0)$, where $a$ is nonozero we can transform
it to the form $(0,\ldots,0,a_i=1,0,\ldots,0|0)$ by using CNOT gates. 
Thus it is easy to transform $(X|Z)$ to $(I_{n-k}0|0)$ using CNOT, $P$ and $H$ gates. 
The final transformation to $(0|I_{n-k}0)$ is achieved by using $H$ gates on the first $n-k$ qubits.
At this point the stabilizer matrix has been transformed to a trivial stabilizer matrix which 
stabilizes the state $\ket{0}^{\otimes^{n-k}}\ket{\psi}$. The encoded operators 
are $(0I_k|0)$ and $(0|0I_k)$. Let $T$ be the sequence of gates applied to transform the 
stabilizer matrix to the trivial stabilizer matrix. Then $T$  applied in the 
reverse order to $\ket{0}^{\otimes^{n-k}}\ket{\psi}$ gives the encoding circuit 
for the stabilizer code. 

Now we shall use this method to encode the subsystem codes. The main difference is 
that instead of considering just the stabilizer we need to consider the entire gauge group.
Let the gauge group be 
$G= \langle S, G_Z, G_X \rangle$,
where $G_Z=\langle  z_{s+1},\ldots,z_{s+r}\rangle$, and 
$G_X=\langle  x_{s+1},\ldots,x_{s+r}\rangle$.
The idea is to transform the gauge group as follows. 
\begin{eqnarray}
G= \left[\begin{array}{c}S \\ \hline G_Z \\ \hline G_X \end{array}\right]
\mapsto \left[ \begin{array}{ccc|ccc}
0&0&0&I_s&0&0\\ \hline
0&0&0&0&I_r&0\\ \hline
0&I_r&0&0&0&0 \end{array} \right].\label{eq:GnormForm}
\end{eqnarray}
At this point the gauge group has been transformed to a group with trivial stabilizer and
trivial encoded operators for the gauge qubits and the encoded qubits. The sequence of
gates required to achieve this transformation in the reverse order will encode the state
$\ket{0}^{\otimes^{s}}\ket{\phi}\ket{\psi}$. The state $\ket{\phi}$ corresponds to the
gauge qubits and it can be initialized to any state, while $\ket{\psi}$ corresponds to the
input. 
\begin{algorithm}[H]
\caption{{\ensuremath{\mbox{\scshape Encoding subsystem codes -- conjugation method}}}}\label{alg:subsysEncConj}
\begin{algorithmic}[1]
\REQUIRE Gauge group, $G= \langle S, G_Z, G_X \rangle$,
where $G_Z=\langle  z_{s+1},\ldots,z_{s+r}\rangle$, and 
$G_X=\langle  x_{s+1},\ldots,x_{s+r}\rangle$
 and stabilizer, $S =\langle z_1,\ldots, z_{n-k-r}\rangle$ of the $[[n,k,r,d]]$ subsystem code.

\ENSURE $[x_i,x_j] =[z_i,z_j] =0$; $[x_i,z_j ]=2x_iz_i \delta_{ij}$
\medskip
\STATE Assume that $G$ is the following form
$$G= \left[\begin{array}{c}S \\ \hline G_Z \\ \hline G_X \end{array}\right]
$$
\FORALL{$i=1$ to $s+r$}
\STATE  Transform $z_i$ to $z_i'=(a_1,\ldots,a_n|0,\ldots,0)$ using Lemma~\ref{lm:qubitConj}
\STATE Transform $z_i'$ to $(0,\ldots, a_i=1,\ldots,0|0)$ using Lemma~\ref{lm:cnotConj}
\STATE Perform Gaussian elimination on column $i$ for rows $j>i$
\ENDFOR
\STATE Apply $H$ gate on each qubit  $i=1$ to $i=s+r$
\FORALL{$i=s+1$ to $s+r$}
\STATE  Transform $x_i$ to $x_i'=(a_1,\ldots,a_n|0,\ldots,0)$ using Lemma~\ref{lm:qubitConj}
\STATE Transform $x_i'$ to $(0,\ldots, a_i=1,\ldots,0|0)$ using Lemma~\ref{lm:cnotConj}
\STATE Perform Gaussian elimination on column $i$ for rows $j>i$
\ENDFOR
\end{algorithmic}
\end{algorithm}

In the above algorithm, we assume that whenever a row is transformed according to Lemma~\ref{lm:qubitConj}
or \ref{lm:cnotConj}, all the other rows are also transformed according to the transformation applied.

\textbf{Correctness of Algorithm~\ref{alg:subsysEncConj}.} 
The correctness of the algorithm is 
straightforward. As $G$ has full rank of $n-k+r$, for each row of $G$, we will be able to find some
nonzero pair $(a,b)$ so that the the transformation in lines 2--6 can be achieved. When $S$ and $G_Z$
are in the form $(0|I_{s+r}0)$, the rows in $G_X$ are in the form 
$$
\left[\begin{array}{ccc|ccc}0&A&B&0&0&D\end{array} \right]. 
$$
The zero columns of $G_X$ are consequence of the requirement to satsify the commutation
relations with (transformed) $S$ and $G_Z$. For instance, 
The first $n-k-r$ are all zero because
they must commute with $(0|I_s 0)$, the elements of the transformed stabilizer.
The submatrix $A$ must have rank $r$, otherwise at this point one of the rows of $G_X$ commutes with
all the rows of $G_Z$ and the condition that we have there are $r$ hyperbolic pairs 
is violated. It is possible therefore to transform $A$  to the form $(0I_r0|0)$. It cannot be any other
form because then we would not have the $r$ hyperbolic pairs. 
The applied transformations transform $G$ to the form given in equation~(\ref{eq:GnormForm}).
The encoded operators for  this gauge group are clearly $(0I_k|0)$ and $(0|0I_k)$. 
We conclude with a simple example that illustrates the process. 
\begin{example}
To compare with the standard form method, we consider
the $[[4,1,1,2]]$ code again. Let the gauge group $G$,
stabilizer $S$ and encoded operators given by $L$.
\begin{eqnarray*}
S&=&\left[ \begin{array}{cccc}X&X&X&X\\Z&Z&Z&Z 
\end{array}\right] =\left[ \begin{array}{c}z_1\\z_2\end{array}\right],\\
G&=&\left[ \begin{array}{cccc} 
X&X&X&X\\Z&Z&Z&Z \\  \hline
I&I&Z&Z\\I&X&I&X 
\end{array}\right] = \left[ \begin{array}{c}  z_1\\z_2\\\hline x_3\\z_3\end{array}\right].
\end{eqnarray*} 
In matrix form $G$ can be written as 
\begin{eqnarray*}
G&=&\left[ \begin{array}{cccc|cccc} 
1&1&1&1&0&0&0&0\\0&0&0&0&1&1&1&1 \\  \hline
0&0&0&0&0&0&1&1\\
0&1&0&1&0&0&0&0
\end{array}\right].
\end{eqnarray*} 
The transformations consisting of $T_1=\textup{CNOT}^{1,2}\textup{CNOT}^{1,3}\textup{CNOT}^{1,4}$ followed by 
$T_2=I\otimes H\otimes H\otimes H$ maps $G$ to 
\begin{eqnarray*}
\stackrel{T_1}{\mapsto}\left[ \begin{array}{cccc|cccc} 
1&0&0&0&0&0&0&0\\0&0&0&0&0&1&1&1 \\  \hline
0&0&0&0&0&0&1&1\\
0&1&0&1&0&0&0&0
\end{array}\right] \stackrel{T_2}{\mapsto}
\left[ \begin{array}{cccc|cccc} 
1&0&0&0&0&0&0&0\\0&1&1&1&0&0&0&0 \\  \hline
0&0&1&1&0&0&0&0\\
0&0&0&0&0&1&0&1
\end{array}\right].
\end{eqnarray*} 
Now transform the second row using $T_3=\textup{CNOT}^{2,3}\textup{CNOT}^{2,4}$. Then 
transform using $T_4=\textup{CNOT}^{4,3}$. We get
\begin{eqnarray*}
\stackrel{T_3}{\mapsto}\left[ \begin{array}{cccc|cccc} 
1&0&0&0&0&0&0&0\\0&1&0&0&0&0&0&0 \\  \hline
0&0&1&1&0&0&0&0\\
0&0&0&0&0&0&0&1
\end{array}\right] \stackrel{T_4}{\mapsto}
\left[ \begin{array}{cccc|cccc} 
1&0&0&0&0&0&0&0\\0&1&0&0&0&0&0&0 \\  \hline
0&0&0&1&0&0&0&0\\
0&0&0&0&0&0&0&1
\end{array}\right].
\end{eqnarray*} 
Applying $T_5=H\otimes H \otimes I \otimes H $ gives us 
\begin{eqnarray*}
\stackrel{T_5}{\mapsto}\left[ \begin{array}{cccc|cccc} 
0&0&0&0&1&0&0&0\\0&0&0&0&0&1&0&0 \\  \hline
0&0&0&1&0&0&0&0\\
0&0&0&0&0&0&0&1
\end{array}\right].
\end{eqnarray*} 
We could have chosen $T_5= H\otimes H\otimes I \otimes I $, since the effect of 
$H$ on the fourth qubit is trivial. The complete circuit is given in 
Figure~\ref{fig:conjEnc4112a}.
\begin{figure}[htb]
\[
\Qcircuit @C=1.4em @R=1.2em {
\lstick{\ket{0}}&\gate{H}&\qw &\qw&\qw&\ctrl{3}&\qw\\
\lstick{\ket{0}}&\gate{H}&\qw&\ctrl{2}&\gate{H}&\targ&\qw\\
\lstick{\ket{\psi}}&\qw&\targ&\targ&\gate{H}&\targ&\qw\\
\lstick{\ket{g}}&\gate{H}&\ctrl{-1}&\targ&\gate{H}&\targ&\qw
}
\]
\caption{Encoding $[[4,1,1,2]]$ code by conjugation method}
\label{fig:conjEnc4112a}
\end{figure}

By switching the target and control qubits of the CNOT gates in $T_3$ and $T_4$
we can show that this circuit is equivalent to the circuit shown in 
Figure~\ref{fig:conjEnc4112b}.
\begin{figure}[htb]
\[
\Qcircuit @C=1.4em @R=1.2em {
\lstick{\ket{0}}&\gate{H}&\qw &\qw&\qw&\ctrl{3}&\qw\\
\lstick{\ket{0}}&\qw&\qw&\targ&\targ&\targ&\qw\\
\lstick{\ket{\psi}}&\gate{H}&\ctrl{1}&\ctrl{-1}&\qw&\targ&\qw\\
\lstick{\ket{g}}&\qw&\targ&\qw&\ctrl{-2}&\targ&\qw \gategroup{3}{2}{4}{3}{.7em}{--}
}
\]
\caption{Encoding $[[4,1,1,2]]$ code by conjugation method}\label{fig:conjEnc4112b}
\end{figure}

It is instructive to compare this circuit with the one given earlier
in Figure~\ref{fig:stdFormEncEx1}. The dotted lines 
show the additional circuitry.
Since the gauge qubit can be initialized to any state, we can initialize
$\ket{g}$ to $\ket{0}$, which then gives the following logical states for the code. 
\begin{eqnarray}
\ket{\ol{0}}&=&\ket{0000}+\ket{1111}+\ket{0011}+\ket{1100},\\
\ket{\ol{1}}&=&\ket{0000}+\ket{1111}-\ket{0011}-\ket{1100}.
\end{eqnarray}
It will be observed that $IIXX$ acts as the logical $Z$ operator while $IZIZ$ acts as the
logical $X$ operator. We could flip these logical operators by absorbing the $H$ gate into
$\ket{\psi}$. If we additionally initialize $\ket{g}$ to $\ket{0}$, we will see that
the two CNOT gates on the second qubit can be removed. The simplified
circuit is shown in Figure~\ref{fig:conjEnc4112c}.
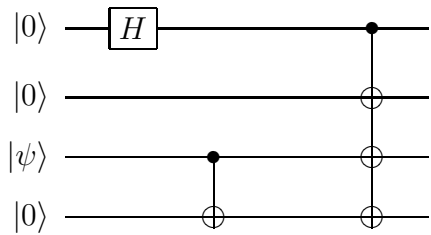
\begin{figure}[htb]
\[
\Qcircuit @C=1.4em @R=1.2em {
\lstick{\ket{0}}&\gate{H}&\qw &\qw&\qw&\ctrl{3}&\qw\\
\lstick{\ket{0}}&\qw&\qw&\qw&\qw&\targ&\qw\\
\lstick{\ket{\psi}}&\qw&\ctrl{1}&\qw&\qw&\targ&\qw\\
\lstick{\ket{0}}&\qw&\targ&\qw&\qw&\targ&\qw 
}
\]
\caption{Encoding $[[4,1,1,2]]$ code by conjugation method -- optimized}
\label{fig:conjEnc4112c}
\end{figure}

This is precisely, the same circuit that we had arrived earlier in Figure~\ref{fig:stdFormEncEx1Opt}
using the standard
form method. 
\end{example}
The preceding example provides additional evidence in the direction that 
it is better to initialize the gauge qubits to zero
and avoid the encoding operators on them. 

\section{Syndrome Measurement for Nonbinary $\F_q$-linear Codes}\label{sec:appendix} 

Decoding of nonbinary quantum codes has not been studied as well as
binary codes.  Encoding of $\F_q$-linear nonbinary quantum codes was
investigated in \cite{grassl03}.  The authors
suggest that the decoder is simply the encoder running backwards.
In this context one important task is that measuring the syndrome
so that appropriate error correction maybe performed. While binary 
codes have been well studied in this regard similar efforts have not
been invested in the nonbinary case. 
Here we give a method that allows us to measure the syndrome for
$\F_q$-linear nonbinary quantum codes.  We also show that an
$\F_q$-linear $[[n,k,r,d]]_q$ code requires $n-k-r$ syndrome
measurements.  But first we need the definition of the following
nonbinary gates, see \cite{grassl03}.
\begin{compactenum}[i)]
\item $X(a)\ket{x}=\ket{x+a}$ 
\item $Z(b)\ket{x}=\omega^{\tr_{q/p}(bx)} \ket{x}$, $\omega=e^{j2\pi/p}$
\item $M(c)\ket{x}=\ket{cx}, c\in \F_q^\times$
\item $F\ket{x}= \frac{1}{\sqrt{q}}\sum_{y\in \F_q} \omega^{\tr_{q/p}(xy)}\ket{y}$
\item $A\ket{x} \ket{y} = \ket{x}\ket{x+y}$
\end{compactenum}
Graphically, these gates are represented below.
\[
\Qcircuit @C=1em @R=.7em {
& \gate{X(a)}&\qw && \gate{Z(b)}&\qw & &\gate{c} &\qw& &\gate{F}&\qw & &\ctrl{1} &\qw \\
& &  & & & & & &&&&&&\targ &\qw \\
& &  & & & & & &&&&&& &\\
& \text{i)}&  & &\text{ii)} & & &\text{iii)} &&&\text{iv)}&&&\text{v)} &
}
\]
Consider the following circuit. 
\[
\Qcircuit @C=1em @R=.7em {
\lstick{\ket{a}}& \qw & \ctrl{1} & \qw  &\qw &\rstick{\ket{a}}\\
\lstick{\ket{y}}& \gate{g_x^{-1}} & \targ{Z} & \gate{g_x} & \qw&\rstick{\ket{y+a g_x}}
}
\]
Alternatively, this circuit maps $\ket{a}\ket{x}$ to $\ket{a}X({a
g_x})\ket{y}$. Observe that this circuit effectively applies $X({a
g_x})$ on the second qudit.  Using the linearity, we can analyze the
following circuit.
\[
\Qcircuit @C=1em @R=.7em {
\lstick{\ket{0}}& \gate{F}& \ctrl{1} & \qw  &\qw & \\ 
\lstick{\ket{y}}&\gate{g_x^{-1}} & \targ{Z} & \gate{g_x} & \qw&\rstick{\sum_{\alpha\in \F_q}\ket{\alpha}\ket{y+\alpha g_x}}
}
\]
The above circuit maps $\ket{0}\ket{y}$ to $\sum_{\alpha \in \F_q} \ket{\alpha} X({\alpha g_x})\ket{y}$. 
Using the fact that $F X(b) F^\dagger =Z(b)$, we can show that the following circuit maps 
$\ket{b}\ket{y}$ to $\ket{b} Z({bg_z}) \ket{y}$. 
\[
\Qcircuit @C=1em @R=.7em {
\lstick{\ket{b}}& \qw&\qw & \ctrl{1} & \qw  &\qw &\qw&\rstick{\ket{b}}\\
\lstick{\ket{y}}& \gate{F^\dagger}&\gate{g_z^{-1}} & \targ{Z} & \gate{g_z} & \gate{F} &\qw&\rstick{Z({bg_z})\ket{y}}
}
\]
If we wanted to apply a general operator $X({ag_x})Z({ag_z})$ to the
second qudit conditioned on the first one, then we can combine the
previous circuits as follows.
\[
\Qcircuit @C=1em @R=.7em {
\lstick{\ket{a}}& \qw&\qw & \ctrl{1} & \qw  &\qw  &\qw&\qw& \ctrl{1} & \qw  &\qw &\rstick{\ket{a}}\\
\lstick{\ket{y}}& \gate{F^\dagger}&\gate{g_z^{-1}} & \targ{Z} & \gate{g_z} & \gate{F} &\qw& \gate{g_x^{-1}} & \targ{Z} & \gate{g_x}&\qw &\rstick{X({ag_x})Z({ag_z})\ket{y}}
}
\]
The above implementation is not optimal in terms of gates, but it will
suffice for our purposes.  Consider an $[[n,k,r,d]]_q$ code. Let $E$
be an error in $G_n$, (see \ref{eq:genPauli}).  
If $E$ is detectable, then $E$ does not
commute with some element(s) in the stabilizer of the code. Let
$$g=(g_x|g_z)=(0,\ldots,0,a_j,\ldots,a_n|0,\ldots,0,b_j,\ldots,
b_n)\in \F_q^{2n},$$ where $(a_j,b_j)\neq (0,0)$, be a generator of
the stabilizer. Then for all detectable errors that do not commute
with a multiple of $g$, the following circuit gives a nonzero value on
measurement.
\[
\Qcircuit @C=1em @R=.7em {
\lstick{\ket{0}}& \gate{F}&\qw & \ctrl{9} & \qw  &\qw  &\qw&\qw& \ctrl{9} & \gate{F^\dagger}  &\qw &\meter\\
& & & &   &  &&& &  & &\\
\lstick{\ket{x_1}}& \qw&\qw & \qw & \qw  &\qw  &\qw&\qw& \qw & \qw  &\qw &\\
& \dots& & &   &\dots  &&& &  \dots &&\\
& \qw&\qw & \qw & \qw  &\qw  &\qw&\qw& \qw & \qw  &\qw &\\
\lstick{\ket{x_j}}& \gate{F^\dagger}&\gate{b_j^{-1}} & \targ{Z} & \gate{b_j} & \gate{F} &\qw& \gate{a_j^{-1}} & \targ{Z} & \gate{a_j}&\qw &\\
& \qw&\qw & \qw & \qw  &\qw  &\qw&\qw& \qw & \qw  &\qw &\\
& \dots& & &   &\dots  &&& & \dots & &\\
& \qw&\qw & \qw & \qw  &\qw  &\qw&\qw& \qw & \qw  &\qw &\\
\lstick{\ket{x_n}}& \gate{F^\dagger}&\gate{b_n^{-1}} & \targ{Z} & \gate{b_n} & \gate{F} &\qw& \gate{a_n^{-1}} & \targ{Z} & \gate{a_n}&\qw \gategroup{1}{2}{10}{10}{.7em}{--}
}
\]
Note that whenever $(a_i,b_i)=(0,0)$, then we leave that qudit alone. Similarly if $a_i$ or 
$b_i$ are zero, then we do not implement the corresponding portion. 
Let the input to the above circuit be $E\ket{\psi}$, where 
$\ket{\psi}$ is an encoded state. 
It can be easily verified that the above circuit maps the state $\ket{0}E\ket{\psi}$ to
$$ 
\sum_{\alpha \in \F_q}F^\dagger\ket{\alpha}X({\alpha g_x})Z({\alpha g_z})E\ket{\psi}.
$$
Let $X({g_x})Z({g_z})E=\omega^{\tr_{q/p}(t)}E X({g_x})Z({g_z})$, where $X({g_x})Z({g_z})$ 
is corresponding matrix representation of $g$. 
 By 
Lemma~\ref{th:commute}.
we have 
$X({\alpha g_x})Z({\alpha g_z})E=\omega^{\tr_{q/p}(\alpha t)}E X({g_x})Z({g_z})$.
Thus we can write 
\begin{eqnarray*}
\sum_{\alpha\in \F_q}\ket{\alpha}  X({\alpha g_x})Z({\alpha g_z})E\ket{\psi} &= &\sum_{\alpha \in \F_q}\ket{\alpha}\omega^{\tr_{q/p}(\alpha t)} E  X({\alpha g_x})Z({\alpha g_z})  \ket{\psi},\\
&=&\left(\sum_{\alpha\in \F_q}\ket{\alpha} \omega^{\tr_{q/p}(\alpha t)} \right)E\ket{\psi},
\end{eqnarray*}
where we have made use of the fact that $X({\alpha g_x})Z({\alpha g_z}) \ket{\psi}=\ket{\psi}$
as  $X({\alpha g_x})Z({\alpha g_z})$ is in the stabilizer.
The final state is  given by 
\begin{eqnarray*}
\sum_{\alpha\in \F_q}F^\dagger\ket{\alpha}  X({\alpha g_x})Z({\alpha g_z})E\ket{\psi} &= &
\sum_{\alpha\in \F_q}F^\dagger\ket{\alpha} \omega^{\tr_{q/p}(\alpha t)} E\ket{\psi},\\
&=&\sum_{\alpha\in \F_q}\sum_{\beta\in \F_q}\omega^{-\tr_{q/p}(\alpha\beta)}\ket{\beta} \omega^{\tr_{q/p}(\alpha t)} E\ket{\psi},\\
&=&\sum_{\beta\in \F_q}\ket{\beta}\sum_{\alpha\in \F_q} \omega^{\tr_{q/p}(\alpha t -\alpha\beta )} E\ket{\psi},\\
&=&\sum_{\beta\in \F_q}\ket{\beta}\sum_{\alpha\in \F_q} \omega^{\tr_{q/p}(\alpha t -\alpha\beta )} E\ket{\psi},\\
&=&\ket{t}E\ket{\psi},
\end{eqnarray*}
where the last equality follows from the property of the characters of
$\F_q$.  Next we observe that the error $\alpha E$, where $\alpha\in
\F_q$ gives $\ket{\alpha t}$ on measurement. Strictly speaking we refer to
the preimage of $\alpha\ol{E}$ in $G_n$.
Hence the syndrome
qudit can take $q$ different values. Since every detectable error
does not commute with some $\F_q$-multiple of a stabilizer generator,
we have the following lemma on the necessary and sufficient number of syndrome 
measurements.
\begin{lemma}\label{lm:suffMeas}
Given an $\F_q$-linear $[[n,k,r,d]]_q$ Clifford subsystem code, 
$n-k-r$ syndrome measurements are required for decoding it completely.
\end{lemma}
\begin{proof}
Let $g$ be a generator of the stabilizer of the subsystem code. 
By Theorem~\ref{th:four} and 
Lemma~\ref{lm:xzBasis}, for every generator $g$ there 
exists at least one detectable error that does not commute with $g$ but commutes 
with all the other generators. This error can be detected only by 
measuring $g$. Thus we need to measure all the
generators of the stabilizer, equivalently $n-k-r$ syndrome
measurements must be performed.

Every correctable error takes the code space into a $q^{k+r}$-dimensional orthogonal subspace in 
the $q^n$-dimensional ambient space. 
Each of these errors will give a distinct syndrome. 
This implies that we can have $q^{n-k-r}$ distinct
syndromes. Since each  syndrome measurement can have $q$ possible outcomes and there are 
$n-k-r$ generators, these measurements are sufficient for performing error correction.
\end{proof}

This parallels the classical case where
an $[n,k,d]_q$ code requires $n-k$ syndrome bits. 
A subtle caveat must be issued to the reader. If we choose to perform 
bounded distance decoding, then it maybe possible that the set of
correctable errors can be distinguished by a smaller number of syndrome
measurements. But even in the case of (classical) bounded distance decoding it is
often the case that we need to measure all the syndrome bits. 

\section{Conclusions}
In this paper, we have demonstrated that the subsystem codes can be encoded 
using the techniques used for stabilizer codes. In particular, we have
considered two methods for encoding stabilizer codes -- the standard form
method and the conjugation method. While the standard form method explored
here required us to initialize the gauge qubits to zero, 
it admits two two variants and seems to have
the potential for lower complexity; the exact gains being determined by
the actual codes under consideration. The conjugation method allows us
to initialize the gauge qubits to any state. The disadvantage seems to be
the increased complexity of encoding. It must be emphasized that the standard
form method is equivalent to the conjugation method and it is certainly
possible to use this method to encode subsystem codes so that the gauge
qubits can be initialized to arbitrary states. However, it appears to 
be a little more cumbersome and for this reason we have not investigated
this in this chapter. There is yet another method for encoding stabilizer codes
based on the teleportation due to Knill. We expect that gauge qubits can 
be exploited even in this method to reduce its complexity. It would be
interesting to investigate fault tolerant 
encoding schemes for subsystem codes and how gauge qubits can be 
used to improve fault tolerant thresholds. Finally, we mention that it is
still open how to leverage the subsystem coding in the one way quantum
computer model. 

\section{Appendix}
\paragraph{The logical states of a stabilizer code.}
We assume that our basis input states are of the 
form $\ket{0}^{\otimes^{n-k}}\ket{\alpha_1\ldots\alpha_k}$, where $\alpha_i\in \{ 0, 1 \}$. 
Clearly, we have freedom in the choice of the states into which each of these
states are encoded to. Additionally, we have freedom in the choice of
the encoded operators though they are not entirely unrelated. 
Perhaps, this is best illustrated through an example.
Let us consider Shor's $[[9,1,3]]_2$ code. 
A choice of the logical states for this code is
\begin{eqnarray*}
\ket{\ol{0}}&=&(\ket{000}+\ket{111})(\ket{000}+\ket{111})(\ket{000}+\ket{111}),\\
\ket{\ol{1}}&=&(\ket{000}-\ket{111})(\ket{000}-\ket{111})(\ket{000}-\ket{111}).
\end{eqnarray*}
For this choice of the encoded states the logical $Z$ operator is 
$X^{\otimes^9}$ and the logical $X$ operator is $Z^{\otimes^9}$. 
On the other hand, let us see what happens if we choose the logical states as follows
\begin{eqnarray*}
\ket{\ol{0}}&=&\ket{000000000}+\ket{000111111}+\ket{111000111}+\ket{111111000},\\
\ket{\ol{1}}&=&\ket{111111111}+\ket{111000000}+\ket{000111000}+\ket{000000111}.
\end{eqnarray*}
In this case the encoded $X$ operator is $X^{\otimes^9}$ and encoded $Z$ operator is
$Z^{\otimes^9}$; they are flipped with respect to the previous choice!

So it becomes apparent that the assignment of the encoded operators as logical 
$Z$ or $X$ is flexible and it seems to depend on the choice of the logical states.
But are we free to choose any basis of the codespace as the encoded 
logical states. We can show that this cannot be. 
For instance let us choose the logical zero state to be a superposition of the 
previous two assignments. Then we have 
\begin{eqnarray*}
\ket{\ol{0}}&=&(\ket{000}+\ket{111})(\ket{000}+\ket{111})(\ket{000}+\ket{111})\\
&+&\ket{000000000}+\ket{000111111}+\ket{111000111}\\&+&\ket{111111000}.
\end{eqnarray*}

The possibilities for the logical $Z$ operator\footnote{Including scalar multiples of $i$
will not change our conclusions.} are $\pm X^{\otimes^9}$, $\pm Z^{\otimes^9}$,
$\pm X^{\otimes^9} Z^{\otimes^9}$. But for none of these operators we have 
$\ol{Z}\ket{\ol{0}}=\ket{\ol{0}}$. As these are the only possible encoded operators (modulo the stabilizer which acts trivially in any case), this is not a valid choice for $\ket{\ol{0}}$.
This raises the question what are all the possible valid choices for the logical states. 
Let us look at yet another choice of logical states. 
\begin{eqnarray*}
\ket{\ol{0}}&=&(\ket{000}-\ket{111})(\ket{000}-\ket{111})(\ket{000}-\ket{111}),\\
\ket{\ol{1}}&=&(\ket{000}+\ket{111})(\ket{000}+\ket{111})(\ket{000}+\ket{111}).
\end{eqnarray*}
In this case, the encoded $Z$ and $X$ operators 
are $-X^{\otimes^9}$ and $Z^{\otimes^9}$ respectively. 
This gives us a clue as to the possible logical all zero states
for a given stabilizer code. The all zero logical state is the state in the 
code space that is fixed by the stabilizer and the logical $Z$ operators. 
Assuming that $S$ is the stabilizer and $C_{\mc{P}_n}(S)$, its centralizer, we can 
can pick any $k$ independent commuting generators in $C_{\mc{P}_n}(S)\setminus S Z({\mc{P}_n})$ as $Z$ operators. 
Hence, we have the following lemma.

\begin{lemma}\label{lm:logicalZero}
Let $S$ be the stabilizer of an $[[n,k,d]]_2$ stabilizer code.
If $L \le C_{\mc{P}_n}(S)$ is any subgroup generated by 
$n$ commuting generators such that $L\cap Z({\mc{P}_n})=I$ and $S\le L$, then 
the state stabilized by $L$ is a valid logical all zero state for the stabilizer 
code defined by $S$.
\end{lemma}

The implicit choice of $\ket{\ol{0}}$ made in Lemma~\ref{lm:stabStdForm} 
(by picking the encoded $Z$ operators, at least the representatives) is convenient in
the sense it allows us to speak of a canonical $\ket{\ol{0}}$
without ambiguity. This $\ket{\ol{0}}$ can be conveniently identified
with the state $P\ket{0}^{\otimes^n}$, where it will be recalled that $P$
is the projector for the stabilizer code given as
\begin{eqnarray}
P &=& \frac{1}{|S|}\sum_{M\in S}  M.
\end{eqnarray}

%% file: chAqeccIsit.tex
\chapter{Quantum LDPC Codes for Asymmetric Channels\footnotemark}\label{ch:aqecc}
\footnotetext{\copyright 2008 IEEE. Reprinted from, P. K. Sarvepalli, M. R\"{o}tteler, and A. Klappenecker. ``Asymmetric quantum LDPC codes''. In {\em Proc. 2008 IEEE  Intl. Symposium on Inform. Theory, Toronto, Canada},
Jul 6--11, pp.~, 2008.}
  Recently, quantum error-correcting codes were proposed that
  capitalize on the fact that many physical error models lead to a
  significant asymmetry between the probabilities for bit flip and
  phase flip errors. An example for a channel which exhibits such
  asymmetry is the combined amplitude damping and dephasing channel,
  where the probabilities of bit flips and phase flips can be related
  to relaxation and dephasing time, respectively. We give systematic
  constructions of asymmetric quantum stabilizer codes that exploit
  this asymmetry. Our approach is based on a CSS construction that
  combines BCH and finite geometry LDPC codes.


In many quantum mechanical systems the mechanisms for the occurrence
of bit flip and phase flip errors are quite different. In a recent
paper Ioffe and M\'{e}zard \cite{ioffe07} postulated that quantum
error-correction should take into account this asymmetry. The
main argument given in \cite{ioffe07} is that most of the known
quantum computing devices have relaxation times ($T_1$) that are
around $1-2$ orders of magnitude larger than the corresponding
dephasing times $(T_2)$. In general, relaxation leads to both bit flip
and phase flip errors, whereas dephasing only leads to phase flip
errors.  This large asymmetry between $T_1$ and $T_2$ suggests that
bit flip errors occur less frequently than phase flip errors and a
well designed quantum code would exploit this asymmetry of errors to
provide better performance.  In fact, this observation and its
consequences for quantum error correction, especially quantum fault
tolerance, have prompted investigations from various other researchers
\cite{aliferis07,evans07,stephens07}.

Our goal will be as in \cite{ioffe07} to construct asymmetric quantum
codes for quantum memories and at present we do not consider the issue of
fault tolerance. We first quantitatively justify how noise processes,
characterized in terms of $T_1$ and $T_2$, lead to an asymmetry in the
bit flip and phase flip errors. As a concrete illustration of this we
consider the amplitude damping and dephasing channel. For this
channel we can compute the probabilities of bit flip and phase flips
in closed form. In particular, by giving explicit expressions for the
ratio of these probabilities in terms of the ratio $T_1/T_2$, we show
how the channel asymmetry arises.

After providing the necessary background,
we give two systematic constructions of asymmetric quantum codes
based on BCH and LDPC codes, as an alternative to the randomized 
construction of \cite{ioffe07}. 

\section{Background}
\noindent
Recall that a quantum channel that maps a state $\rho$ to 
\begin{eqnarray}
(1-p_x-p_y-p_z)\rho+ p_x\X\rho\X+p_y\Y\rho\Y + p_zZ\rho Z,\label{pauli-channel}
\end{eqnarray}
with 
$  \one= \left[\begin{smallmatrix} 1&0\\0&1\end{smallmatrix}\right]$,
$\X= \left[\begin{smallmatrix} 0&1 \\1&0\end{smallmatrix}\right]$,
$\Y= \left[\begin{smallmatrix} 0&-i \\i&0\end{smallmatrix}\right]$,
$Z= \left[\begin{smallmatrix} 1&0 \\0&-1\end{smallmatrix}\right]$
is called a \textit{Pauli channel}. For a Pauli channel, one can
respectively determine the probabilities $p_x, p_y, p_z$ that an input
qubit in state $\rho$ is subjected to a Pauli $X$, $Y$, or $Z$ error.

A combined \textit{amplitude damping and dephasing channel}
$\mathcal{E}$ with relaxation time $T_1$ and dephasing time $T_2$ that
acts on a qubit with density matrix $\rho=(\rho_{ij})_{i,j\in
\{0,1\}}$ for a time $t$ yields the density matrix
$$ \mathcal{E}(\rho) = 
\left[
\begin{array}{cc}
1-\rho_{11}e^{-t/T_1} & \rho_{01} e^{-t/T_2} \\
\rho_{10} e^{-t/T_2} 
& \rho_{11}e^{-t/T_1} 
\end{array}
\right]. 
$$ This channel is interesting as it models common decoherence
processes fairly well. We would like to determine the probability
$p_x$, $p_y$, and $p_z$ such that an $X$, $Y$, or $Z$ error occurs in
a combined amplitude damping and dephasing channel. However, it turns
out that this question is not well-posed, since $\mathcal{E}$ is not a
Pauli channel, that is, it cannot be written in the form
(\ref{pauli-channel}). However, we can obtain a Pauli channel
$\mathcal{E}_T$ by a technique called twirling
\cite{ESR+:2007,dankert06}. In our case, the twirling consists of
conjugating the channel $\mathcal{E}$ by Pauli matrices and averaging
over the results. The resulting channel $\mathcal{E}_T$ is called the
Pauli-twirl of $\mathcal{E}$ and is explicitly given by
$$ \mathcal{E}_T(\rho) = \frac{1}{4} \sum_{A \in \{ \one, \X,\Y,Z\}}
A^\dagger \mathcal{E}(A\rho A^\dagger )A.$$ 

\begin{theorem} Given a combined amplitude damping and dephasing 
channel $\mathcal{E}$ as above, the associated Pauli-twirled channel
is of the form
$$\mathcal{E}_T(\rho) = (1-p_x-p_y-p_z)\rho+ p_x\X\rho\X+p_y\Y\rho\Y +
p_z Z\rho Z,$$
where $p_x=p_y=(1-e^{-t/T_1})/4$ and $p_z=1/2-p_x-\frac{1}{2}e^{-t/T_2}$. 
In particular,
$$ \frac{p_z}{p_x} =1+2 \frac{1-e^{t/T_1(1-T_1/T_2)}}{e^{t/T_1}-1}.
$$ 
If $t\ll T_1$, then we can approximate this ratio as $2T_1/T_2-1$.
\end{theorem}
\begin{proof}
  The Kraus operator decomposition \cite{nielsen00} of $\mathcal{E}$ is
\begin{eqnarray}
\mathcal{E}(\rho) = \sum_{k=0}^2A_k\rho A_k^\dagger,\label{eq:krausDecomp}
\end{eqnarray}
where  $A_0= \left[\begin{smallmatrix} 1&0
\\0&\sqrt{1-\lambda-\gamma}\end{smallmatrix}\right]; A_1 = \left[\begin{smallmatrix} 0&
0\\0&\sqrt{\lambda}\end{smallmatrix}\right];A_2 = \left[\begin{smallmatrix} 0&
\sqrt{\gamma}\\0& 0\end{smallmatrix}\right],$
and $\sqrt{1-\gamma-\lambda}= e^{-t/T_2}$,  $1-\gamma=e^{-t/T_1}$. We can rewrite the Kraus
operators $A_i$  as 
\begin{eqnarray*}
A_0=\frac{1+\sqrt{1-\lambda-\gamma}}{2}\one +\frac{1-\sqrt{1-\lambda-\gamma}}{2} Z,\\
A_1=\frac{\sqrt{\lambda}}{2}\one -\frac{\sqrt{\lambda}}{2} Z, \quad A_2=\frac{\sqrt{\gamma}}{2}\X -\frac{\sqrt{\gamma}}{2i}\Y.
\end{eqnarray*}
Rewriting $\E(\rho)$ in terms of Pauli matrices 
yields 
\begin{eqnarray}
\E(\rho) &=& \frac{2-\gamma+2\sqrt{1-\lambda-\gamma}}{4} \rho + \frac{\gamma}{4}\X\rho\X +\frac{\gamma}{4}\Y\rho\Y \nonumber \\ 
&+& \frac{2-\gamma-2\sqrt{1-\lambda-\gamma}}{4} Z\rho Z \nonumber \\
&-& \frac{\gamma}{4}\one\rho Z -\frac{\gamma}{4} Z\rho\one +\frac{\gamma}{4i} \X\rho\Y
-\frac{\gamma}{4i}\Y\rho\X.
\end{eqnarray}
It follows that the Pauli-twirl channel $\mathcal{E}_T$ is of the
claimed form, see~\cite[Lemma 2]{dankert06}.
Computing the ratio $p_z/p_x$ we get
\begin{eqnarray*}
	\frac{p_z}{p_x} &=&\frac{2-\gamma-2\sqrt{1-\lambda-\gamma}}{\gamma} =\frac{1+e^{-t/T_1}-2e^{-t/T_2}}{1-e^{-t/T_1}},\\
	&=&1+2\frac{e^{-t/T_1}-e^{-t/T_2}}{1-e^{-t/T_1}} = 1+2 \frac{1-e^{t/T_1-t/T_2}}{e^{t/T_1}-1} \\
	&=&1+2 \frac{1-e^{t/T_1(1-T_1/T_2)}}{e^{t/T_1}-1}.
\end{eqnarray*}
If $t\ll T_1$, then we can approximate the ratio as $2T_1/T_2-1$, as claimed. 
\end{proof}

Thus, an asymmetry in the $T_1$ and $T_2$ times does translate to an
asymmetry in the occurrence of bit flip and phase flip errors.  Note
that $p_x=p_y$ indicating that the $\Y$ errors are as unlikely as the
$\X$ errors.  We shall refer to the ratio $p_z/p_x$ as the channel
asymmetry and denote this parameter by $A$.

Asymmetric quantum codes use the fact that the phase flip errors are much
more likely than the bit flip errors or the combined bit-phase flip
errors.  Therefore the code has different error correcting capability
for handling different type of errors. We require the code to correct
many phase flip errors but it is not required to handle the same
number of bit flip errors. If we assume a CSS code
\cite{calderbank98}, then we can meaningfully speak of $X$-distance
and $Z$-distance. A CSS stabilizer code that can detect all $X$ errors
up to weight $d_x-1$ is said to have an $X$-distance of $d_x$.
Similarly if it can detect all $Z$ errors upto weight $d_z-1$, then it
is said to have a $Z$-distance of $d_z$. We shall denote such a code
by $[[n,k,d_x/d_z]]_q$ to indicate it is an asymmetric code, see also
\cite{steane96} who was the first to use a notation that allowed to
distinguish between $X$- and $Z$-distances.  We could also view this
code as an $[[n,k,\min\{ d_x,d_z\}]]_q$ stabilizer code.  Further
extension of these metrics to an additive non-CSS code is an
interesting problem, but we will not go into the details here.

Recall that in the CSS construction a pair of codes are used, one for
correcting the bit flip errors and the other for correcting the phase
flip errors. Our choice of these codes will be such that the code for
correcting the phase flip errors has a larger distance than the code
for correcting the bit flip errors. We restate the CSS construction in
a form convenient for asymmetric stabilizer codes.

\begin{lemma}[CSS Construction \cite{calderbank98}]\label{lm:css}
  Let $C_x, C_z$ be linear codes over $\F_q^n$ with the parameters
  $[n,k_x]_q$, and $[n,k_z]_q$ respectively. Let $C_x^\perp\subseteq
  C_z$.  Then there exists an $[[n,k_x+k_z-n,d_x/d_z]]_q$ asymmetric
  quantum code, where $d_x=\wt(C_x\setminus C_z^\perp)$ and
  $d_z=\wt(C_z\setminus C_x^\perp)$.
\end{lemma}
If in the above construction $d_x=\wt(C_x)$ and $d_z=\wt(C_z)$, then
we say that the code is pure. 

In the theorem above and elsewhere in this paper  $\F_q$ denotes a finite field with $q$ elements.
We also denote a $q$-ary narrow-sense primitive BCH code of length $n=q^m-1$
and design distance $\delta$ as $\bch(\delta)$. 

\section{Asymmetric Quantum Codes from LDPC Codes}

In \cite{ioffe07}, Ioffe and M\'{e}zard used a combination of BCH and
LDPC codes to construct asymmetric codes. The intuition being that the
stronger LDPC code should be used for correcting the phase flip errors
and the BCH code can be used for the infrequent bit flips. This
essentially reduces to finding a good LDPC code such that the dual of
the LDPC code is contained in the BCH code. They solve this problem by
randomly choosing codewords in the BCH code which are of low weight
(so that they can be used for the parity check matrix of the LDPC
code).  However, this method leaves open how good the resulting LDPC
code is.  For instance, the degree profiles of the resulting code are
not regular and there is little control over the final degree profiles
of the code.  Furthermore, it is not apparent what ensemble or degree
profiles one will use to analyze the code.

We propose an alternate scheme that uses LDPC codes to construct asymmetric
stabilizer codes. We propose two families of quantum codes based on LDPC codes.  In the
first case we use LDPC codes for both the $X$ and $Z$ channel while in the
second construction we will use a combination of BCH and LDPC codes. But 
first, we will need the following facts about generalized Reed-Muller 
codes, (\cite{kasami68}) and finite geometry LDPC codes, (\cite{kou01,tang05}).

\subsection{Finite Geometry LDPC Codes } 
Let us denote by $\eg(m,p^s)$ the Euclidean finite geometry over
$\F_{p^s}$ consisting of $p^{ms}$ points. For our purposes it suffices
to use the fact that this geometry is equivalent to the vector space
$\F_{p^s}^m$.  A $\mu$-dimensional subspace of $\F_{p^s}^m$ or its
coset is called a \textsl{$\mu$-flat}. Assume that $0\leq \mu_1<
\mu_2\leq m$.  Then we denote by $N_{\eg}(\mu_2,\mu_1,s,p)$ the number
of $\mu_1$-flats in a $\mu_2$-flat and by $A_{\eg}(m,\mu_2,\mu_1,
s,p)$, the number of $\mu_2$-flats that contain a given $\mu_1$-flat.
These are given by (see \cite{tang05})
\begin{eqnarray}
N_{\eg}(\mu_2,\mu_1,s,p)&=& q^{(\mu_2-\mu_1)}\prod_{i=1}^{\mu_1}\frac{q^{\mu_2-i+1}-1}{q^{\mu_1-i+1}-1},\label{eq:nEG}\\
A_{\eg}(m,\mu_2,\mu_1,s,p)&=& \prod_{i=\mu_1+1}^{\mu_2}\frac{q^{m-i+1}-1}{q^{\mu_2-i+1}-1}, \label{eq:aEG}
\end{eqnarray}
where $q=p^s$.  Index all the $\mu_1$-flats from $i=1$ to
$n=N_{\eg}(m,\mu_1,s,p)$ as $F_i$.  Let $F$ be a $\mu_2$-flat in
$\eg(m,p^s)$.  Then we can associate an incidence vector to $F$ with
respect to the $\mu_1$ flats as follows.
$$
\mathbf{i}_F = \left\{i_j\mid \begin{array}{cl}i_j =1 &\mbox{ if $F_j$ is contained in } F \\
i_j=0& \text{otherwise.}\end{array}
\right\}.
$$
Index the $\mu_2$-flats from $j=1$ to $J=N_{\eg}(m,\mu_2,s,p)$.
Construct the $J\times n $ matrix $H_{\eg}^{(1)}(m,\mu_2,\mu_1,s,p)$
whose rows are the incidence vectors of all the $\mu_2$-flats with
respect to the $\mu_1$-flats. This matrix is also referred to as the
incidence matrix.  Then the type-I Euclidean geometry code from
$\mu_2$-flats and $\mu_1$-flats is defined to be the null space,
i.\,e., Euclidean dual code) of the $\F_p$-linear span of
$H_{\eg}^{(1)}(m,\mu_2,\mu_1,s,p)$. This is denoted as
$C_{\eg}^{(1)}(m,\mu_2,\mu_1,s,p)$.  Let
$H_{\eg}^{(2)}(m,\mu_2,\mu_1,s,p) =
H_{\eg}^{(1)}(m,\mu_2,\mu_1,s,p)^t.$ 
The type-II Euclidean
geometry code $C_{\eg}^{(2)}(m,\mu_2,\mu_1,s,p)$ is defined 
as the
null space of $H_{\eg}^{(2)}(m,\mu_2,\mu_1,s,p)$. Let us now consider
the $\mu_2$-flats and $\mu_1$-flats that do not contain the origin of
$\eg(m,p^s)$. Now form the incidence matrix of the $\mu_2$-flats with
respect to the $\mu_1$-flats not containing the origin.  The null
space of this incidence matrix gives us a quasi-cyclic code in
general, which we denote by $C_{\eg,c}^{(1)}(m,\mu_2,\mu_1,s,p)$, see
\cite{tang05}.

\subsection{Generalized Reed-Muller Codes} 
Let $\alpha$ be a primitive element in $\F_{q^{m}}$.  The cyclic
generalized Reed-Muller code of length $q^m-1$ and order $\nu$ is
defined as the cyclic code with the generator polynomial whose roots
$\alpha^j$ satisfy $0<j\leq m(q-1)-\nu-1$. The generalized Reed-Muller
code is the singly extended code of length $q^m$.  It is denoted as
$\Rm_q(\nu,m)$. The dual of a GRM code is also a GRM code
\cite{assmus98,blahut03,kasami68}.  It is known that
\begin{eqnarray}
\Rm_q(\nu,m)^\perp = \Rm_q(\nu^\perp,m),
\end{eqnarray}
where $\nu^\perp = m(q-1)-1-\nu$.

Let $C$ be a linear code over $\F_{q^s}^n$. Then we define
$C|_{\F_q}$, the \textsl{subfield subcode} of $C$ over $\F_q^n$ as the
codewords of $C$ which are entirely in $\F_q^n$, (see
\cite[pages~116-120]{huffman03}). Formally this can be expressed as
\begin{eqnarray}
C|_{\F_q} = \{c\in C\mid c\in \F_q^n \}.
 \end{eqnarray}
Let $C\subseteq \F_{q^l}^n$. The the \textsl{trace code} of $C$ over $\F_q$ is defined  as
\begin{eqnarray}
\tr_{q^l/q}(C) = \{\tr_{q^l/q}(c) \mid c\in C\}.
\end{eqnarray}
There are interesting relations between the trace code and the
subfield subcode. One of which is the following result which we will
need later.
\begin{lemma}\label{lm:inclusion}
Let $C\subseteq \F_{q^l}^n$. Then $C|_{\F_q}$, the subfield subcode of $C$
is contained in $\tr_{q^l/q}(C)$, the trace code of $C$. In other words
$$ C|_{\F_q}\subseteq \tr_{q^l/q}(C).$$
\end{lemma}
\begin{proof}
  Let $c\in C|_{\F_q} \subseteq \F_q^n$ and $\alpha \in \F_{q^l}$.
  Then $\tr_{q^l/q}(\alpha c)= c\tr_{q^l/q}(\alpha)$ as $c\in \F_q^n$.
  Since trace is a surjective form, there exists some $\alpha\in
  \F_{q^l}$, such that $\tr_{q^l/q}(\alpha)=1$. This implies that
  $c\in \tr_{q^l/q}(C)$. Since $c$ is an arbitrary element in
  $C|_{\F_q}$ it follows that $ C|_{\F_q}\subseteq \tr_{q^l/q}(C)$.
\end{proof}

Let $q=p^s$, then the Euclidean geometry code of order $r$ over
$\eg(m,p^s)$ is defined as the dual of the subfield subcode of
$\Rm_{q}((q-1)(m-r-1),m)$, \cite[page~448]{blahut03}. The type-I LDPC
code $C_{\eg}^{(1)}(m,\mu,0,s,p)$ code is an Euclidean geometry code
of order $\mu-1$ over $\eg(m,p^s)$, see \cite{tang05}.
Hence its dual is the subfield subcode of $\Rm_q((q-1)(m-\mu),m)$
code.  In other words,
\begin{eqnarray}
C_{\eg}^{(1)}(m,\mu,0,s,p)^\perp = 
\Rm_q((q-1)(m-\mu),m)|_{\F_p}.\label{eq:fgSubfiledCode}
\end{eqnarray}
Further, Delsarte's theorem \cite{delsarte75} tells us that 
\begin{eqnarray*}
C_{\eg}^{(1)}(m,\mu,0,s,p)&=&\Rm_q((q-1)(m-\mu),m)|_{\F_p}^\perp ,\\
&=&\tr_{q/p}\left( \Rm_q((q-1)(m-\mu),m)^\perp\right)\\
&=&  \tr_{q/p}(\Rm_q(\mu(q-1)-1,m)).
\end{eqnarray*}
Hence,  $C_{\eg}^{(1)}(m,\mu,0,s,p)$ code can also be  related to $\Rm_q(\mu(q-1)-1,m)$ as 
\begin{eqnarray}
  C_{\eg}^{(1)}(m,\mu,0,s,p)=\tr_{q/p}(\Rm_q(\mu(q-1)-1),m).\label{eq:fgTraceCode}
\end{eqnarray}

\subsection{New Families of Asymmetric Quantum Codes}
With the previous preparation we are now ready to construct asymmetric quantum
codes from finite geometry LDPC codes. 

\begin{theorem}[Asymmetric EG LDPC Codes]\label{th:asymQldpc}
  Let $p$ be a prime, with $q=p^s$ and $s\geq 1,m\geq 2$. Let $1<
  \mu_z <m$ and $m-\mu_z+1 \leq \mu_x<m$.  Then there exists an
$$
[[p^{ms},k_x+k_z-p^{ms} , d_x/d_z]]_p
$$
asymmetric EG LDPC code, where 
$$k_x=\dim C_{\eg}^{(1)}(m,\mu_x,0,s,p); \quad  k_z=\dim C_{\eg}^{(1)}(m,\mu_z,0,s,p).$$ For the distances 
$d_x\geq A_{\eg}(m,\mu_x,\mu_x-1,s,p)+1$ and $d_z\geq
A_{\eg}(m,\mu_z,\mu_z-1,s,p)+1$ hold.
\end{theorem}
\begin{proof}
  Let $C_z=C_{\eg}^{(1)}(m,\mu_z,0,s,p)$.  Then from
  equation~(\ref{eq:fgTraceCode}) we have
\begin{eqnarray*}
C_z &=& \tr_{q/p}(\Rm_q(\mu_z(q-1)-1,m).
\end{eqnarray*}
By Lemma~\ref{lm:inclusion} we know that 
\begin{eqnarray*}
C_z &\supseteq&  \Rm_q(\mu_z(q-1)-1,m)|_{\F_p},\\
C_z&\supseteq & \Rm_q((q-1)(m-(m-\mu_z+1)),m)|_{\F_p},
\end{eqnarray*}
where the last inclusion follows from the nesting property of the
generalized Reed-Muller codes. For any order $\mu_x$ such that
$m-\mu_z+1\leq \mu_x <m$, let $C_x=C_{\eg}^{(1)}(m,\mu_x,0,s,p)$. Then
$C_x$ is an LDPC code whose dual $C_x^\perp =
\Rm_q((q-1)(m-\mu_x),m)|_{\F_p}$ is contained in $C_z$. Thus we can
use Lemma~\ref{lm:css} to form an asymmetric code with the parameters
$$ 
[[p^{ms},k_x+k_z-p^{ms} , d_x/d_z]]_p
$$
The distance of $C_z$ and $C_x$ are at lower bounded as 
$d_x \geq A_{\eg}(m,\mu_x,\mu_x-1,s,p)+1$ 
and $d_z \geq A_{\eg}(m,\mu_z,\mu_z-1,s,p)+1$  (see \cite{tang05}).
\end{proof}

In the construction just proposed, we should choose $C_z$ to be a
stronger code compared to $C_x$. We have given the construction over 
a nonbinary alphabet even though the case $p=2$ might be of particular 
interest.

We briefly turn our attention back to the depolarizing 
channel. The LDPC codes designed for the asymmetric channels 
will not in general perform well on the depolarizing channel. 
In fact constructing good quantum LDPC codes for the depolarizing
channel remains a difficult problem and a satisfactory solution 
is yet to be advanced. We contribute to the ongoing discussion
in this topic by drawing upon the finite geometry LDPC codes
as we did for the asymmetric codes. The codes presented in 
Theorem~\ref{th:asymQldpc} can under certain conditions 
lead to LDPC codes that are suitable for use on the depolarizing
channel. 

\begin{corollary}[EG LDPC Codes for Depolaring Channel]\label{co:symQldpc}
 Let $p$ be a prime, with $q=p^s$ and $s\geq 1,m\geq 2$. Let
   $\ceil{ (m+1)/2} \leq \mu <m$.  Then there exists an
$
[[p^{ms},2k-p^{ms}, d]]_p
$
symmetric EG LDPC code, where 
$k=\dim C_{\eg}^{(1)}(m,\mu,0,s,p)$. For the distance
$d \geq A_{\eg}(m,\mu,\mu-1,s,p)+1$  holds.
\end{corollary}

Our next construction makes use of the cyclic finite geometry
codes.  Our goal will be to find a small BCH code whose dual is
contained in a cyclic Euclidean geometry LDPC code. For solving this
problem we need to know the cyclic structure of
$C_{\eg,c}^{(1)}(m,\mu,0,s,p)$. Let $\alpha$ be a primitive element in
$\F_{p^{ms}}$. Then the roots of the generator polynomial of
$C_{\eg,c}^{(1)}(m,\mu,0,s,p)$ are given by
\cite[Theorem~6]{kasami71}, see also \cite{kasami68b,lin04}. Now, 
$$
Z= \{\alpha^h\mid 0< \max_{0\leq l<s} W_{p^s}(h p^l) \leq
(p^s-1)(m-\mu) \},
$$
where $W_{q}(h)$ is the $q$-ary weight of $h = h_0+h_1q+\cdots +
h_kq^{k-1}$, i.\,e., $W_q(h)=\sum h_i$.  The finite geometry code
$C_{\eg,c}^{(1)}(m,\mu,0,s,p)$ is actually an $(\mu-1,p^s)$ Euclidean
geometry code.  The roots of the generator polynomial of the dual code
are given by
$$
Z^{\perp}= \{\alpha^h \mid \min_{0\leq l<s} W_{p^s}(hp^l) < \mu(p^s-1) \}.
$$
In fact, the dual code is the even-like subcode of a primitive
polynomial code of length $p^{ms}-1$ over $\F_p$ and order $m-\mu$,
whose generator polynomial, by \cite[Theorem~6]{kasami68b}, has the
roots
$$
Z_p= \{\alpha^h \mid 0< \min_{0\leq l<s} W_{p^s}(hp^l) < \mu(p^s-1) \}.
$$
Thus $Z^\perp = Z_p\cup \{ 0\}$.  Now by \cite[Theorem~11]{kasami68b},
$Z_p$ and therefore $Z^\perp$ contain the sequence of consecutive
roots, $\alpha,\alpha^2,\ldots, \alpha^{\delta_0-1}$, where
$\delta_0=(R+1)p^{Qs}-1$ and $m(p^s-1)-(m-\mu)(p^s-1) = Q(p^s-1)+R$.
Simplifying, we see that $R=0$ and $Q=\mu$ giving $\delta_0 = p^{\mu s}-1$.
It follows that 
\begin{eqnarray*}
C_{\eg,c}^{(1)}(m,\mu,0,s,p)^\perp &=& \Rm_q(m,(q-1)(m-\mu))|_{\F_p} \\
&\subseteq& \bch(\delta_0).
\end{eqnarray*}

Thus we have solved the problem of construction of the asymmetric
stabilizer codes in a dual fashion to that of \cite{ioffe07}. Instead of
finding an LDPC code whose parity check matrix is contained in a given
BCH code, we have found a BCH code whose parity check matrix is
contained in a given finite geometry LDPC code. This gives us the
following result.
\begin{theorem}[Asymmetric BCH-LDPC stabilizer codes]
  Let $C_z=C_{\eg,c}^{(1)}(m,\mu,0,s,p)$ and $\delta\leq \delta_0=p^{\mu s}-1$.  
  Let $n=p^{ms}-1$ and $C_x=\bch(\delta)\subseteq \F_p^n$. Then there exists an
$$[[n,k_x+k_z-n,d_x/d_z]]_p$$
asymmetric stabilizer code where $d_z\geq A_{\eg}(m,\mu,\mu-1,s,p)$,
$d_x\geq \delta$  and $k_x=\dim C_x$, $k_z=\dim C_z$.
\end{theorem}
Perhaps an example will be helpful at this juncture.
\begin{example}
  Let $m=s=p=2$ and $\mu=1$. Then $C_{\eg,c}^{(1)}(2,1,0,2,2)$ is a
  cyclic code whose generator polynomial has roots given by
\begin{eqnarray*}
  Z&=&\{\alpha^h|0<\max_{0\leq l<2 }W_{2^2}(2^lh) \leq (m-\mu)(p^s-1)=3 \}\\
  &=&\{  \alpha^1, \alpha^2, \alpha^3, \alpha^4, \alpha^6, \alpha^8, \alpha^9, \alpha^{12} \}.
\end{eqnarray*}
As there are 4 consecutive roots and $|Z|=8$, it defines a $[15,7,\geq
5]$ code.  The roots of the generator polynomial of the dual code are
given by
\begin{eqnarray*}
Z^\perp&=& \{\alpha^h|0<\min_{0\leq l<2 }W_{2^2}(2^lh) \leq \mu(p^s-1)=(2^2-1) \}\\
&=&\{ \alpha^0, \alpha^1, \alpha^2, \alpha^4, \alpha^5, \alpha^8, \alpha^{10} \}.
\end{eqnarray*}
We see that $Z^\perp$ has two consecutive roots excluding $1$,
therefore the dual code is contained in a narrowsense BCH code with
design distance 3. Note that $p^{\mu s}-1=3$.  Thus we can choose
$C_x=\bch(3)$ and $C_z=C_{\eg,c}^{(1)}(2,1,0,2,2)$ and apply
Lemma~\ref{lm:css} to construct a $[[15,3,3/5]]_2$ asymmetric code.
\end{example}

We can also state the above construction as in \cite{ioffe07}, that is
given a primitive BCH code of design distance $\delta$, find an
LDPC code whose dual is contained in it. It must be pointed out that
in case of asymmetric codes derived from LDPC codes, the asymmetry
factor $d_x/d_z$ is not as indicative of the code performance as in
the case of bounded distance decoders.  For $m=p=2$, we can derive
explicit relations for the parameters of the codes.
\begin{corollary}\label{co:2dCyclic}
  Let $C=C_{\eg,c}^{(1)}(2,1,0,s,2)$ and $\delta =2t+1\leq 2^{s}-1$.
  Then there exists an
$$[[2^{2s}-1,2^{2s}-3^s -s(\delta-1),\delta/2^s+1]]_2 $$ asymmetric stabilizer code. 
\end{corollary}
\begin{proof}
  The parameters of $C$ are $[2^{2s}-1,2^{2s}-3^s,2^{s}+1]_2$, see
  \cite{lin04}. Since $C^\perp $ is contained in a BCH code of length
  $2^{2s}-1$ whose design distance $\delta\leq 2^s-1$, we can compute
  the dimension of the BCH code as $2^{2s}-1 -s(\delta-1)$, see
  \cite[Corollary~8]{macwilliams77}. By Lemma~\ref{lm:css} the quantum
  code has the dimension $2^{2s}-3^s -s(\delta-1)$.
\end{proof}

\begin{example}
  For $m=p=2$ and $s=4$ we can obtain a $[255,175,17]$ LDPC code. We
  can choose any BCH code with design distance $\delta \leq 2^4-1=15$
  to construct an asymmetric code. Table~\ref{tab:bchLdpcAqecc} lists possible codes.
\end{example}

\begin{table}[htb]
\caption{Asymmetric BCH-LDPC stabilizer codes}
\label{tab:bchLdpcAqecc}
\begin{center}
\begin{tabular}{c|c|c|l|c}
$s$&$\delta$  & Code &   Asymmetry &Rate\\
&  & $[[n,k,d_x/d_z]]_2$ &    $d_z/d_x$ &\\\hline
4&  15 & $[[255,119,15/17]]_2$   &   $\approx 1$&0.467\\
4&13  & $[[255,127,13/17]]_2 $&  $\approx 1.25$&0.498\\
4&11  & $[[255,135,11/17]]_2 $&  $\approx 1.5$&0.529\\
4&9  & $[[255,143,9/17]]_2 $&  $\approx 2$ &0.561\\
4&7  & $[[255,151,7/17]]_2 $&  $\approx 2.5$&0.592\\
4&5  & $[[255,159,5/17]]_2 $&  $\approx 3$&0.624\\
4&3  & $[[255,167,3/17]]_2 $&  $\approx 6$&0.655\\
\end{tabular}
\end{center}
\end{table}

\section{Performance Results}
We now study the  performance of the codes constructed
in the previous section. 
\nix{
Due to space constraints the discussion will 
be rather brief, but more details will be supplied in a forthcoming
paper.
} 
We assume that the overall probability of error in the channel is
given by $p$, while the individual probabilities of $X$, $Y$, and $Z$
errors are $p_x=p/(A+2)$, $p_y=p/(A+2)$ and $p_z=pA/(A+2)$
respectively. The exact performance would require us to simulate a
$4$-ary channel and also account for the fact that some errors can be
estimated modulo the stabilizer. However, we do not account for this
and in that sense these results provide an upper bound on the actual
error rates.  
The 4-ary channel can be modeled as two
binary symmetric channels -- one modeling the bit flip channel and the
other the phase flip channel. For exact performance, these two
channels should be dependent, however, a good approximation is to
model the channel as two independent BSCs with cross over
probabilities $p_x+p_y=2p/(A+2)$ and $p_y+p_z=p(A+1)/(A+2)$. In this
case the overall error rate in the quantum channel is the sum of the
error rates in the two BSCs. While this approach is going to slightly
overestimate the error rates, nonetheless it is useful and has been
used before \cite{mackay04}.  Since the $X$-channel uses a BCH code
and decoded using a bounded distance decoder, we can just compute
$P_e^x$ the $X$ error rate, in closed form. 
The error rate in the Z channel, $P_e^z$  is obtained through simulations. The overall
error rate is 
$$P_e=1-(1-P_e^x)(1-P_e^z)=P_e^x+P_e^z-P_e^xP_e^z\approx P_e^x+P_e^z.$$ 

\paragraph{Decoding LDPC Codes.} The LDPC code was decoded using the
an algorithm similar to the hard decision bit flipping algorithm given
in \cite{kou01}. This is an instance of the bit flipping algorithm
originally given by Gallager.  The maximum number of iterations for
decoding is set to 50. A small modification had to be made to
accommodate the special situation of quantum syndrome decoding. By
measuring the generators of the stabilizer group, we obtain a
classical syndrome, which due to the fact that only $\pm 1$
eigenspaces occur in all of the generators, is hard information. We
use the syndrome as shown in Figure \ref{fig:bp} and initialize all
the bit nodes with $0$ at the start of the algorithm. Then the
algorithm proceeds in the usual fashion as in \cite{kou01}. We
implemented this algorithm and ran several simulations which are
described next.

\begin{figure}[htb]
\begin{center}
\epsfig{file=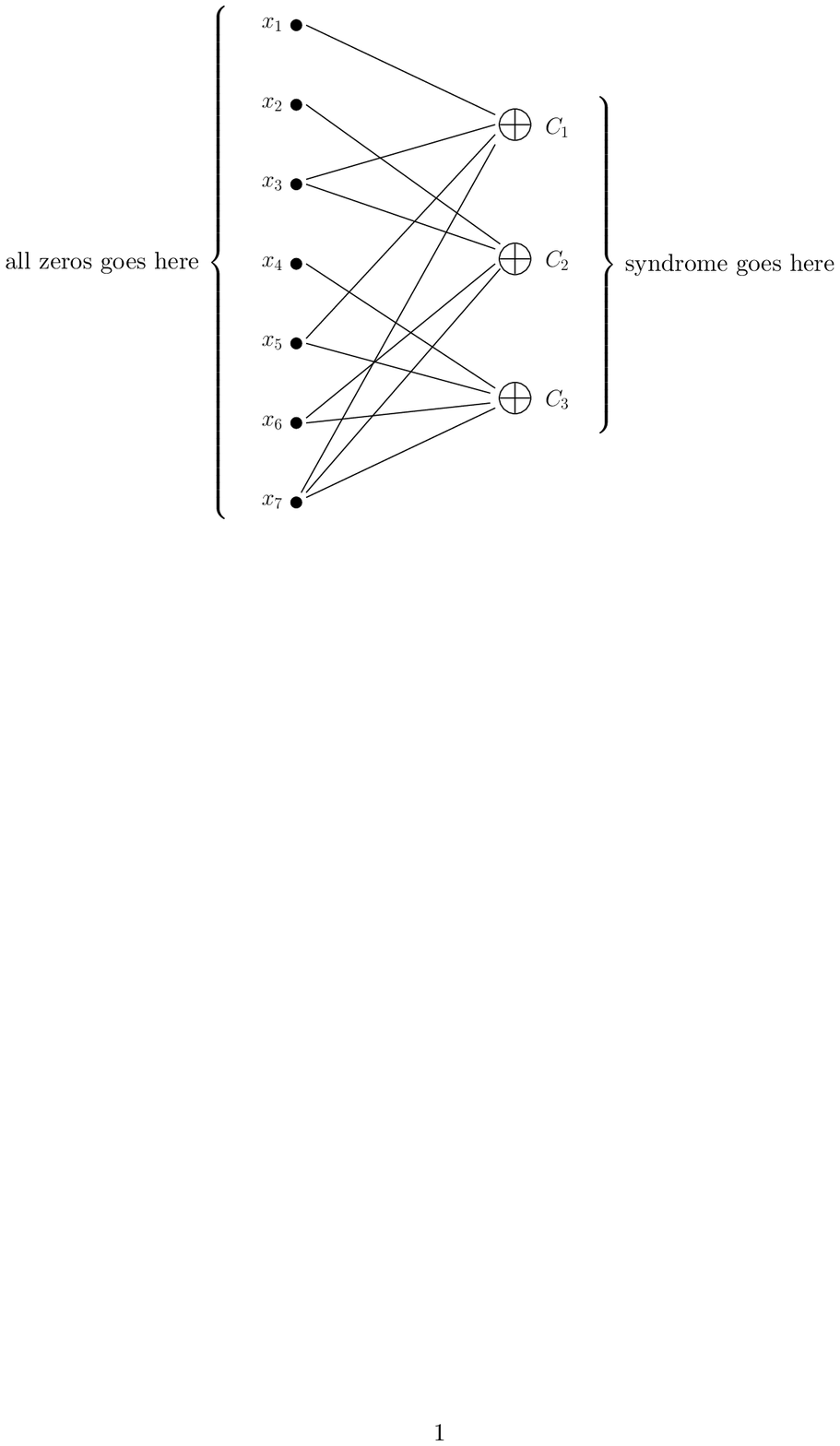,scale=1.0} 
\caption{Modification of the iterative message passing algorithm to
  the quantum case. The initialization step is different from the
  classical case as no soft information from the channel is available
  but rather only hard information about the measured syndrome is
  available. The algorithms begins with initializing all bit nodes to
  $0$ and the check nodes with the syndrome. From then on, any
  classically known method for iterative decoding can be applied. In
  the figure this principle is shown for the example of a classical
  [7,4,3] Hamming code. Application to the quantum case is
  straightforward as the decoding algorithm only works with classical
  information to compute the most likely error.}\label{fig:bp}
\end{center}
\end{figure}

In figure~\ref{fig:fg255d5} we see the performance of
$[[255,159,5/17]]$ as the channel asymmetry is varied from 1 to 100.
We see that as we increase the
asymmetry the code starts to perform better.  As the asymmetry is
increased eventually the performance of the quantum code approaches
the performance of the classical LDPC code.

\begin{figure}
\begin{minipage}[b]{1\linewidth}
\centering
\epsfig{file=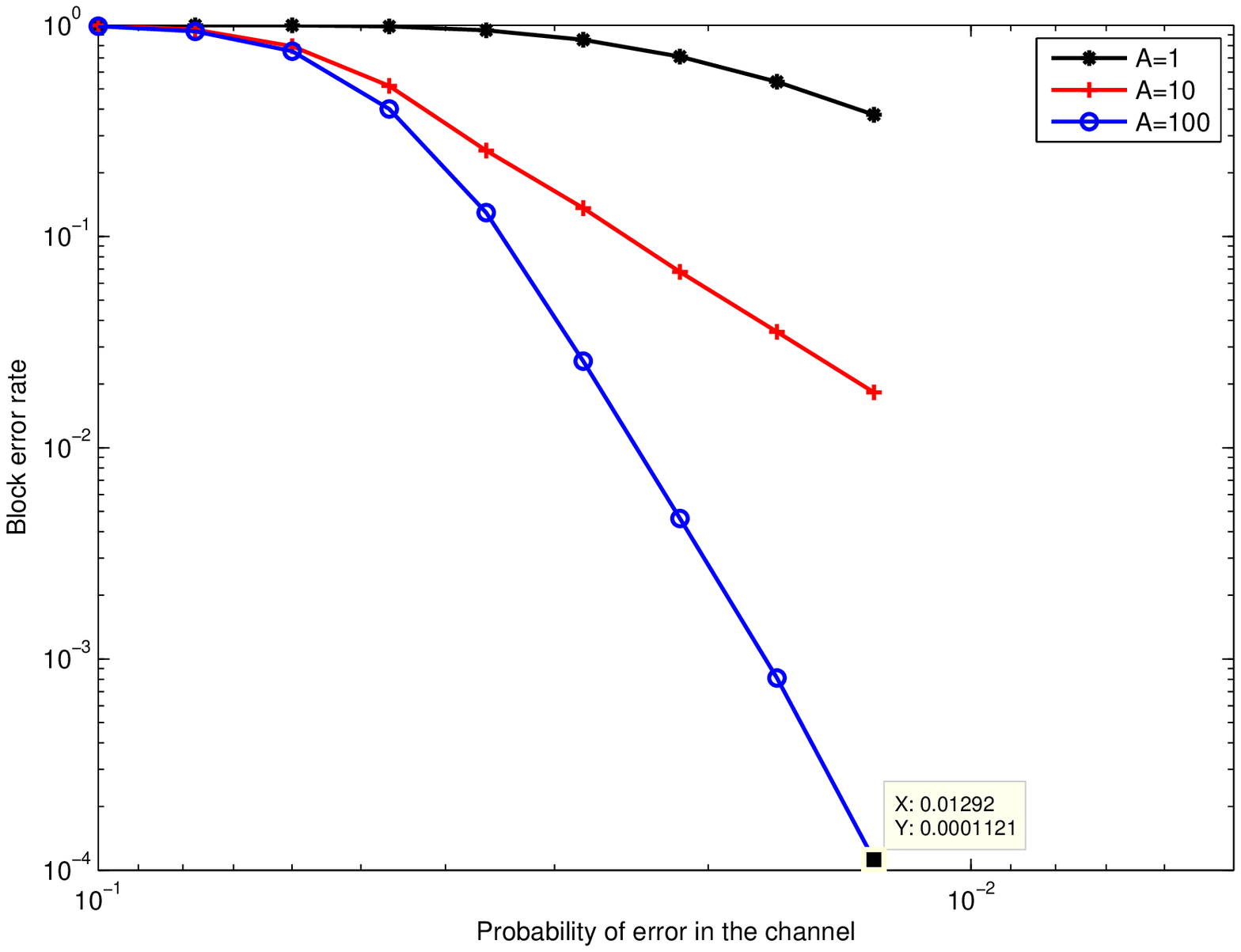,scale=0.6} 
\caption{Performance of a $[[255,159,5/17]]$ code described in the
  text for choices $A=1,10,100$ of the channel asymmetry.}\label{fig:fg255d5}
\end{minipage}
\hspace{0.5cm} 
\begin{minipage}[b]{1\linewidth}
\centering
\epsfig{file=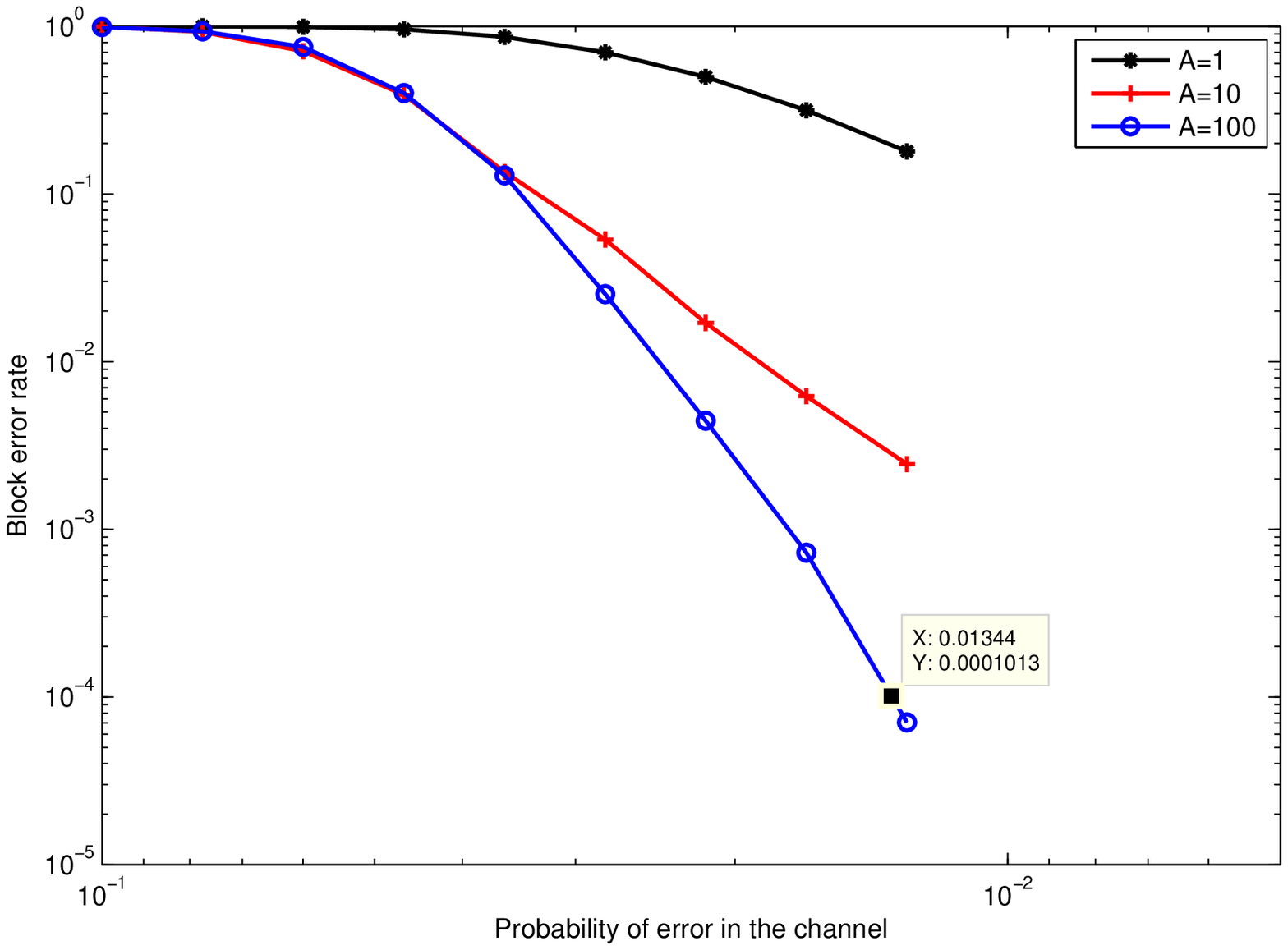,scale=0.6}
\caption{Performance of a $[[255,151,7/17]]$ code described in the
  text for choices $A=1,10,100$ of the channel asymmetry.}
\label{fig:fg255d7}
\end{minipage}
\end{figure}

Tolerating a little rate loss improves the performance as can be seen
from figure~\ref{fig:fg255d7}.  If we increase the distance of the BCH
code the code becomes more tolerant to variations in channel asymmetry
as can be seen by the performance of $[[255,143,9/17]]$ in
figure~\ref{fig:fg255d9}.  This plot also illustrates an important
point. Our channel model assumes that as we vary the channel asymmetry
we keep the total probability of error in the channel fixed.  This
implies that while the probability of $X$ errors goes down, the
probability of $Z$ errors tends to $p$, the total probability of
error. Hence, the reduction in error rate in the $X$ channel must more
than compensate for the increase in $Z$ error rate.  If on the other
hand, we had fixed the probability of error in the $Z$ channel and
varied the channel asymmetry then we would observe a monotonic
improvement in the error rate because on one hand the $Z$ error rate
does not change but the $X$ error rate does.
We note that with larger lengths we can get an even steep drop in the error rate as is
apparent from the performance of 
$[[1023,731,11/33]]$ code shown in 
Figure~\ref{fig:n1023A100}.

\begin{figure}    
\begin{minipage}[b]{1\linewidth}  
\centering
\epsfig{file=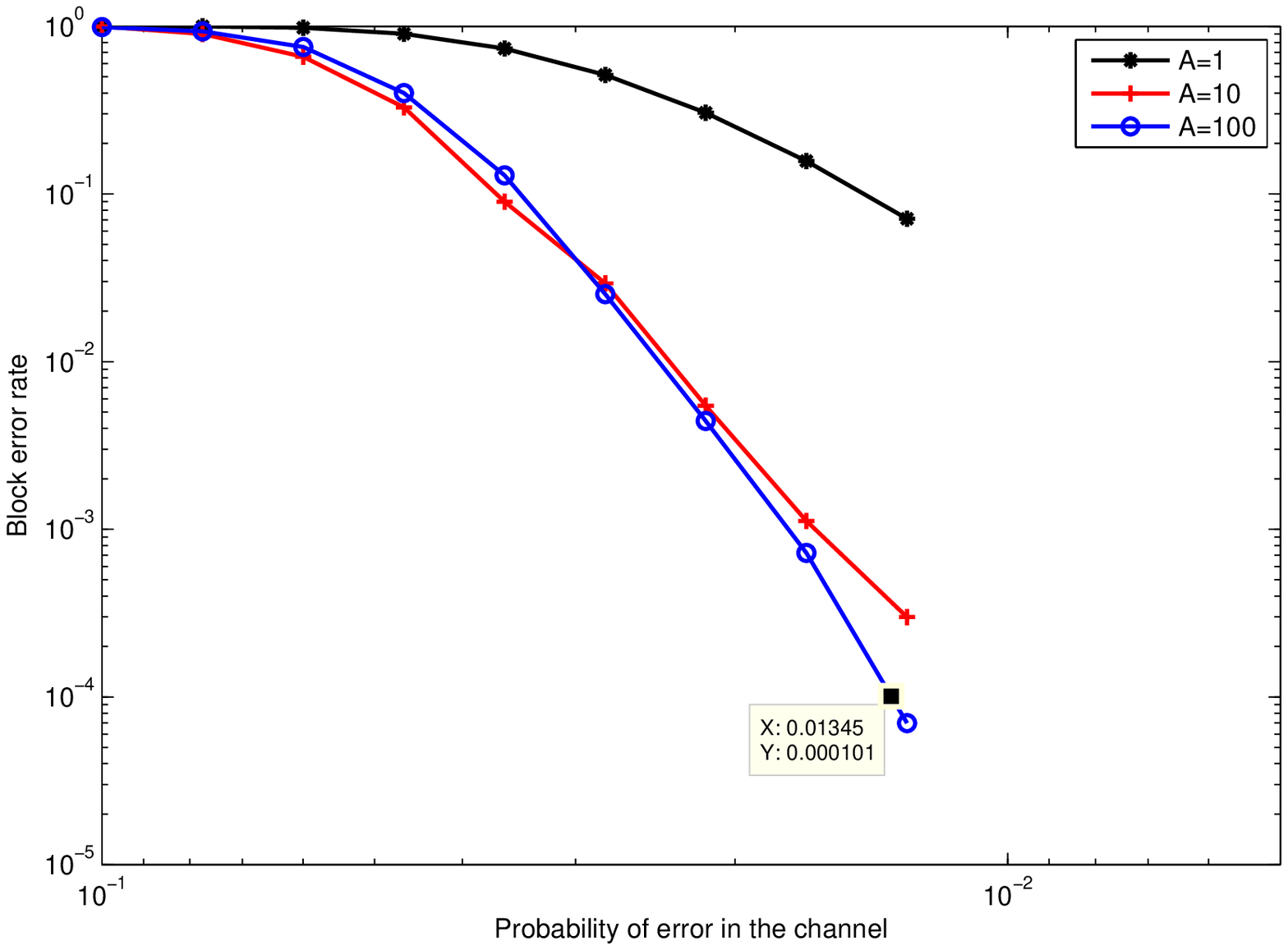,scale=0.6}
\caption{Performance of a $[[255,143,9/17]]$ code described in the
  text for choices $A=1,10,100$ of the channel asymmetry.}\label{fig:fg255d9}
\end{minipage}
\hspace{0.5cm} 
\begin{minipage}[b]{1\linewidth}
\centering
\epsfig{file=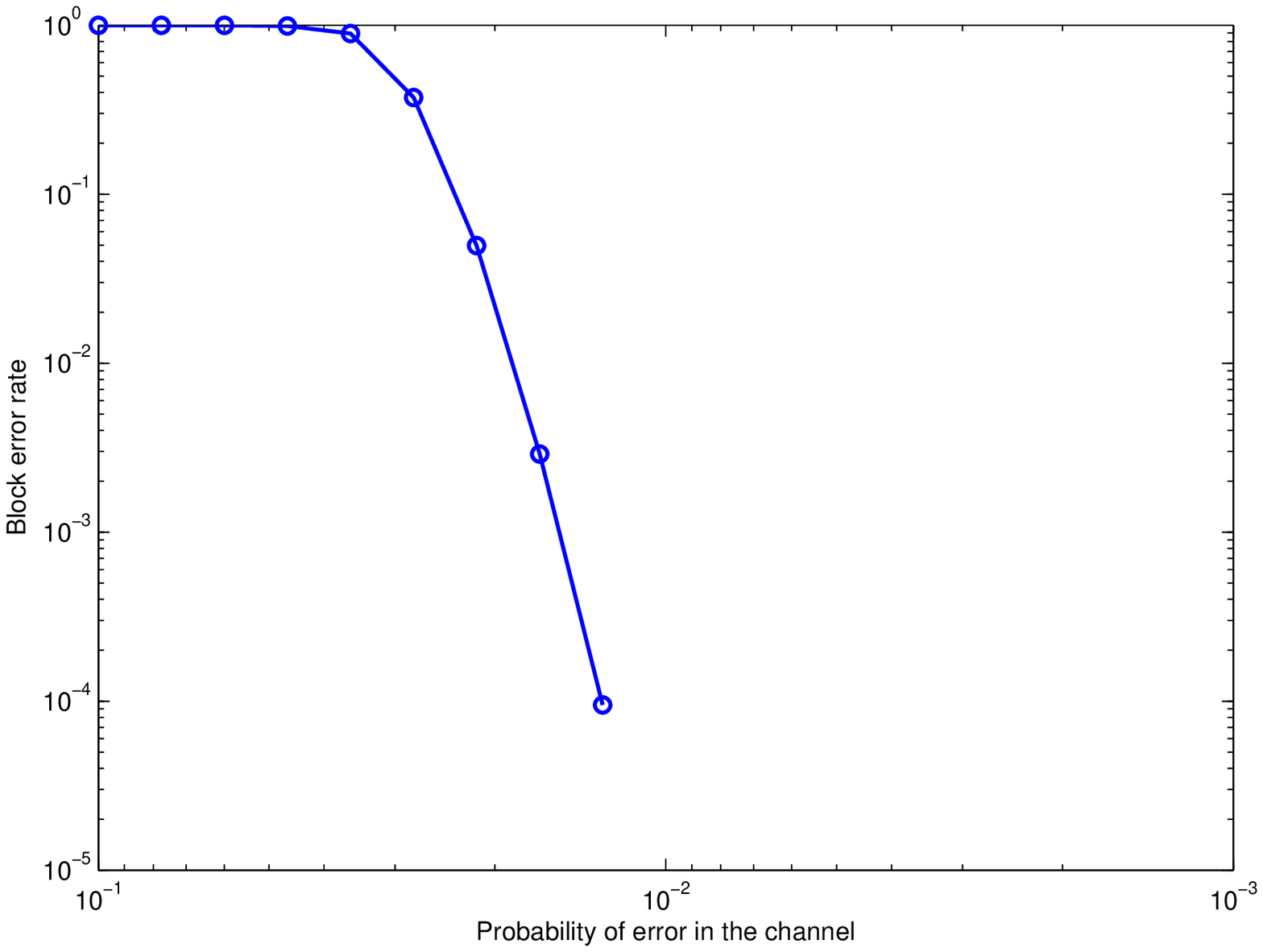,scale=0.6}
\caption{Performance of  $[[1023,731,11/33]]$ code for $A=100$.}\label{fig:n1023A100}
\end{minipage}
\end{figure}

The question naturally raises how do these codes compare with the
codes proposed in \cite{ioffe07}. Strictly speaking both constructions
have regimes where they can perform better than the other. But it
appears that the algebraically constructed asymmetric codes have the
following benefits with respect to the randomly constructed ones of
\cite{ioffe07}.
\begin{itemize}
\item They give comparable performance and higher data rates with
  shorter lengths.
\item The benefits of classical algebraic LDPC codes are inherited, giving for instance
lower error floors compared to the random constructions.
\item The code construction is systematic.
\end{itemize}
Our codes also offer flexibility in the rate and performance of the
code because we can choose many possible BCH codes for a given finite
geometry LDPC code or vice versa.  The flip side however is that the
codes given here have higher complexity of decoding.

%% file: chBch.tex
\chapter{New Results on BCH Codes\footnotemark}\label{ch:bch}
\footnotetext{\copyright 2007 IEEE. Reprinted with permission from
S. A. Aly, A. Klappenecker and P. K. Sarvepalli, ``On quantum and classical BCH
codes''. {\em IEEE Trans. Inform. Theory}, vol 53, no. 3, pp.~1183--1188, 2007. }
\nix{
Classical BCH codes that contain their (Euclidean or Hermitian) dual codes can be
used to construct quantum stabilizer codes; this chapter studies the
properties of such codes.  It is shown that a BCH code of length $n$ can contain its
dual code only if its designed distance $\delta=O(\sqrt{n})$, and the converse is
proved in the case of narrow-sense codes. Furthermore, the dimension of narrow-sense
BCH codes with small design distance is completely determined, and -- consequently
-- the bounds on their minimum distance are improved.  These results make it
possible to determine the parameters of quantum BCH codes in terms of their design
parameters.
}

The Bose-Chaudhuri-Hocquenghem (BCH)
codes~\cite{Bose60a,Bose60b,Gorenstein61,Hocquenghem59} are a well-studied class of
cyclic codes that have found numerous applications in classical and more recently in
quantum information processing. Recall that a cyclic code of length $n$ over a
finite field $\F_q$ with $q$ elements, and $\gcd(n,q)=1$, is called a \textsl{BCH code with
designed distance $\delta$} if its generator polynomial is of the form
$$g(x)=\prod_{z\in Z} (x-\alpha^z), \qquad Z=C_b\cup \cdots \cup
C_{b+\delta-2},$$ where $C_x=\{ xq^k\bmod n \,|\, k\in \Z, k\ge 0\,\}$ denotes the
$q$-ary cyclotomic coset of $x$ modulo~$n$, $\alpha$ is a primitive element of
$\F_{q^m}$, and $m=\ord_n(q)$ is the multiplicative order of $q$ modulo $n$.  Such a
code is called primitive if $n=q^m-1$, and narrow-sense if $b=1$.

An attractive feature of a (narrow-sense) BCH code is that one can derive many
structural properties of the code from the knowledge of the parameters $n$, $q$, and
$\delta$ alone.  Perhaps the most well-known facts are that such a code has minimum
distance $d\ge \delta$ and dimension $k\ge n-(\delta-1)\ord_n(q)$.  In this
chapter, we will show that a necessary condition for a narrow-sense BCH code
which contains its Euclidean dual code is that its designed distance
$\delta=O(qn^{1/2})$. We also derive a sufficient condition for dual containing BCH
codes. Moreover, if the codes are primitive, these conditions are same. These
results allow us to derive families of quantum stabilizer codes. Along the way, we
find new results concerning the minimum distance and dimension of classical BCH
codes.

To put our results into context, we give a brief overview of related work.  This
chapter was motivated by problems concerning quantum BCH codes; specifically,
our goal was to derive the parameters of the quantum codes as a function of the
design parameters.  Examples of certain binary quantum BCH codes have been given by
many authors, see, for example,~\cite{calderbank98,grassl99b,grassl00,steane96}.
Steane \cite{steane99} gave a simple criterion to decide when a binary narrow-sense
primitive BCH code contains its dual, given the design distance and the length of
the code.  We generalize Steane's result in various ways, in particular, to
narrow-sense (not necessarily primitive) BCH codes over arbitrary finite fields with
respect to Euclidean and Hermitian duality. These results allow one to derive
quantum BCH codes; however, it remains to determine the dimension, purity, and
minimum distance of such quantum codes.

The dimension of a classical BCH code can be bounded by many different standard
methods, see~\cite{berlekamp68,huffman03,macwilliams77} and the references therein.
An upper bound on the dimension was given by Shparlinski~\cite{shparlinski88}, see
also~\cite[Chapter~17]{konyagin99}. More recently, the dimension of primitive
narrow-sense BCH codes of designed distance $\delta< q^{\lceil m/2\rceil}+1$ was
apparently determined by Yue and Hu~\cite{yue96}, according to
reference~\cite{yue00}.  We generalize their result and determine the dimension of
narrow-sense BCH codes that are not necessarily primitive for a certain range of
designed distances. As desired, this result allows us to explicitly obtain the
dimension of the quantum codes without computation of cyclotomic cosets.

The purity and minimum distance of a quantum BCH code depend on the minimum distance
and dual distance of the associated classical code. In general, it is a difficult
problem to determine the true minimum distance of BCH codes, see~\cite{charpin98}.
A lower bound on the dual distance can be given by the Carlitz-Uchiyama-type bounds
when the number of field elements is prime, see, for
example,~\cite[page~280]{macwilliams77} and \cite{stichtenoth94}. Many authors have
determined the true minimum distance of BCH codes in special cases, see, for
instance,~\cite{peterson72},\cite{yue00}.

This chapter also extends our previous work on \textit{primitive} narrow-sense BCH
codes \cite{aly06a}, simplifies some of the proofs and generalizes many of the
results to the nonprimitive case.

\medskip
\textit{Notation.} We denote the ring of integers by $\mathbf{Z}$ and the finite field
with $q$ elements by $\mathbf{F}_q$. We use the bracket notation of Iverson and
Knuth that associates to $[\textit{statement}\,]$ the value 1 if \textit{statement}
is true, and 0 otherwise. For instance, we have $[k \text{ even}]= k-1\bmod 2$ and
$[k \text{ odd}]= k\bmod 2$ for an integer $k$. The Euclidean dual code $C^\perp$ of
a code $C\subseteq \F_q^n$ is given by $C^\perp = \{ y\in \F_q^n \,|\, x\cdot y =0
\mbox{ for all } x \in C \},$ while the Hermitian dual of $C\subseteq \F_{q^2}^n$ is
defined as $C^{\perp_h}=\{ y\in \F_{q^2}^n\,|\, y^q\cdot x = 0 \mbox{ for all } x
\in C\}$. We denote a narrow-sense BCH code of length $n$ over $\F_q$ with designed
distance $\delta$ by $\B(n,q;\delta)$, and we omit the parameter $q$ if the finite
field is clear from the context.

\section{Euclidean Dual Codes}
Recall that one can construct quantum stabilizer codes using classical codes that
contain their duals. In this section, our goal is to find such classical codes.
Steane showed that a primitive, narrow-sense, binary BCH code of length $2^m-1$
contains its dual if and only if its designed distance $\delta$ satisfies $\delta
\leq 2^{\lceil m/2\rceil}-1$, see~\cite{steane99}. We generalize this result in
various ways.

\begin{lemma}\label{th:selforthogonal}
Let $C$ be a cyclic code of length $n$ over the finite field $\F_q$ such that
$\gcd(n,q)=1$, and let $Z$ be the defining set of $C$.  The code $C$ contains its
Euclidean dual code if and only if $Z\cap Z^{-1} = \emptyset$, where $Z^{-1}$
denotes the set $Z^{-1}=\{-z\bmod n\mid z \in Z \}$.
\end{lemma}
\begin{proof}
See \cite[Theorem~2]{grassl97}. See also~\cite[Theorem 4.4.11]{huffman03}.
\end{proof}

Let us first consider narrow-sense BCH codes of length $n$ such that the
multiplicative order of $q$ modulo $n$ equals 1; for example, Reed-Solomon codes
belong to this class of codes. We can avoid some special cases in our subsequent
arguments by treating this case separately. Furthermore, the next lemma nicely
illustrates the proof technique that will be used throughout this section, so it can
serve as a warm-up exercise.

\begin{lemma}
Suppose that $q$ is a power of a prime and $n$ is a positive integer such that
$q\equiv 1\bmod n$. We have $\B(n,q;\delta)^\perp\subseteq \B(n,q;\delta)$ if and
only if the designed distance $\delta$ is in the range $2\le \delta\le
\delta_{\max}=\lfloor (n+1)/2 \rfloor$.
\end{lemma}
\begin{proof}
The defining set $Z$ of $\B(n,q;\delta)$ is given by $Z=\{1, \ldots,\delta-1\}$,
since $q$ has multiplicative order 1 modulo $n$, and therefore all cyclotomic cosets
are singleton sets.  If $\B(n,q;\delta)^\perp\subseteq \B(n,q;\delta)$, then by
Lemma~\ref{th:selforthogonal}, $Z\cap Z^{-1}=\emptyset$. If $x\in Z$, then
$n-x\not\in Z$ and $n-x>x$; hence, $\delta_{\max}\leq \lfloor (n+1)/2 \rfloor$.
Conversely, if $\delta \leq\lfloor (n+1)/2\rfloor$, then $\min Z^{-1} =\min
\{n-1,\ldots, n-\delta+1\} = n-\delta+1 \geq n-\lfloor (n+1)/2\rfloor +1 =\lceil
(n+1)/2
 \rceil \geq \delta_{\max}$; hence, $Z\cap Z^{-1}=\emptyset$ and  Lemma~\ref{th:selforthogonal} implies that $\B(n,q;\delta)^\perp\subseteq \B(n,q;\delta)$.
\end{proof}

If the multiplicative order $m$ of $q$ modulo $n$ is larger than 1, then the
defining set of the code has a more intricate structure, so proofs become more
involved. The next theorem gives a sufficient condition on the designed distances
for which the dual code of a narrow-sense BCH code is self-orthogonal.

\begin{theorem}\label {th:sufficientEdual}
Suppose that $m=\ord_n(q)$. If the designed distance $\delta$ is in the range $2\leq
\delta\leq \delta_{\max}=\floor{\kappa}$, where
\begin{equation}\label{ddistbound}
\kappa = \frac{n}{q^{m}-1} (q^{\lceil m/2\rceil}-1-(q-2)[m \textup{ odd}]),
\end{equation}
then $\B(n,q;\delta)^\perp \subseteq \B(n,q;\delta)$.
\end{theorem}
\begin{proof}
It suffices to show that $\B(n,q;\delta_{\max})^\perp\subseteq
\B(n,q;\delta_{\max})$ holds, since $\B(n,q;\delta)$ contains
$\B(n,q;\delta_{\max})$, and the claim follows from these two facts.

Seeking a contradiction, we assume that $\B(n,q;\delta_{\max})$ does not contain its
dual.  Let $Z=C_1\cup \cdots \cup C_{\delta_{\max}-1}$ be the defining set of
$\B(n,q;\delta_{\max})$.  By Lemma~\ref{th:selforthogonal}, $Z\cap Z^{-1}\neq
\emptyset$, which means that there exist two elements $x,y\in
\{1,\ldots,\delta_{\max}-1\}$ such that $y\equiv -xq^j\bmod n$ for some $j\in
\{0,1,\ldots, m-1 \}$, where $m$ is the multiplicative order of $q$ modulo $n$.
Since $\gcd(q,n)=1$ and $q^m\equiv 1 \bmod n$, we also have $x\equiv -yq^{m-j}\bmod
n$.  Thus, exchanging $x$ and $y$ if necessary, we can even assume that $j$ is in
the range $0 \leq j\leq \lfloor m/2 \rfloor$.  It follows from (\ref{ddistbound})
that
\begin{eqnarray*}
1&\leq&  xq^j \le (\delta_{\max}-1)q^j \\ &\le&
\frac{n}{q^m-1}(q^m-q^j-q^j(q-2)[m\text{ odd}])-q^j\\ &<&n,
\end{eqnarray*}
 for all $j$ in the range $0\le j\le \floor{m/2}$. Since $1\le xq^j<n$
and $1\le y<n$, we can infer from $y\equiv -xq^j\bmod n$ that $y=n-xq^j$. But this
implies
\begin{eqnarray*}
\begin{array}{lcl}
y &\ge&\ds n-xq^{\floor{m/2}} \\
&\ge& \ds n-\frac{n}{q^m-1}(q^m-q^{\floor{m/2}}
 - q^{\floor{m/2}}(q-2)[m\text{ odd}])+q^{\floor{m/2}}\\
&=& \ds \frac{n}{q^m-1}( q^{\floor{m/2}}-1+q^{\floor{m/2}}(q-2)[m \text{ odd}]  )\\ && +q^{\floor{m/2}}\\
&\ge& \delta_{\max}\, ,
\end{array}
\end{eqnarray*}
contradicting the fact that $y<\delta_{\max}$.
\end{proof}

Now we will derive a necessary condition on the design distance of narrow-sense,
nonprimitive BCH codes that contain their duals.

\begin{theorem}\label{th:necessaryEdual}
Suppose that $m=\ord_n(q)$. If the designed distance $\delta$ exceeds
$\delta_{\max}=\floor{qn^{1/2}}$, then $\B(n,q;\delta)^\perp \not\subseteq
\B(n,q;\delta)$.
\end{theorem}
\begin{proof}
Let $n=n_0+n_1q+\cdots+n_{d-1}q^{d-1}$, where $0\leq n_i\leq q-1$ and $\delta \geq
\delta_{\max}+1$. Then the defining set $Z\supseteq \{1,\ldots, \lfloor
qn^{1/2}\rfloor \}$. We will show that $Z\cap Z^{-1}\neq \emptyset$. Let,
\begin{eqnarray*}
s&=&\sum_{i=\lfloor d/2 \rfloor}^{ d-1}n_{i}q^{i-\lfloor d/2\rfloor},\\
s& \leq & (q-1)\sum_{i=\lfloor d/2\rfloor}^{ d-1} q^{i-\lfloor d/2\rfloor}
=q^{\lceil d/2\rceil}-1 < q^{\lceil d/2\rceil}.
\end{eqnarray*}
Since $q^{d-1}<n< q^d$, we have $q^{(d+1)/2}< qn^{1/2}< q^{(d+2)/2}$. If $d$ is even
then $\lceil d/2\rceil < (d+1)/2$ and if $d$ is odd, then $\lceil d/2 \rceil \leq
(d+1)/2$. Hence we have $s< q^{\lceil d/2\rceil} \leq q^{(d+1)/2}< qn^{1/2}$.
Therefore $s\in Z$. Now consider,
\begin{eqnarray*}
s'=n-sq^{\lfloor d/2 \rfloor} &=& \sum_{i=0}^{d-1}n_iq^i-q^{\floor{ d/2}}
\sum_{i=\lfloor d/2\rfloor}^{ d-1}n_{i}q^{i-\lfloor d/2\rfloor},\\
&=&\sum_{i=0}^{\lfloor d/2 -1\rfloor}n_i q^i < q^{\floor{ d/2}} \\&<&
q^{(d+1)/2}<qn^{1/2}.
\end{eqnarray*}
Hence $s'\in Z$ and by definition $s'\in Z^{-1}$, which implies $Z\cap
Z^{-1}\neq\emptyset$; by Lemma~\ref{th:selforthogonal} it follows that
$\B(n,q;\delta)^\perp\not\subseteq \B(n,q;\delta)$.
\end{proof}

The condition we just derived can be strengthened under some restrictions.
Especially, if the constant $\kappa$ in equation~(\ref{ddistbound}) is integral,
then we can derive a necessary and sufficient condition as shown below:

\begin{theorem}\label{th:nonprimitiveduals2}
We keep the notation of Theorem~\ref{th:necessaryEdual}. Suppose that $\kappa$ is
integral, and that $m\ge 2$. We have $\B(n,q;\delta)^\perp \subseteq \B(n,q;\delta)$
if and only if the designed distance $\delta$ is in the range $2\leq \delta\leq
\delta_{\max}=\kappa$.
\end{theorem}
\begin{proof}
Suppose that $\B(n,q;\delta)^\perp \subseteq \B(n,q;\delta)$. Seeking a
contradiction, we assume that $\delta>\delta_{\max}$; thus, $\delta_{\max}$ is
contained in the defining set $Z$ of $\B(n,q;\delta)$. If $m$ is even, then
\begin{eqnarray*} -\delta_{\max}q^{\floor{m/2}} &\equiv&
-\frac{nq^{\floor{m/2}}}{q^{\floor{m/2}}+1} \equiv -n+\frac{n}{q^{\floor{m/2}}+1}\\
&\equiv& \delta_{\max}\pmod{n},
\end{eqnarray*} hence, $\delta_{\max}\in Z\cap Z^{-1}\neq \emptyset$. If $m$ is
odd, then
$$\begin{array}{lcl}
-\delta_{\max}q^{\floor{m/2}}&\equiv&
-n(q^{m}-q^{\ceil{m/2}}+q^{\floor{m/2}})/(q^{m}-1)\\
&\equiv& n(q^{\ceil{m/2}}-q^{\floor{m/2}}-1)/(q^m-1) \\ &\equiv& s \pmod{n}.
  \end{array}
$$
By definition, $s\in Z^{-1}$; furthermore,  $s<\delta_{\max}$, so $s\in Z\cap
Z^{-1}\neq \emptyset$.
In both cases, $m$ even and odd, we found that $Z\cap Z^{-1}$ is not empty, so
$\B(n,q;\delta)$ cannot contain its Euclidean dual code, contradiction.
The converse follows from Theorem~\ref{th:sufficientEdual}.
\end{proof}
As a consequence of Theorem~\ref{th:nonprimitiveduals2} we have the following test
for primitive narrow-sense BCH codes that contain their duals.
\begin{corollary}
A primitive narrow-sense BCH code of length $n=q^m-1$, $m\ge 2$, over the finite
field\/ $\F_q$ contains its Euclidean dual code if and only if its designed distance
$\delta$ satisfies
$$2\le \delta \leq
\delta_{\max}=q^{\lceil m/2\rceil}-1-(q-2)[m \textup{ odd}].$$
\end{corollary}

We observe that a narrow-sense BCH code containing its Euclidean dual code must have
a small designed distance ($\delta=O(\sqrt{n}))$, when the multiplicative order of
$q$ modulo $n$ is greater than one.  This raises the question whether one can allow
larger designed distances by considering non-narrow-sense BCH codes. Our next result
shows that this is not possible, at least in the case of primitive codes.

\begin{theorem}\label {th:nonnarrowEdual}
Let $C$ be a primitive (not necessarily narrow-sense) BCH code of length $n=q^m-1$
over $\F_q$ with designed distance $\delta$. If $m>1$ and $\delta$ exceeds
$$
\delta_{\max}=\left\{ \begin{array}{ll}
q^{m/2}-1,& m \equiv 0\bmod 2,\\
2(q^{(m+1)/2}-q+1), & m \equiv 1 \bmod 2,
\end{array} \right. $$ then $C$
cannot contain its Euclidean dual.
\end{theorem}

\begin{proof}
Let the defining set of $C$ be $Z=C_{b}\cup C_{b+1}\cup \cdots \cup C_{b+\delta-2}$.
We will show that if $\delta > \delta_{\max}$ then $Z\cap Z^{-1}\neq \emptyset$. If
$0\in Z$, then $0\in Z^{-1}$, so $Z\cap Z^{-1}\neq \emptyset$. Therefore, we can
henceforth assume that $0\not\in Z$, which implies $b\ge 1$ and $b+\delta-2<n$.
\begin{enumerate}
\item Suppose that $m$ is even; thus, $\delta_{\max}=q^{m/2}-1$.  If
$\delta > \delta_{\max}$ then the defining set $Z$ contains an element of the form
$s=\alpha \delta_{\max}$ for some integer $\alpha$. However,
\begin{eqnarray*}
-sq^{m/2} &\equiv& -\alpha(q^{m/2}-1)q^{m/2} \equiv
 \alpha(q^{m/2}-1)\\ &\equiv& s \pmod n.
\end{eqnarray*}
Hence, $s\in Z\cap Z^{-1}\neq \emptyset$.
\item
Suppose that $m >1$ is odd; thus, $\delta_{\max}=2q^{(m+1)/2}-2q+2$. If $\delta>
\delta_{\max}$ then there exists an integer $\alpha$ such that two multiples of
$\delta'=\delta_{\max}/2$ are contained in the range $b\leq (\alpha-1) \delta' <
\alpha\delta' \leq b+\delta-2$. Since $b\ge 1$ and $\alpha \delta'<n$, it follows
that $2 \leq \alpha\leq q^{(m-1)/2}$.

The defining set $Z$ of the code contains the element $s=\alpha\delta'$. The number
$s'=\alpha(q^{(m+1)/2}-q^{(m-1)/2}-1)$ lies in the range $0\le s' \leq s$ and
satisfies $-s q^{(m-1)/2}\equiv s' \bmod n$, so $s'\in Z^{-1}$.

Suppose that $b \leq s'$. Then $s'\in Z$, which implies $Z\cap Z^{-1}\neq
\emptyset$.

Suppose that $s'< b$. Since $b\leq (\alpha-1)\delta'$, we obtain the inequality
$s'<(\alpha-1)\delta'$; solving for $\alpha$ shows that $\alpha\geq q$; thus, $q\leq
\alpha\leq q^{(m-1)/2}$.  Let $t'=(\alpha-1)(q^{(m+1)/2}-1)+q^{(m-1)/2}-1$; it is
easy to check that $t'$ is in the range $(\alpha-1)\delta'\le t'\le \alpha \delta'$
when $\alpha\geq q$; thus, $t'\in Z$.  Further, let $t=s-(\alpha-q+1)$; since $t\geq
s-\delta'$, we have $t\in Z$ as well. Since $-tq^{(m-1)/2}\equiv t'\bmod n$, we can
conclude that $t'\in Z\cap Z^{-1}\neq \emptyset$.
\end{enumerate}
Therefore, we can conclude that if the designed distance of $C$ is greater than
$\delta_{\max}$, then $Z\cap Z^{-1}\neq \emptyset$, which proves the claim thanks to
Lemma~\ref{th:selforthogonal}.
\end{proof}

\section{Dimension and Minimum Distance}
While the results in the previous section are sufficient to tell us when we can
construct quantum BCH codes, they are still unsatisfactory because we do not know
the dimension of these codes. To this end, we determine the dimension of
narrow-sense BCH codes of length $n$ with minimum distance $d=O(n^{1/2})$. It turns
out that these results on dimension also allow us to sharpen the estimates of the
true distance of some BCH codes.

First, we make some simple observations about cyclotomic cosets that are essential
in our proof.
\begin{lemma}\label{th:bchnpcosetsize}
Let $n$ be a positive integer and $q$ be a power of a prime such that $\gcd(n,q)=1$
and $q^{\lfloor m/2\rfloor} <n \leq q^m-1$, where $m=\ord_n(q)$. The cyclotomic
coset $C_x=\{ xq^j\bmod n \mid 0\le j<m\}$ has cardinality $m$ for all $x$ in the
range $1\leq x\leq nq^{\lceil m/2\rceil}/(q^m-1).$
\end{lemma}
\begin{proof}
If $m=1$, then $|C_x|=1$ for all $x$ and the statement is trivially true. Therefore,
we can assume that $m>1$.  Seeking a contradiction, we suppose that $|C_x|<m$,
meaning that there exists a divisor $j$ of $m$ such that $xq^{j}\equiv x \bmod n$,
or, equivalently, that $x(q^{j}-1)\equiv 0\bmod n$ holds.

Suppose that $m$ is even. The divisor $j$ of $m$ must be in the range $1\le j\le
m/2$. However, $x(q^{j}-1) \leq nq^{m/2}(q^{m/2}-1)/(q^{m}-1) <n $; hence
$x(q^j-1)\not\equiv 0 \bmod n$, contradicting the assumption $|C_x|<m$.

Suppose that $m$ is odd. The divisor $j$ of $m$ must be in the range $1\le j\le
m/3$. Since $q^{(m+1)/2}\le q^{2m/3}$ for $m\ge 3$, we have $x(q^j-1)\le
nq^{(m+1)/2}(q^{m/3}-1)/(q^m-1)\le nq^{2m/3}(q^{m/3}-1)/(q^m-1)<n$. Therefore,
$x(q^j-1)\not\equiv 0\bmod n$, contradicting the assumption $|C_x|<m$.
\end{proof}

The following observation tells us when some cyclotomic cosets are disjoint.
\begin{lemma}\label{th:npdisjointcosets}
Let $n\ge 1$ be an integer and $q$ be a power of a prime such that $\gcd(n,q)=1$ and
$q^{\lfloor m/2\rfloor} <n \leq q^m-1$, where $m=\ord_n(q)$. If $x$ and $y$ are
distinct integers in the range $1\leq x,\, y\leq \min\{ \lfloor nq^{\lceil
m/2\rceil}/(q^m-1)-1\rfloor, n-1\}$ such that $x,y\not\equiv 0 \bmod q$, then the
$q$-ary cyclotomic cosets of $x$ and $y$ modulo $n$ are distinct.
\end{lemma}
\begin{proof}
If $m=1$, then clearly $C_x=\{ x\}$, $C_y=\{y \}$ and distinct $x,y$ implies that
$C_x$ and $C_y$ are disjoint. If $m>1$, then $x,y\leq \lfloor nq^{\lceil
m/2\rceil}/(q^m-1)-1 \rfloor <n-1$. The set $S=\{ xq^j\bmod n, yq^j \bmod n\,|\,
0\leq j\leq \lfloor m/2\rfloor\}$ contains $2(\lfloor m/2\rfloor+1) \geq m+1$
elements, since $q^{\floor{m/2}}\times \lfloor nq^{\lceil m/2\rceil}
/(q^m-1)-1\rfloor < n$ and, thus, no two elements are identified modulo $n$.  If we
assume that $C_x=C_y$, then the preceding observation would imply that
$|C_x|=|C_y|\geq |S| \geq m+1$, which is impossible since the maximal size of a
cyclotomic coset is $m$. Hence, the cyclotomic cosets $C_x$ and $C_y$ must be
disjoint.
\end{proof}

With these results in hand, we can now derive the dimension of narrow-sense BCH
codes.
\begin{theorem}\label{th:bchnpdimension}
Let $q$ be a prime power and $\gcd(n,q)=1$ with $\ord_n(q)=m$. Then a narrow-sense
BCH code of length $q^{\lfloor m/2\rfloor} <n \leq q^m-1$ over $\F_q$ with designed
distance $\delta$ in the range $2 \leq \delta \le \min\{ \lfloor nq^{\lceil m/2
\rceil}/(q^m-1)\rfloor,n\} $ has dimension
\begin{equation}\label{eq:npdimension}
k=n-m\lceil (\delta-1)(1-1/q)\rceil.
\end{equation}
\end{theorem}
\begin{proof}
Let the defining set of $\B(n,q;\delta)$ be $Z=C_1\cup C_2\cdots \cup C_{\delta-1}$;
a union of at most $\delta -1$ consecutive cyclotomic cosets. However, when $1\leq
x\leq \delta-1$ is a multiple of $q$, then $C_{x/q}=C_x$. Therefore, the number of
cosets is reduced by $\lfloor(\delta-1)/q \rfloor$. By
Lemma~\ref{th:npdisjointcosets}, if $x, y\not\equiv 0 \bmod q$ and $x\neq y$, then
the cosets $C_x$ and $C_y$ are disjoint. Thus, $Z$ is the union of
$(\delta-1)-\lfloor (\delta-1)/q\rfloor= \lceil (\delta-1)(1-1/q)\rceil$ distinct
cyclotomic cosets. By Lemma~\ref{th:bchnpcosetsize}, all these cosets have
cardinality~$m$. Therefore, the degree of the generator polynomial is $m\lceil
(\delta-1)(1-1/q)\rceil$, which proves our claim about the dimension of the code.
\end{proof}

As a consequence of the dimension result, we can tighten the bounds on the minimum
distance of narrow-sense BCH codes generalizing a result due to Farr,
see~\cite[p.~259]{macwilliams77}.

\begin{corollary}
A $\B(n,q;\delta)$ code
\begin{compactenum}[i)]
\item with length in the range $q^{\lfloor m/2\rfloor} <n \leq q^m-1$,
$m=\ord_n(q)$,
\item and designed distance in the range
$2 \leq \delta \le \min\{ \lfloor nq^{\lceil m/2 \rceil}/(q^m-1)\rfloor,n\} $
\item such that
\begin{eqnarray}\label{eqa3}
\sum_{i=0}^{\lfloor (\delta+1)/2\rfloor} \binom{n}{i} (q-1)^i
>q^{m\lceil (\delta-1)(1-1/q)\rceil},
\end{eqnarray}
\end{compactenum}
has minimum distance $d= \delta$ or $\delta+1$; if $\delta\equiv 0\bmod q$, then
$d=\delta+1$.
\end{corollary}
\begin{proof}
Seeking a contradiction, we assume that the minimum distance~$d$ of the code
satisfies $d \geq \delta+2$. We know from Theorem~\ref{th:bchnpdimension} that the
dimension of the code is $k=n-m\lceil (\delta-1)(1-1/q)\rceil.$ If we substitute
this value of $k$ into the sphere-packing bound $q^{k}   \sum_{i=0}^{\lfloor
(d-1)/2\rfloor} \binom{n}{i} (q-1)^i \leq q^n$, then we obtain

\begin{eqnarray*}
\sum_{i=0}^{\lfloor (\delta+1)/2\rfloor}\binom{n}{i}(q-1)^i &\le&
\sum_{i=0}^{\lfloor (d-1)/2\rfloor}\binom{n}{i}(q-1)^i \\&\le& q^{m\lceil
(\delta-1)(1-1/q)\rceil},
\end{eqnarray*}
but this contradicts condition~(\ref{eqa3}); hence, $\delta\le d\le \delta+1$.

If $\delta\equiv 0\bmod q$, then the cyclotomic coset $C_\delta$ is contained in the
defining set $Z$ of the code because $C_\delta=C_{\delta/q}$. Thus, the BCH bound
implies that the minimum distance must be at least $\delta+1$.
\end{proof}

We conclude this section with  a minor result on the dual distance of BCH codes
which will be needed later for determining the purity of quantum codes.
\begin{lemma}\label{th:Edualdist}
Suppose that $C$ is a narrow-sense BCH code of length $n$ over $\F_q$ with designed
distance $2\leq \delta\le \delta_{\max}=\lfloor n(q^{\lceil m/2\rceil}-1-(q-2)[m
\textup{ odd}])/(q^m-1) \rfloor$, then the dual distance $d^\perp \geq \delta_{\max}
+ 1$.
\end{lemma}
\begin{proof}
Let $N=\{0,1,\ldots,n-1 \}$ and $Z_{\delta}$ be the defining set of $C$. We know
that $Z_{\delta_{\max}}\supseteq Z_{\delta}\supset \{1,\ldots,\delta-1 \}$.
Therefore $N\setminus Z_{\delta_{\max}} \subseteq N\setminus Z_{\delta}$.  Further,
we know that $Z\cap Z^{-1}=\emptyset$ if $2\leq \delta\leq \delta_{\max }$ from
Lemma~\ref{th:selforthogonal} and Theorem~\ref{th:sufficientEdual}. Therefore,
$Z^{-1}_{\delta_{\max}}\subseteq N\setminus Z_{\delta_{\max}}\subseteq N\setminus
Z_{\delta}$.

Let $T_{\delta}$ be the defining set of the dual code. Then $T_{\delta}=(N\setminus
Z_{\delta})^{-1} \supseteq Z_{\delta_{\max}}$. Moreover $\{0\}\in N\setminus
Z_{\delta}$ and therefore $T_{\delta}$. Thus there are at least $\delta_{\max}$
consecutive roots in $T_{\delta}$. Thus the dual distance $d^\perp \geq
\delta_{\max}+1$.
\end{proof}

\section{Hermitian Dual Codes}
Suppose that $C$ is a linear code of length $n$ over $\F_{q^2}$. Recall that its
Hermitian dual code is defined by $C^{\perp_h}=\{ y\in \F_{q^2}^n\,|\, y^q\cdot x =
0 \mbox{ for all } x \in C\}$, where $y^q=(y_1^q,\dots,y_n^q)$ denotes the conjugate
of the vector $y=(y_1,\dots,y_n)$.

\begin{lemma}\label{th:hermitianBch}
Assume that $\gcd(n,q)=1$. A cyclic code of length $n$ over $\F_{q^2}$ with defining
set $Z$ contains its Hermitian dual code if and only if $Z\cap Z^{-q} = \emptyset$,
where $Z^{-q}=\{-qz \bmod n \mid z \in Z \}$.
\end{lemma}
\begin{proof}
Let $N=\{0,1,\dots,n-1\}$. If $g(x)=\prod_{z\in Z} (x-\alpha^z)$ is the generator
polynomial of a cyclic code $C$, then $h^\dagger(x)=\prod_{z\in N\setminus Z}
(x-\alpha^{-qz})$ is the generator polynomial of $C^{\perp_h}$.  Thus,
$C^{\perp_h}\subseteq C$ if and only if $g(x)$ divides $h^\dagger(x)$. The latter
condition is equivalent to $Z\subseteq \{ -qz\,|\, z\in N\setminus Z\}$, which can
also be expressed as $Z\cap Z^{-q}=\emptyset$.
\end{proof}

Now similar to Theorem~\ref {th:sufficientEdual} we will derive a sufficient
condition for BCH codes that contain their Hermitian duals.

\begin{theorem}\label {th:sufficientHdual}
Suppose that $m=\ord_n(q^2)$. If the designed distance $\delta$ satisfies
$2\leq \delta  \leq \delta_{\max}$, where 
\begin{eqnarray*}
\delta_{\max} &=& \left\lfloor \frac{n}{q^{2m}-1} (q^{
m+[\textup{m even}] }-1-(q^2-2)[m \textup{ even}])\right\rfloor,
\end{eqnarray*}
 then
$\B(n,q^2;\delta)^{\perp_h} \subseteq \B(n,q^2;\delta).$
\end{theorem}
\begin{proof}
Since $\B(n,q^2;\delta)$ contains $\B(n,q^2;\delta_{\max})$, it suffices to show
that the relation
$\B(n,q^2;\delta_{\max})^{\perp_h} \subseteq \B(n,q^2;\delta_{\max})$ holds.

Seeking a contradiction, we assume that $\B(n,q^2;\delta_{\max})$ does not contain
its dual. Let $Z=C_1 \cup C_2\cup \dots \cup C_{\delta_{\max}-1}$ be the defining
set of $\B(n,q^2;\delta_{\max} )$. By Lemma~\ref{th:hermitianBch}, $Z\cap Z^{-q}\neq
\emptyset$, which means that there exist two elements $x,y \in
\{1,...,\delta_{\max}-1\}$ such that $y=-xq^{2j+1} \bmod n$ for some $j \in
\{0,1,...,m-1\}$, where $m=\ord_n(q)$. Since $\gcd(q,n)=1$ and $q^{2m} \equiv 1
\bmod n$, we also have $y \equiv - xq^{2m-2j-1} \bmod n$, so we can assume without
loss of generality that $j$ lies in the range $0 \leq j \leq \lfloor (m-1)/2
\rfloor$. It follows that
\begin{eqnarray*}
xq^{2j+1} &\leq& (\delta_{\max}-1)q^{2j+1} \\ &=& \frac{nq^{2j+1}}{q^{2m}-1} (q^{
m+[\textup{m  even}] }-1 
-(q^2-2)[m \textup{ even}]) -q^{2j+1} \\&<& n
\end{eqnarray*}
holds  for all $j$ in the range $0 \leq j \leq \lfloor (m-1)/2\rfloor$.

Since $1 \leq xq^{2j+1} < n$, the congruence $y \equiv -xq^{2j+1} \bmod n$ implies
that $y=n-xq^{2j+1}$. Therefore, $y\ge n-(\delta_{\max}-1)q^{2\floor{(m-1)/2}+1}$,
which is equivalent to
\begin{eqnarray*} y&\geq& n -  \frac{n q^{2\lfloor
(m-1)/2\rfloor+1}}{q^{2m}-1} (q^{ m+[\textup{m even}]}-1\\&&-(q^2-2)[m \textup{
even}]) +q^{2\floor{(m-1)/2}+1}.
\end{eqnarray*}
If $m$ is odd, this yields
\begin{eqnarray*}
y&\geq& n - \frac{n q^m}{q^{2m}-1} (q^{ m }-1)+q^m \\&=&
\frac{n}{q^{2m-1}}(q^{m}-1)+q^m\ge \delta_{\max}\, .
\end{eqnarray*}
Similarly, if $m$ is even, then
\begin{eqnarray*}
   y&\geq&\frac{n}{q^{2m}-1}(q^{m+1}-q^{m-1}-1)+q^{m-1} \\ &\ge& \delta_{\max}.
\end{eqnarray*}
Both cases contradict the assumption $0\le y<\delta_{\max}$. Therefore, we can
conclude that $\B(n,q;\delta_{\max})$ contains its Hermitian dual code.
\end{proof}
Arguing as in Theorem~\ref{th:necessaryEdual} we can show that a BCH code must have
its designed distance $\delta=O(q^2n^{1/2})$ if it contains its Hermitian dual. As
the arguments are very similar we illustrate it for a simpler case as shown below:

\begin{lemma}
Let $C\subseteq \F_{q^2}^n$ be a nonnarrow-sense, nonprimitive BCH code of length
$n\equiv 0\bmod q^{m}+1$, where $m=\ord_n(q^2)$. If its design distance $\delta \geq
\delta_{\max}=n/(q^{m}+1)$, then $C$ cannot contain its Hermitian dual.
\end{lemma}
\begin{proof}
The defining set $Z=C_b\cup \ldots \cup C_{b+\delta-2}$ contains $\{
b,\ldots,b+\delta-2\}$. If $\delta>\delta_{\max}=n/(q^{m}+1)$, then there exists an
element $s=\alpha\delta_{\max}\in Z$ for some positive integer $\alpha$.  Then
$-qs(q^2)^{(m-1)/2}\equiv-\alpha nq^{m}/(q^m+1) \equiv \alpha n/(q^m+1)\equiv s \mod
{n}$. Therefore, $Z\cap Z^{-q}\neq \emptyset$; hence, $C$ cannot contain its
Hermitian dual code.
\end{proof}
Finally, we conclude this section on Hermitian duals by proving as in the Euclidean
case nonnarrow-sense BCH codes that contain their Hermitian duals cannot have too
large design distances.

\begin{theorem}\label{th:nonnarrowhdual}
Let $C \subseteq \F_{q^2}^n$ be a primitive (not necessarily narrow-sense) BCH code
of length $n=q^{2m}-1$, $m=\ord_n(q)$, and designed distance $\delta$. If $\delta$
exceeds
$$\delta_{\max}=\left\{ \begin{array}{ll}
q^{m}-1 & \text{if\/ $m$ is odd},\\
2(q^{m+1}-q^2+1) & \text{if\/ $m\neq 2$ is even},
\end{array} \right. $$ then $C$ cannot contain its Hermitian dual code.
\end{theorem}
\begin{proof}
Suppose that the defining set of $C$ is given by $Z=C_b\cup \cdots \cup
C_{b+\delta-2}$, where $C_x=\{ xq^{2j}\bmod n\,|\, j\in \Z \}$, and that
$\delta>\delta_{\max}$. Seeking a contradiction, we assume that $C^\hdual\subseteq
C$, which means that $Z\cap Z^{-q}=\emptyset$. It follows that $0\not\in Z$, for
otherwise $0\in Z\cap Z^{-q}$; therefore, $b\ge 1$ and $b+\delta-2<n$.

If $m$ is odd, then there exists an integer $\alpha$ such that $b\le
\alpha\delta_{\max}\le b+\delta-2$. We have $-q\alpha\delta_{\max}q^{m-1} \equiv
\alpha (1-q^m)q^m \equiv \alpha (q^m-1) \equiv \alpha\delta_{\max}\bmod n$; thus,
$\alpha\delta_{\max}\in Z\cap Z^{-q}\neq \emptyset.$

If $m>2$ is even and $\delta> \delta_{\max}=2q^{m+1}-2q^2+2$, then there exists an
integer $\alpha$ such that two multiples of $\delta'=\delta_{\max}/2$ are contained
in the range $b\leq (\alpha-1) \delta' < \alpha\delta' \leq b+\delta-2$. Since $b\ge
1$ and $\alpha \delta'<n$, it follows that $2 \leq \alpha\leq q^{m-1}$ (which holds
only if $m>2$).

Clearly $s=\alpha\delta'\in Z$. Let $s'\equiv -qsq^{m-2}\bmod n$, so $s'\in Z^{-q}$,
then $1\leq s'=\alpha(q^{m+1}-q^{m-1}-1) \leq s$ for $m>2$.

Suppose that $b \leq s'$. Then $s'\in Z$, which implies $Z\cap Z^{-q}\neq
\emptyset$.

Suppose that $s'< b$. Since $b\leq (\alpha-1)\delta'$, we obtain the inequality
$s'<(\alpha-1)\delta'$; solving for $\alpha$ shows that $\alpha\geq q^2$; thus,
$q^2\leq \alpha\leq q^{m-1}$.  Let $t'=(\alpha-1)(q^{m+1}-1)+q^{(m-1)/2}-1$; it is
easy to check that $t'$ is in the range $(\alpha-1)\delta'\le t'\le \alpha \delta'$
when $\alpha\geq q^2$; thus, $t'\in Z$.  Further, let $t=s-(\alpha-q^2+1)$; since
$t\geq s-\delta'$, we have $t\in Z$ as well. Since $-qtq^{m-2}\equiv t'\bmod n$, we
can conclude that $t'\in Z\cap Z^{-q}\neq \emptyset$. Hence, by
Lemma~\ref{th:hermitianBch} we conclude that $C$ cannot contain its Hermitian dual if
its design distance exceeds $\delta_{\max}$
\end{proof}

\section{Families of Quantum BCH Codes}
In this section we shall study the construction of (nonbinary) quantum BCH codes.
Calderbank, Shor, Rains and Sloane outlined  the construction of binary quantum BCH
codes  in \cite{calderbank98}. Grassl, Beth and Pellizari developed the theory
further by  formulating a nice condition for determining which BCH codes can be used
for constructing quantum codes \cite{grassl97,grassl99b}. The dimension and the
purity of the quantum codes constructed were determined by numerical computations.
Steane simplified it further for the special case of binary narrow-sense primitive
BCH codes \cite{steane99} and gave a very simple criterion based on the design
distance alone. Very little was done with respect to the nonprimitive and nonbinary
quantum BCH codes.

In this section we show how the results we have developed in the previous sections
help us to generalize the previous work on quantum codes and give very simple
conditions based on design distance alone. Further, we give precisely the dimension
and tighten results on the purity of the quantum codes.  The reader can refer
to Chapters~\ref{ch:stabq1} and \ref{ch:stabq2} for constructions on
stabilizer codes.
\nix{
But, first we review the
methods of constructing quantum codes from classical codes.

\begin{lemma}[Quantum Code Constructions] \label{th:qconst} \ \par
\begin{compactenum}[a)]
\item If there exists classical linear codes $C_1\subseteq C_2 \subseteq \F_q^n$,
then there exists an $[[n,k_2-k_1,d]]_q$ quantum code where $d=\min \wt\{
(C_2\setminus C_1)\cup (C_1^\perp\setminus C_2^\perp)\}$.

\item If there exists a classical linear $[n,k,d]_q$ code $C$ such that
$C^\perp\subseteq C$, then there exists an $[[n,2k-n,\ge d]]_q$ stabilizer code that
is pure to $d$. If the minimum distance of $C^\perp$ exceeds $d$, then the
stabilizer code is pure and has minimum distance $d$.

\item If there exists a classical linear $[n,k,d]_{q^2}$ code $D$ such
that $D^\hdual\subseteq D$, then there exists an $[[n,2k-n,\ge d]]_{q}$ stabilizer
code that is pure to $d$. If the minimum distance $d^\hdual$ of $D^\hdual$ exceeds
$d$, then the stabilizer code is pure and has minimum distance $d$.
\end{compactenum}
\end{lemma}
\begin{proof}
See, for instance,~\cite{ketkar06} for the proofs. Part a) is commonly referred to
as the CSS construction \cite[Lemma~20]{ketkar06} and b) is a special case of a);
part c) is the Hermitian code construction \cite[Corollary~19]{ketkar06}.
\end{proof}
}

\begin{theorem}\label{sh:nested}
Let $m=\ord_n(q)\geq 2$, where $q$ is a power of a prime and $\delta_1,\delta_2$ are
integers such that $2\le \delta_1 < \delta_2\le \delta_{\max}$ where
$$\delta_{\max} = \frac{n}{q^m-1}(q^{\lceil m/2\rceil}-1-(q-2)[m \textup{ odd}]),$$
then there exists a quantum code with parameters
$$[[n,m(\delta_2-\delta_1-\lfloor(\delta_2-1)/q\rfloor+\lfloor(\delta_1-1)/q \rfloor),\geq \delta_1]]_q$$
pure to $\delta_2$.
\end{theorem}
\begin{proof}
By Theorem~\ref{th:bchnpdimension}, there exist BCH codes $\B(n,q;\delta_i)$ with
the parameters $[n,n-m(\delta_i-1)+m\lfloor(\delta_i-1)/q\rfloor,\geq \delta_i]_q$
for $i\in \{ 1,2\}$. Further, $\B(n,q;\delta_2)\subset \B(n,q;\delta_1)$. Hence by
the CSS construction there exists a quantum code with the parameters
 $$[[n,m(\delta_2-\delta_1
-\lfloor(\delta_2-1)/q\rfloor +\lfloor(\delta_1-1)/q\rfloor),\geq \delta_1]]_q.$$
The purity follows due to the fact that $\delta_2>\delta_1$ and
Lemma~\ref{th:Edualdist} by which the dual distance of either BCH code is $\geq
\delta_{\max}+1 >\delta_2$.
\end{proof}
When the BCH codes contain their duals, then we can derive the following codes. Note
that these cannot be obtained as a consequence of Theorem~\ref{sh:nested}.

\begin{theorem}\label{sh:euclid}
Let $m=\ord_n(q)$  where $q$ is a power of a prime and $2\le \delta\le
\delta_{\max},$ with
$$\delta_{\max}=\frac{n}{q^m-1}(q^{\lceil m/2\rceil}-1-(q-2)[m  \textup{ odd}]),$$
then there exists a quantum code with parameters
$$[[n,n-2m\lceil(\delta-1)(1-1/q)\rceil,\ge \delta]]_q$$ pure to $\delta_{\max}+1$
\end{theorem}
\begin{proof}
Theorems~\ref {th:sufficientEdual} and \ref{th:bchnpdimension} imply that there
exists a classical BCH code with parameters $[n,n-m\lceil(\delta-1)(1-1/q)\rceil,\ge
\delta]_q$ which contains its dual code. 
By Corollary~\ref{th:css2}
an $[n,k,d]_q$
code that contains its dual code implies the existence of the quantum code with
parameters $[[n,2k-n,\ge d]]_q$. The purity follows from Lemma~\ref{th:Edualdist} by
which the dual distance $\geq \delta_{\max}+1 > \delta$.
\end{proof}

Before we can construct quantum codes via the Hermitian construction, we will need
the following lemma.
\begin{lemma}\label{th:Hdualdist}
Suppose that $C$ is a primitive, narrow-sense BCH code of length $n=q^{2m}-1$ over
$\F_{q^2}$ with designed distance 
$2\le \delta \le
\delta_{\max}=\lfloor n(q^{m}-1)/(q^{2m}-1)\rfloor$,
then the dual distance $d^\perp \geq
\delta_{\max} + 1$.
\end{lemma}
\begin{proof}
The proof is analogous to the one of Lemma~\ref{th:Edualdist}; just keep in mind
that the defining set $Z_\delta$ is invariant under multiplication by $q^2$ modulo
$n$.
\end{proof}

\begin{theorem}\label{sh:hermite}
Let $m=\ord_n(q^2) \geq 2$ where $q$ is a power of a prime and  $2\le \delta \le
\delta_{\max}=\lfloor n(q^{m}-1)/(q^{2m}-1)\rfloor$, then there exists a quantum
code with parameters
$$ [[n, n-2m\lceil(\delta-1)(1-1/q^2)\rceil ,\ge \delta]]_q$$
that is pure up to $\delta_{\max}+1$.
\end{theorem}
\begin{proof}
It follows from Theorems~\ref{th:bchnpdimension} and~\ref {th:sufficientHdual} that
there exists a primitive, narrow-sense $[n,n-1-m\lceil
(\delta-1)(1-1/q^2)\rceil,\ge\delta]_{q^2}$ BCH code that contains its Hermitian
dual code.  
By Corollary~\ref{co:classical} 
a classical $[n,k,d]_{q^2}$ code that
contains its Hermitian dual code implies the existence of an $[[n,2k-n,\ge d]]_q$
quantum code. By Lemma~\ref{th:Hdualdist} the quantum code is pure to
$\delta_{\max}+1$.
\end{proof}
In the above theorem, quantum codes can also be constructed when the design distance
exceeds the given value of $\delta_{\max}$, however we do not have exact knowledge
of the dimension in all those cases, hence we have not included them to keep the
theorem precise.

These are not the only possible families of quantum codes that can be derived from
BCH codes. As pointed out in \cite{grassl99b}, we can expand BCH codes over
$\F_{q^l}$ to get codes over $\F_q$. Once again the dimension and duality results of
BCH codes makes it very easy to specify such codes. We will just give one example in
the Euclidean case. Similar results can be derived for the Hermitian case.

\begin{theorem}\label{sh:euclidexpansion}
Let $m=\ord_n(q^l)$  where $q$ is a power of a prime and $2\le \delta\le
\delta_{\max},$ with
$$\delta_{\max}=\frac{n}{q^{lm}-1}(q^{l\lceil m/2\rceil}-1-(q^l-2)[m
\textup{ odd}]),$$ then there exists a quantum code with parameters
$$[[ln,ln-2lm\lceil(\delta-1)(1-1/q^l)\rceil,\ge \delta]]_q$$
that is pure up to $\delta$.
\end{theorem}
\begin{proof}
By Theorem~\ref{sh:euclid} there exists a quantum BCH code with parameters
$[[n,n-2m\lceil(\delta-1)(1-1/q^l)\rceil,\ge \delta]]_{q^l}$. An $[[n,k,d]]_{q^l}$
quantum code implies the existence of the quantum code with parameters $[[ln,lk,\ge
d]]_q$ by 
Lemma~\ref{th:codeexpansion}
and the code follows.
\end{proof}

\section{Conclusions}
In this chapter  we have identified the classes of BCH codes that contain their
Euclidean (Hermitian) duals by a careful analysis of the cyclotomic cosets. In the
process we have been able to shed more light on the structure of dual containing BCH
codes. We were able to derive a formula for the dimension of narrow-sense BCH codes
when the designed distance is small. These results allowed us to identify easily
which classical BCH codes can be used for construct quantum codes. Further, the
parameters of these quantum codes are easily specified in terms of the design
distance.

\nix{
\IEEEtriggeratref{15}

}

%% file: bib.tex
\def\cprime{$'$}

%% file: vita.tex
\vita{Pradeep Kiran Sarvepalli received his Ph.D. in Computer Science from Texas A\&M
University (2008). His dissertation was in quantum error correction under the supervision
of Dr. Andreas Klappenecker. His primary research interests are quantum computing,
coding theory, and circuit design. He received his M.S. in Electrical Engineering
from Texas A\&M University (2003) and his B.Tech in Electrical Engineering from 
Indian Institute of Technology, Madras (1997). Before he enrolled at Texas A\&M
University to pursue his graduate studies, he was an IC Design Engineer at 
Texas Instruments India, Bangalore, where among other things he designed 
analog and mixed signal circuits.  His permanent address is 
23/1235, Sodhan Nagar, Nellore, Andhra Pradesh, India 524 003.
He can be reached at {sarvepalli.pradeep@gmail.com}.

}